\def\NoPrintedContents{1}
\documentclass{article}
\usepackage{graphicx}
\usepackage{placeins}
\usepackage{subfiles}
\usepackage[T1]{fontenc}
\usepackage[margin=1in]{geometry}
\usepackage{amsmath}
\usepackage{amssymb}
\usepackage{amsthm}
\usepackage{needspace}
\usepackage{refcount}
\usepackage{array}
\usepackage{booktabs}
\usepackage{tabularx}
\usepackage{aliascnt} %
\usepackage{microtype}
\usepackage[sorting=none]{biblatex}
\usepackage{xcolor}
\usepackage{graphicx}
\usepackage{tikz}
\usetikzlibrary{arrows.meta,positioning}
\usepackage{xr-hyper}
\usepackage{hyperref}
\definecolor{linkblue}{RGB}{0,62,140}
\definecolor{citegreen}{RGB}{0,105,60}
\definecolor{urlblue}{RGB}{0,80,130}
\hypersetup{
  colorlinks=true,
  linkcolor=linkblue,  %
  citecolor=citegreen, %
  urlcolor=urlblue     %
}
\usepackage[capitalise]{cleveref}
\crefname{equation}{Eq.}{Eqs.}
\Crefname{equation}{Eq.}{Eqs.}

\newcommand{\eps}{\varepsilon}
\newcommand{\R}{\mathbb{R}}

\newcommand{\C}{\mathbb{C}}
\newcommand{\BI}{\mathbb{I}}

\newcommand{\BB}{\mathcal{B}}
\newcommand{\HH}{\mathcal{H}}
\newcommand{\FF}{\mathcal{F}}
\newcommand{\GG}{\mathcal{G}}
\newcommand{\LL}{\mathcal{L}}
\newcommand{\KK}{\mathcal{K}}
\newcommand{\MM}{\mathcal{M}}
\newcommand{\NN}{\mathcal{N}}
\newcommand{\SSS}{\mathcal{D}_1}

\DeclareMathOperator{\Tr}{Tr}
\DeclareMathOperator{\EP}{EP}
\DeclareMathOperator{\Rea}{Re}
\DeclareMathOperator{\spec}{spec}
\DeclareMathOperator{\id}{id}

\DeclareMathOperator{\polylog}{polylog}

\newcommand{\faml}[1]{(#1_N)_N}
\newcommand{\ketbra}[2]{\lvert #1\rangle\!\langle #2\rvert}
\newcommand{\Diss}[1]{\mathcal D[#1]}

\newcommand{\osc}[1]{|\mkern-1mu|\mkern-1mu|#1|\mkern-1mu|\mkern-1mu|}
\newcommand{\oscm}[1]{|\mkern-1mu|\mkern-1mu|#1|\mkern-1mu|\mkern-1mu|_{\operatorname{m}}}

\newaliascnt{proposition}{theorem}
\newtheorem{proposition}[proposition]{Proposition}
\aliascntresetthe{proposition}
\newaliascnt{corollary}{theorem}
\newtheorem{corollary}[corollary]{Corollary}
\aliascntresetthe{corollary}
\newaliascnt{lemma}{theorem}
\newtheorem{lemma}[lemma]{Lemma}
\aliascntresetthe{lemma}
\newaliascnt{conjecture}{theorem}
\newtheorem{conjecture}[conjecture]{Conjecture}
\aliascntresetthe{conjecture}
\theoremstyle{plain}
\newaliascnt{definition}{theorem}
\newtheorem{definition}[definition]{Definition}
\aliascntresetthe{definition}
\theoremstyle{definition}
\newaliascnt{example}{theorem}
\newtheorem{example}[example]{Example}
\aliascntresetthe{example}
\theoremstyle{remark}
\newaliascnt{remark}{theorem}

\aliascntresetthe{remark}

\newcounter{proof}

\crefname{proof}{Proof}{Proofs}
\Crefname{proof}{Proof}{Proofs}
\newcommand{\prooflabel}[2]{%
  \renewcommand{\theproof}{\getrefnumber{#1}}%
  \refstepcounter{proof}%
  \label{#2}%
}

\title{Is rapid mixing stable?}

\author{Jordi A. Montañà-López\textsuperscript{1,2,\textdagger}, Barbara Roos\textsuperscript{3},\\
Sebastian Stengele\textsuperscript{2,4}, Ángela Capel\textsuperscript{5,6,*}, Rahul Trivedi\textsuperscript{1,2,*}\\[0.8em]
\begin{minipage}{0.97\textwidth}
\centering\small
\textsuperscript{1}Max Planck Institute of Quantum Optics, Hans-Kopfermann-Straße 1, 85748 Garching, Germany\endgraf
\textsuperscript{2}Munich Center for Quantum Science and Technology, Schellingstraße 4, 80799 München, Germany\endgraf
\textsuperscript{3}Gran Sasso Science Institute, Viale Francesco Crispi 7, 67100 L'Aquila\endgraf
\textsuperscript{4}Departments of Mathematics and Physics, TU München, 85747 Garching, Germany\endgraf
\textsuperscript{5}Department of Applied Mathematics and Theoretical Physics, University of Cambridge, Cambridge CB3~0WA, United Kingdom\endgraf
\textsuperscript{6}Fachbereich Mathematik, Universität Tübingen, 72076 Tübingen, Germany
\end{minipage}}
\AddToHook{cmd/maketitle/after}{%
  \begingroup
  \footnotetext[2]{Corresponding author: \href{mailto:jordi.montana-lopez@mpq.mpg.de}{jordi.montana-lopez@mpq.mpg.de}.}%
  \footnotetext[1]{Equal contribution.}%
  \endgroup
}

\graphicspath{{.}}
\newcommand{\resulttoc}[2]{%
  \addcontentsline{toc}{subsection}{%
    \texorpdfstring{\protect\cref*{#1}}{Result \getrefnumber{#1}}: #2}%
}
\newcommand{\subsectionresulttoc}[2]{%
  \addcontentsline{toc}{subsubsection}{%
    \texorpdfstring{\protect\cref*{#1}}{Result \getrefnumber{#1}}: #2}%
}

\providecommand{\papercheckpoint}[9]{}
\makeatletter
\newcommand{\paperpart}[2]{%
  \immediate\write\@auxout{\string\papercheckpoint{#1}%
    {\number\value{section}}{\number\value{subsection}}{\number\value{equation}}%
    {\number\value{figure}}{\number\value{table}}%
    {\number\value{theorem}}{\number\value{proof}}{\number\value{footnote}}}%
}
\makeatother

\newif\ifpaperhascites
\DeclareBibliographyCategory{partcited}
\ifSubfilesClassLoaded{%
  \AtEveryCitekey{%
    \global\paperhascitestrue
    \addtocategory{partcited}{\thefield{entrykey}}%
  }%
}{}
\newcommand{\partbibliography}{%
  \ifSubfilesClassLoaded{%
    \ifpaperhascites
      \clearpage
      \printbibliography[heading=bibintoc,title=References,category=partcited]%
    \fi
  }{}%
}

\begin{document}
\ifdefined\NoPrintedContents
\pagenumbering{arabic}
\else
\pagenumbering{roman}
\fi

\maketitle

\begin{abstract}
Several physically interesting Lindbladians, such as those describing generic open quantum system models as well as engineered state preparation processes, are often rapidly mixing, that is, they relax to steady states in time scaling polylogarithmically with the system size. Owing to the short mixing times, it is physically expected that local perturbations do not substantially increase the mixing time. In this work we question this expectation and ask whether rapid mixing survives. We provide counterexamples showing that rapid mixing is not stable in general, even when the perturbation does not change the steady state. We also show that arbitrarily weak local Hamiltonian perturbations can destroy rapid mixing. Complementing these results, we develop an array of sufficient conditions which imply stable rapid mixing, including several classes of quasifree fermionic systems, commuting Lindbladians, and local spin systems. We also identify conditions under which modified logarithmic Sobolev inequalities (MLSIs), a stronger form of convergence, lead to stable rapid mixing. As a further implication of our results, we show that rapid mixing does not imply a positive MLSI constant.
\end{abstract}

\ifdefined\NoPrintedContents
\else
\clearpage
\tableofcontents
\clearpage
\pagenumbering{arabic}
\fi

\paperpart{sections-introduction}{0}
\section{Introduction}\label{sec:introduction}
The evolution of a quantum system is often influenced by its interactions with the environment. When the system is weakly coupled to a large environment, one assumes the environment remains at equilibrium upon interacting with the system. Under the additional bath correlation-time and coarse-graining assumptions of the Markov approximation, the evolution of the system alone follows Markovian dynamics --- the subsequent time evolution only depends on the current state of the system, not on its history --- which are described by the Lindblad master equation \cite{GoriniKossakowskiSudarshan1976,Lindblad1976}.

Markovian quantum dynamics model physical processes such as spontaneous and collective decay in quantum optics \cite{Dicke1954}, as well as thermalization of local quantum systems on a lattice coupled to a reservoir via the Davies generator \cite{Davies1974}. In the context of quantum technologies, they are also used as a common model for noise arising from decoherence, depolarization and many-body noise processes \cite{GardinerZoller2004}. For relaxing dynamics, the system approaches a steady state at long times. While this can be detrimental if depolarizing noise eventually sends the system to the maximally mixed state, if one can tune the couplings with the environment it is possible to change the steady state of the dynamics. Engineered dissipation studies how to control the system-environment couplings such that the steady state is a desired state we want to prepare, turning dissipation into a tool for state preparation \cite{VerstraeteWolfCirac2009,kraus2008preparation}.

A basic question in the settings above is how quickly the dynamics forgets its initial state and approaches its steady-state subspace. This time scale is the mixing time. We call a family rapidly mixing when this time grows at most polylogarithmically in the system size. We separately say that it is \emph{rapidly mixing with constant decay rate} when its trace-norm distance to the steady-state subspace is bounded by a polynomial prefactor times an exponential with a size-independent rate. For local Lindbladians with unique stationary states, this stronger property imposed uniformly on local restrictions yields stability of local observables under small local perturbations \cite{CLMP2015}. Lindbladians satisfying such bounds have been designed for thermal states \cite{RouzeFrancaAlhambra2024b, SmidMeisterBertaBondesan2025rm,KochanowskiAlhambraCapelRouze2024}, free-fermionic ground states \cite{ZhanDingHuhnGrayPreskillChanLin2025}, injective tensor network states \cite{baruah2026dissipative} and matrix product density operators \cite{liu2026parent}.

The spectral gap controls the slowest asymptotic decay rate, but by itself need not give a sharp estimate of the mixing time \cite{SzehrReebWolf2015}. Functional inequalities, including logarithmic Sobolev inequalities and modified logarithmic Sobolev inequalities (MLSIs), give quantitative convergence bounds for classical Markov chains and Glauber dynamics \cite{DiaconisSaloffCoste1996,BobkovTetali2006,StroockZegarlinski1992a,StroockZegarlinski1992b,GuionnetZegarlinski2003,MartinelliOlivieri1994,CaputoParisi2021}. Quantum logarithmic Sobolev inequalities similarly control entropy decay and mixing, including for semigroups with nonunique stationary states \cite{KT2013,CM2017,BR2022,GR22}. Bounds relating global relative entropy to conditional relative entropies on smaller regions provide a route to such estimates \cite{CapelLuciaPerezGarcia2018,BardetCapelRouze2022}. Applications include heat-bath dynamics and Davies generators for commuting quantum spin systems \cite{BardetCapelLuciaPerezGarciaRouze2021,CapelRouzeFranca2021,BCG2024,KochanowskiAlhambraCapelRouze2024}. Oscillator-norm methods originate in classical lattice spin dynamics \cite{StroockZegarlinski1992a,StroockZegarlinski1992b,GuionnetZegarlinski2003}; they have since been used to show ergodicity of infinite-size limits \cite{majewski1995quantum,BertiniDeSolePostaPresilla2025}, as well as to prove rapid mixing of quantum Gibbs samplers in several settings \cite{KB2016,RouzeFrancaAlhambra2024b,SmidMeisterBertaBondesan2025rm,ZhanDingHuhnGrayPreskillChanLin2025}.

Beyond equilibration, local dissipative evolution over polylogarithmic times has recently been used to define equivalence of mixed states. The standard notion of a zero-temperature phase is formulated in terms of ground states of local Hamiltonians connected by a gapped path, and does not extend directly to mixed states arising at nonzero temperature, under noise, or as steady states of driven systems, where neither a parent Hamiltonian nor an energy gap need exist. Local dissipative evolution provides an alternative: generating correlations across the system from initially short-range-correlated states with bounded local interactions requires a time at least proportional to the system diameter \cite{BravyiHastingsVerstraete2006,Poulin2010}, so that two mixed states may be regarded as equivalent when local Lindbladian evolutions transform each into the other in polylogarithmic time \cite{CP2019}. These transformations concern the specified initial states and need not mix rapidly from every initial state. Local evolution on these timescales thus plays a role analogous to that of a non-closing gap for ground states. Related criteria define stable steady-state phases through the response to dynamical perturbations \cite{RGK2023}, or through a Markov length that stays finite along the evolution, allowing that local evolution to be inverted quasi-locally and hence does not leave the phase \cite{SangHsieh2025}. Such definitions also have algorithmic consequences, enabling learning of local expectation values across a phase \cite{ORFW2023,ORFW2023a}.

To use efficient dissipative preparation and dynamical descriptions of mixed-state phases in practice, we need their favorable time scaling to survive local implementation errors. We therefore consider perturbations with uniformly bounded local strength. This robustness question has a well-established Hamiltonian analogue: a gap that remains open along a path of local Hamiltonians permits quasi-adiabatic continuation of the ground-state sector \cite{HastingsWen2005}, while gap-stability theorems establish robustness under local perturbations for classes of Hamiltonians satisfying suitable structural assumptions \cite{BravyiHastingsMichalakis2010}. In open quantum systems, rapid mixing plays a related role: under uniform local mixing hypotheses, it guarantees stability of local observables and correlation functions \cite{CubittLuciaMichalakisPerezGarcia2015,LuciaCubittMichalakisPerezGarcia2015}. Quantitative perturbation bounds for generators and their fixed points are also available \cite{SzehrWolf2013}, but they do not establish stability of rapid mixing under perturbations with uniformly bounded local strength. This question has been studied in the classical setting of Glauber dynamics, a widely studied model of thermalization for spin systems. For the ferromagnetic Ising model in the uniqueness regime away from the critical point, exponential relaxation of local observables and uniqueness of the stationary measure persist under small perturbations of the dynamics, including perturbations that break detailed balance \cite{CrawfordDeRoeck2018}. The stability in the Hamiltonian case, together with the role of rapid mixing in quantum stability results, suggests that short mixing times may themselves be robust. This would allow the favorable time scaling of dissipative preparation to survive local implementation errors, motivating the question:

\begin{center}
\emph{When is rapid mixing stable under a local perturbation of the Lindbladian?}
\end{center}

We show that rapid mixing and locality of the generator alone do not imply such stability, and that it continues to fail under several additional assumptions that one might expect to be sufficient. We establish stability criteria for commuting models, local spin systems, quasifree fermionic dynamics, and generators controlled by functional inequalities.

\noindent \Cref{tab:instability-summary,tab:positive-summary} summarize the negative and positive results. The counterexamples are presented in \cref{sec:counterexamples} and the stability criteria in \cref{sec:dissipative-covers,sec:stability_quasifree,sec:commuting,sec:fast}; the full proofs are collected in the appendices. \Cref{sec:open} discusses open questions.

\partbibliography

\paperpart{sections-preliminaries}{0}
\section{Preliminaries}\label{sec:setting}
Let $\HH$ be a finite-dimensional Hilbert space. We denote by $\BB(\HH)$ the algebra of linear operators on $\HH$ and by $\mathcal D_1(\HH)$ the set of states on $\HH$, i.e., the positive semidefinite operators $\rho\in\BB(\HH)$ satisfying $\Tr(\rho)=1$. We write $A\succeq B$ when $A-B$ is positive semidefinite and $A\succ B$ when it is positive definite. In particular, $\rho\succ0$ means that the density matrix $\rho$ has full rank. For spin systems, $\Lambda$ is a finite set of sites equipped with a metric $d_\Lambda$. Each site $x\in\Lambda$ carries a finite-dimensional Hilbert space $\HH_x\simeq\C^d$, and $\HH_\Lambda=\bigotimes_{x\in\Lambda}\HH_x$. An operator is supported on $X\subseteq\Lambda$ if it belongs to $\BB(\HH_X)\otimes\BI_{X^c}$, where $\HH_X=\bigotimes_{x\in X}\HH_x$. We write $\operatorname{diam}(X):=\max_{x,y\in X}d_\Lambda(x,y)$. A collection of terms is $k$-local if every support $X$ has $|X|\leq k$, and is geometrically local with range $r$ if every support has $\operatorname{diam}(X)\leq r$.

A \emph{Lindbladian} is a generator of the form
\begin{align}
    \label{eq:lindblad} \LL(\rho)=-i[H,\rho]+\sum_j\Diss{L_j}(\rho),\qquad
    \Diss{A}(\rho):=A\rho A^\dagger-\tfrac12\{A^\dagger A,\rho\}.
\end{align}
Here $H=H^\dagger$, and the jumps can be chosen traceless after a compensating change of $H$ \cite{GoriniKossakowskiSudarshan1976,Lindblad1976}. We write $^*$ for the Hilbert--Schmidt adjoint of a superoperator and retain $^\dagger$ for the conjugate transpose of an operator. Thus the Heisenberg-picture generator is
\begin{align}
    \label{eq:lindblad_heis} \LL^*(O)=i[H,O]+\sum_j\mathcal D^*[L_j](O),\qquad
    \mathcal D^*[A](O):=A^\dagger O A-\tfrac12\{A^\dagger A,O\}.
\end{align}
For fermionic systems, $\Lambda$ labels fermionic modes with creation and annihilation operators $a_x^\dagger,a_x$ satisfying the canonical anticommutation relations $\{a_x,a_y\}=0$ and $\{a_x,a_y^\dagger\}=\delta_{xy}\BI$. The local algebra on $X\subseteq\Lambda$ is the CAR algebra $\mathcal A_X$ generated by $a_x,a_x^\dagger$ for $x\in X$, and geometric locality refers to the same metric $d_\Lambda$ on the mode labels. Let $\Pi:=(-1)^{\sum_x a_x^\dagger a_x}$ be the parity operator. A physical density matrix is even, meaning that $[\rho,\Pi]=0$, or equivalently that $\rho$ is block diagonal in the two parity sectors. A Lindbladian preserves this block-parity structure whenever it is parity covariant,
\begin{align}
    \LL(\Pi\rho\Pi)=\Pi\LL(\rho)\Pi.
\end{align}
A sufficient condition is that $[H,\Pi]=0$ and every jump operator has definite parity, $\Pi L_j\Pi=\pm L_j$. Odd jump operators can transfer population between the two parity sectors, so parity covariance does not in general imply conservation of the parity expectation $\Tr(\Pi\rho)$.

Every finite-dimensional Lindbladian has at least one stationary density matrix, and its spectrum lies in the closed left half-plane \cite{Lindblad1976}. If the stationary state is unique, we denote it by $\sigma_\LL$: it is globally attractive, so $0$ is the only eigenvalue on the imaginary axis \cite{schirmer2010stabilizing}. More generally, we say that $\LL$ has a \emph{well-defined infinite-time limit} when it has no nonzero purely imaginary eigenvalues, $\operatorname{spec}(\LL)\cap i\mathbb R=\{0\}$. This is equivalent to the existence of the pointwise limit
\begin{align}
    P_\LL&:=\lim_{t\to\infty}e^{t\LL},
\end{align}
which is the projection onto the steady-state subspace $\ker\LL=\{X\in\BB(\HH):\LL(X)=0\}$. We prove this equivalence in \cref{lem:infinite-time-limit} for completeness. Since it is a limit of quantum channels, $P_\LL$ maps density matrices to density matrices. When the stationary state is not unique, convergence to $P_\LL(\rho)$ concerns the entire steady-state subspace and does not assert mixing between distinct stationary states. Its adjoint $P_\LL^*=\lim_{t\to\infty}e^{t\LL^*}$ projects observables onto the subspace $\ker\LL^*$ of conserved quantities. When the stationary state is unique, these projections take the simpler form
\begin{align}
    P_\LL(X)&=\Tr(X)\sigma_\LL, & P_\LL^*(O)&=\Tr(\sigma_\LL O)\BI.
\end{align}

Here, for an operator $X$, we use $\lVert X\rVert_1$ for the trace norm, $\lVert X\rVert$ for the operator norm, and $\lVert X\rVert_2$ for the Hilbert--Schmidt norm. For a linear map $\Phi$, we write $\lVert\Phi\rVert_{p\to p}:=\sup_{X\ne0}\lVert\Phi(X)\rVert_p/\lVert X\rVert_p$ and use
\begin{align*}
    \lVert\Phi\rVert_\diamond&:=\sup_{k\geq1}\lVert\Phi\otimes\operatorname{id}_k\rVert_{1\to1},
    &\lVert\Phi\rVert_{\operatorname{cb}}&:=\sup_{k\geq1}\lVert\Phi\otimes\operatorname{id}_k\rVert_{\infty\to\infty}
\end{align*}
for the diamond norm and the completely bounded norm, respectively \cite{Watrous2018}.

\begin{definition}[Mixing time]
The \emph{mixing time} of a Lindbladian $\mathcal{L}$ with a well-defined infinite-time limit $P_\LL$ for $\varepsilon>0$ is
\begin{align*}
    \tau_{\operatorname{mix}}^{\LL}(\varepsilon):=\inf\Bigl\{t>0:\sup_{\rho\in\SSS(\HH)} \|e^{t\LL}(\rho)-P_\LL(\rho)\|_1<\varepsilon\Bigr\}.
\end{align*}
\end{definition}

Writing $d_\LL(t):=\sup_{\rho\in\SSS(\HH)}\lVert e^{t\LL}(\rho)-P_\LL(\rho)\rVert_1$, trace-norm contractivity and the semigroup property give, for $0<\varepsilon<1/2$,
\begin{align*}
    d_\LL(t+s)\leq d_\LL(t)d_\LL(s)\quad\Longrightarrow\quad \tau_{\operatorname{mix}}^\LL(\varepsilon)\leq\left\lceil\log_2(\varepsilon^{-1})\right\rceil\tau_{\operatorname{mix}}^\LL(1/2).
\end{align*}
It is thus customary to consider mixing to a fixed accuracy: we write $\tau_{\operatorname{mix}}^\LL:=\tau_{\operatorname{mix}}^\LL(1/2)$. Following the infinite-lattice convention of \cite[Definitions~3.3--3.4]{CubittLuciaMichalakisPerezGarcia2015}, a family of Lindbladians is specified by a fixed generator on the infinite lattice, whose restrictions to finite regions are equipped with prescribed boundary conditions. The bulk interaction terms are fixed independently of the region; boundary terms may depend on the region, subject to uniform locality and strength bounds. Open and periodic boundary conditions are included.

\begin{definition}[Rapid mixing]
For a family $\faml{\LL}$ with well-defined infinite-time limits, we say that the family is \emph{rapidly mixing} if there are constants $a,q>0$, independent of $N$, such that
\begin{align}
    \label{eq:rapid-mixing-time} \tau_{\operatorname{mix}}^{\LL_N}\leq a\left[\log(eN)\right]^q.
\end{align}
\end{definition}

If instead $\tau_{\operatorname{mix}}^{\LL_N}\leq aN^q$ for constants $a,q>0$ independent of $N$, we call the family \emph{fast mixing}.

\begin{definition}[Rapid mixing with constant decay rate]
For a family $\faml{\LL}$ with well-defined infinite-time limits, we say that the family is \emph{rapidly mixing with constant decay rate} if there are constants $\gamma>0$ and $c\geq0$, independent of $N$, and a prefactor $C_N\geq1$ satisfying $C_N=O(N^c)$ such that
\begin{align}
    \label{eq:sufficient} \sup_{\rho\in\SSS(\HH)}\|e^{t\LL_N}(\rho)-P_{\LL_N}(\rho)\|_1\leq C_N e^{-\gamma t},\qquad t\geq0.
\end{align}
\end{definition}
If instead \cref{eq:sufficient} holds with $C_N=O(e^{cN})$ for constants $c,\gamma>0$ independent of $N$, we call the family \emph{fast mixing with constant decay rate}. For brevity, when the family is clear we often suppress the subscript $N$ and say that $\LL$ is rapidly mixing. Rapid mixing with constant decay rate implies rapid mixing, since \cref{eq:sufficient} gives
\begin{align}
    \tau_{\operatorname{mix}}^{\LL_N}\leq\gamma^{-1}\log(2C_N)=O(\log N).
\end{align}
The converse need not hold: the decay rate may close as an inverse power of $\log N$ while the mixing time remains polylogarithmic. Given a Lindbladian $\LL$, its spectral gap
\begin{align}
    \operatorname{gap}(\LL):=\min\{-\operatorname{Re}z:z\in\operatorname{spec}(\LL),\ z\neq0\}
\end{align}
provides a general lower bound on the mixing time. Indeed, the spectral-radius bound for $e^{t\LL}-P_\LL$ gives \cite[Theorem~III.2]{SzehrReebWolf2015}
\begin{align}
    d_\LL(t)\geq e^{-t\operatorname{gap}(\LL)},\qquad
    \tau_{\operatorname{mix}}^\LL(\varepsilon)\geq\frac{\log(\varepsilon^{-1})}{\operatorname{gap}(\LL)},\qquad 0<\varepsilon<1,
    \label{eq:gap-mixing-lower-bound}
\end{align}
Consequently, a gap closing as an inverse power of $N$ rules out rapid mixing.

We will also consider rapid mixing uniformly for restrictions to lattice balls, following the local restrictions in \cite[Eq.~(11)]{CLMP2015}. Let $b_x(R)$ denote the ball of lattice sites of radius $R$ centered at $x$. For $A\subseteq\Lambda$, let $\mathcal{L}_A$ be obtained by retaining the Hamiltonian terms and jump operators supported entirely on $A$, acting on $\HH_A=(\C^d)^{\otimes |A|}$.

\begin{definition}[Uniform rapid mixing]
For a family $\faml{\LL}$ of Lindbladians, we say that the family is \emph{uniformly rapidly mixing} if there exist constants $a,q,R_0>0$ such that, for every lattice $\Lambda$ in the family, every $x\in\Lambda$, and every $R\geq R_0$, the restriction to $A=b_x(R)$ has a well-defined infinite-time limit and
\begin{align}
    \tau_{\operatorname{mix}}^{\LL_A}\leq a\left[\log(e|A|)\right]^q,
\end{align}
where the constants are independent of $\Lambda,x$, and $R$. We say that the family is \emph{uniformly rapidly mixing with constant decay rate} if there are constants $C,\gamma>0$ and $p\geq0$, uniform in the same parameters, such that
\begin{align}
    \label{eq:uniform-rapid-mixing} \sup_{\rho\in\SSS(\HH_A)}\|e^{t\LL_A}(\rho)-P_{\LL_A}(\rho)\|_1\le C|A|^p e^{-\gamma t},\qquad t\geq0,
\end{align}
\end{definition}

Whenever the family contains the full lattice as one of the permitted restrictions, uniform rapid mixing implies rapid mixing of the global family.

Functional inequalities provide rapid-mixing bounds for several classes of Lindbladians; here we focus on modified logarithmic Sobolev inequalities \cite{KT2013,Bardet2017,BR2022}. We use natural logarithms and write $D(\rho\|\sigma):=\Tr[\rho(\log\rho-\log\sigma)]$ for quantum relative entropy when $\operatorname{supp}\rho\subseteq\operatorname{supp}\sigma$, with $0\log0=0$; otherwise $D(\rho\|\sigma)=+\infty$.

\begin{definition}[MLSI]
Let $\sigma\succ0$ be a stationary state of $\LL$. We say that $\LL$ satisfies a \emph{modified logarithmic Sobolev inequality} (MLSI) with constant $\alpha>0$ if, for every density matrix $\rho\in\SSS(\HH)$,
\begin{align}
    \label{eq:MLSI} 2\alpha\,D(\rho\|\sigma)\le \operatorname{EP}_\LL(\rho).
\end{align}
Here $\operatorname{EP}_\LL$, called the entropy production, is given by
\begin{align}
    \operatorname{EP}_\LL(\rho):=-\Tr[\LL(\rho)(\log\rho-\log\sigma)].
\end{align}
The optimal constant is denoted by $\alpha(\LL)$ and is called the modified log-Sobolev constant.
\end{definition}

MLSI gives useful bounds on the mixing time when the smallest singular value of the full-rank stationary state is controlled. Suppose that the smallest singular value satisfies $(\sigma_{\LL_N})_{\min}\geq e^{-cN^q}$ for size-independent constants $c,q>0$. Then the following estimates hold for every density matrix $\rho$ \cite{KT2013}.
\begin{align*}
    \tfrac12\|e^{t\LL_N}(\rho)-\sigma_{\LL_N}\|_1^2
    &\overset{(1)}{\le} D(e^{t\LL_N}(\rho)\|\sigma_{\LL_N})
    \overset{(2)}{\le} e^{-2\alpha(\LL_N)t}D(\rho\|\sigma_{\LL_N}),\\
    &\overset{(3)}{\le} e^{-2\alpha(\LL_N)t}\log(\sigma_{\LL_N})_{\min}^{-1}
    \overset{(4)}{\le} cN^q e^{-2\alpha(\LL_N)t},
\end{align*}
where $(1)$ uses Pinsker's inequality, $(2)$ uses the MLSI together with Gr\"onwall's inequality, $(3)$ uses $D(\rho\|\sigma)=-S(\rho)-\Tr(\rho\log\sigma)\leq\log\lambda_{\min}(\sigma)^{-1}$, since $S(\rho)\geq0$, and $(4)$ uses the assumed lower bound on $(\sigma_{\LL_N})_{\min}$. Thus
\begin{align}
    \lVert e^{t\LL_N}(\rho)-\sigma_{\LL_N}\rVert_1&\leq\sqrt{2c}\,N^{q/2}e^{-\alpha(\LL_N)t}
    \qquad\text{and}\qquad
    \tau_{\operatorname{mix}}^{\LL_N}(\varepsilon)\leq\frac{1}{\alpha(\LL_N)}\left(\frac{q}{2}\log N+\frac12\log(2c)+\log(\varepsilon^{-1})\right).
\end{align}
Consequently, $\alpha(\LL_N)=\Omega((\log N)^{-r})$ for any fixed $r\geq0$ implies rapid mixing; the case $r=0$ gives rapid mixing with constant decay rate. Establishing useful MLSI lower bounds is, however, hard in general. The inequality is known for specific structured classes: tensor products of depolarizing noise \cite{KastoryanoTemme2013}, one-dimensional heat-bath dynamics under mixing and entropy-factorization assumptions \cite{BardetCapelLuciaPerezGarciaRouze2021}, Davies dynamics for finite-range commuting spin chains \cite{KochanowskiAlhambraCapelRouze2024}, and commuting nearest-neighbour models at high temperature \cite{CapelRouzeFranca2021}.

\partbibliography

\paperpart{sections-instability_results}{0}
\section{Instability results}\label{sec:instability_results}\label{sec:counterexamples}

In this paper we study the following question: if $\LL$ is rapidly mixing and $\KK$ is another Lindbladian, when is $\LL+\KK$ rapidly mixing? A straightforward perturbation analysis yields the following stability result.
\begin{lemma}[Mixing under a small contractive perturbation]
\label{lem:perturbation_fixed_point}\resulttoc{lem:perturbation_fixed_point}{Mixing under a small contractive perturbation}
Let $\mathcal L_0$ be a Lindbladian with a unique stationary state $\sigma_0$ and a well-defined infinite-time limit. Set $T:=\tau_{\operatorname{mix}}^{\mathcal L_0}(1/2)$, and let $V$ be a Hermiticity-preserving, trace-annihilating linear map.
Define the restricted induced trace norm on traceless Hermitian operators by
\begin{align*}
    \|V\|_{1\to1,0}:=\sup_{\substack{X=X^\dagger,\ \Tr X=0\\X\neq0}}
    \frac{\|V(X)\|_1}{\|X\|_1}.
\end{align*}
For $\varepsilon\geq0$, set $\mathcal L_\varepsilon:=\mathcal L_0+\varepsilon V$ and suppose that $e^{t\mathcal L_\varepsilon}$ is a trace-norm contraction on Hermitian operators for every $t\geq0$. If
\begin{align}
    \varepsilon \|V\|_{1\to1,0}\leq\frac{1}{4T}, \label{eq:generic-perturbation-threshold}
\end{align}
then $\mathcal L_\varepsilon$ has a unique stationary state $\sigma_\varepsilon$ and, for all states $\rho$ and $t\geq0$,
\begin{align}
    \lVert e^{t\mathcal L_\varepsilon}(\rho)-\sigma_\varepsilon\rVert_1
    \leq\frac43e^{-\gamma t}\lVert\rho-\sigma_\varepsilon\rVert_1,
    \qquad \gamma:=\frac{\log(4/3)}{T}. \label{eq:generic-perturbation-bound}
\end{align}
Mixing times $T=O(\operatorname{polylog}N)$ imply rapid mixing; $T=O(1)$ gives constant decay rate.
\end{lemma}
The proof is given in \cref{proof:perturbation_fixed_point}. For a local perturbation with uniformly bounded interaction strength, $\|V\|_{1\to1,0}=O(N)$. If the unperturbed family is rapidly mixing, the lemma therefore guarantees stability for perturbations whose local strength decreases as an inverse polynomial in $N$. We now study perturbations of constant local strength and give counterexamples under progressively stronger hypotheses.

Let us first consider the case that $\LL$ and $\KK$ are rapidly mixing. If each can rapidly erase information on its own, it is natural to expect their combined dissipative evolution to do the same: adding relaxation channels should only improve convergence. It makes a positive answer to the following question seem natural:
\begin{center}
\emph{$(\star)$ If $\LL$ and $\KK$ are rapidly mixing Lindbladians then is $\LL+\KK$ rapidly mixing?}
\end{center}

We show that $(\star)$ is false even after imposing several hypotheses that would normally be expected to prevent competing long-time behavior. We begin with the most direct setting: both summands are primitive, have the maximally mixed state as their unique stationary state, and mix rapidly with constant decay rate.

\begin{lemma}[A common full-rank fixed point does not make rapid mixing additive]
\label{lem:XZeps}\resulttoc{lem:XZeps}{A common full-rank fixed point does not make rapid mixing additive}
There exist two local, unital, primitive Lindbladian families $(\LL_N)_N$ and $(\KK_N)_N$ that are both rapidly mixing with constant decay rate and have the same unique fixed point $\sigma_N=2^{-N}\BI$, whereas $\LL_N+\KK_N$ is primitive with fixed point $\sigma_N$ but satisfies
\begin{align}
    \tau_{\operatorname{mix}}^{\LL_N+\KK_N}=\Omega(N^2).
\end{align}
In particular, the family $(\LL_N+\KK_N)_N$ is not rapidly mixing.
\end{lemma}

\begin{example}[Spin-chain construction]\label{ex:XZeps}
On a spin chain of $N$ sites, write $\LL=\LL_N$ and $\KK=\KK_N$ for the local generators of the families $(\LL_N)_N$ and $(\KK_N)_N$, with jump operator $Z_j$ at every site $j$ and Hamiltonians $X_j$ and $(\eps_j-1)X_j$, respectively:
\begin{align}
    \LL(\rho)=\sum_j\bigl(-i[X_j,\rho]+\Diss{Z_j}(\rho)\bigr),\qquad \KK(\rho)=\sum_j\bigl(-i[(\eps_j-1)X_j,\rho]+\Diss{Z_j}(\rho)\bigr),
\end{align}
where $\eps_j=1/(j+1)$. Both generators have the unique full-rank fixed point $\sigma=(\BI/2)^{\otimes N}$, and both families are rapidly mixing. For every finite $N$, the sum $\LL+\KK$ is primitive and gapped, but its gap is $\Theta(N^{-2})$ and its mixing time is $\Omega(N^2)$, so $(\LL_N+\KK_N)_N$ is not rapidly mixing.
\end{example}
The proof is given in \cref{proof:XZeps}. The construction proving \cref{lem:XZeps} identifies a first mechanism: addition cancels the order-one Hamiltonian rotations that make the individual generators relax rapidly, exposing a mode whose rate vanishes with the system size. One might regard the example as nongeneric because the residual coefficients $\eps_j$ decrease along the chain. Translation invariance and fixed local coefficients would appear to rule out precisely this fine tuning, while uniform rapid mixing of all restrictions suggests that no slow mode is hidden in a finite region. The next result shows that these assumptions are still insufficient.

\begin{lemma}[Translation invariance and fixed local coefficients do not make rapid mixing additive]
\label{lem:fermionic-addition}\resulttoc{lem:fermionic-addition}{Translation invariance and fixed local coefficients do not make rapid mixing additive}
There exist two finite-range, translation-invariant fermionic Lindbladian families $(\LL_N^+)_N$ and $(\LL_N^-)_N$, indexed by odd $N\geq5$ and with local coefficients independent of $N$, such that both families are uniformly rapidly mixing with constant decay rate and have the same unique full-rank fixed point $\sigma_N=4^{-N}\BI$. The sum family $(\mathcal S_N)_N$, with $\mathcal S_N=\LL_N^++\LL_N^-$, is primitive with the same fixed point at each size, but for every $0<\varepsilon<1$,
\begin{align}
    \tau_{\operatorname{mix}}^{\mathcal S_N}(\varepsilon)
    \geq \frac{N^2}{2\pi^2}\log(\varepsilon^{-1}).
\end{align}
The lower bound persists when the supremum in the mixing time is restricted to parity-even states.
\end{lemma}

\begin{example}[Fixed-coefficient fermionic construction]\label{ex:fermionic-addition}
For odd $N\geq5$, let $\Lambda_N=\mathbb Z/N\mathbb Z$, with two fermionic annihilation operators $a_x,b_x$ in each cell $x$. With indices understood modulo $N$, define
\begin{align}
    H_{\lambda,N}&=\sum_{x\in\Lambda_N}\left[\lambda a_x^\dagger b_x+\frac12(a_{x+1}^\dagger+a_{x-1}^\dagger)b_x+\mathrm{h.c.}\right],\label{eq:fermionic-addition-H-main}\\
    \LL_{\lambda,N}(\rho)&=-i[H_{\lambda,N},\rho]+\sum_{x\in\Lambda_N}\left(\Diss{a_x}(\rho)+\Diss{a_x^\dagger}(\rho)\right).\label{eq:fermionic-addition-L-main}
\end{align}
Define the families $(\LL_N^\pm)_N$ and $(\mathcal S_N)_N$ by $\LL_N^\pm=\LL_{\pm2,N}$ and $\mathcal S_N=\LL_N^++\LL_N^-$, and set $\sigma_N=4^{-N}\BI$. Then
\begin{align}
    \lVert e^{t\LL_N^\pm}-P_{\LL_N^\pm}\rVert_\diamond\leq12Ne^{-t/2},\qquad P_{\LL_N^\pm}(X)=\Tr(X)\sigma_N.
\end{align}
The same estimate, with $N$ replaced by $|A|$, holds for the restriction to every nonempty set of cells $A$, retaining the on-site jumps and the Hamiltonian terms supported in $A$. The sum $\mathcal S_N=2\LL_{0,N}$ is primitive with fixed point $\sigma_N$, but
\begin{align}
    \tau_{\operatorname{mix}}^{\mathcal S_N}(\varepsilon)&\geq\frac{\log(\varepsilon^{-1})}{2r_N}
    \geq\frac{N^2}{2\pi^2}\log(\varepsilon^{-1}),
    &r_N&=1-\sqrt{1-4\sin^2\left(\frac{\pi}{2N}\right)}.
\end{align}
The state establishing this bound has even parity.
\end{example}
The proof is given in \cref{proof:fermionic-addition}. At the single-particle level, the fixed on-site couplings of each summand keep its modes uniformly damped. Their addition cancels those couplings and leaves a long-wavelength mode with decay rate of order $N^{-2}$. Thus even translation invariance, fixed local coefficients, a common faithful stationary state, and uniform rapid mixing of both summands do not make rapid mixing additive. The preceding examples nevertheless exploit coherent Hamiltonian cancellation. It is therefore natural to ask whether restricting to purely dissipative generators restores the intuition that adding relaxation channels can only improve convergence.

\phantomsection\label{par:ising-kernel-mismatch}
The low-temperature two-dimensional Ising model provides an instructive illustration of why this expectation can fail \cite{GKZ2024}. Let $\Lambda_L=\{1,\ldots,L\}^2$ have free boundary conditions, with even $L\geq4$ and $N=L^2$, and let
\begin{align}
    H_L=-\sum_{\langle x,y\rangle}Z_xZ_y, \qquad \sigma_{\beta,L}=\frac{e^{-\beta H_L}}{\Tr e^{-\beta H_L}}.
\end{align}
Splitting the conditional heat-bath updates between the two checkerboard sublattices gives geometrically 5-local, zero-Hamiltonian families $(\LL_{1,N})_N$ and $(\LL_{2,N})_N$. At each size, write $\LL_i=\LL_{i,N}$; each generator preserves the full-rank state $\sigma_{\beta,L}$ and satisfies
\begin{align}
    \label{eq:ising-layer-rm-main} \sup_\rho\|e^{t\LL_i}(\rho)-P_{\LL_i}(\rho)\|_1 \leq 2N e^{-t},\qquad i\in\{1,2\}.
\end{align}
Their fixed-point spaces are degenerate and strictly larger than $\operatorname{span}\{\sigma_{\beta,L}\}$: for each summand, the relevant sectors correspond to the positive- and negative-magnetization phases on the checkerboard sublattice left frozen by that summand. Their sum is the full heat-bath dynamics and is primitive. The low-temperature bound of \cite[Corollary~3.5 and Lemma~4.14]{GKZ2024} gives, for every sufficiently large fixed $\beta$ and all sufficiently large $L$,
\begin{align}
    \label{eq:ising-slow-main} \tau_{\operatorname{mix}}^{\LL_1+\LL_2}(1/4) \geq e^{c_\beta L}=e^{c_\beta\sqrt N}
\end{align}
for some $c_\beta>0$. We provide a detailed analysis in \cref{app:ex:ising-kernel-mismatch}.

At low temperature, the ferromagnetic Ising model has two macroscopically magnetized phases. Each summand updates only one checkerboard sublattice while treating the other as frozen boundary data, so it rapidly equilibrates conditionally on that data. Their sum updates both sublattices, but moving between the positive- and negative-magnetization phases requires creating an interface of length of order $L$. The associated free-energy barrier makes such transitions exponentially rare in $L$, producing the exponentially long mixing time. A tempting diagnosis is therefore that the slowdown is caused by the degenerate steady-state subspaces of the summands. Requiring a unique stationary state for each generator removes that source of conserved information; removing Hamiltonian terms also rules out coherent cancellation. One could therefore expect additivity to hold for dissipation-only generators with unique steady states. The following example shows otherwise.

\begin{figure}[!t]
\centering
\begin{minipage}[t]{0.49\linewidth}
a)\par\smallskip
\centering
\includegraphics[width=\linewidth]{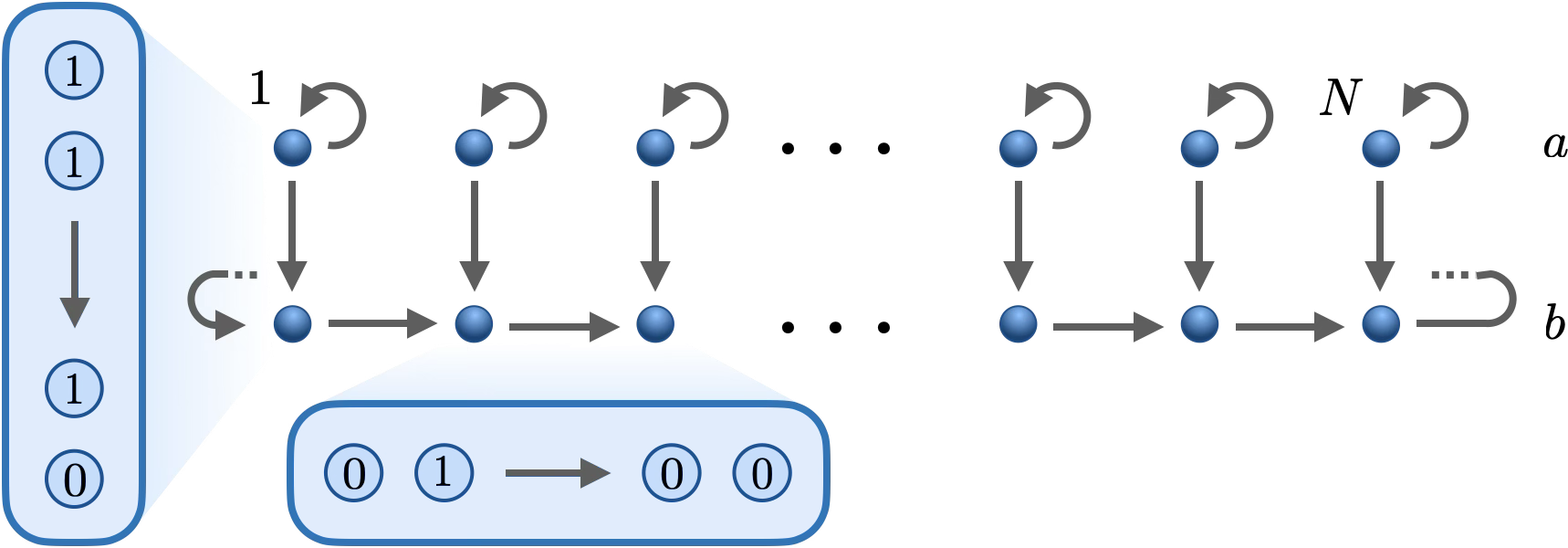}
\end{minipage}\hfill
\begin{minipage}[t]{0.49\linewidth}
b)\par\smallskip
\centering
\vspace*{0.048\linewidth}
\includegraphics[width=\linewidth]{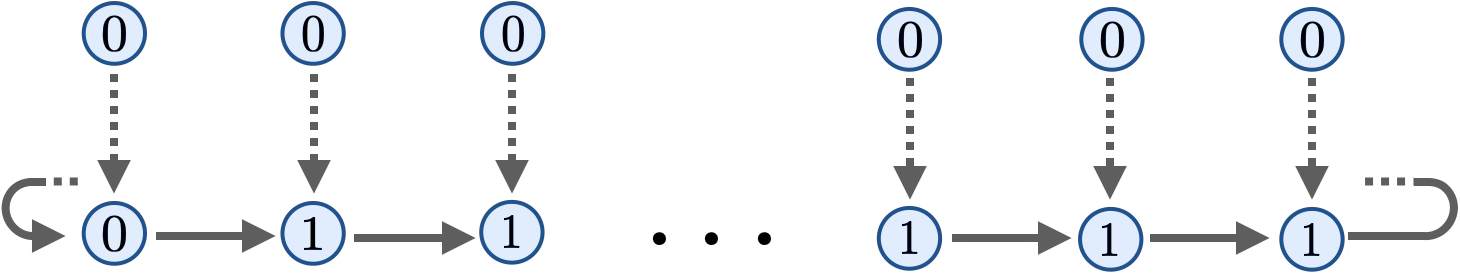}
\end{minipage}
\caption{Mechanism of the dissipation-only construction on the $a$ and $b$ qubits. a) In both $\LL$ and $\KK$, the $L_{a,b,i}$ jumps between the $a$ and $b$ qubits introduce conditional resets on the $b$ qubits, while the $L_{b,i}$ jumps propagate zeros around the $b$ qubits. The generators differ in their single-qubit action on the $a$ qubits and rapidly mix toward the product states $\sigma_\LL=\eta_a^{\otimes N}\otimes\ketbra{0^N}{0^N}_b$ and $\sigma_\KK=(Z\eta Z)_a^{\otimes N}\otimes\ketbra{0^N}{0^N}_b$, respectively. The insets show that $L_{a,b,i}$ maps $|11\rangle_{a_i b_i}$ to $|10\rangle_{a_i b_i}$ and $L_{b,i}$ maps $|01\rangle_{b_i b_{i+1}}$ to $|00\rangle_{b_i b_{i+1}}$. b) The sum $\LL+\KK$ rapidly drives the $a$ qubits toward $|0^N\rangle_a$. The conditional resets $L_{a,b,i}$ then switch off, and the $a$ and $b$ qubits evolve independently. Starting from the displayed state $|01\cdots1\rangle_b$, relaxation requires applying the bottom jumps successively around the ring, so the mixing time is $\Omega(N)$.}
\label{fig:rm_sum_dissip_only}
\end{figure}

\begin{lemma}[Dissipation-only rapid mixing is not additive]
\label{lem:rm_sum_dissip_only}\resulttoc{lem:rm_sum_dissip_only}{Dissipation-only rapid mixing is not additive}
There exist two geometrically local, translation-invariant, zero-Hamiltonian Lindbladian families $(\LL_N)_N$ and $(\KK_N)_N$ on periodic lattices, each with a unique stationary state and rapid mixing with constant decay rate, such that the stationary states are different and the infinite-time limit of $\LL_N+\KK_N$ exists, but
\begin{align}
    \tau_{\operatorname{mix}}^{\LL_N+\KK_N}\geq\frac{N-1}{8}.
\end{align}
The sum family $(\LL_N+\KK_N)_N$ has a degenerate steady-state subspace at each size.
\end{lemma}

\begin{example}[Dissipation-only construction]\label{ex:rows-dissip}\label{ex:rm_sum_dissip_only}
Consider a periodic chain of $N\geq2$ sites, with two qubits $a_i,b_i$ at each site $i$. Define the families $(\LL_N)_N$ and $(\KK_N)_N$ by writing $\LL=\LL_N$ and $\KK=\KK_N$ at each size:
\begin{align}
    \LL&=\sum_{i=1}^N\mathcal{D}[L_{a,i}]+\sum_{i=1}^N\mathcal{D}[L_{b,i}]+\sum_{i=1}^N\mathcal{D}[L_{a,b,i}],\notag\\
    \KK&=\sum_{i=1}^N\mathcal{D}[K_{a,i}]+\sum_{i=1}^N\mathcal{D}[L_{b,i}]+\sum_{i=1}^N\mathcal{D}[L_{a,b,i}],
\end{align}
with jump operators
\begin{align}
    L_{a,i}&=\begin{pmatrix}1&2\\
    0&-1\end{pmatrix}_{a_i},\qquad 1\le i\le N,\notag\\
    K_{a,i}&=\begin{pmatrix}1&-2\\
    0&-1\end{pmatrix}_{a_i},\qquad 1\le i\le N,\notag\\
    L_{b,i}&=\ketbra{0}{0}_{b_i}\otimes\ketbra{0}{1}_{b_{i+1}},\qquad 1\le i\le N,\notag\\
    L_{a,b,i}&=\sqrt{\gamma}\,\ketbra{1}{1}_{a_i}\otimes\ketbra{0}{1}_{b_i},\qquad 1\le i\le N.
\end{align}
Here $\gamma>0$ is independent of $N$, and the site labels are understood modulo $N$, so $b_{N+1}=b_1$. The bottom jump changes $b_{i+1}$ from $1$ to $0$ when $b_i$ is in state $0$, while the cross jump changes $b_i$ from $1$ to $0$ when $a_i$ is in state $1$. Thus $\LL$ and $\KK$ are translation invariant and have the same bottom and cross jumps, while the off-diagonal entries of their $a_i$ jumps have opposite signs.

Let $\eta=\begin{pmatrix}\frac56&-\frac13\\-\frac13&\frac16\end{pmatrix}$ be the unique fixed point of the single-qubit dissipator with jump operator $L_a=\begin{pmatrix}1&2\\0&-1\end{pmatrix}$. The Lindbladian $\LL$ has the unique fixed point $\sigma_\LL=\eta_a^{\otimes N}\otimes\ketbra{0^N}{0^N}_b$, while $\KK$ has the unique fixed point $\sigma_\KK=(Z\eta Z)_a^{\otimes N}\otimes\ketbra{0^N}{0^N}_b$. Both families $(\LL_N)_N$ and $(\KK_N)_N$ are dissipation-only and rapidly mixing with constant decay rate, but $(\LL_N+\KK_N)_N$ is not rapidly mixing: for every fixed $0<\delta<1/2$, $\tau_{\operatorname{mix}}^{\LL+\KK}(\delta)=\Omega(N)$.

Their sum has kernel
\begin{align}
    \ker(\LL+\KK)&=\ketbra{0^N}{0^N}_a\otimes\BB(\mathcal S_b),\qquad \mathcal S_b:=\operatorname{span}\{|0^N\rangle_b,|1^N\rangle_b\},
\end{align}
where $\BB(\mathcal S_b)$ is embedded in the $b$-qubit operator algebra and includes coherences between the strings. Its infinite-time limit exists, and the mixing time to this kernel is $\Omega(N)$.
\end{example}
The proof is given in \cref{proof:rm_sum_dissip_only}, and the mechanism is illustrated in \cref{fig:rm_sum_dissip_only}. In the construction shown in \cref{fig:rm_sum_dissip_only}a), the individual generators provide conditional resets throughout the $b$ row. Adding $\KK$ makes the dynamics approach $|0\rangle$ on all $a_i$ qubits. The projectors in the cross jumps then turn off the resets from $a_i$ to $b_i$, and the $b$ qubits evolve independently of the $a$ qubits. A $b$-qubit $1$ can be removed only after the site immediately to its left has become $0$, so a string such as $|01\cdots1\rangle_b$ relaxes through the directed front shown in \cref{fig:rm_sum_dissip_only}b) that needs linear time to travel around the ring. This gives a physical picture of the slowdown: the individual generators provide local resets throughout the bottom row, whereas their sum confines the remaining relaxation to information propagation from a single zero seed. One might still attribute this behavior to adding a second dissipative dynamics with a different stationary state. The next example shows that the same transport bottleneck can be triggered by a coherent Hamiltonian perturbation of a single rapidly mixing dissipative generator.

\Needspace{10\baselineskip}
\begin{lemma}[A local Hamiltonian can destroy rapid mixing]
\label{lem:rm_sum_hamiltonian}\resulttoc{lem:rm_sum_hamiltonian}{A local Hamiltonian can destroy rapid mixing}
There exist a geometrically local, translation-invariant, zero-Hamiltonian Lindbladian family $(\LL_N)_N$ on periodic lattices with a unique stationary state and rapid mixing with constant decay rate, and a translation-invariant on-site Hamiltonian family $(\KK_N)_N$ of uniformly bounded local strength, such that the infinite-time limit of $\LL_N+\KK_N$ exists but
\begin{align}
    \tau_{\operatorname{mix}}^{\LL_N+\KK_N}\geq\frac{N-1}{4}.
\end{align}
\end{lemma}

\begin{example}[Hamiltonian construction]\label{ex:rows-ham}\resulttoc{ex:rows-ham}{Rapid mixing lost due to the Hamiltonian}\label{ex:rm_sum_hamiltonian}
Let $\LL$ be the Lindbladian from \cref{ex:rm_sum_dissip_only}. It has the unique fixed point $\sigma_\LL=\eta_a^{\otimes N}\otimes\ketbra{0^N}{0^N}_b$ and satisfies $\tau_{\operatorname{mix}}^\LL(\varepsilon)=\mathcal{O}(\log(N/\varepsilon))$. Let $H_i=Y_{a_i}$ and define $(\KK_N)_N$ by $\KK_N=\KK=-i[\sum_{i=1}^NH_i,\cdot]$. Then $(\LL_N+\KK_N)_N$ is not rapidly mixing: for every fixed $0<\delta<1/2$, $\tau_{\operatorname{mix}}^{\LL+\KK}(\delta)=\Omega(N)$.
\end{example}
The proof is given in \cref{proof:rows-ham}. All of the adverse perturbations above have order-one local strength: they establish nonadditivity, but do not yet exclude a perturbative stability threshold. In particular, the Hamiltonian in \cref{lem:rm_sum_hamiltonian} is strong enough to make the steady state rank-deficient. This leaves the perturbative question: can an arbitrarily weak, but constant, local perturbation turn a $\operatorname{polylog}(N)$ mixing time into a $\operatorname{poly}(N)$ mixing time?

\Needspace{12\baselineskip}
\begin{lemma}[An arbitrarily weak local Hamiltonian can destroy rapid mixing]
\label{lem:sparse-boundary-potential}\resulttoc{lem:sparse-boundary-potential}{An arbitrarily weak local Hamiltonian can destroy rapid mixing}
There exist a primitive local Lindbladian family, evaluated on chains of $m_N=\Theta(\log N)$ qubits and denoted by $(\LL_N)_N$, and, for each $N$, an on-site Hamiltonian generator $\KK_N$ of uniformly bounded local strength, such that, for every fixed $\eps>0$, $\LL_N+\eps\KK_N$ is primitive with the same maximally mixed stationary state as $\LL_N$, and
\begin{align}
    \tau_{\operatorname{mix}}^{\LL_N}=O((\log N)^{10}),\qquad
    \tau_{\operatorname{mix}}^{\LL_N+\eps\KK_N}\geq N^{c_\eps+o(1)}
\end{align}
for some $c_\eps>0$. Here $m_N$ is the physical chain length and $N$ is an auxiliary parameter.
\end{lemma}

\begin{example}[Boundary dissipation of a localized mode]\label{ex:sparse-boundary-potential}
Fix $a,J>0$ and set $\gamma=J/2$. For all sufficiently large $N$, let
\begin{align}
    m=m_N:=3\left\lfloor\frac{a\log N}{3}\right\rfloor,
    \qquad \ell:=\frac{m}{3}.
\end{align}
Consider one open chain of $m$ qubits, with sites labelled $1,\ldots,m$, and define
\begin{align}
    H&=-J\sum_{j=1}^{m-1}\bigl(\sigma_j^+\sigma_{j+1}^-+\sigma_j^-\sigma_{j+1}^+\bigr),\label{eq:sparse-H0}\\
    \LL_N(\rho)&=-i[H,\rho]+\gamma\sum_{s\in\{1,m\}}\bigl(\Diss{\sigma_s^-}(\rho)+\Diss{\sigma_s^+}(\rho)\bigr).\label{eq:sparse-L0}
\end{align}
The dissipation-free interval between the two endpoints has length $m-2=\Theta(\log N)$. Define the on-site potential profile and its Hamiltonian generator by
\begin{align}
    V_N&=\sum_{j=1}^{\ell}n_j+\sum_{j=2\ell+1}^{3\ell}n_j,
    \qquad n_j=\sigma_j^+\sigma_j^-,\label{eq:sparse-potential}\\
    \KK_N&=-i[V_N,\,\cdot\,].\label{eq:sparse-Leps}
\end{align}
For a fixed $\eps>0$, the perturbation $\eps\KK_N$ adds an on-site potential of strength $\eps$ on the left and right thirds of the chain. Here $N$ is an auxiliary parameter used to set the physical chain length $m_N=\Theta(\log N)$. Both $\LL_N$ and $\LL_N+\eps\KK_N$ are primitive and have the unique fixed point $\tau_N=2^{-m}\BI$. For every $0<\delta<1$ and, for the lower bound, $\ell\geq\max\{2,4J/\eps\}$, we have
\begin{align}
    \tau_{\operatorname{mix}}^{\LL_N}(\delta)
    &\leq\frac{9(m+1)^9}{2J}\left[\left(\frac m2+1\right)\log2+\log\frac1\delta\right],\label{eq:sparse-main-upper}\\
    \tau_{\operatorname{mix}}^{\LL_N+\eps\KK_N}(1/2)
    &\geq\frac1{16J}\left(1+\frac\eps{2J}\right)^{\ell-1}
    =N^{a\kappa_\eps/3+o(1)},
    \qquad \kappa_\eps=\log\left(1+\frac\eps{2J}\right)>0.\label{eq:sparse-main-lower}
\end{align}
Hence $\LL_N$ has polylogarithmic mixing time in $N$, whereas for every fixed $\eps>0$, the mixing time of $\LL_N+\eps\KK_N$ is bounded below by a positive power of $N$.
\end{example}
The proof is given in \cref{proof:sparse-boundary-potential}. The example exhibits a localization mechanism under boundary-driven dissipation, illustrated in \cref{fig:sparse_boundary_potential}. In the unperturbed chain, shown in \cref{fig:sparse_boundary_potential}a), relaxation proceeds by transporting excitations to the dissipative endpoints and the mixing time is polylogarithmic in $N$. The weak potential localizes a mode in the bulk, as shown in \cref{fig:sparse_boundary_potential}b), making its overlap with the boundary dissipation exponentially small in the physical chain length $m_N=\Theta(\log N)$; this becomes a polynomial mixing-time lower bound in $N$. The perturbation preserves the maximally mixed steady state. Here the time $T=\tau_{\operatorname{mix}}^{\LL_N}(1/2)$ in \cref{lem:perturbation_fixed_point} grows with $N$, so its sufficient bound $\eps \|V\|_{1\to1,0}\leq1/(4T)$ shrinks with size and does not protect against a fixed local perturbation strength.

\begin{figure}[!t]
\centering
\begin{minipage}[t]{0.39\linewidth}
a)\par\smallskip
\centering
\parbox[b][26mm][b]{\linewidth}{\raisebox{-2.3mm}{\includegraphics[width=\linewidth]{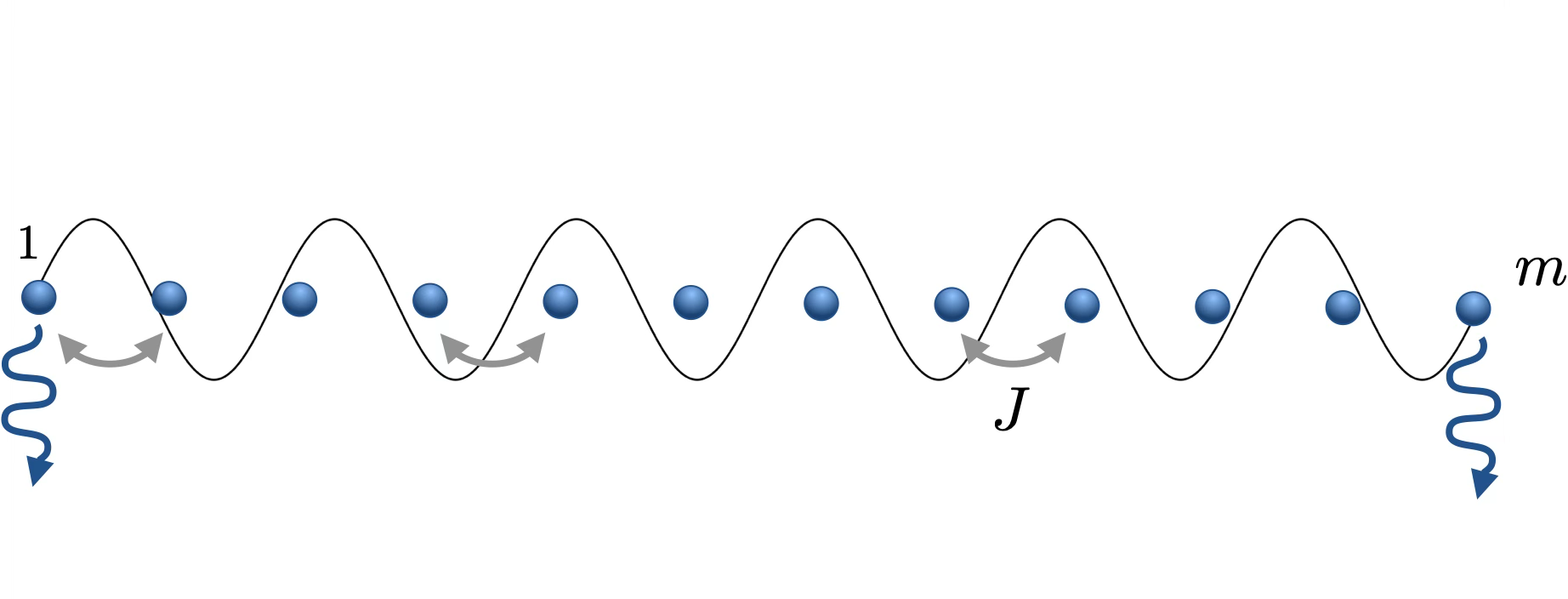}}}
\end{minipage}\hfill
\begin{minipage}[t]{0.59\linewidth}
b)\par\smallskip
\centering
\parbox[b][26mm][b]{\linewidth}{\includegraphics[width=\linewidth]{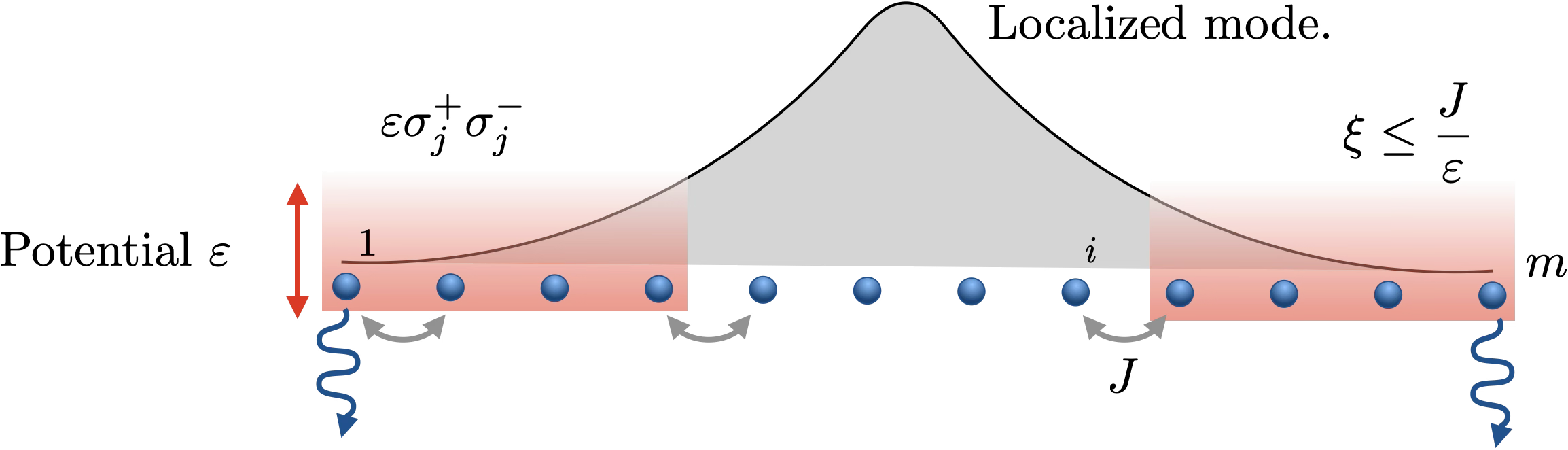}}
\end{minipage}
\caption{Boundary-driven relaxation with and without a local potential. Both panels show a chain of qubits of length $m$ with nearest neighbor hopping of strength $J$ and dissipation via equal-rate raising and lowering jump operators on the first and last qubit. The unique fixed point is the fully mixed state. a) Without a potential, coherences are transported towards the dissipative boundaries, and the mixing time is at most polynomial in $m$. b) Adding a single-site potential $\varepsilon$ on the left and right thirds of the chain produces a mode localized in the middle third. The curve is a schematic amplitude profile: the actual witness has $v_1=v_m=0$, and its neighboring amplitudes satisfy $|v_2|,|v_{m-1}|\leq\exp[-(m/3-1)\log(1+\varepsilon/(2J))]$. This makes the dynamics take $\exp(\Omega(m))$ time to reach the fully mixed state. Setting $m=\Theta(\log N)$, with $N$ an auxiliary size parameter, gives a single chain with a polylogarithmic mixing-time upper bound in $N$ without a potential and a polynomial mixing-time lower bound in $N$ for every fixed $\varepsilon>0$.}
\label{fig:sparse_boundary_potential}
\end{figure}

The preceding examples motivate stronger forms of rapid mixing that one might expect to provide stability. Functional inequalities are a natural candidate. Does rapid mixing with constant decay rate force a positive MLSI constant? Conversely, does imposing an MLSI provide the stability under adding another dissipative generator that rapid mixing alone lacks? The Hamiltonian-cancellation construction answers the first question negatively.

\begin{lemma}[Rapid mixing does not imply MLSI]\label{lemma:rm_MLSI}\resulttoc{lemma:rm_MLSI}{Rapid mixing without MLSI}
The Lindbladian family $(\LL_N)_N$ defined by
\begin{align}
    \LL_N(\rho)=\sum_{j=1}^N\bigl(-i[X_j,\rho]+\Diss{Z_j}(\rho)\bigr)
\end{align}
has a unique full-rank fixed point $\BI/2^N$ and is rapidly mixing with constant decay rate, but it does not have a positive MLSI constant.
\end{lemma}
The proof is given in \cref{proof:rm_MLSI}. The intuition is that MLSI requires exponential contraction of relative entropy with prefactor one:
\begin{align}
    D(e^{t\LL}(\rho)\|\sigma)&\leq e^{-2\alpha(\LL)t}D(\rho\|\sigma),\qquad t\geq0.
\end{align}
For a full-rank state $\rho\neq\sigma$, differentiating at $t=0$ therefore requires a strictly negative initial entropy derivative, i.e., strictly positive entropy production. But if $\rho$ is diagonal in the computational basis, the dissipators vanish and $\LL(\rho)=-i[H,\rho]$. This initial Hamiltonian motion has zero entropy production relative to $\sigma=\BI/2^N$, so no positive MLSI constant is possible. Thus rapid mixing with constant decay rate need not imply MLSI. Conversely, the Ising example analyzed in \cref{app:ex:ising-kernel-mismatch} shows that the individual local Lindbladians each satisfy a degenerate fixed-point MLSI with constant at least $1/2$, yet their sum mixes in time $e^{\Omega(\sqrt N)}$. MLSI by itself therefore does not provide stability under local perturbations.

Taken together, these examples show that rapid mixing is considerably more fragile than one might expect a priori. Across several settings in which the available structure might suggest stability, adding a local generator can instead produce a parametrically longer mixing time. The question for the remainder of the paper is which additional conditions prevent these mechanisms and yield stability of rapid mixing under local perturbations with uniformly bounded strength and sufficiently small fixed $\eps>0$. \Cref{tab:instability-summary,tab:positive-summary} collect the negative and positive results, respectively.

\clearpage
\begin{table}[p]
\centering
\caption{Instability results. Here ``RM'' denotes rapid mixing. In the localized-mode example, $N$ is auxiliary and the physical chain has $m=\Theta(\log N)$ sites; the slowdown is polynomial to exponential in $m$.}
\label{tab:instability-summary}
\small
\renewcommand{\arraystretch}{1.15}
\begin{tabularx}{\textwidth}{@{}>{\raggedright\arraybackslash}p{0.17\textwidth} >{\raggedright\arraybackslash}X >{\raggedright\arraybackslash}p{0.19\textwidth} >{\raggedright\arraybackslash}p{0.17\textwidth} >{\raggedright\arraybackslash}p{0.14\textwidth}@{}}
\toprule
Setting & Unperturbed generator & Perturbation & Failure & Reference \\
\midrule
Hamiltonian cancellation & A primitive RM generator with constant decay rate and a full-rank fixed point & A second generator with the same properties and the same full-rank fixed point & Mixing time becomes $\Omega(N^2)$ & \cref{lem:XZeps} \\

Fermionic chain & A translation-invariant, CAR-local, uniformly RM generator with constant decay rate, uniformly bounded coefficients, and a full-rank fixed point & A second generator with the same properties and the same full-rank fixed point & Mixing time becomes $\Omega(N^2)$ & \cref{lem:fermionic-addition} \\

Dissipation-only sum & A translation-invariant, dissipation-only RM generator with constant decay rate on a periodic chain & A second translation-invariant, dissipation-only RM generator with constant decay rate & Kernel becomes degenerate; mixing time becomes $\Omega(N)$ & \cref{lem:rm_sum_dissip_only} \\

Hamiltonian perturbation & A geometrically local, translation-invariant, dissipation-only RM generator with constant decay rate on a periodic chain & Translation-invariant on-site Hamiltonian & Mixing time becomes $\Omega(N)$ & \cref{lem:rm_sum_hamiltonian} \\

Localized mode & Primitive generator with polylogarithmic mixing time in auxiliary $N$ & Arbitrarily weak local Hamiltonian & Mixing time becomes $N^{\Omega(1)}$ & \cref{lem:sparse-boundary-potential} \\

Low-temperature Ising model & An RM dissipative Davies generator with MLSI and a full-rank fixed state & A second RM dissipative Davies generator with MLSI and the same full-rank fixed state & Mixing time becomes $e^{\Omega(\sqrt N)}$ & \cref{app:ex:ising-kernel-mismatch}; \cite{GKZ2024} \\

RM versus MLSI & Primitive RM generator with constant decay rate & --- & No MLSI & \cref{lemma:rm_MLSI} \\
\bottomrule
\end{tabularx}
\end{table}
\clearpage

\partbibliography

\addtocontents{toc}{\protect\setcounter{tocdepth}{3}}
\hypersetup{bookmarksdepth=3}
\clearpage
\section{Stability results}\label{sec:stability-results}

\paperpart{sections-commuting_generators_mlsi}{0}
\subsection{Commuting generators}\label{sec:commuting}
The commuting setting is a natural starting point because the simplest local-to-global proof of rapid mixing already has this structure. Let
\begin{align}
    \LL=\sum_{x\in\Lambda}\LL_x. \label{eq:commuting-single-qubit-generator}
\end{align}
be a sum of single-qubit Lindbladians with well-defined infinite-time limits, where \(\LL_x\) acts only on qubit \(x\). Order the sites as \(\Lambda=\{x_1,\ldots,x_N\}\). Since the terms act on different qubits, their semigroups commute and factorize. Denote their asymptotic projections by \(P_{\LL_{x_j}}\). The local-to-global step is the telescoping sum
\begin{align}
    \bigotimes_{j=1}^N e^{t\LL_{x_j}}-\bigotimes_{j=1}^N P_{\LL_{x_j}}
    &=\sum_{j=1}^N\left(\bigotimes_{k<j}e^{t\LL_{x_k}}\right)\otimes
    \left(e^{t\LL_{x_j}}-P_{\LL_{x_j}}\right)\otimes
    \left(\bigotimes_{k>j}P_{\LL_{x_k}}\right), \label{eq:single-qubit-telescoping}
\end{align}
and hence, using the triangle inequality, multiplicativity of the diamond norm, and contractivity of the channels in the other tensor factors,
\begin{align*}
    \left\|\bigotimes_{j=1}^N e^{t\LL_{x_j}}-\bigotimes_{j=1}^N P_{\LL_{x_j}}\right\|_\diamond
    &\leq \sum_{j=1}^N\left\|e^{t\LL_{x_j}}-P_{\LL_{x_j}}\right\|_\diamond.
\end{align*}
Thus, a uniform single-qubit estimate \(\|e^{t\LL_x}-P_{\LL_x}\|_\diamond\leq Ce^{-\lambda t}\) yields a global bound \(C|\Lambda|e^{-\lambda t}\), and hence a logarithmic mixing time. This is a basic local-to-global argument: local relaxation estimates for the individual qubits can be assembled into global relaxation.

This result extends to commuting Lindbladians \cite{KastoryanoWolf2012}. If $\LL,\KK$ have well-defined infinite-time limits and $[\LL,\KK]=0$, then $P_{\LL+\KK}=P_\LL P_\KK$ and, for every state $\rho$,
\begin{align}
    \|e^{t(\LL+\KK)}(\rho)-P_{\LL+\KK}(\rho)\|_1 \leq \|e^{t\LL}(\rho)-P_{\LL}(\rho)\|_1 +\|e^{t\KK}(\rho)-P_{\KK}(\rho)\|_1,
\end{align}
and the mixing times satisfy:
\begin{align}
    \tau_{\operatorname{mix}}^{\LL+\KK}(\delta) \leq \max\{\tau_{\operatorname{mix}}^{\LL}(\delta/2),\tau_{\operatorname{mix}}^{\KK}(\delta/2)\} .
\end{align}
We include the proof in \cref{proof:commuting-sum} for completeness. The next two lemmas give related consequences when the steady states of $\LL$ are also steady states of $\KK$.

\begin{lemma}[Commuting addition to a relaxing generator]\label{lem:commuting-unique-stationary}\subsectionresulttoc{lem:commuting-unique-stationary}{Commuting addition to a relaxing generator}
Let $\LL,\KK$ be commuting Lindbladians and suppose that $\LL$ has a unique stationary state $\sigma$ and a well-defined infinite-time limit. Then $\KK(\sigma)=0$, $\LL+\KK$ has the unique stationary state $\sigma$, and, for every $\rho\in\SSS(\HH)$ and $t\ge0$,
\begin{align*}
    \|e^{t(\LL+\KK)}(\rho)-\sigma\|_1\le \|e^{t\LL}(\rho)-\sigma\|_1.
\end{align*}
In particular, every rapid-mixing bound for $\LL$ is inherited by $\LL+\KK$.
\end{lemma}
The proof of \cref{lem:commuting-unique-stationary} is in \cref{proof:commuting-unique-stationary}. Uniqueness makes the kernel inclusion automatic.

\begin{lemma}[Commuting sum with kernel inclusion]\label{lem:commuting-kernel}\subsectionresulttoc{lem:commuting-kernel}{Commuting sum with kernel inclusion}
Let $\LL,\KK$ be commuting Lindbladians, suppose that $P_\LL=\lim_{t\to\infty}e^{t\LL}$ exists, and assume $\ker(\LL)\subseteq\ker(\KK)$. Then $P_{\LL+\KK}$ exists and equals $P_\LL$, and for all $\rho\in\SSS(\HH)$ and all $t\ge0$,
\begin{align*}
    \|e^{t(\LL+\KK)}(\rho)-P_{\LL+\KK}(\rho)\|_1\le \|e^{t\LL}(\rho)-P_\LL(\rho)\|_1.
\end{align*}
\end{lemma}
The proof of \cref{lem:commuting-kernel} is in \cref{proof:commuting-kernel}.
\begin{proposition}[Stability under commuting additions]
\label{prop:commuting-stability}\subsectionresulttoc{prop:commuting-stability}{Stability under commuting additions}
Let $(\LL_N)_N$ and $(\KK_N)_N$ be Lindbladian families with $[\LL_N,\KK_N]=0$ for every $N$. If both families are rapidly mixing, then $(\LL_N+\KK_N)_N$ is rapidly mixing. If $(\LL_N)_N$ is rapidly mixing and, for every $N$, $\LL_N$ either has a unique stationary state or satisfies $\ker(\LL_N)\subseteq\ker(\KK_N)$, then its mixing bound is inherited by $(\LL_N+\KK_N)_N$.
\end{proposition}
The proposition follows from the commuting-sum estimate above and \cref{lem:commuting-unique-stationary,lem:commuting-kernel}. These statements can fail without commutativity, as the counterexamples of \cref{sec:instability_results} show.

\partbibliography

\paperpart{sections-strong_bulk_dissipation}{0}
\subsection{Strong bulk dissipation}\label{sec:dissipative-covers}
Now suppose that the single-qubit \(\LL\) in \cref{eq:commuting-single-qubit-generator} is perturbed by a local Lindbladian \(\KK\). One can ask whether, for sufficiently small \(\varepsilon\), rapid mixing of \(\LL+\varepsilon\KK\) follows from a local-to-global argument as above. We will show that this holds more generally for local Lindbladians \(\LL,\KK\) whenever \(\LL\) dissipates information sufficiently quickly in local regions throughout the system.

The intuitive reason for this is already suggested by the examples in \cref{sec:instability_results}: sufficiently strong and uniformly distributed local dissipation should prevent the mechanisms by which rapid mixing is lost upon adding a perturbation. In \cref{lem:sparse-boundary-potential}, for example, the dissipation acts only at the boundaries, which are separated by a distance growing with $N$, and a perturbation can create a mode localized far from them. This motivates requiring bounded local regions throughout the system to contain Lindbladian terms that erase their local degrees of freedom. If these terms are strong enough compared with the interactions coupling nearby regions, their local relaxation can again be assembled into global rapid mixing with constant decay rate, robustly under small local perturbations. We introduce bulk-dissipative covers to formalize this local-to-global mechanism. These covers generalize previous oscillator norm methods, which originate in classical lattice spin dynamics \cite{StroockZegarlinski1992a,StroockZegarlinski1992b,GuionnetZegarlinski2003} and have also been used to study ergodicity of infinite-size limits \cite{majewski1995quantum,BertiniDeSolePostaPresilla2025}. These methods have been applied to Gibbs samplers in several settings \cite{KB2016,RouzeFrancaAlhambra2024b,SmidMeisterBertaBondesan2025rm}, typically comparing a sum of single-qubit detailed-balance generators with a quasilocal detailed-balance Gibbs sampler.

We say that a Lindbladian $\LL=\sum_{A\in\mathcal S_{\LL}}\LL_A$ is $l$-local with interaction strength bounded by $s_{\LL}$ if every $\LL_A$ is a Lindbladian supported on $A$, acting on at most $|A|\leq l$ sites, and the interaction strength per site is defined by
\begin{align}
    s_{\LL} &:= \max_{x\in\Lambda}\sum_{\substack{A\in\mathcal S_{\LL}\\
    x\in A}}\lVert\LL_A\rVert_{1\to1}.
\end{align}

In this section, induced norms of local maps are computed after tensoring with the identity map on the rest of the lattice: thus $\|\Phi_A\|_{1\to1}$ denotes $\|\Phi_A\otimes\operatorname{id}_{A^c}\|_{1\to1}$. Choose terms $\MM_\alpha:=\LL_{A_\alpha}$, $\alpha=1,\ldots,M$, from the decomposition indexed by $\mathcal S_{\LL}$. For each term $\MM_\alpha$, we call the largest region on which every stationary observable is trivial its bulk, illustrated in \cref{fig:main-dissipative-cover}a):
\begin{align}
    A_\alpha^\circ&:=\left\{x\in A_\alpha:\ker\MM_\alpha^*\subseteq\BI_x\otimes\mathcal B(\mathcal H_{x^c})\right\}. \label{eq:main-bulk-dissipation}
\end{align}
For a particular term $\MM_\alpha$, the bulk may be empty. We make three assumptions. First, as illustrated in \cref{fig:main-dissipative-cover}b), the bulk regions cover the lattice:
\begin{align}
    \bigcup_{\alpha=1}^M A_\alpha^\circ&=\Lambda.
\end{align}
Second, every term $\MM_\alpha$ has a well-defined infinite-time limit and, for some $t_{\mathrm{loc}}>0$, satisfies the local contraction estimate:
\begin{align}
    \lVert e^{t_{\mathrm{loc}}\MM_\alpha}-P_{\MM_\alpha}\rVert_\diamond &\leq \frac12, \label{eq:main-local-bulk-decay}
\end{align}
uniformly in $\alpha$, where $P_{\MM_\alpha}:=\lim_{t\to\infty}e^{t\MM_\alpha}$. These first two assumptions define a bulk-dissipative cover with local mixing time $t_{\mathrm{loc}}$. To state the third assumption, fix this cover and the local Lindblad decomposition $\LL=\sum_{B\in\mathcal S_{\LL}}\LL_B$, and define
\begin{align}
    n_{\LL}(B)&:=\#\{\alpha\in\{1,\ldots,M\}:[\MM_\alpha,\LL_B]\neq0\}, \label{eq:main-noncommutation-count}\\
    \kappa_{\LL}&:=4\max_{1\leq\nu\leq M}\sum_{\substack{B\in\mathcal S_{\LL}\\ B\cap A_\nu^\circ\neq\varnothing}}n_{\LL}(B)\lVert\LL_B\rVert_{1\to1}. \label{eq:main-local-influence}
\end{align}
Thus $\kappa_{\LL}$ bounds the total rate at which the remaining terms of $\LL$ introduce fluctuations into a locally relaxing region, as illustrated in \cref{fig:main-dissipative-cover}c). It measures the influence of noncommuting terms; terms that commute with every cover generator $\MM_\alpha$ contribute zero, regardless of their strength. Third, we assume the influence condition
\begin{align}
    \kappa_{\LL}t_{\mathrm{loc}}&<\frac12, \label{eq:strong-bulk-unperturbed-condition}
\end{align}
Under these three assumptions, $\LL$ has a unique fixed point. For covers satisfying $t_{\mathrm{loc}}=O(1)$ and $M=\operatorname{poly}(|\Lambda|)$, a size-independent positive margin in \cref{eq:strong-bulk-unperturbed-condition} gives rapid mixing with constant decay rate. A more general statement is given in \cref{app:prop:stability-strong-bulk-dissipation}.
\begin{figure}[!t]
\centering
\begin{minipage}[t]{0.58\linewidth}
a)\par\smallskip
\includegraphics[width=\linewidth]{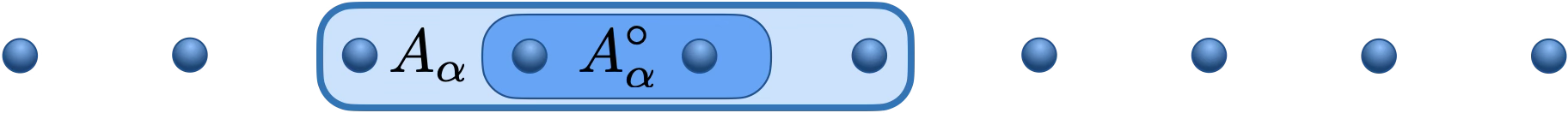}
\end{minipage}\par\medskip
\begin{minipage}[t]{0.58\linewidth}
b)\par\smallskip
\includegraphics[width=\linewidth]{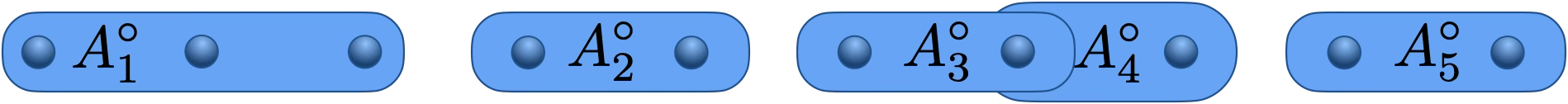}
\end{minipage}\par\medskip
\begin{minipage}[t]{0.58\linewidth}
c)\par\smallskip
\begingroup
\setlength{\unitlength}{\linewidth}
\begin{picture}(1,0.133684)
\put(0,0){\includegraphics[width=\linewidth]{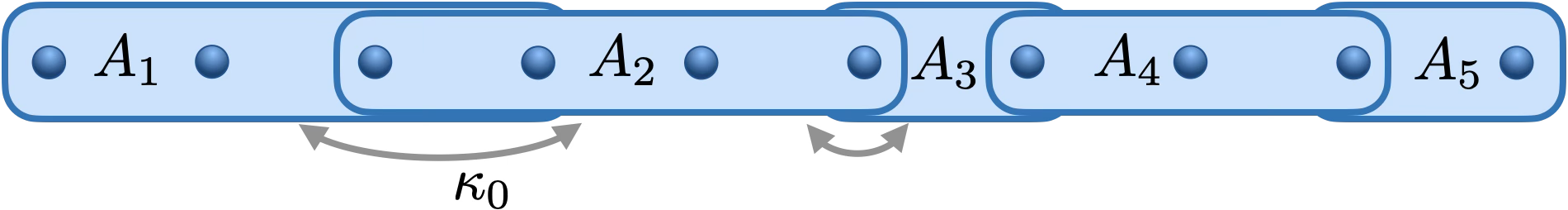}}
\put(0.284,0){\color{white}\rule{0.052\linewidth}{0.026\linewidth}}
\put(0.307,0.004){\makebox(0,0)[b]{\scriptsize $\kappa_{\LL}$}}
\end{picture}
\endgroup
\end{minipage}
\caption{Bulk-dissipative covers and the rapid-mixing conditions. a) The local generator $\MM_\alpha$ acts on $A_\alpha$, and every stationary observable acts trivially on its bulk $A_\alpha^\circ$. We assume uniform local relaxation: $\|e^{t_{\mathrm{loc}}\MM_\alpha}-P_{\MM_\alpha}\|_\diamond\leq1/2$ at a time $t_{\mathrm{loc}}$ for all $\alpha$. b) The bulk regions may overlap and must cover the lattice, $\bigcup_\alpha A_\alpha^\circ=\Lambda$, so every site is locally dissipated by at least one cover generator. c) Interactions between overlapping supports can reintroduce fluctuations into the locally relaxing regions. The coefficient $\kappa_{\LL}$ bounds the total influence of these couplings, as defined in \cref{eq:main-local-influence}. The condition $\kappa_{\LL}t_{\mathrm{loc}}<1/2$ ensures that local relaxation dominates this influence.}
\label{fig:main-dissipative-cover}
\end{figure}
The support size and incidence of the cover are quantified by
\begin{align}
    a &:=\max_\alpha|A_\alpha|,\qquad \mathfrak d:=\max_{x\in\Lambda}\#\{\alpha\in\{1,\ldots,M\}:x\in A_\alpha\},
\end{align}
respectively.\par\smallskip

For a $k$-local Lindbladian perturbation $\KK=\sum_{B\in\mathcal S_{\KK}}\KK_B$ with interaction strength bounded by $s_{\KK}$, keep the cover fixed and define $\kappa_{\KK}$ and $\kappa_{\LL+\varepsilon\KK}$ by \cref{eq:main-noncommutation-count,eq:main-local-influence}, replacing $\LL$ by $\KK$ and $\LL+\varepsilon\KK$, respectively. Then, for $\varepsilon\geq0$, the definition in \cref{eq:main-local-influence} gives
\begin{align}
    \kappa_{\LL+\varepsilon\KK}&\leq\kappa_{\LL}+\varepsilon\kappa_{\KK}
    \leq\kappa_{\LL}+\varepsilon\cdot4a\mathfrak d\,k s_{\KK}. \label{eq:main-perturbed-local-influence}
\end{align}

\begin{proposition}[Stability of strong bulk dissipation]
\label{prop:stability-strong-bulk-dissipation}\subsectionresulttoc{prop:stability-strong-bulk-dissipation}{Stability of strong bulk dissipation}
Let $(\LL_N)_N$ be an $l$-local Lindbladian family with bulk-dissipative covers of uniformly bounded support size $a$ and incidence $\mathfrak d$, and local mixing time $t_{\mathrm{loc}}=O(1)$. Suppose that $\kappa_{\LL_N}t_{\mathrm{loc}}\leq1/2-\Delta$ for some $\Delta>0$ independent of $N$. For every fixed locality $k$ and interaction-strength bound $s_{\KK}$, there exists a constant $\varepsilon_*>0$, independent of $N$, such that, for every $k$-local Lindbladian family $(\KK_N)_N$ with interaction strength at most $s_{\KK}$, the family $(\LL_N+\varepsilon\KK_N)_N$ remains rapidly mixing with constant decay rate and has a unique stationary state for all $0\leq\varepsilon\leq\varepsilon_*$.
\end{proposition}
The proof of \cref{prop:stability-strong-bulk-dissipation} is in \cref{proof:stability-strong-bulk-dissipation}. If $\kappa_{\LL}t_{\mathrm{loc}}\leq1/2-\Delta$ for some $\Delta>0$ and $s_{\KK}>0$, one may take
\begin{align}
    \varepsilon_*&:=\frac{\Delta}{8a\mathfrak d\,k s_{\KK}t_{\mathrm{loc}}}. \label{eq:main-constant-perturbation-threshold}
\end{align}

\begin{corollary}[Commuting strong bulk dissipation]
\label{cor:commuting-strong-bulk-dissipation}\subsectionresulttoc{cor:commuting-strong-bulk-dissipation}{Commuting strong bulk dissipation}
Let $(\LL_N)_N$ be an $l$-local Lindbladian family whose local terms commute and contain bulk-dissipative covers of uniformly bounded support size $a$ and incidence $\mathfrak d$, and local mixing time $t_{\mathrm{loc}}=O(1)$. For every fixed locality $k$ and interaction-strength bound $s_{\KK}$, there exists a constant $\varepsilon_*>0$, independent of $N$, such that, for every $k$-local Lindbladian family $(\KK_N)_N$ with interaction strength at most $s_{\KK}$, the family $(\LL_N+\varepsilon\KK_N)_N$ remains rapidly mixing with constant decay rate and has a unique stationary state for all $0\leq\varepsilon\leq\varepsilon_*$.
\end{corollary}
The proof of \cref{cor:commuting-strong-bulk-dissipation} is in \cref{proof:commuting-strong-bulk-dissipation}. Here $[\LL_A,\LL_{A'}]=0$ for every pair of local terms, so $\kappa_{\LL}=0$ and $\kappa_{\LL+\varepsilon\KK}\leq\varepsilon\cdot4a\mathfrak d\,k s_{\KK}$. Thus one may use \cref{eq:main-constant-perturbation-threshold} with $\Delta=1/2$.

In Appendix~\ref{app:dissipative-covers} we treat more general perturbations and provide sharper bounds. For example, we analyze differences of local Lindbladians and perturbations of Hamiltonians and jump operators, including small reductions of local dissipation rates; see \cref{cor:app-difference,cor:app-perturbed-hamiltonians-jumps}. We also obtain stability of rapid mixing for the following state-preparation protocols.

\begin{corollary}[Stability of high-temperature Gibbs samplers]
\label{cor:main-gibbs-stability}\subsectionresulttoc{cor:main-gibbs-stability}{Stability of high-temperature Gibbs samplers}
Let $(H_N)_N$ be a family of finite-range Hamiltonians on a $D$-dimensional qubit lattice, with uniformly bounded interaction strength and a uniformly bounded number of terms meeting each site. Let $(\LL_N^{(\beta)})_N$ be the Gibbs sampler family of \cite{RouzeFrancaAlhambra2024b} with single-qubit Pauli bath couplings, at inverse temperature $\beta$. There exists a size-independent $\beta_*>0$ with the following property: for every $k$ and $s_{\KK}$, there is a size-independent $\varepsilon_*>0$ such that, for every $k$-local Lindbladian perturbation family $(\KK_N)_N$ with interaction strength at most $s_{\KK}$, $(\LL_N^{(\beta)}+\varepsilon\KK_N)_N$ remains rapidly mixing with constant decay rate and has a unique stationary state whenever $0\leq\beta<\beta_*$ and $0\leq\varepsilon<\varepsilon_*$.
\end{corollary}
The proof and quantitative condition are given in \cref{proof:app-rfa-gibbs-sampler-stability}.

\begin{corollary}[Stability of simple injective MPDO parent dynamics]
\label{cor:main-mpdo-stability}\subsectionresulttoc{cor:main-mpdo-stability}{Stability of simple injective MPDO parent dynamics}
Consider the parent Lindbladian family $(\LL_N)_N$ of \cite[Section~V~A]{liu2026parent} for a simple injective matrix product density operator at a renormalization fixed point on a periodic chain of $N\geq5$ sites. For every fixed $k$ and $s_{\KK}$, there is a size-independent $\varepsilon_*>0$ such that, for $0\leq\varepsilon\leq\varepsilon_*$ and every $k$-local Lindbladian perturbation family $(\KK_N)_N$ with interaction strength at most $s_{\KK}$, $(\LL_N+\varepsilon\KK_N)_N$ remains rapidly mixing with constant decay rate and has a unique stationary state.
\end{corollary}
The precise construction and proof are given in \cref{cor:app-mpdo-parent-stability,proof:app-mpdo-parent-stability}. Strong local noise gives a complementary application: it can obstruct dissipative preparation of long-range order.

\begin{corollary}[Loss of long-range order under strong local noise]
\label{cor:main-local-noise}\subsectionresulttoc{cor:main-local-noise}{Loss of long-range order under strong local noise}
Let $(\LL_N)_N$ and $(\NN_N)_N$, with $\NN_N=\sum_B\NN_B$, be geometrically local Lindbladian families on a fixed-dimensional lattice, both with uniformly bounded per-site interaction strength. Suppose that the supports of the noise terms have uniformly bounded size and incidence and cover the lattice. Suppose each $\NN_B$ has a unique local stationary state and a common local mixing time bound $t_{\mathrm{noise}}=O(1)$. Let $\kappa_{\NN}$ be defined by \cref{eq:main-local-influence} for the cover consisting of the noise terms. Assume that $\kappa_{\NN}t_{\mathrm{noise}}\leq1/2-\Delta$ for a size-independent $\Delta>0$. Then there exists a size-independent $\varepsilon_*>0$ such that, for every fixed $\varepsilon>\varepsilon_*$, the family $(\LL_N+\varepsilon\NN_N)_N$ has a unique stationary state and is rapidly mixing with constant decay rate. The influence condition is always satisfied when the noise terms commute.
\end{corollary}
The proof is in \cref{proof:main-local-noise}. Proposals for generating long-range ordered states with bounded-range dynamics include boundary driving of spin chains into states with long-range spin correlations \cite{ProsenPizorn2008} and local dissipative preparation of topologically ordered states \cite{VerstraeteWolfCirac2009}. Under its noise hypotheses, \cref{cor:main-local-noise} shows that sufficiently strong local noise forces rapid relaxation. Preparing nonvanishing connected correlations over distances proportional to the system diameter requires time at least proportional to that diameter \cite{BravyiHastingsVerstraete2006,Poulin2010}, excluding such correlations in these rapidly mixing stationary states. The same preparation-time argument excludes the topological target states covered by \cite{KoenigPastawski2014}, at sufficiently small fixed trace-norm error; see \cref{proof:main-local-noise}.

\partbibliography

\paperpart{sections-quasifree_lindbladians}{0}
\subsection{Quasifree fermionic dynamics}\label{sec:stability_quasifree}
The bulk-dissipation criterion applies when local terms erase the degrees of freedom in a cover of the lattice. A complementary natural setting is provided by quasifree fermionic dynamics. Quadratic fermionic models describe free particles and superconducting mean-field states, and provide a basic setting for studying transport and topological phases, including Majorana boundary modes in quantum wires \cite{Kitaev2001}. Their open-system counterparts encompass exactly solvable boundary-driven transport \cite{Prosen2008}, dissipative preparation of topological quantum wires \cite{DiehlRicoBaranovZoller2011}, fermionic linear optics \cite{bravyi2011classical}, as well as dissipative preparation of free-fermionic ground states \cite{ZhanDingHuhnGrayPreskillChanLin2025}. While we prove the results for quasifree fermions, we expect analogous results for quasifree bosons, with suitable control of energy and dynamical stability.

Consider $N$ fermionic modes written as $m=2N$ Majoranas $c_1,\ldots,c_m$ satisfying $\{c_i,c_j\}=2\delta_{i,j}\BI$.
\begin{definition}[Quasifree fermionic Lindbladian]
The Lindbladian
\begin{align*}
    \LL(\rho)=-i[H,\rho]+\sum_\mu\left(2L_\mu\rho L_\mu^\dagger-\{L_\mu^\dagger L_\mu,\rho\}\right)
\end{align*}
is called \emph{quasifree} if $H$ is quadratic in the Majorana operators and $L_\mu$ are linear:
\begin{align*}
    H=\frac{i}{4}\sum_{j,k=1}^m \mathbf{H}_{j,k}c_jc_k,\qquad L_\mu=\sum_{j=1}^m l_{\mu,j}c_j,
\end{align*}
where $\mathbf{H}=-\mathbf{H}^T\in\mathbb R^{m\times m}$.
\end{definition}
We follow the conventions of \cite{Prosen2008,bravyi2011classical}. The Hamiltonian operator is denoted by $H$, and its coefficient matrix by $\mathbf{H}$. Define the positive semidefinite Hermitian matrix $\mathbf{M}$ and the real matrix $\mathbf{X}$ as
\begin{align}
    \mathbf{M}_{i,j}:=\sum_\mu l_{\mu,i}l_{\mu,j}^*,\qquad \mathbf{X}:=-\mathbf{H}-2(\mathbf{M}+\mathbf{M}^T)=-\mathbf{H}-4\operatorname{Re}(\mathbf{M}). \label{eq:quasifree_X}
\end{align}
Unlike the antisymmetric Hamiltonian matrix $\mathbf H$, the drift matrix $\mathbf X$ is generally nonnormal. Quasifree dynamics therefore exhibits a rich range of non-Hermitian phenomena, from exceptional points where eigenvalues and eigenvectors coalesce \cite{AshidaGongUeda2020} to boundary accumulation of decay modes and directional relaxation fronts \cite{SongYaoWang2019}. Third quantization provides an exact framework for studying such dynamics, including non-diagonalizable generators and their Jordan structure \cite{Prosen2010,BarthelZhang2022}.

For a state $\rho$, its covariance matrix is the real antisymmetric matrix
\begin{align}
    \Gamma_{j,k}(\rho)&:=\frac{i}{2}\Tr\!\left(\rho[c_j,c_k]\right).
\end{align}
Writing $\Gamma_t(\rho):=\Gamma(e^{t\LL}(\rho))$, the two-point correlation functions obey the closed differential equation
\begin{align}
    \frac{d}{dt}\Gamma_t(\rho)&=\mathbf X^T\Gamma_t(\rho)+\Gamma_t(\rho)\mathbf X+8\operatorname{Im}(\mathbf M). \label{eq:quasifree-covariance-evolution}
\end{align}
Thus $\mathbf X$ governs the homogeneous propagation of the covariance matrix, while $\operatorname{Im}(\mathbf M)$ is the dissipative source term. Equivalently,
\begin{align*}
    \Gamma_t(\rho)&=e^{t\mathbf X^T}\Gamma(\rho)e^{t\mathbf X}
    +8\int_0^t e^{(t-s)\mathbf X^T}\operatorname{Im}(\mathbf M)e^{(t-s)\mathbf X}\,ds,
\end{align*}
so the stationary covariance matrix satisfies the corresponding Lyapunov equation.

For a Gaussian initial state $\rho$ and a Gaussian stationary state $\sigma$, the covariance evolution and the trace-norm bound of \cite{Bittel_2025} give
\begin{align*}
    \|e^{t\LL}(\rho)-\sigma\|_1\leq\frac12\|\Gamma_t(\rho)-\Gamma(\sigma)\|_1
    \leq\frac12\|e^{t\mathbf X}\|^2\|\Gamma(\rho)-\Gamma(\sigma)\|_1.
\end{align*} For arbitrary initial states, however, even identical covariance matrices need not imply that the states are close. A simple example is given by two mixtures of occupation configurations with exactly two particles in four fermionic modes:
\begin{align}
    \rho_1&=\frac12\bigl(|1100\rangle\langle1100|+|0011\rangle\langle0011|\bigr),&
    \rho_2&=\frac12\bigl(|1010\rangle\langle1010|+|0101\rangle\langle0101|\bigr).\label{eq:qf-same-covariance-states}
\end{align}
Both states are supported entirely in the even-parity sector. Both have occupation $1/2$ in each mode and vanishing off-diagonal quadratic expectations, so $\Gamma(\rho_1)=\Gamma(\rho_2)=0$, but their orthogonal supports give $\|\rho_1-\rho_2\|_1=2$. Both are non-Gaussian: the Gaussian state with zero covariance is $\BI/16$, whereas their occupation correlations do not factorize. Thus covariance relaxation alone does not control the full state for non-Gaussian inputs. Here we seek trace-norm bounds valid for arbitrary initial states. For this purpose, we introduce Majorana quasiderivations and extend the oscillator-norm method to general quasifree fermionic Lindbladians.

We show that if $\sigma$ is a stationary state of a quasifree Lindbladian $\LL$, then for every state $\rho$ and every $t\geq0$,
\begin{align*}
    \|e^{t\LL}(\rho)-\sigma\|_1\leq m\|e^{t\mathbf{X}}\|\cdot||\rho-\sigma||_1.
\end{align*}
In particular, a size-independent positive lower bound on $\lambda_{\min}(\operatorname{Re}\mathbf M)$ gives rapid mixing with constant decay rate. The precise result is given in \cref{prop:quasifree_rapid_mixing}, with its proof in \cref{proof:quasifree_rapid_mixing}. Prior results obtained rapid mixing of quasifree Lindbladians when the fixed point is a Bogoliubov vacuum \cite{ZhanDingHuhnGrayPreskillChanLin2025}; our result uses different quasiderivations and applies to any quasifree Lindbladian.

\begin{proposition}[Stability under quasifree additions]
\label{cor:sum-gaussian}\subsectionresulttoc{cor:sum-gaussian}{Stability under quasifree additions}
Let $(\LL_N)_N$ be a quasifree fermionic Lindbladian family on $N$ fermionic modes, with $\operatorname{Re}(\mathbf M_{\LL_N})\succeq\mu\BI$ for some $\mu>0$ independent of $N$. For every quasifree Lindbladian family $(\KK_N)_N$ on the same modes, $(\LL_N+\KK_N)_N$ is rapidly mixing with constant decay rate and a unique stationary state.
\end{proposition}
The proof of \cref{cor:sum-gaussian} is in \cref{proof:sum-gaussian}. More explicitly, on $N$ fermionic modes, $\operatorname{Re}(\mathbf M_{\LL_N})\succeq\mu\BI$ gives
\begin{align}
    \lVert e^{t(\LL_N+\KK_N)}(\rho)-\sigma_{\LL_N+\KK_N}\rVert_1\leq 2Ne^{-4\mu t}\|\rho-\sigma_{\LL_N+\KK_N}\|_1. \label{eq:qf-addition-mixing-bound}
\end{align} Stability holds for any quasifree $\KK$: we do not require a bound on its strength, locality or any other property beyond the fact that it is quasifree in the same Majoranas. In fact, this last condition can be relaxed. Replacing $c_i$ by $c_i'=\sum_jO_{i,j}c_j$ with a real orthogonal $m\times m$ matrix $O$ again gives Majoranas, and the new matrices are
\begin{align*}
    \mathbf{H}=O^T\mathbf{H}'O,\qquad \mathbf{M}=O^T\mathbf{M}'O,\qquad \mathbf{X}=O^T\mathbf{X}'O.
\end{align*}
Hence $\lambda_{\min}(\operatorname{Re}(\mathbf{M}))$ and $\|e^{t\mathbf{X}}\|$ are basis independent, and \cref{cor:sum-gaussian} applies whenever $\KK$ is quasifree in any orthogonally related basis. Translation invariance gives a second stability result for quasifree models.

\begin{proposition}[Stability of translation-invariant quasifree dynamics under Lindbladian perturbations]
\label{prop:qf-translation-invariant-stability}\subsectionresulttoc{prop:qf-translation-invariant-stability}{Translation-invariant quasifree stability under perturbations}
Let $\faml{\LL}$ be a translation-invariant quasifree Lindbladian family on periodic lattices $\Lambda_N=(\mathbb Z/N\mathbb Z)^D$, with one fermionic mode per unit cell and uniformly bounded interaction strength $s_{\LL}$. Let $\mathbf X_{\LL_N}$ be its single-particle matrix. Suppose that for some $\gamma>0$, independently of $N$,
\begin{align}
    \max_{z\in\operatorname{spec}(\mathbf X_{\LL_N})}\operatorname{Re}z&\leq-\gamma\qquad\text{for every }N, \label{eq:qf-ti-symbol-gap}
\end{align}
Then:
\begin{enumerate}
\item[a)] The family $(\LL_N)_N$ is rapidly mixing with constant decay rate.
\item[b)] For every interaction-strength bound $s_{\KK}$, there exists a constant $\varepsilon_*>0$, independent of $N$, such that, for every quasifree Lindbladian family $(\KK_N)_N$ with interaction strength at most $s_{\KK}$, not necessarily translation invariant, the family $(\LL_N+\varepsilon\KK_N)_N$ remains rapidly mixing with constant decay rate for all $0\leq\varepsilon\leq\varepsilon_*$.
\end{enumerate}
\end{proposition}
The proof of \cref{prop:qf-translation-invariant-stability} is in \cref{proof:qf-translation-invariant-stability}.

Here the perturbation strength refers to the decomposition
\begin{align}
    \KK_N&=\sum_{Z\subseteq\Lambda_N}\KK_N^{(Z)},
    &s_{\KK}&:=\sup_N\sup_{x\in\Lambda_N}\sum_{Z\ni x}\|\KK_N^{(Z)}\|_{1\to1}<\infty, \label{eq:qf-ti-local-perturbation}
\end{align}
where each $\KK_N^{(Z)}$ is a quasifree Lindbladian supported on $Z$, and each local norm is computed on the modes in $Z$. The reference strength $s_{\LL}$ is defined analogously from a decomposition of $\LL_N$ into supported quasifree Lindbladians. The proof gives a size-independent time $t_0$ with $\sup_N\|e^{t_0\mathbf X_{\LL_N}}\|\leq1/2$. For $s_{\KK}>0$, one may take
\begin{align}
    \varepsilon_*&=\frac{1}{4t_0s_{\KK}}. \label{eq:qf-ti-perturbation-threshold}
\end{align}

A spectral gap alone does not generally imply rapid mixing: relaxation times can grow with system size without gap closing \cite{HagaNakagawaHamazakiUeda2021}, and relaxation can also be much slower than the inverse spectral gap in boundary-dissipated many-body systems \cite{MoriShirai2020}. The classical East process with open boundary conditions, at fixed vacancy density and with a facilitating boundary also has a size-independent positive gap but a worst-case mixing time growing linearly with the chain length \cite{GangulyLubetzkyMartinelli2013}. A quasifree example is the open fermionic chain with linear gain and loss studied in \cite{YangJiangBergholtz2022}: in its weak-dissipation regime, the Liouvillian gap remains positive independently of chain length, while the reported relaxation time of a boundary occupation grows linearly with chain length. In contrast, translation-invariant, finite-range quasifree fermionic systems with periodic boundary conditions have exponentially decaying connected correlations in their unique steady states \cite{ZhangBarthel2022Criticality}. Moreover, \cref{prop:qf-translation-invariant-stability} shows that, for translation-invariant quasifree families with periodic boundaries and uniformly bounded interaction strength, a uniform gap in the single-particle sector implies rapid mixing.

The role of the gap is made explicit by the estimates in \cref{prop:quasifree_rapid_mixing}. Write $\mathbf X=SJS^{-1}$, let $L$ be the largest Jordan-block size, and set
\begin{align*}
    \gamma(\mathbf X):=\min_{z\in\operatorname{spec}(\mathbf X)}-\operatorname{Re}z>0,
    \qquad \kappa(S):=\|S\|\|S^{-1}\|.
\end{align*}
Then
\begin{align*}
    \|e^{t\LL}(\rho)-\sigma\|_1
    \leq m\,\kappa(S)e^{-t\gamma(\mathbf X)}\left(\sum_{j=0}^{L-1}\frac{t^j}{j!}\right)\|\rho-\sigma\|_1.
\end{align*}
Thus, for well-conditioned systems ($\kappa(S)=\operatorname{poly}(N)$) with bounded Jordan blocks, \cref{prop:quasifree_rapid_mixing} shows that a size-independent gap in the single-particle sector yields rapid mixing. In our conventions, $\gamma(\mathbf X)$ coincides with the spectral gap of $\LL$ on the full operator space \cite{Prosen2010}. Together with the translation-invariant setting, these are two rare instances in which a Lindbladian gap yields rapid mixing.

\partbibliography

\paperpart{sections-mlsi_stability}{0}
\subsection{MLSI}\label{sec:mlsi-stability}
The preceding results obtain stability from structural information about the generator: a commuting decomposition, strong local dissipation, or quasifree closure. A different natural possibility is to strengthen the form of convergence itself. The modified logarithmic Sobolev inequality gives exponential contraction of relative entropy and is a central tool for proving thermalization. We prove stability when the stationary projections commute and, separately, under kernel inclusion.

We first introduce the notation needed to treat possibly nonprimitive dynamics. A \emph{conditional expectation} onto a unital $*$-subalgebra $\NN\subseteq\MM$ of a von Neumann algebra $\MM$ is a completely positive unital map $E^*:\MM\to\NN$ satisfying $E^*[X\Phi Y]=XE^*[\Phi]Y$ for $X,Y\in\NN$, $\Phi\in\MM$; its predual is denoted $E$. The \emph{decoherence-free subalgebra} of $\LL$ is the maximal algebra on which the dual semigroup acts as a $*$-automorphism \cite{Frigerio1978}:
\begin{align}
    \label{eq:dfs} \FF(\LL)=\Bigl\{X\in\BB(\HH): e^{t\LL^*}(X^\dagger X)=e^{t\LL^*}(X)^\dagger e^{t\LL^*}(X),\ e^{t\LL^*}(XX^\dagger)=e^{t\LL^*}(X)e^{t\LL^*}(X)^\dagger\ \ \forall t\ge0\Bigr\}.
\end{align}
\begin{proposition}[Conditional expectation and infinite-time limit, Proposition 8 of \cite{carbone2013decoherence}, Theorem 19 \cite{carbone2015environment}]\label{thm:CE}
Assume $e^{t\LL}$ has a full-rank invariant state $\sigma$. Then there is a unique conditional expectation $E^*_\LL:\BB(\HH)\to\FF(\LL)$ with $E_\LL(\sigma)=\sigma$, and for every $X\in\BB(\HH)$,
\begin{align*}
    \lim_{t\to+\infty}e^{t\LL^*}(X-E^*_\LL[X])=0, \qquad\text{equivalently}\qquad \lim_{t\to+\infty}e^{t\LL}(\rho-E_\LL(\rho))=0\quad\forall\rho\in\SSS(\HH).
\end{align*}
If $\LL$ has no nonzero purely imaginary eigenvalues, then $E_\LL=P_\LL$.
\end{proposition}
The proof of \cref{thm:CE} is in \cref{proof:CE}. For a Lindbladian $\LL$ with a full-rank invariant state and a well-defined infinite-time limit, we write $E_\LL=P_\LL$ for the state predual of the corresponding conditional expectation $E_\LL^*$. The stationary state need not be unique: for each initial state $\rho$, $E_\LL(\rho)$ is the stationary state in $\ker\LL$ selected by its infinite-time evolution. The fixed-point MLSI controls relative entropy to this input-dependent stationary state:
\begin{align}
    2\alpha\,D(\rho\|E_\LL(\rho))&\leq\EP_\LL(\rho),\qquad \EP_\LL(\rho):=-\Tr[\LL(\rho)(\log\rho-\log E_\LL(\rho))], \label{eq:fixed-point-MLSI}
\end{align}
for all density matrices $\rho$, and $\alpha(\LL)$ denotes its optimal constant. The following proposition bounds the MLSI constant of a sum when the stationary projections commute.
\begin{proposition}\label{prop:mlsi-commuting}\subsectionresulttoc{prop:mlsi-commuting}{MLSI for commuting stationary projections}
Let $\HH$ be a finite-dimensional Hilbert space, and let $\LL$ and $\KK$ be Lindbladians on $\BB(\HH)$, each with a full-rank invariant state, a well-defined infinite-time limit, and a positive MLSI constant. Suppose that
\begin{align*}
    E_\LL E_\KK=E_\KK E_\LL.
\end{align*}
Then the infinite-time limit of $\LL+\KK$ exists, and
\begin{align}
    E_{\LL+\KK}=E_\LL E_\KK,\qquad
    \ker(\LL+\KK)=\ker\LL\cap\ker\KK,\qquad
    \alpha(\LL+\KK)\geq\min\{\alpha(\LL),\alpha(\KK)\}.
\end{align}
If $E_\LL=E_\KK$, the stronger bound
\begin{align*}
    \alpha(\LL+\KK)\geq\alpha(\LL)+\alpha(\KK)
\end{align*}
holds.
\end{proposition}
The proof of \cref{prop:mlsi-commuting} is in \cref{proof:mlsi-commuting}. A stronger monotonicity statement holds when the perturbation leaves every steady state of $\LL$ fixed: under $\ker\LL\subseteq\ker\KK$, a positive MLSI for $\LL$ ensures that $\LL+\KK$ converges to the same stationary projection and inherits its MLSI.
\begin{proposition}\label{prop:mlsi-kernel}\subsectionresulttoc{prop:mlsi-kernel}{MLSI under kernel inclusion}
Let $\LL$ be a Lindbladian with a full-rank invariant state, a well-defined infinite-time limit $E_\LL$, and $\alpha(\LL)>0$. Let $\KK$ be a Lindbladian satisfying $\ker\LL\subseteq\ker\KK$. Then the infinite-time limit of $\LL+\KK$ exists, and
\begin{align*}
    E_{\LL+\KK}=E_\LL,\qquad
    \ker(\LL+\KK)=\ker\LL,\qquad
    \alpha(\LL+\KK)\geq\alpha(\LL).
\end{align*}
\end{proposition}
The proof of \cref{prop:mlsi-kernel} is in \cref{proof:mlsi-kernel}. In particular, every Lindbladian addition preserving the unique full-rank stationary state of $\LL$ inherits its MLSI.

\partbibliography

\paperpart{sections-fast_mixing}{0}
\subsection{The \texorpdfstring{$\chi^2$}{chi-squared} divergence}\label{sec:fast}
MLSI provides a robust entropy-contraction mechanism, but it can fail even for primitive rapidly mixing dynamics, as \cref{lemma:rm_MLSI} shows. This motivates looking for a weaker functional inequality that still has a monotonicity property under addition. Quantum $\chi^2$-divergences provide such a setting: they measure quadratic fluctuations around a faithful stationary state and lead to a Poincar\'e-type contraction bound governed by a self-adjoint operator. This is natural for physical relaxation problems in which the stationary state is faithful but a full relative-entropy contraction estimate is unavailable.

The $\chi^2$-divergences of \cite{TKRW2010} provide upper bounds on the trace-norm distance to a faithful stationary state, with exponential contraction controlled by a self-adjoint operator constructed from the Lindbladian and the stationary state. The resulting bound gives fast mixing with constant decay rate and does not by itself imply the system-size-independent stability of local observables established in \cite{CLMP2015}. Nevertheless, we show that adding any Lindbladian preserving the stationary state can only improve this $\chi^2$ contraction bound. Thus fast mixing with constant decay rate obtained through this mechanism is stable under such additions.

Let
\begin{align*}
    K=\left\{k:(0,\infty)\to(0,\infty)| -k\text{ is operator monotone},\ k(w^{-1})=wk(w)\text{ for all }w>0,\ k(1)=1\right\}.
\end{align*}
A prototypical example is $k_\alpha(w)=\frac12(w^{-\alpha}+w^{\alpha-1})$, with $0\le\alpha\le1$. On a finite-dimensional Hilbert space, define right and left multiplication operators $R_\sigma(A)=A\sigma$, $L_\rho(A)=\rho A$, respectively, and $\Delta_{\rho,\sigma}=L_\rho R_\sigma^{-1}$ for states $\rho$ and $\sigma\succ0$. The commuting positive self-adjoint maps $L_\rho$ and $R_\sigma^{-1}$ make $\Delta_{\rho,\sigma}$ positive and self-adjoint with respect to $\langle A,B\rangle=\Tr(A^\dagger B)$. For $k\in K$, define
\begin{align*}
    \chi_k^2(\rho,\sigma)=\bigl\langle\rho-\sigma,\Omega_\sigma^k(\rho-\sigma)\bigr\rangle,\qquad \Omega_\sigma^k=R_\sigma^{-1}k(\Delta_{\sigma,\sigma}).
\end{align*}
Here $\Omega_\sigma^k$ is a positive definite self-adjoint superoperator. In this subsection, inequalities between self-adjoint superoperators are understood with respect to the Hilbert--Schmidt inner product: $\Phi\preceq\Psi$ means that $\langle A,(\Psi-\Phi)(A)\rangle\geq0$ for every operator $A$.

The contraction estimate proved in \cref{thm:tkrw} has the following form.
Let $\mathcal L$ be a Lindbladian with a stationary state $\sigma\succ0$, and let $k\in K$. Write $\mathcal L^*$ for its Hilbert--Schmidt adjoint and define
\begin{align}
    \label{eq:Lambda} \Lambda_k(\mathcal L)=[\Omega_\sigma^k]^{1/2}\circ\mathcal L\circ[\Omega_\sigma^k]^{-1/2}+[\Omega_\sigma^k]^{-1/2}\circ\mathcal L^*\circ[\Omega_\sigma^k]^{1/2}.
\end{align}
Then $\Lambda_k(\mathcal L)\preceq0$ and $\Lambda_k(\mathcal L)(\sigma^{1/2})=0$. If $l_1^k(\mathcal L)\le0$ denotes its second largest eigenvalue, counting multiplicity, then every state $\rho$ satisfies
\begin{align}
    \label{eq:tkrw} \lVert e^{t\mathcal L}(\rho)-\sigma\rVert_1^2\le\chi_k^2(e^{t\mathcal L}(\rho),\sigma)\le e^{l_1^k(\mathcal L)t}\chi_k^2(\rho,\sigma),\qquad t\ge0.
\end{align}
This is shown in \cref{thm:tkrw}, with its proof in \cref{proof:tkrw}. The self-adjoint formulation gives a direct addition rule: if $\LL$ and $\KK$ have a common faithful stationary state $\sigma$, then the generalized-gap eigenvalue of their sum is no larger than the sum of their individual eigenvalues, computed with the same $k$ and $\sigma$. Thus the corresponding decay rates add.

\begin{proposition}[Stability of the $\chi^2$ mixing bound]\label{prop:tkrw-stability}\subsectionresulttoc{prop:tkrw-stability}{Stability of the \texorpdfstring{$\chi^2$}{chi-squared} mixing bound}
Let $\mathcal L$ and $\mathcal K$ be Lindbladians satisfying $\mathcal L(\sigma)=\mathcal K(\sigma)=0$ for the same state $\sigma\succ0$. For every $k\in K$, using this same $\sigma$ to define all three operators, we have
\begin{align*}
    \Lambda_k(\mathcal L+\mathcal K)&=\Lambda_k(\mathcal L)+\Lambda_k(\mathcal K)\preceq\Lambda_k(\mathcal L),\\
    l_1^k(\mathcal L+\mathcal K)&\le l_1^k(\mathcal L)+l_1^k(\mathcal K).
\end{align*}
If $l_1^k(\mathcal L)+l_1^k(\mathcal K)<0$, then $\mathcal L+\mathcal K$ converges exponentially to $\sigma$, which is therefore its unique stationary state.
\end{proposition}
The proof of \cref{prop:tkrw-stability} is in \cref{proof:tkrw-stability}. To make the resulting system-size dependence explicit, Eqs. (7) and (20) of \cite{TKRW2010} give
\begin{align}
    \chi_k^2(\rho,\sigma)\le\chi_0^2(\rho,\sigma)=\Tr(\rho^2\sigma^{-1})-1\le\lambda_{\min}(\sigma)^{-1},
\end{align}
where the last inequality uses $\sigma^{-1}\preceq\lambda_{\min}(\sigma)^{-1}I$ and $\Tr(\rho^2)\le1$. Applying \cref{eq:tkrw} to the sum, substituting this estimate, and using \cref{prop:tkrw-stability} yields
\begin{align}
    \sup_\rho\|e^{t(\mathcal L+\mathcal K)}(\rho)-\sigma\|_1
    \le\lambda_{\min}(\sigma)^{-1/2}e^{[l_1^k(\mathcal L)+l_1^k(\mathcal K)]t/2}.
\end{align}

In particular, for a family $(\LL_N)_N$ on $N$ qubits, if $-l_1^k(\mathcal L_N)\ge\gamma>0$ and $\lambda_{\min}(\sigma_N)\ge e^{-cN}$ with $c,\gamma$ independent of $N$, then every Lindbladian family $(\KK_N)_N$ satisfying $\mathcal K_N(\sigma_N)=0$ yields a fast-mixing family $(\LL_N+\KK_N)_N$ with constant decay rate:
\begin{align}
    \sup_\rho\|e^{t(\mathcal L_N+\mathcal K_N)}(\rho)-\sigma_N\|_1
    &\le e^{cN/2}e^{[l_1^k(\mathcal L_N)+l_1^k(\mathcal K_N)]t/2},\notag\\
    &\le e^{cN/2}e^{-\gamma t/2}.
\end{align}

Primitivity alone does not imply $l_1^k(\mathcal L)<0$. For example, the Lindbladian in \cref{lemma:rm_MLSI}, $\mathcal L(\rho)=-i[X,\rho]+\Diss{Z}(\rho)$, is primitive with stationary state $\BI/2$: its eigenvalues on traceless matrices are $-2$ and $-1\pm i\sqrt3$. Nevertheless,
\begin{align}
    \Delta_{\BI/2,\BI/2}&=L_{\BI/2}\circ R_{\BI/2}^{-1}=\tfrac12\,\mathrm{id}\circ2\,\mathrm{id}=\mathrm{id},\notag\\
    \Omega_{\BI/2}^k&=R_{\BI/2}^{-1}\circ k(\Delta_{\BI/2,\BI/2})=2k(1)\,\mathrm{id}=2\,\mathrm{id},
\end{align}
so we obtain
\begin{align}
    \Lambda_k(\mathcal L)(A)=\mathcal L(A)+\mathcal L^*(A)=2(ZAZ-A),
\end{align}
whose kernel contains both $\BI$ and $Z$. Hence $l_1^k(\mathcal L)=0$ for every $k\in K$.

\partbibliography

\clearpage
\begin{table}[p]
\centering
\caption{Positive rapid mixing and stability criteria proved here. The cover rows assume a polynomial number of regions, a size-independent positive margin $\Delta$, and $t_{\mathrm{loc}}=O(1)$.}
\label{tab:positive-summary}
\small
\renewcommand{\arraystretch}{1.15}
\begin{tabularx}{\textwidth}{@{}>{\raggedright\arraybackslash}p{0.18\textwidth} >{\raggedright\arraybackslash}X >{\raggedright\arraybackslash}p{0.27\textwidth} >{\raggedright\arraybackslash}p{0.16\textwidth}@{}}
\toprule
Setting & Sufficient hypothesis & Conclusion & Reference \\
\midrule
Dissipative cover & The incoming influence satisfies $\kappa_{\LL}t_{\mathrm{loc}}\leq1/2-\Delta$ & Unique fixed point and RM with constant decay rate & \cref{prop:stability-strong-bulk-dissipation} \\

Perturbed dissipative cover & The combined influence satisfies $\kappa_{\LL+\varepsilon\KK}t_{\mathrm{loc}}\leq1/2-\Delta$ & Unique fixed point and RM with constant decay rate & \cref{prop:stability-strong-bulk-dissipation,cor:commuting-strong-bulk-dissipation} \\

Gibbs-state preparation & High-temperature Gibbs samplers under the hypotheses of \cref{cor:main-gibbs-stability} & RM with constant decay rate, stable under sufficiently weak $k$-local Lindbladian perturbations & \cref{cor:main-gibbs-stability} \\

MPDO preparation & Simple injective fixed-point MPDO parent Lindbladian & RM with constant decay rate, stable under sufficiently weak $k$-local Lindbladian perturbations & \cref{cor:main-mpdo-stability} \\

Quasifree fermions: gapped generator & $\gamma(X)=\Omega(1)$, the Jordan transformation is polynomially well-conditioned, and the largest Jordan block has size $L=O(\log(N+1))$ & RM with constant decay rate; inverse-polylogarithmic $\gamma(X)$ still gives RM for $L=O(1)$ & \cref{prop:quasifree_rapid_mixing} \\

Quasifree fermions: gapped dissipation & $\lambda_{\min}(\operatorname{Re}\mathbf M)\geq\mu>0$ uniformly in system size & RM with constant decay rate, stable under arbitrary quasifree Lindbladian additions & \cref{prop:quasifree_rapid_mixing,cor:sum-gaussian} \\

Quasifree fermions: translation-invariant reference & The single-particle matrices $\mathbf X_{\LL_N}$ have a uniform gap as in \cref{eq:qf-ti-symbol-gap} & RM with constant decay rate, stable under sufficiently weak quasifree Lindbladian perturbations with uniformly bounded local strength & \cref{prop:qf-translation-invariant-stability} \\

Generic perturbation & A relaxing RM family with a unique stationary state; a Hermiticity-preserving, trace-annihilating perturbation of polynomial norm; contractive perturbed dynamics and sufficiently small inverse-polynomial strength & Unique, possibly different stationary state and RM & \cref{lem:perturbation_fixed_point} \\

MLSI stability & Positive MLSI, full-rank invariant states and the stated base limits; commuting stationary projections or $\ker\LL\subseteq\ker\KK$ & MLSI at least the minimum for commuting projections, the sum for identical projections, and $\alpha(\LL)$ under kernel inclusion & \cref{prop:mlsi-commuting,prop:mlsi-kernel} \\

$\chi^2$ framework & A common full-rank stationary state $\sigma$, with the same $k$ & $\chi^2$ contraction rate preserved or improved & \cref{prop:tkrw-stability} \\
\bottomrule
\end{tabularx}
\end{table}
\clearpage
\addtocontents{toc}{\protect\setcounter{tocdepth}{2}}
\hypersetup{bookmarksdepth=2}

\paperpart{sections-open_questions}{0}
\section{Conclusions and outlook}\label{sec:open}
Rapid mixing is important for thermalization, stability of local observables, efficient state preparation, and dynamical descriptions of mixed-state phases. Our results show that robustness of these mixing times requires assumptions beyond locality and rapid mixing of the ideal generator. Even two rapidly mixing generators with the same unique full-rank stationary state can have a sum whose mixing time grows polynomially with the system size. Conversely, our stability criteria identify regimes in which rapid mixing persists, including systems satisfying the bulk-dissipation condition, quasifree fermionic families, and generators with compatible stationary spaces and functional inequalities. Sufficiently strong local noise can also force rapid relaxation and obstruct the preparation of macroscopic order. These results motivate treating robustness of rapid mixing as an additional property of dissipative dynamics.

Size-independent perturbation thresholds are particularly relevant to scalability: preserving a favorable relaxation time need not require reducing the allowed local error as the system grows. Our applications to the high-temperature Gibbs sampler of \cite{RouzeFrancaAlhambra2024b} and the simple injective MPDO parent dynamics at renormalization fixed points of \cite{liu2026parent} give concrete guarantees for thermal-state and tensor-network-state preparation. Extending the stability guarantees beyond high temperature and MPDO renormalization fixed points would make them relevant to more strongly correlated thermal and tensor-network states. Allowing interactions in perturbations of quasifree dynamics would address errors that take an implementation outside its exactly solvable description.

The tools developed here may also be useful for dynamical descriptions of mixed-state phases, where equivalence is formulated through fast local dissipative transformations between specified states \cite{CP2019}. Although we have considered rapid mixing from arbitrary initial states, these tools may help bound relaxation times between particular states and establish the stability of those transformations under local perturbations. Such an extension would address the robustness of evolutions defining a phase even when convergence from every initial state is unnecessary.

A central remaining question is how broadly rapid mixing with constant decay rate is stable under arbitrarily weak local perturbations. The boundary-driven example shows that arbitrarily weak potentials can destroy rapid mixing, but does not settle stability with constant decay rate or under uniform rapid mixing. Resolving this question under physically natural assumptions, such as fixed local interactions and uniform mixing of local restrictions, would help distinguish fragile relaxation from robust rapidly mixing dissipation.

\begin{conjecture}[Rapid mixing is unstable under arbitrarily weak local perturbations]
\label{conj:uniform-rm-instability}
There exist geometrically local Lindbladian families $(\LL_N)_N$ and $(\KK_N)_N$ on lattices of $N$ sites, with range and local interaction strength bounded independently of $N$, and a constant $\eps_*>0$, such that $(\LL_N)_N$ is uniformly rapidly mixing, but $(\LL_N+\eps\KK_N)_N$ has a well-defined infinite-time limit and is not rapidly mixing for every fixed $0<\eps<\eps_*$. The two families are fixed independently of $\eps$.\footnote{The boundary-dissipation example in \cref{ex:sparse-boundary-potential} does not furnish a uniformly rapidly mixing family of this kind. The distance between dissipation sites grows with $N$, leaving arbitrarily large interior regions whose restricted dynamics is purely Hamiltonian.}
\end{conjecture}

Several questions remain open.
\begin{itemize}
\item \textbf{Davies dynamics for commuting Hamiltonians.} Let $(\LL_N)_N$ be a uniformly rapidly mixing Davies family for a fixed finite-range commuting Hamiltonian, at fixed temperature and with fixed local bath couplings and rates. Given a bound on the interaction strength and range, is there a size-independent $\eps_*>0$ such that, for every Lindbladian family $(\KK_N)_N$ satisfying these bounds, $(\LL_N+\eps\KK_N)_N$ remains rapidly mixing for $0\leq\eps<\eps_*$?
\item \textbf{Zero Hamiltonian part.} Let $(\KK_N)_N$ and $(\LL_N)_N$ be geometrically local Lindbladian families with uniformly bounded interaction strength and zero Hamiltonian part. If both families are rapidly mixing, have the same kernel, and share a full-rank invariant state, is $(\LL_N+\KK_N)_N$ rapidly mixing?
\item \textbf{Detailed balance and MLSI.} Does there exist a detailed-balance family that is rapidly mixing with a full-rank fixed point but has no uniformly positive MLSI constant?\footnote{Under the detailed-balance hypotheses of \cite{GR22}, the MLSI constant is positive at each finite size; the question concerns its uniformity in system size.}
\item \textbf{Rapid mixing without MLSI for purely dissipative dynamics.} Does rapid mixing of a Lindbladian with zero Hamiltonian part and a unique full-rank stationary state imply a positive fixed-point MLSI constant? Does uniform rapid mixing imply a size-independent positive MLSI constant under these assumptions?
\item \textbf{The $\chi^2$ divergence.} Does rapid mixing imply a gap of $\Lambda_k$ under suitable conditions?
\item \textbf{Relation to stability of expectations.} Are there examples showing that the results on stability of local observables \cite{CubittLuciaMichalakisPerezGarcia2015,LuciaCubittMichalakisPerezGarcia2015} can hold while stability of rapid mixing fails?
\item \textbf{Generic stability.} The counterexamples required fine tuning. Is rapid mixing \emph{generically} stable?
\end{itemize}

\section*{Acknowledgements}
We thank Georgios Styliaris for insightful discussions. B.R., S.St. and A.C. were financially supported by the DFG under grant TRR 352 -- Project-ID 470903074.
B.R. acknowledges financial support from the European Research Council through the ERC Starting Grant MaTCh, grant agreement n.~101117299.
B.R. also warmly acknowledges the GNFM (Gruppo Nazionale per la Fisica Matematica) - INDAM. J.A.M.L and R.T. acknowledge funding from the Munich Center for Quantum Science and Technology (MCQST), funded by the Deutsche Forschungsgemeinschaft (DFG) under Germany's Excellence Strategy (EXC2111-390814868). The research is part of the Munich Quantum Valley, which is supported by the Bavarian State Government with funds from the High tech Agenda Bayern Plus. This project was funded within the QuantERA II Programme, which received funding from the European Union's Horizon 2020 research and innovation programme under Grant Agreement No.~101017733.

AI tools were used for typesetting, straightforward calculations and analysis, and the refinement of individual steps in the proofs. The counterexamples were designed without the use of AI and predate ChatGPT 5.6 Sol.

\partbibliography

\appendix
\addtocontents{toc}{\protect\setcounter{tocdepth}{1}}
\hypersetup{bookmarksdepth=1}

\paperpart{appendices-counterexamples}{1}
\section{Instability results}\label{app:counterexamples}

\begin{example}[Restatement of \cref{ex:XZeps}]\label{app:ex:XZeps}
On a spin chain of $N$ sites, write $\LL=\LL_N$ and $\KK=\KK_N$ for the local generators of the families $(\LL_N)_N$ and $(\KK_N)_N$, with jump operator $Z_j$ at every site $j$ and Hamiltonians $X_j$ and $(\eps_j-1)X_j$, respectively:
\begin{align}
    \LL(\rho)=\sum_j\bigl(-i[X_j,\rho]+\Diss{Z_j}(\rho)\bigr),\qquad \KK(\rho)=\sum_j\bigl(-i[(\eps_j-1)X_j,\rho]+\Diss{Z_j}(\rho)\bigr),
\end{align}
where $\eps_j= 1/(j+1)$. Both generators have the unique full-rank fixed point $\sigma=(\BI/2)^{\otimes N}$, and both families are rapidly mixing. For every finite $N$, the sum $\LL+\KK$ is primitive and gapped, but its gap is $\Theta(N^{-2})$ and its mixing time is $\Omega(N^2)$, so $(\LL_N+\KK_N)_N$ is not rapidly mixing.
\end{example}

\begin{proof}[Proof of \cref{lem:XZeps}]
\prooflabel{lem:XZeps}{proof:XZeps}
We first prove rapid mixing of the two summands by obtaining a single-site mixing bound that is uniform in $j$ and then telescoping the tensor-product channels. We then compute the gap of their sum and use its slowest mode to rule out rapid mixing.

\emph{Rapid mixing of each summand.} For a single qubit, write
\begin{align}
    \mathcal G_{h,d}(\rho)&=-ih[X,\rho]+d\Diss{Z}(\rho),\qquad d>0.
\end{align}
The generator fixes $\BI/2$. Using $[X,Y]=2iZ$, $[X,Z]=-2iY$, $ZXZ=-X$, and $ZYZ=-Y$, its action on the traceless Pauli basis is
\begin{align}
    \mathcal G_{h,d}(X)&=-2dX,\notag\\
    \mathcal G_{h,d}(Y)&=-2dY+2hZ,\notag\\
    \mathcal G_{h,d}(Z)&=-2hY.
\end{align}
Thus its matrix on $(X,Y,Z)$ and its characteristic polynomial are
\begin{align}
    A_{h,d}&=\begin{pmatrix}-2d&0&0\\
    0&-2d&-2h\\
    0&2h&0\end{pmatrix},\notag\\
    \det(z\BI-A_{h,d})&=(z+2d)(z^2+2dz+4h^2).
\end{align}
The eigenvalues on the traceless subspace are consequently
\begin{align}
    \operatorname{spec}(A_{h,d})&=\{-2d,-d\pm\sqrt{d^2-4h^2}\}. \label{eq:XZeps-single-site-spectrum}
\end{align}
For $h\ne0$, all three eigenvalues have strictly negative real part: if $d^2-4h^2\geq0$, then $\sqrt{d^2-4h^2}<d$, while otherwise the two roots have real part $-d$. Hence $\mathcal G_{h,d}$ is primitive with unique stationary state $\BI/2$.

The single-site terms of $\LL$ and $\KK$ have $d=1$ and respectively $h=1$ and $h=\eps_j-1=-j/(j+1)$, so in both cases $1/2\leq|h|\leq1$. We need a mixing bound uniform in this range, including $|h|=1/2$, where the two eigenvalues of the $(Y,Z)$ block coincide. Write this block as $-\BI+B_h$, with
\begin{align}
    B_h&=\begin{pmatrix}-1&-2h\\
    2h&1\end{pmatrix}, & B_h^2&=-(4h^2-1)\BI.
\end{align}
The diagonal and off-diagonal parts have norms $1$ and $2|h|$, respectively, so $\|B_h\|\leq1+2|h|\leq3$. Set $\omega_h=\sqrt{4h^2-1}$. Expanding the exponential into its even and odd powers gives
\begin{align}
    \bigl\|e^{t(-\BI+B_h)}\bigr\| &\overset{(1)}{=}e^{-t}\left\|\cos(\omega_ht)\BI+ \frac{\sin(\omega_ht)}{\omega_h}B_h\right\|,\notag\\
    &\overset{(2)}{\leq}e^{-t}\left(|\cos(\omega_ht)|+ \left|\frac{\sin(\omega_ht)}{\omega_h}\right|\|B_h\|\right),\notag\\
    &\overset{(3)}{\leq}(1+3t)e^{-t}, \label{eq:XZeps-block-decay}
\end{align}
where $(1)$ uses $B_h^2=-\omega_h^2\BI$ and commutativity with the identity, $(2)$ is the triangle inequality, and $(3)$ uses $|\cos(\omega_ht)|\leq1$, $|\sin(\omega_ht)|\leq\omega_ht$, and $\|B_h\|\leq3$. At $\omega_h=0$, the quotient is interpreted by continuity as $t$, and the same estimate holds.

Let $P(O)=\tfrac12\Tr(O)\BI$. The map $e^{t\mathcal G_{h,1}}-P$ vanishes on the identity and has matrix $e^{tA_{h,1}}$ on the traceless Pauli basis. Since the space of single-qubit linear maps has fixed dimension, equivalence of norms gives a constant $C_0$, independent of $h,j,N$, such that
\begin{align}
    \|e^{t\mathcal G_{h,1}}-P\|_\diamond &\overset{(1)}{\leq}C_0\|e^{tA_{h,1}}\|,\notag\\
    &\overset{(2)}{\leq}C_0(1+3t)e^{-t},\notag\\
    &\overset{(3)}{\leq}Ce^{-t/2}, \label{eq:XZeps-local-mixing}
\end{align}
where $(1)$ is this norm comparison in the fixed Pauli basis, $(2)$ uses \cref{eq:XZeps-block-decay} and the decay $e^{-2t}$ on $X$, and $(3)$ uses $C:=C_0\sup_{s\geq0}(1+3s)e^{-s/2}<\infty$.

For either choice of the summand, let $T_{j,t}$ be its single-site channel. The tensor-product difference has the telescoping expansion
\begin{align}
    \bigotimes_{j=1}^NT_{j,t}-P^{\otimes N} &=\sum_{j=1}^NP^{\otimes(j-1)}\otimes(T_{j,t}-P) \otimes\bigotimes_{k=j+1}^NT_{k,t}.
\end{align}
Therefore
\begin{align}
    \left\|\bigotimes_{j=1}^NT_{j,t}-P^{\otimes N}\right\|_\diamond &\overset{(1)}{\leq}\sum_{j=1}^N \left\|P^{\otimes(j-1)}\otimes(T_{j,t}-P) \otimes\bigotimes_{k=j+1}^NT_{k,t}\right\|_\diamond,\notag\\
    &\overset{(2)}{\leq}\sum_{j=1}^N\|T_{j,t}-P\|_\diamond,\notag\\
    &\overset{(3)}{\leq}CNe^{-t/2},
\end{align}
where $(1)$ is the triangle inequality applied to the telescoping expansion, $(2)$ uses multiplicativity of the diamond norm and the fact that $P,T_{k,t}$ are channels with diamond norm one, and $(3)$ uses \cref{eq:XZeps-local-mixing}. Since $P^{\otimes N}(\rho)=\sigma=(\BI/2)^{\otimes N}$ for every state $\rho$, this proves rapid mixing with constant decay rate for both $\LL$ and $\KK$ towards $\sigma$. Stationarity follows from the single-site calculation, and convergence to $\sigma$ proves uniqueness. Both generators are unital and $\sigma$ is full rank.

\emph{Closing gap of the sum.} The Hamiltonian terms add to $\eps_jX_j$, while the dephasing strength doubles:
\begin{align}
    (\LL+\KK)(\rho) &=\sum_{j=1}^N\bigl(-i[\eps_jX_j,\rho]+2\Diss{Z_j}(\rho)\bigr).
\end{align}
Its single-site generator is $\mathcal G_{\eps_j,2}$. Substituting $h=\eps_j$ and $d=2$ into \cref{eq:XZeps-single-site-spectrum} gives the nonzero eigenvalues
\begin{align}
    -4,\qquad -2\pm2\sqrt{1-\eps_j^2}.
\end{align}
They are strictly negative because $0<\eps_j\leq1/2$. Hence the sum is primitive for every finite $N$, with the same stationary state $\sigma$. Since the single-site generators act on different tensor factors, the full spectrum consists of sums of their eigenvalues. Its gap is thus
\begin{align}
    g_N:=\operatorname{gap}(\LL+\KK) &\overset{(1)}{=}\min_{1\leq j\leq N}2\left(1-\sqrt{1-\eps_j^2}\right),\notag\\
    &\overset{(2)}{=}2\left(1-\sqrt{1-\eps_N^2}\right),\notag\\
    &\overset{(3)}{=}\frac{2\eps_N^2}{1+\sqrt{1-\eps_N^2}} \overset{(4)}{=}\Theta(N^{-2}), \label{eq:XZeps-sum-gap}
\end{align}
where $(1)$ uses the tensor-sum spectrum and the smallest single-site decay rate, $(2)$ uses that $\eps_j=1/(j+1)$ decreases with $j$ and $x\mapsto1-\sqrt{1-x^2}$ increases on $(0,1)$, $(3)$ rationalizes the numerator, and $(4)$ uses $\eps_N=(N+1)^{-1}$ and $1\leq1+\sqrt{1-\eps_N^2}\leq2$.

We finally check explicitly why this closing gap rules out rapid mixing. The real $(Y,Z)$ block of $\mathcal G_{\eps_N,2}$ has a real eigenvector with eigenvalue $-g_N$. Its corresponding Pauli combination gives a traceless Hermitian eigenoperator $F_N$, which we normalize by $\|F_N\|=1$. Its two eigenvalues are then $1$ and $-1$, so $(\BI+F_N)/2$ is a state and $\|F_N\|_1=2$. Choose the initial state
\begin{align}
    \rho_N&=(\BI/2)^{\otimes(N-1)}\otimes(\BI+F_N)/2.
\end{align}
Only the last site evolves away from its stationary state, and
\begin{align}
    \|e^{t(\LL+\KK)}(\rho_N)-\sigma\|_1 &\overset{(1)}{=}\left\|(\BI/2)^{\otimes(N-1)}\otimes \tfrac12 e^{t\mathcal G_{\eps_N,2}}(F_N)\right\|_1,\notag\\
    &\overset{(2)}{=}\tfrac12e^{-g_Nt} \left\|(\BI/2)^{\otimes(N-1)}\right\|_1\|F_N\|_1,\notag\\
    &\overset{(3)}{=}e^{-g_Nt},
\end{align}
where $(1)$ uses the tensor-product evolution, stationarity of $\BI/2$, and linearity, $(2)$ uses the eigenoperator equation and multiplicativity of the trace norm, and $(3)$ uses that a state has trace norm one and $\|F_N\|_1=2$. Therefore
\begin{align}
    \tau_{\operatorname{mix}}^{\LL+\KK}(1/2)\geq\frac{\log2}{g_N}=\Omega(N^2),
\end{align}
which rules out rapid mixing under \cref{eq:rapid-mixing-time}.
\end{proof}

\begin{example}[Restatement of \cref{ex:fermionic-addition}]
\label{app:ex:fermionic-addition}
For odd $N\geq5$, let $\Lambda_N=\mathbb Z/N\mathbb Z$, with two fermionic annihilation operators $a_x,b_x$ in each cell $x$. With indices understood modulo $N$, define
\begin{align}
    H_{\lambda,N}&=\sum_{x\in\Lambda_N}\left[\lambda a_x^\dagger b_x+\frac12(a_{x+1}^\dagger+a_{x-1}^\dagger)b_x+\mathrm{h.c.}\right],\label{eq:fermionic-addition-H}\\
    \LL_{\lambda,N}(\rho)&=-i[H_{\lambda,N},\rho]+\sum_{x\in\Lambda_N}\left(\Diss{a_x}(\rho)+\Diss{a_x^\dagger}(\rho)\right).\label{eq:fermionic-addition-L}
\end{align}
Define the families $(\LL_N^\pm)_N$ and $(\mathcal S_N)_N$ by $\LL_N^\pm=\LL_{\pm2,N}$ and $\mathcal S_N=\LL_N^++\LL_N^-$, and set $\sigma_N=4^{-N}\BI$. Then
\begin{align}
    \lVert e^{t\LL_N^\pm}-P_{\LL_N^\pm}\rVert_\diamond\leq12Ne^{-t/2},\qquad P_{\LL_N^\pm}(X)=\Tr(X)\sigma_N.\label{eq:fermionic-addition-individual}
\end{align}
The same estimate, with $N$ replaced by $|A|$, holds for the restriction to every nonempty set of cells $A$, retaining the on-site jumps and the Hamiltonian terms supported in $A$. The sum $\mathcal S_N=2\LL_{0,N}$ is primitive with fixed point $\sigma_N$, but
\begin{align}
    \tau_{\operatorname{mix}}^{\mathcal S_N}(\varepsilon)&\geq\frac{\log(\varepsilon^{-1})}{2r_N}
    \geq\frac{N^2}{2\pi^2}\log(\varepsilon^{-1}),
    &r_N&=1-\sqrt{1-4\sin^2\left(\frac{\pi}{2N}\right)}.\label{eq:fermionic-addition-lower}
\end{align}
The state establishing this bound has even parity.
\end{example}

We use the oscillator estimate of \cref{lem:qf-oscillator-mixing}. For a unital quasifree generator $\mathcal G$ on $m=2N$ Majoranas, its proof also gives the channel estimate needed here. Let $\mathbf X$ be its single-particle matrix and $P(B)=2^{-N}\Tr(B)\BI$. The expectation maps and commutator bounds used in that proof remain contractive after tensoring with the identity on any ancillary system. Since $P^*e^{t\mathcal G^*}=P^*$, the telescope in \cref{eq:qf-telescoping} and the closed evolution in \cref{eq:X_evolution} give
\begin{align}
    \|e^{t\mathcal G}-P\|_\diamond
    &=\|e^{t\mathcal G^*}-P^*\|_{\mathrm{cb}}\notag\\
    &\leq\sum_{j=1}^m\|\delta_j e^{t\mathcal G^*}\|_{\mathrm{cb}}
    \leq m\|e^{t\mathbf X}\|. \label{eq:fermionic-clifford-bound}
\end{align}
The first equality uses the duality between the diamond norm and the completely bounded norm. Here $\|\cdot\|_{\mathrm{cb}}$ is the operator norm maximized over ancillary extensions; each term in the sum is bounded by $\|e^{t\mathbf X}\|$ by the commutator argument in \cref{proof:qf-oscillator-mixing}.

\begin{proof}[Proof of \cref{lem:fermionic-addition}]
\prooflabel{lem:fermionic-addition}{proof:fermionic-addition}
We verify the claims for the construction in \cref{ex:fermionic-addition}.

\emph{Uniform rapid mixing of the summands.} For $q_k=2\pi k/N$, define
\begin{align}
    a_k=N^{-1/2}\sum_xe^{-iq_kx}a_x,\qquad b_k=N^{-1/2}\sum_xe^{-iq_kx}b_x.
\end{align}
Unitary mixing of the jump list leaves the dissipator invariant, and therefore
\begin{align}
    H_{\lambda,N}&=\sum_{k=0}^{N-1}g_\lambda(q_k)(a_k^\dagger b_k+b_k^\dagger a_k),
    &g_\lambda(q)&=\lambda+\cos q,\notag\\
    \sum_x(\Diss{a_x}+\Diss{a_x^\dagger})&=\sum_k(\Diss{a_k}+\Diss{a_k^\dagger}).
\end{align}
With $\Pi$ as in \cref{eq:qf-parity-dressing}, for $a=(c_1+ic_2)/2$ one has $\Diss{a}+\Diss{a^\dagger}=\tfrac12(\Diss{c_1}+\Diss{c_2})$, so the generators are unital. More generally, if the balanced gain and loss term has coefficient $\kappa>0$, the matrix governing the evolution of the linear observables in the dressed basis $(i\Pi a_k,i\Pi b_k)$ is
\begin{align}
    B_\kappa(g)&=\begin{pmatrix}-\kappa&-ig\\-ig&0\end{pmatrix},
    &z_\pm(g)&=-\frac\kappa2\pm\sqrt{\frac{\kappa^2}{4}-g^2}.\label{eq:fermionic-addition-drift}
\end{align}
For $\lambda=\pm2$, $1\leq|g_\lambda(q)|\leq3$. With $\kappa=1$ and $\omega_g=\sqrt{g^2-1/4}$,
\begin{align}
    e^{tB_1(g)}&=e^{-t/2}\left[\cos(\omega_gt)\BI+\frac{\sin(\omega_gt)}{\omega_g}\left(B_1(g)+\tfrac12\BI\right)\right],\notag\\
    \lVert e^{tB_1(g)}\rVert&\leq e^{-t/2}\left(1+\frac{|g|+1/2}{\sqrt{g^2-1/4}}\right)\leq3e^{-t/2}.
\end{align}
There are $4N$ Majoranas, so \cref{eq:fermionic-clifford-bound} gives \cref{eq:fermionic-addition-individual}, which proves primitivity and rapid mixing with constant decay rate.

For a restriction to cells $A$, define the real symmetric matrix
\begin{align}
    (J_A)_{xy}:=\frac12\mathbf1\bigl[\{x,y\}\text{ is an edge of the subgraph induced by }A\bigr],\qquad x,y\in A.
\end{align}
It satisfies $\lVert J_A\rVert\leq1$: the sum of the absolute values in each row is at most one, and $J_A$ is symmetric. Write $J_A=U_A\operatorname{diag}(\nu_1,\ldots,\nu_{|A|})U_A^T$, where $U_A$ is real orthogonal and $|\nu_j|\leq1$. Define modes
\begin{align*}
    \widetilde a_j&=\sum_{x\in A}(U_A)_{xj}a_x,\qquad
    \widetilde b_j=\sum_{x\in A}(U_A)_{xj}b_x.
\end{align*}
The same change of basis on both species diagonalizes the restricted hopping matrix $\pm2\BI+J_A$. Thus
\begin{align*}
    H_{\pm2,A}&=\sum_{j=1}^{|A|}(\pm2+\nu_j)
    (\widetilde a_j^\dagger\widetilde b_j+\widetilde b_j^\dagger\widetilde a_j).
\end{align*}
Orthogonality gives $\sum_j(U_A)_{xj}(U_A)_{yj}=\delta_{xy}$, so expanding both the jump and anticommutator terms yields
\begin{align*}
    \sum_{x\in A}(\Diss{a_x}+\Diss{a_x^\dagger})
    &=\sum_{j=1}^{|A|}(\Diss{\widetilde a_j}+\Diss{\widetilde a_j^\dagger}).
\end{align*}
The restricted drift therefore has blocks $B_1(\pm2+\nu_j)$ and their complex conjugates. Since $1\leq|\pm2+\nu_j|\leq3$, the bound above gives $\|e^{t\mathbf X_A}\|\leq3e^{-t/2}$, independently of the shape or connectedness of $A$. There are $4|A|$ Majoranas, and \cref{eq:fermionic-clifford-bound} now gives
\begin{align*}
    \|e^{t\LL_A^\pm}-P_{\LL_A^\pm}\|_\diamond
    &\leq12|A|e^{-t/2},\qquad
    P_{\LL_A^\pm}(X)=4^{-|A|}\Tr(X)\BI.
\end{align*}
This proves uniform rapid mixing with constant decay rate.

\emph{Slow mixing of the sum.} Adding the two generators cancels the on-site hopping and doubles the intercell hopping and the dissipator. Thus the matrices governing the linear-observable evolution are
\begin{align}
    B^{\mathrm{sum}}(q)=2B_1(\cos q)=\begin{pmatrix}-2&-2i\cos q\\-2i\cos q&0\end{pmatrix},
\end{align}
with eigenvalues $-1\pm\sqrt{1-4\cos^2q}$. Since $N$ is odd,
\begin{align}
    \min_k|\cos q_k|=\sin\left(\frac{\pi}{2N}\right)>0.
\end{align}
Every eigenvalue of $B^{\mathrm{sum}}(q_k)$ has strictly negative real part, so \cref{eq:fermionic-clifford-bound} gives convergence to the full-rank state $\sigma_N$ and hence primitivity of $\mathcal S_N$.

It remains to produce an even slow mode. Let $k_*\in\arg\min_k|\cos q_k|$, write $f=(a_{k_*},b_{k_*})^T$, and choose a unit eigenvector $w$ of $(B^{\mathrm{sum}}(q_{k_*}))^\dagger$ with eigenvalue $-r_N$. For a Hermitian $2\times2$ matrix $Q$, put
\begin{align}
    O_Q=f^\dagger Qf-\frac12\Tr(Q)\BI.
\end{align}
To verify its evolution, write $B=-ih-2P_a$, where $h=2\cos(q_{k_*})\left(\begin{smallmatrix}0&1\\1&0\end{smallmatrix}\right)$ and $P_a=\left(\begin{smallmatrix}1&0\\0&0\end{smallmatrix}\right)$. The CAR give the Hamiltonian contribution
\begin{align}
    i[H,f^\dagger Qf]&=f^\dagger i[h,Q]f,
\end{align}
while direct substitution in the balanced gain and loss dissipator gives
\begin{align}
    2(\Diss{a_{k_*}}+\Diss{a_{k_*}^\dagger})^*(f^\dagger Qf)
    &=-2f^\dagger(P_aQ+QP_a)f+2\Tr(P_aQ)\BI.
\end{align}
Combining the two identities and centering by $\Tr(Q)/2$ yields
\begin{align}
    \mathcal S_N^*(O_Q)=O_{B^\dagger Q+QB},\qquad B=B^{\mathrm{sum}}(q_{k_*}).
\end{align}
Taking $Q=ww^\dagger$ and $W=2O_Q$ yields $W=W^\dagger$, $W^2=\BI$, $\Tr W=0$, and $\mathcal S_N^*(W)=-2r_NW$. The state $\rho_N=\sigma_N(\BI+W)$ is positive, normalized, and has even parity. Trace-norm duality gives
\begin{align}
    \lVert e^{t\mathcal S_N}(\rho_N)-\sigma_N\rVert_1
    &\geq\left|\Tr\left[W(e^{t\mathcal S_N}(\rho_N)-\sigma_N)\right]\right|=e^{-2r_Nt}.
\end{align}
Finally, with $s=\sin(\pi/(2N))<1/2$,
\begin{align}
    2s^2\leq r_N=\frac{4s^2}{1+\sqrt{1-4s^2}}\leq4s^2\leq\frac{\pi^2}{N^2}.
\end{align}
This proves \cref{eq:fermionic-addition-lower} and completes the proof.
\end{proof}

\begin{lemma}[Repeated \cref{lemma:rm_MLSI}]\label{app:lemma:rm_MLSI}
The Lindbladian family $(\LL_N)_N$ defined by
\begin{align}
    \LL_N(\rho)=\sum_{j=1}^N\bigl(-i[X_j,\rho]+\Diss{Z_j}(\rho)\bigr)
\end{align}
has the unique full-rank fixed point $\BI/2^N$ and is rapidly mixing with constant decay rate, but it does not have a positive MLSI constant.
\end{lemma}

\begin{proof}[Proof of \cref{lemma:rm_MLSI}]
\prooflabel{lemma:rm_MLSI}{proof:rm_MLSI}
Rapid mixing with constant decay rate and the unique full-rank stationary state $\sigma=2^{-N}\BI$ follow from \cref{ex:XZeps}, proved in \cref{proof:XZeps}. It remains to show that $\LL$ has no positive MLSI constant. Write $\LL=\LL_N=\sum_{j=1}^N\LL_j$, where $\LL_j(\rho)=-i[X_j,\rho]+\Diss{Z_j}(\rho)$.

Let $\omega=\frac12(\BI+cZ)$ with $c\in(-1,0)\cup(0,1)$ and $\rho=\omega^{\otimes N}$. The entropy production of $\rho$ is
\begin{align}
    \operatorname{EP}(\rho)=-\Tr[\LL(\rho)(\log\rho-\log(2^{-N}\BI))] =-\Tr[\LL(\rho)\log\rho],
\end{align}
since $\LL(\rho)$ is traceless and $\log(2^{-N}\BI)$ is a multiple of $\BI$. Using $\LL(\omega^{\otimes N})=\sum_{j=1}^N\omega^{\otimes\hat j}\otimes\LL_j(\omega)$, where $\omega^{\otimes\hat j}$ denotes the tensor product over all sites except $j$, we get:
\begin{align*}
    \operatorname{EP}(\rho) &=\sum_{j=1}^N-\Tr\bigl[(\omega^{\otimes\hat j}\otimes\LL_j(\omega)) \log(\omega^{\otimes N})\bigr],\\
    &=\sum_{j=1}^N-\Tr_j[\LL_j(\omega)\log\omega],
\end{align*}
where in the last expression we expand the logarithm $\log(A\otimes B)=\log A\otimes\BI+\BI\otimes\log B$ and notice that the cross terms vanish because $\Tr_j(\LL_j(\omega))=0$. For the single-site state $\omega$, the dissipative term vanishes because $Z\omega Z=\omega$, so $\LL_j(\omega)=-i[X,\omega]=-cY$, while $\log\omega=d_1\BI+d_2Z$ for $d_1=\frac12\log\frac{1-c^2}{4}$, $d_2=\frac12\log\frac{1+c}{1-c}$; hence
\begin{align*}
    \Tr_j[\LL_j(\omega)\log\omega] =\Tr[-cY(d_1\BI+d_2Z)]=0,
\end{align*}
because $\Tr Y=0=\Tr(YZ)$. Thus $\operatorname{EP}(\rho)=0$. On the other hand $\rho\ne2^{-N}\BI$, so by Klein's inequality $D(\rho\|2^{-N}\BI)>0$. Hence no positive $\alpha$ can satisfy the MLSI $2\alpha D(\rho\|2^{-N}\BI)\le \operatorname{EP}(\rho)$ for such $\rho$.
\end{proof}

\begin{example}[Checkerboard decomposition of Ising dynamics, \cite{GKZ2024}]
\label{app:ex:ising-kernel-mismatch}
Let $\Lambda_L=\{1,\ldots,L\}^2$, with free boundary conditions, even $L\geq4$, and $N=L^2$, and let
\begin{align}
    H_L=-\sum_{\langle x,y\rangle}Z_xZ_y, \qquad \sigma_{\beta,L}=\frac{e^{-\beta H_L}}{\Tr e^{-\beta H_L}}
\end{align}
be the Hamiltonian and Gibbs state of the two-dimensional Ising model.
There are two geometrically 5-local, zero-Hamiltonian families $(\LL_{1,N})_N$ and $(\LL_{2,N})_N$, written $\LL_i=\LL_{i,N}$ at each size, with the following properties:
\begin{enumerate}
\item each has $\sigma_{\beta,L}$ as a full rank fixed state, satisfies the fixed-point MLSI with a constant at least $1/2$, and obeys
\begin{align}
    \label{eq:app-ising-layer-rm} \sup_\rho\|e^{t\LL_i}(\rho)-P_{\LL_i}(\rho)\|_1 \leq 2N e^{-t},\qquad i\in\{1,2\};
\end{align}
\item their fixed-point spaces are different and strictly larger than the span of $\sigma_{\beta,L}$;
\item for every sufficiently large fixed $\beta$ and all sufficiently large $L$, their sum is primitive but
\begin{align}
    \label{eq:app-ising-slow} \tau_{\operatorname{mix}}^{\LL_1+\LL_2}(1/4) \geq e^{c_\beta L}=e^{c_\beta\sqrt N}
\end{align}
for some $c_\beta>0$.
\end{enumerate}
\end{example}

\begin{proof}[Details of the checkerboard decomposition]
\prooflabel{app:ex:ising-kernel-mismatch}{proof:ising-kernel-mismatch}
Let $C_1,C_2$ be the two checkerboard colors. For each site $x$, let $T_x$ be the conditional heat-bath reset: measure its neighbors in the $Z$ basis and replace spin $x$ by its conditional Gibbs state, acting as the identity elsewhere. Write $\partial x$ for the neighbors of $x$ and $\eta\in\{\pm1\}^{\partial x}$ for their spin configuration, with projector $\Pi_\eta:=\bigotimes_{y\in\partial x}\frac{\BI_y+\eta_yZ_y}{2}$. In the basis $Z|s\rangle=s|s\rangle$, the operator form is
\begin{align}
    p_x(s\mid\eta)&:=\frac{e^{\beta s\sum_{y\in\partial x}\eta_y}}{2\cosh\!\left(\beta\sum_{y\in\partial x}\eta_y\right)},\qquad s\in\{\pm1\},\notag\\
    J_{x,\eta,s,u}&:=\sqrt{p_x(s\mid\eta)}\,|s\rangle\langle u|_x\otimes\Pi_\eta,\qquad s,u\in\{\pm1\},\notag\\
    T_x(\rho)&:=\sum_{\eta\in\{\pm1\}^{\partial x}}\sum_{s,u\in\{\pm1\}}J_{x,\eta,s,u}\rho J_{x,\eta,s,u}^\dagger, \label{eq:ising-conditional-reset}
\end{align}
where each $J_{x,\eta,s,u}$ is tensored with the identity outside $\{x\}\cup\partial x$. Set
\begin{align}
    \LL_i&:=\sum_{x\in C_i}(T_x-\operatorname{id}),\qquad E_i:=\prod_{x\in C_i}T_x. \label{eq:ising-layer-construction}
\end{align}
These are geometrically 5-local, purely dissipative generators. Since $\sigma_{\beta,L}$ is diagonal, measuring the neighbors leaves it unchanged; resampling spin $x$ from its conditional Gibbs distribution also preserves it. Thus $T_x(\sigma_{\beta,L})=\sigma_{\beta,L}$ for every $x$, and $\LL_i(\sigma_{\beta,L})=0$ for $i=1,2$. The resets are idempotent and commute within each color. The standard conditional-reset MLSI gives $\alpha(T_x-\operatorname{id})\geq1/2$ \cite[Example~3.2]{Bardet2017}, and \cref{prop:mlsi-commuting} therefore gives $\alpha(\LL_i)\geq1/2$.

We now derive the mixing estimate in \cref{eq:app-ising-layer-rm}. For $x\in C_i$, idempotence gives
\begin{align}
    S_x(t):=e^{t(T_x-\operatorname{id})}
    =T_x+e^{-t}(\operatorname{id}-T_x).
\end{align}
Because the resets commute within one color, if $C_i=\{x_1,\ldots,x_m\}$ then
\begin{align}
    e^{t\LL_i}=\prod_{k=1}^m S_{x_k}(t),
    \qquad
    P_{\LL_i}=E_i=\prod_{k=1}^m T_{x_k}.
\end{align}
The difference of these products has the telescoping expansion
\begin{align}
    e^{t\LL_i}-E_i
    =\sum_{k=1}^m
    \left(\prod_{j<k}S_{x_j}(t)\right)
    \bigl(S_{x_k}(t)-T_{x_k}\bigr)
    \left(\prod_{j>k}T_{x_j}\right).
\end{align}
Every $S_x(t)$ and $T_x$ is a quantum channel and therefore has diamond norm one. Moreover,
$S_x(t)-T_x=e^{-t}(\operatorname{id}-T_x)$ and
$\|\operatorname{id}-T_x\|_\diamond\leq2$. Hence submultiplicativity and the triangle inequality give
\begin{align}
    \|e^{t\LL_i}-P_{\LL_i}\|_\diamond
    &\leq\sum_{k=1}^m e^{-t}\|\operatorname{id}-T_{x_k}\|_\diamond
    \leq2|C_i|e^{-t}\leq2Ne^{-t}.
\end{align}
Applying this channel bound to any state $\rho$ gives \cref{eq:app-ising-layer-rm}.

The map $P_{\LL_i}=E_i$ preserves an arbitrary diagonal marginal on the opposite color and puts $C_i$ in its conditional Gibbs state. Thus the two fixed-point spaces are different and nontrivial.

We use the criterion in \cite{Frigerio1978}: a finite-dimensional Lindbladian with a full-rank stationary state is primitive if the only operators commuting with its Hamiltonian, all jumps, and their adjoints are scalar multiples of the identity. Here the Hamiltonian vanishes, and $p_x(s\mid\eta)>0$ for every spin configuration. The jumps in \cref{eq:ising-conditional-reset} satisfy
\begin{align*}
    \sum_\eta p_x(s\mid\eta)^{-1/2}J_{x,\eta,s,u}
    =|s\rangle\langle u|_x\otimes\BI_{x^c}.
\end{align*}
Thus an operator commuting with every jump commutes with every single-site matrix unit, and must be a scalar multiple of $\BI$. Since $\sigma_{\beta,L}$ is full rank, the sum $\LL_1+\LL_2$ is primitive. The intersection of the fixed-point spaces is the span of $\sigma_{\beta,L}$.
The low-temperature mixing-time lower bound in \cref{eq:app-ising-slow} follows from \cite[Corollary~3.5 and Lemma~4.14]{GKZ2024} for the two-dimensional Ising model.
\end{proof}

\needspace{4\baselineskip}
\begin{example}[Restatement of \cref{ex:rm_sum_dissip_only}]\label{app:ex:rm_sum_dissip_only}
Consider a periodic chain of $N\geq2$ sites, with two qubits $a_i,b_i$ at each site $i$. Define the families $(\LL_N)_N$ and $(\KK_N)_N$ by writing $\LL=\LL_N$ and $\KK=\KK_N$ at each size:
\begin{align}
    \LL&=\sum_{i=1}^N\mathcal{D}[L_{a,i}]+\sum_{i=1}^N\mathcal{D}[L_{b,i}]+\sum_{i=1}^N\mathcal{D}[L_{a,b,i}],\notag\\
    \KK&=\sum_{i=1}^N\mathcal{D}[K_{a,i}]+\sum_{i=1}^N\mathcal{D}[L_{b,i}]+\sum_{i=1}^N\mathcal{D}[L_{a,b,i}],\label{rows-eq:two-dissipators}
\end{align}
with jump operators
\begin{align}
    L_{a,i}&=\begin{pmatrix}1&2\\
    0&-1\end{pmatrix}_{a_i},\qquad 1\le i\le N,\notag\\
    K_{a,i}&=\begin{pmatrix}1&-2\\
    0&-1\end{pmatrix}_{a_i},\qquad 1\le i\le N,\notag\\
    L_{b,i}&=\ketbra{0}{0}_{b_i}\otimes\ketbra{0}{1}_{b_{i+1}},\qquad 1\le i\le N,\notag\\
    L_{a,b,i}&=\sqrt{\gamma}\,\ketbra{1}{1}_{a_i}\otimes\ketbra{0}{1}_{b_i},\qquad 1\le i\le N.\label{rows-eq:two-dissipator-jumps}
\end{align}
Here $\gamma>0$ is independent of $N$, and the site labels are understood modulo $N$, so $b_{N+1}=b_1$. The bottom jump changes $b_{i+1}$ from $1$ to $0$ when $b_i$ is in state $0$, while the cross jump changes $b_i$ from $1$ to $0$ when $a_i$ is in state $1$. Thus $\LL$ and $\KK$ are translation invariant and have the same bottom and cross jumps, while the off-diagonal entries of their $a_i$ jumps have opposite signs.

Let $\eta=\begin{pmatrix}\frac56&-\frac13\\-\frac13&\frac16\end{pmatrix}$ be the unique fixed point of the single-qubit dissipator with jump operator $L_a=\begin{pmatrix}1&2\\0&-1\end{pmatrix}$. The Lindbladian $\LL$ has the unique fixed point $\sigma_\LL=\eta_a^{\otimes N}\otimes\ketbra{0^N}{0^N}_b$, while $\KK$ has the unique fixed point $\sigma_\KK=(Z\eta Z)_a^{\otimes N}\otimes\ketbra{0^N}{0^N}_b$. Both families $(\LL_N)_N$ and $(\KK_N)_N$ are dissipation-only and rapidly mixing with constant decay rate, but $(\LL_N+\KK_N)_N$ is not rapidly mixing: for every fixed $0<\delta<1/2$, $\tau_{\operatorname{mix}}^{\LL+\KK}(\delta)=\Omega(N)$.

Their sum has kernel
\begin{align}
    \ker(\LL+\KK)&=\ketbra{0^N}{0^N}_a\otimes\BB(\mathcal S_b),\qquad \mathcal S_b:=\operatorname{span}\{|0^N\rangle_b,|1^N\rangle_b\},
\end{align}
where $\BB(\mathcal S_b)$ is embedded in the $b$-qubit operator algebra and includes coherences between the strings. Its infinite-time limit exists, and the mixing time to this kernel is $\Omega(N)$.
\end{example}

\begin{proof}[Proof of \cref{lem:rm_sum_dissip_only}]
\prooflabel{lem:rm_sum_dissip_only}{proof:rm_sum_dissip_only}
\emph{Rapid mixing of each summand.} We prove rapid mixing for $\LL$ with the following strategy. First we bound how quickly the $b$ qubits approach $\ketbra{0^N}{0^N}_b$. At the intermediate time $s:=t/2$, we then replace the evolved state $\rho_s:=e^{s\LL}(\rho)$ by a nearby state $\zeta=\zeta_a\otimes\ketbra{0^N}{0^N}_b$ whose $b$ qubits have reached the steady state. Finally, on this exact $b$-vacuum, the bottom and cross terms vanish and the $a$ qubits evolve as the tensor product $T_s^{\otimes N}$ of the one-qubit semigroup $T_s:=e^{s\mathcal{D}[L_a]}$. The bounds proved below give
\begin{align}
    \|e^{t\LL}(\rho)-\sigma_\LL\|_1&\overset{(1)}{\le}\|\rho_s-\zeta\|_1+\|T_s^{\otimes N}(\zeta_a)-\eta_a^{\otimes N}\|_1,\notag\\
    &\overset{(2)}{\le}2\sqrt{8N}\,e^{-\lambda t}+CNe^{-ct/2},\notag\\
    &\overset{(3)}{\le}(2\sqrt8+C)Ne^{-at},\qquad a:=\min\{\lambda,c/2\},\label{rows-eq:upper}
\end{align}
where $(1)$ uses the triangle inequality proved in \cref{rows-eq:two-stage-triangle}, $(2)$ uses the $b$-qubit replacement bound \cref{rows-eq:projection-applied} and the $a$-qubit tensor-product bound \cref{rows-eq:top-product}, and $(3)$ uses $\sqrt N\le N$ and the definition of $a$. This is rapid mixing with constant decay rate. We now prove each of these bounds.

It is enough to prove this estimate for times satisfying $8Ne^{-2\lambda t}<1$. Indeed, for the complementary times the trivial estimate $\|e^{t\LL}(\rho)-\sigma_\LL\|_1\le2$ can be absorbed into the prefactor in \cref{rows-eq:upper}. This restriction will ensure that the normalized projection used below is well defined.

We first study how quickly the $b$ qubits approach $\ketbra{0^N}{0^N}_b$. Let $n_{b_j}:=\ketbra{1}{1}_{b_j}$ and $\rho_t:=e^{t\LL}(\rho)$. Our goal is to bound the decay of $\Tr(n_{b_j}\rho_t)$, the probability that qubit $b_j$ is occupied. It suffices to find a positive one-qubit matrix $h$ and a constant $\lambda>0$ such that the observable $h_{a_j}n_{b_j}$ on the pair $a_j,b_j$ satisfies
\begin{align}
    \LL^*(h_{a_j}n_{b_j})&\preceq-4\lambda h_{a_j}n_{b_j},\label{rows-eq:required-adjoint}\\
    \tfrac14 n_{b_j}&\preceq h_{a_j}n_{b_j}\preceq2n_{b_j}.\label{rows-eq:occupation-comparison}
\end{align}
The first inequality controls its decay; the second compares it with the occupation. If these inequalities hold, then
\begin{align}
    \frac{d}{dt}\Tr(h_{a_j}n_{b_j}\rho_t)&\overset{(1)}{\le}-4\lambda\Tr(h_{a_j}n_{b_j}\rho_t),\notag\\
    \Tr(h_{a_j}n_{b_j}\rho_t)&\overset{(2)}{\le}e^{-4\lambda t}\Tr(h_{a_j}n_{b_j}\rho),\label{rows-eq:weighted-decay}
\end{align}
where $(1)$ follows by taking the expectation of \cref{rows-eq:required-adjoint} in $\rho_t$, and $(2)$ follows by integration. Using this decay estimate, we obtain
\begin{align}
    \Tr(n_{b_j}\rho_t)&\overset{(1)}{\le}4\Tr(h_{a_j}n_{b_j}\rho_t),\notag\\
    &\overset{(2)}{\le}4e^{-4\lambda t}\Tr(h_{a_j}n_{b_j}\rho),\notag\\
    &\overset{(3)}{\le}8e^{-4\lambda t},\label{rows-eq:site-bound}
\end{align}
where $(1)$ uses the lower bound in \cref{rows-eq:occupation-comparison}, $(2)$ uses \cref{rows-eq:weighted-decay}, and $(3)$ uses the upper bound in \cref{rows-eq:occupation-comparison} together with $\Tr(n_{b_j}\rho)\le1$. We now construct $h$ and prove \cref{rows-eq:required-adjoint} and \cref{rows-eq:occupation-comparison}.

To obtain \cref{rows-eq:required-adjoint}, we look for a real symmetric matrix $h$ with $h_{00}=1$. On $h_{a_j}n_{b_j}$, the top jump contributes $\mathcal{D}[L_a]^*(h)_{a_j}n_{b_j}$, where $L_a=\left(\begin{smallmatrix}1&2\\0&-1\end{smallmatrix}\right)$. The cross jump contributes $-\gamma\{\ketbra{1}{1},h\}_{a_j}n_{b_j}/2$: its jump term vanishes because $\ketbra{1}{0}\ketbra{1}{1}\ketbra{0}{1}=0$. The incoming bottom edge removes the occupation at $b_j$ when $b_{j-1}$ is empty, while the outgoing edge and its adjoint commute with $h_{a_j}n_{b_j}$. All other jumps have disjoint support. Therefore
\begin{align}
    \LL^*(h_{a_j}n_{b_j})&=\bigl(\mathcal{D}[L_a]^*(h)-\tfrac\gamma2\{\ketbra{1}{1},h\}\bigr)_{a_j}n_{b_j}
    -h_{a_j}\ketbra{0}{0}_{b_{j-1}}n_{b_j},\label{rows-eq:pair-action}
\end{align}
where $j-1$ is understood modulo $N$.
For $h=\left(\begin{smallmatrix}1&h_{01}\\h_{01}&h_{11}\end{smallmatrix}\right)$, direct multiplication gives
\begin{align}
    \mathcal{D}[L_a]^*(h)-\tfrac\gamma2\{\ketbra{1}{1},h\}&=\begin{pmatrix}-2h_{01}&1-h_{11}-(4+\gamma/2)h_{01}\\
    1-h_{11}-(4+\gamma/2)h_{01}&4-6h_{01}-(4+\gamma)h_{11}\end{pmatrix}.\label{rows-eq:drift-matrix}
\end{align}
We make this matrix diagonal and negative. First impose $4-(4+\gamma)h_{11}=-\gamma/2$, so its lower-right entry becomes $-\gamma/2-6h_{01}$. Then cancel its off-diagonal entries. Solving these two equations in order gives
\begin{align}
    h_{11}&=\frac{4+\gamma/2}{4+\gamma}=\frac{\gamma+8}{2(\gamma+4)},\qquad h_{01}=\frac{1-h_{11}}{4+\gamma/2}=\frac{\gamma}{(\gamma+4)(\gamma+8)}>0.\notag
\end{align}
Thus both diagonal entries in \cref{rows-eq:drift-matrix} are negative. We have obtained
\begin{align}
    h&=\begin{pmatrix}1&\dfrac{\gamma}{(\gamma+4)(\gamma+8)}\\
    \dfrac{\gamma}{(\gamma+4)(\gamma+8)}&\dfrac{\gamma+8}{2(\gamma+4)}\end{pmatrix},\qquad \lambda:=\frac{h_{01}}4=\frac{\gamma}{4(\gamma+4)(\gamma+8)}.\notag
\end{align}
Here $1/2\le h_{11}\le1$ and $0<h_{01}\le1/8$. Bounding the cross term of its quadratic form by $2|u_0u_1|\le|u_0|^2+|u_1|^2$ gives $\frac38\BI\preceq h\preceq\frac98\BI$. In particular, $h$ is positive and $h\preceq2\BI$; also $0<\lambda\le1/32$. Substituting the chosen entries into \cref{rows-eq:drift-matrix} therefore yields
\begin{align}
    \mathcal{D}[L_a]^*(h)-\tfrac\gamma2\{\ketbra{1}{1},h\}&=\operatorname{diag}(-8\lambda,-\gamma/2-24\lambda)\preceq-8\lambda\BI\preceq-4\lambda h.\label{rows-eq:killed}
\end{align}
The last term of \cref{rows-eq:pair-action} is negative semidefinite because it is a tensor product of positive factors with a minus sign. Hence $\LL^*(h_{a_j}n_{b_j})\preceq-4\lambda h_{a_j}n_{b_j}$, proving \cref{rows-eq:required-adjoint}.

For \cref{rows-eq:occupation-comparison}, the bounds on $h$ give $\frac14\BI\preceq h\preceq2\BI$. Since $h_{a_j}$ and $n_{b_j}$ act on different qubits, tensoring with the positive projector $n_{b_j}$ gives $\frac14 n_{b_j}\preceq h_{a_j}n_{b_j}\preceq2n_{b_j}$. Both required inequalities are therefore established, so \cref{rows-eq:site-bound} holds for every initial state.

With $P:=\BI_a\otimes \ketbra{0^N}{0^N}_b$, the quantity $p_t:=\Tr((\BI-P)\rho_t)$ is the probability that at least one $b$ qubit is occupied. The indicator of this event is bounded by the number of occupied sites, so $\BI-P\preceq\sum_jn_{b_j}$. Summing the preceding estimate with \cref{rows-eq:site-bound} gives
\begin{align}
    p_t&\le\sum_{j=1}^N\Tr(n_{b_j}\rho_t)\le8N e^{-4\lambda t}.\label{rows-eq:clearing}
\end{align}
\cref{rows-eq:clearing} controls the weight outside the $b$-vacuum $\ketbra{0^N}{0^N}_b$. We now study how large $t$ needs to be so that projecting to a state with the $b$ qubits exactly in $\ketbra{0^N}{0^N}_b$ is a good enough approximation for rapid mixing. Set $s:=t/2$ and $\rho_s:=e^{s\LL}(\rho)$. We will construct a state $\zeta$ supported on $P$ and close to $\rho_s$, of the form $\zeta=\zeta_a\otimes\ketbra{0^N}{0^N}_b$. If we find such a $\zeta$, all bottom and cross jumps vanish on the $b$-vacuum, so only the top jumps remain and the $a$ qubits evolve independently. Thus, writing $T_u:=e^{u\mathcal{D}[L_a]}$, we have $e^{s\LL}(\zeta)=T_s^{\otimes N}(\zeta_a)\otimes\ketbra{0^N}{0^N}_b$. The triangle inequality and contraction of the trace norm under $e^{s\LL}$ therefore give
\begin{align}
    \|e^{t\LL}(\rho)-\sigma_\LL\|_1&=\|e^{s\LL}(\rho_s)-\sigma_\LL\|_1,\notag\\
    &\overset{(1)}{\le}\|e^{s\LL}(\rho_s)-e^{s\LL}(\zeta)\|_1+\|e^{s\LL}(\zeta)-\sigma_\LL\|_1,\notag\\
    &\overset{(2)}{\le}\|\rho_s-\zeta\|_1+\|\bigl(T_s^{\otimes N}(\zeta_a)-\eta_a^{\otimes N}\bigr)\otimes\ketbra{0^N}{0^N}_b\|_1,\notag\\
    &\overset{(3)}{=}\|\rho_s-\zeta\|_1+\|T_s^{\otimes N}(\zeta_a)-\eta_a^{\otimes N}\|_1,\label{rows-eq:two-stage-triangle}
\end{align}
where $(1)$ is the triangle inequality, $(2)$ uses contractivity for the first term and the expressions for $e^{s\LL}(\zeta)$ and $\sigma_\LL$, and $(3)$ uses multiplicativity of the trace norm together with $\|\ketbra{0^N}{0^N}_b\|_1=1$. Thus the first term in \cref{rows-eq:two-stage-triangle} is the error from replacing the $b$ qubits by the exact vacuum, while the second is the remaining $a$-qubit mixing error. We now construct $\zeta$ and bound these two terms.

For the first term, we use the following normalized-projector form of the gentle-measurement lemma \cite[Lemma~9]{winter1999coding}. For any state $\xi$, let $p:=\Tr((\BI-P)\xi)$ and, when $p<1$, let $\zeta:=P\xi P/(1-p)$. Then
\begin{align}
    \|\xi-\zeta\|_1&\le2\sqrt p.\label{rows-eq:projection}
\end{align}
For completeness, we give the short purification proof. Let $\lvert\Psi\rangle$ be a purification of $\xi$ on the system and a reference space, so $\Tr_{\mathrm{ref}}\ketbra{\Psi}{\Psi}=\xi$, and put $q:=\Tr(P\xi)=1-p$. First,
\begin{align}
    \|(P\otimes\BI)\lvert\Psi\rangle\|^2&\overset{(1)}{=}\langle\Psi\vert(P\otimes\BI)^2\vert\Psi\rangle,\notag\\
    &\overset{(2)}{=}\langle\Psi\vert(P\otimes\BI)\vert\Psi\rangle,\notag\\
    &\overset{(3)}{=}\Tr\!\left((P\otimes\BI)\ketbra{\Psi}{\Psi}\right)=\Tr(P\xi)=q,\label{rows-eq:projected-norm}
\end{align}
where $(1)$ is the definition of the squared norm, $(2)$ uses $P^2=P$, and $(3)$ uses the defining partial-trace identity for the purification. Hence, when $q>0$, the projected vector
\begin{align}
    \lvert\Phi\rangle&:=\frac{(P\otimes\BI)\lvert\Psi\rangle}{\sqrt q}.\notag
\end{align}
is normalized. It is a purification of $\zeta$ because
\begin{align}
    \Tr_{\mathrm{ref}}\ketbra{\Phi}{\Phi}&\overset{(1)}{=}\frac{P\bigl(\Tr_{\mathrm{ref}}\ketbra{\Psi}{\Psi}\bigr)P}{q},\notag\\
    &\overset{(2)}{=}\frac{P\xi P}{q}=\zeta,
\end{align}
where $(1)$ substitutes the definition of $\lvert\Phi\rangle$ and pulls the system operators through the partial trace, while $(2)$ uses $\Tr_{\mathrm{ref}}\ketbra{\Psi}{\Psi}=\xi$. Its overlap with the original purification is
\begin{align}
    \langle\Psi\vert\Phi\rangle&\overset{(1)}{=}\frac{\langle\Psi\vert(P\otimes\BI)\vert\Psi\rangle}{\sqrt q},\notag\\
    &\overset{(2)}{=}\frac{\Tr(P\xi)}{\sqrt q},\notag\\
    &\overset{(3)}{=}\frac{q}{\sqrt q}=\sqrt q=\sqrt{1-p},\label{rows-eq:purification-overlap}
\end{align}
where $(1)$ uses the definition of $\lvert\Phi\rangle$, $(2)$ uses the partial-trace identity, and $(3)$ uses $q=\Tr(P\xi)=1-p$. In particular, there is no hidden phase: $\langle\Psi\vert(P\otimes\BI)\vert\Psi\rangle$ is the expectation of a positive projector, so it is the real nonnegative number $q$. Contractivity under the partial trace and the pure-state trace-distance formula now give
\begin{align}
    \|\xi-\zeta\|_1&\overset{(1)}{\le}\|\ketbra{\Psi}{\Psi}-\ketbra{\Phi}{\Phi}\|_1,\notag\\
    &\overset{(2)}{=}2\sqrt{1-|\langle\Psi\vert\Phi\rangle|^2},\notag\\
    &\overset{(3)}{=}2\sqrt p,
\end{align}
where $(1)$ uses contractivity of the partial trace, $(2)$ is the trace distance between two pure states, and $(3)$ uses \cref{rows-eq:purification-overlap}. This proves \cref{rows-eq:projection} and is exactly where positivity of $\xi$ controls the coherences between the vacuum and nonvacuum sectors. In the time regime fixed above, \cref{rows-eq:clearing} gives $p_s\le8Ne^{-4\lambda s}=8Ne^{-2\lambda t}<1$, so the normalized projected state is well defined. Applying the estimate to $\xi=\rho_s$ gives
\begin{align}
    \|\rho_s-\zeta\|_1&\overset{(1)}{\le}2\sqrt{p_s},\notag\\
    &\overset{(2)}{\le}2\sqrt{8N}\,e^{-2\lambda s},\notag\\
    &\overset{(3)}{=}2\sqrt{8N}\,e^{-\lambda t},\label{rows-eq:projection-applied}
\end{align}
where $(1)$ uses \cref{rows-eq:projection}, $(2)$ uses \cref{rows-eq:clearing} at time $s$, and $(3)$ uses $s=t/2$.

It remains to bound the second term in \cref{rows-eq:two-stage-triangle}. The one-qubit semigroup $T_t$ is primitive with stationary state $\eta$: indeed, $\eta$ is positive with trace $1$, direct substitution gives $\mathcal{D}[L_a](\eta)=0$, and direct diagonalization shows that the eigenvalue $0$ of $\mathcal{D}[L_a]$ is simple, while its nonzero eigenvalues are $-2,-4,-6$. Since this is a fixed finite-dimensional semigroup, there are constants $C,c>0$, independent of $N$, such that its asymptotic channel $E_a(B):=\eta\Tr B$ satisfies
\begin{align}
    \|T_t-E_a\|_\diamond&\le Ce^{-ct}.\label{row:Tt-Etop}
\end{align}
We may therefore telescope directly in diamond norm:
\begin{align}
    T_t^{\otimes N}-E_a^{\otimes N}&=\sum_{j=1}^NE_a^{\otimes(j-1)}\otimes(T_t-E_a)\otimes T_t^{\otimes(N-j)}.
\end{align}
Both $T_t$ and $E_a$ are channels and hence have diamond norm $1$. The triangle inequality and multiplicativity of the diamond norm thus give
\begin{align}
    \|T_t^{\otimes N}(\zeta_a)-\eta_a^{\otimes N}\|_1&\overset{(1)}{\le}\|T_t^{\otimes N}-E_a^{\otimes N}\|_\diamond,\notag\\
    &\overset{(2)}{\le}\sum_{j=1}^N\|T_t-E_a\|_\diamond,\notag\\
    &\overset{(3)}{\le}CNe^{-ct},\label{rows-eq:top-product}
\end{align}
for every $a$-qubit state $\zeta_a$, where $(1)$ uses $E_a^{\otimes N}(\zeta_a)=\eta_a^{\otimes N}$ and the definition of the diamond norm, $(2)$ uses the triangle inequality applied to the telescoping identity together with $\|T_t\|_\diamond=\|E_a\|_\diamond=1$, and $(3)$ uses \cref{row:Tt-Etop}.

Together, \cref{rows-eq:projection-applied,rows-eq:top-product} establish the second step of \cref{rows-eq:upper}, and hence prove the stated constant-decay-rate rapid-mixing estimate for $\LL$.

It remains to verify stationarity and transfer the bound to $\KK$. The stationary $a$-qubit state and the $b$-vacuum restriction show that $\LL(\sigma_\LL)=0$, and convergence in \cref{rows-eq:upper} proves uniqueness. Put $U:=Z_a^{\otimes N}\otimes\BI_b$ and $\mathcal U(B):=UBU^\dagger$. We have $\KK=\mathcal U\LL\mathcal U^{-1}$ because $ZL_aZ=K_a$ and $Z\ketbra{1}{1}Z=\ketbra{1}{1}$, while $U$ acts trivially on the $b$ qubits. Its unique stationary state is therefore $\sigma_\KK=\mathcal U(\sigma_\LL)$, and unitary invariance of trace norm gives the same bound \cref{rows-eq:upper} for $\KK$. Thus both generators are rapidly mixing with the same constant decay rate.

\emph{Stationary space and infinite-time limit of the sum.} Write $s_i=\ketbra{0}{1}_{a_i}$, $n_i^a=\ketbra{1}{1}_{a_i}$, and $n_i^b=\ketbra{1}{1}_{b_i}$. Expanding the two top dissipators gives
\begin{align}
    \mathcal D[Z+2s]+\mathcal D[Z-2s]&=2\mathcal D[Z]+8\mathcal D[s],\notag\\
    (\LL+\KK)^*(n_i^a)&=-8n_i^a.
\end{align}
Thus every stationary state has zero $a$-qubit occupation and, by positivity, is supported on the $a$-vacuum. On this sector the cross dissipators vanish. The total $b$-qubit occupation then satisfies
\begin{align}
    (\LL+\KK)^*\left(\sum_{i=1}^N n_i^b\right) &=-2\sum_{i=1}^{N}(\BI-n_i^b)n_{i+1}^b,
\end{align}
where this identity is restricted to the $a$-vacuum sector and $n_{N+1}^b=n_1^b$. Each summand on the right is a positive projector onto a cyclic neighboring $01$ pattern. A stationary state must therefore be supported on cyclic strings without $01$. Any nonconstant cyclic binary string contains both a transition from $0$ to $1$ and a transition from $1$ to $0$, so the only such strings are $|0^N\rangle_b$ and $|1^N\rangle_b$, precisely the strings spanning $\mathcal S_b$. Conversely, each bottom jump annihilates $\mathcal S_b$, so both its jump term and its anticommutator vanish on every operator in $\BB(\mathcal S_b)$.

To establish convergence on all operators, order the computational matrix units $|x\rangle\langle y|$ by the total number of ones in $x$ and $y$. Each lowering jump lowers this number when its jump term is nonzero; the anticommutator terms and the $a$-qubit $Z$ dissipators are diagonal in this basis. The generator is therefore triangular with real nonpositive diagonal. Its zero diagonal entries are exactly the matrix units supported on the $a$-vacuum and $\mathcal S_b$ in both ket and bra. The semigroup is bounded in induced trace norm, so zero has no nontrivial Jordan block: such a block would produce polynomial growth on its generalized eigenspace. All other eigenvalues are strictly negative. This proves the asserted kernel, and $P_{\LL+\KK}$ exists by \cref{lem:infinite-time-limit}.

\emph{Mixing-time lower bound for the sum.} We prove the lower bound on the mixing time of $\LL+\KK$ by finding an initial state that takes $\Omega(N)$ time to mix. For $0\le j\le N-1$, write $\lvert d_j\rangle_b:=\lvert0^{N-j}1^j\rangle_b$ and define
\begin{align}
    \omega_j&:=\ketbra{0^N}{0^N}_a\otimes\ketbra{d_j}{d_j}_b.
\end{align}
We take $\rho_{\mathrm{in}}:=\omega_{N-1}$. Since $L_a=Z+2\ketbra{0}{1}$ and $K_a=Z-2\ketbra{0}{1}$, the mixed terms in their dissipators cancel, and the remaining top dissipators vanish on $\ketbra{0}{0}$. Every cross jump contains the $a$-qubit projector $\ketbra{1}{1}$, which annihilates $\lvert0\rangle$, so both the jump and anticommutator terms of its dissipator vanish on the $a$-vacuum. Finally, for $1\leq j\leq N-1$, the cyclic string $d_j=0^{N-j}1^j$ has a unique neighboring pattern $01$, at sites $(N-j,N-j+1)$; the wraparound pair is $10$. Thus only $L_{b,N-j}$ acts. Since each bottom dissipator appears once in $\LL$ and once in $\KK$, we obtain
\begin{align}
    (\LL+\KK)(\omega_j)&=2(\omega_{j-1}-\omega_j)\qquad(1\le j\le N-1),\notag\\
    (\LL+\KK)(\omega_0)&=0.\label{rows-eq:dissip-omega-action}
\end{align}
Thus the span of the $\omega_j$ is invariant and the evolution on it is classical. In particular,
\begin{align}
    e^{t(\LL+\KK)}(\rho_{\mathrm{in}})&=\sum_{j=0}^{N-1}q_j(t)\omega_j,
\end{align}
where $q_{N-1}(0)=1$, $q_j(0)=0$ for $j<N-1$, and the coefficients satisfy the system of differential equations
\begin{align}
    q_{N-1}'(t)&=-2q_{N-1}(t),\notag\\
    q_j'(t)&=2q_{j+1}(t)-2q_j(t)\qquad(1\le j<N-1),\notag\\
    q_0'(t)&=2q_1(t).\label{rows-eq:bottom-ode}
\end{align}
The first equation gives $q_{N-1}(t)\to0$, and the middle equations then give $q_j(t)\to0$ successively for every $j\ge1$. Since $\sum_jq_j(t)=1$, we have $q_0(t)\to1$. Thus this trajectory converges to $\sigma_{\LL+\KK}:=\omega_0$, and the existence of $P_{\LL+\KK}$ just proved implies $P_{\LL+\KK}(\rho_{\mathrm{in}})=\sigma_{\LL+\KK}=\omega_0$.

With the $b$-vacuum projector $P:=\BI_a\otimes\ketbra{0^N}{0^N}_b$ defined above, let $p_t:=\Tr((\BI-P)e^{t(\LL+\KK)}(\rho_{\mathrm{in}}))$ be the weight outside that subspace. Since $\omega_0$ has all $b$ qubits in state $0$ and every $\omega_j$ with $j\ge1$ is supported outside this subspace,
\begin{align}
    p_t&=\sum_{j=1}^{N-1}q_j(t)=1-q_0(t).\label{rows-eq:front-weight}
\end{align}
Substituting the expressions for the evolved state and its limit gives
\begin{align}
    \|e^{t(\LL+\KK)}(\rho_{\mathrm{in}})-P_{\LL+\KK}(\rho_{\mathrm{in}})\|_1&\overset{(1)}{=}\left\|\sum_{j=0}^{N-1}q_j(t)\omega_j-\omega_0\right\|_1,\notag\\
    &\overset{(2)}{=}\left\|(q_0(t)-1)\omega_0+\sum_{j=1}^{N-1}q_j(t)\omega_j\right\|_1,\notag\\
    &\overset{(3)}{=}|q_0(t)-1|+\sum_{j=1}^{N-1}q_j(t),\notag\\
    &\overset{(4)}{=}2p_t,\label{rows-eq:front-identity}
\end{align}
where $(1)$ substitutes $e^{t(\LL+\KK)}(\rho_{\mathrm{in}})=\sum_jq_j(t)\omega_j$ and $P_{\LL+\KK}(\rho_{\mathrm{in}})=\omega_0$, $(2)$ separates the $j=0$ term, $(3)$ uses that the $\omega_j$ have mutually orthogonal supports, so the operator is block diagonal and its trace norm is the sum of the trace norms of its blocks, and $(4)$ uses \cref{rows-eq:front-weight}.

It remains to lower-bound $p_t$. Define the weighted mass
\begin{align}
    a(t)&:=\sum_{j=0}^{N-1}(N-1-j)q_j(t).
\end{align}
Differentiating and using \cref{rows-eq:bottom-ode} gives
\begin{align}
    a'(t)&\overset{(1)}{=}2(N-1)q_1(t)+2\sum_{j=1}^{N-2}(N-1-j)\bigl(q_{j+1}(t)-q_j(t)\bigr),\notag\\
    &\overset{(2)}{=}2\sum_{j=1}^{N-1}(N-j)q_j(t)-2\sum_{j=1}^{N-2}(N-1-j)q_j(t),\notag\\
    &\overset{(3)}{=}2\sum_{j=1}^{N-2}\bigl((N-j)-(N-1-j)\bigr)q_j(t)+2q_{N-1}(t),\notag\\
    &\overset{(4)}{=}2\sum_{j=1}^{N-1}q_j(t)=2p_t\le2,\qquad a(0)=0,\label{rows-eq:weighted-mass-derivative}
\end{align}
where $(1)$ differentiates the definition of $a(t)$ and substitutes \cref{rows-eq:bottom-ode}; the term containing $q_{N-1}'(t)$ has coefficient $N-1-(N-1)=0$. In $(2)$ we shift the index in the incoming terms. In $(3)$ we combine the coefficients of each $q_j(t)$; for $1\le j\le N-2$, the coefficient is $2((N-j)-(N-1-j))=2$, while the endpoint $q_{N-1}(t)$ also has coefficient $2$. Finally, $(4)$ uses \cref{rows-eq:front-weight}.

Hence $a(t)\le2t$. Since every term in $a(t)$ is nonnegative and its first term is $(N-1)q_0(t)$,
\begin{align}
    q_0(t)&\le\frac{a(t)}{N-1}\le\frac{2t}{N-1},\notag\\
    p_t&=1-q_0(t)\ge1-\frac{2t}{N-1}.\label{rows-eq:front-probability}
\end{align}
Combining \cref{rows-eq:front-identity,rows-eq:front-probability} gives
\begin{align}
    \|e^{t(\LL+\KK)}(\rho_{\mathrm{in}})-P_{\LL+\KK}(\rho_{\mathrm{in}})\|_1&\ge2\left(1-\frac{2t}{N-1}\right).\label{rows-eq:front-distance}
\end{align}
For every $0\le t\le(N-1)/8$, this distance is at least $3/2$. The strict mixing-time definition therefore yields
\begin{align}
    \tau_{\operatorname{mix}}^{\LL+\KK}(\delta)&\ge\frac{N-1}{8}\qquad(N\ge2,\ 0<\delta<3/2).\notag
\end{align}
\end{proof}

\begin{example}[Restatement of \cref{ex:rows-ham}]\label{app:ex:rows-ham}
Let $\LL$ be the Lindbladian from Example~\ref{ex:rm_sum_dissip_only}. It has the unique fixed point $\sigma_\LL=\eta_a^{\otimes N}\otimes\ketbra{0^N}{0^N}_b$ and satisfies $\tau_{\operatorname{mix}}^\LL(\varepsilon)=\mathcal{O}(\log(N/\varepsilon))$.

Let $H_i=Y_{a_i}$ and define $(\KK_N)_N$ by $\KK_N=\KK=-i[\sum_{i=1}^NH_i,\cdot]$. Then $(\LL_N+\KK_N)_N$ is not rapidly mixing: for every fixed $0<\delta<1/2$, $\tau_{\operatorname{mix}}^{\LL+\KK}(\delta)=\Omega(N)$.
\end{example}

\begin{proof}[Proof of \cref{lem:rm_sum_hamiltonian}]
\prooflabel{lem:rm_sum_hamiltonian}{proof:rows-ham}
The stationary state and mixing bound for $\LL$ were proved in Example~\ref{ex:rm_sum_dissip_only}. First, the Hamiltonian cancels the action of the top dissipator on $P_0=\ketbra{0}{0}$:
\begin{align}
    \mathcal D[Z+2\ketbra{0}{1}](P_0)&=-X,\qquad -i[Y,P_0]=X.
\end{align}
For completeness, the full one-site top generator $\mathcal G=\mathcal D[Z+2\ketbra{0}{1}]-i[Y,\cdot]$ acts on the matrix units $F_{uv}=\ketbra{u}{v}$ as
\begin{align}
    \mathcal G(F_{00})&=0,\qquad \mathcal G(F_{01})=-4F_{01},\qquad \mathcal G(F_{10})=-4F_{10},\notag\\
    \mathcal G(F_{11})&=4F_{00}-4F_{11}-4F_{01}-4F_{10}.
\end{align}
Order the computational matrix units first by their ket occupation pattern and then by their bra occupation pattern, with lower occupations first. The displayed top-site action preserves a matrix unit or sends it to earlier basis elements. For every bottom or cross jump $L$, the term $L\rho L^\dagger$ has the same property, while the anticommutator term is diagonal because the corresponding $L^\dagger L$ is diagonal in the computational basis. The full generator is therefore triangular in this ordered basis, with real nonpositive diagonal. A nontrivial Jordan block at eigenvalue zero would make $e^{t(\LL+\KK)}$ grow polynomially in $t$, contradicting the uniform boundedness of this quantum channel. Thus zero is semisimple, and $P_{\LL+\KK}$ exists by \cref{lem:infinite-time-limit}.

For the lower bound, we use the same initial state $\rho_{\mathrm{in}}=\omega_{N-1}$ and invariant family as in the preceding example. The cancellation above makes the total $a$-qubit generator vanish on the $\omega_j$, and the cross terms vanish because their $a$-qubit projector annihilates the vacuum. Only $L_{b,N-j}$ acts on $\omega_j$ when $j\ge1$, and here each bottom dissipator appears only once. Therefore
\begin{align}
    (\LL+\KK)(\omega_j)&=\omega_{j-1}-\omega_j\qquad(1\le j\le N-1),\notag\\
    (\LL+\KK)(\omega_0)&=0.\label{rows-eq:ham-bitstring-action}
\end{align}
Comparing \cref{rows-eq:ham-bitstring-action} with \cref{rows-eq:dissip-omega-action} shows that, on the invariant span of the $\omega_j$, the present generator is exactly one half of the generator in the preceding example. Its entire evolution is therefore the earlier evolution with $t$ replaced by $t/2$:
\begin{align}
    e^{t(\LL+\KK)}(\rho_{\mathrm{in}})&=\sum_{j=0}^{N-1}q_j(t/2)\omega_j.
\end{align}
Since $q_0(t/2)\to1$, the preceding identity gives $P_{\LL+\KK}(\rho_{\mathrm{in}})=\sigma_{\LL+\KK}:=\omega_0$. The dynamics are therefore exactly those of the preceding example with $t$ replaced by $t/2$. Hence the same initial-state distance remains at least $3/2$ for twice as long, and the mixing-time lower bound in the preceding example doubles:
\begin{align}
    \tau_{\operatorname{mix}}^{\LL+\KK}(\delta)&\ge\frac{N-1}{4}\qquad(N\ge2,\ 0<\delta<3/2).\notag
\end{align}
\end{proof}

\begin{example}[Restatement of \cref{ex:sparse-boundary-potential}]\label{ex:sparse-boundary-potential-app} Consider an open chain of $m\geq 2$ qubits with a hopping Hamiltonian and boundary dissipation:
\begin{align}
    H&=-J\sum_{j=1}^{m-1}(\sigma_j^+\sigma_{j+1}^-+\sigma_j^-\sigma_{j+1}^+),\notag\\
    \LL^{(0)}&=-i[H,\,\cdot\,]+\gamma\sum_{s\in\{1,m\}}\bigl(\Diss{\sigma_s^-}+\Diss{\sigma_s^+}\bigr),\label{eq:app-sparse-H0}
\end{align}
with $J>0$ and $\gamma=J/2$. The generator $\LL^{(0)}$ is primitive with fixed state $\sigma=2^{-m}\BI$ and has mixing time
\begin{align}
    \tau_{\operatorname{mix}}^{\LL^{(0)}}(\delta)\leq\frac{9(m+1)^9}{2J}\left[\left(\frac m2+1\right)\log2+\log\frac1\delta\right],\qquad 0<\delta<1.\label{eq:sparse-uniform-time-bound}
\end{align}
For $m=3\ell$ and $\eps>0$, let $n_j:=\sigma_j^+\sigma_j^-$ be the occupation operator at site $j$ and define
\begin{align}
    V&=\eps\left(\sum_{j=1}^{\ell}n_j+\sum_{j=2\ell+1}^{3\ell}n_j\right),\notag\\
    \LL^{(V)}&=\LL^{(0)}-i[V,\,\cdot\,].\label{eq:app-sparse-potential}
\end{align}
This generator is also primitive with fixed state $\sigma$. Whenever $\ell\geq\max\{2,4J/\eps\}$, it satisfies
\begin{align}
    \tau_{\operatorname{mix}}^{\LL^{(V)}}(1/2)\geq\frac1{16J}\left(1+\frac\eps{2J}\right)^{\ell-1}.\label{eq:app-sparse-main-lower}
\end{align}
Thus, for fixed $J,\eps>0$, the unperturbed mixing time is at most $\mathrm{poly}(m)$, while the perturbed mixing time is at least $\exp(\Omega(m))$. With physical chain length $m=3\lfloor a\log N/3\rfloor$ for fixed $a>0$ and auxiliary parameter $N$, these give a polylogarithmic upper bound and a polynomial lower bound in $N$, respectively, as in \cref{ex:sparse-boundary-potential}.
\end{example}

\begin{proof}[Proof of \cref{ex:sparse-boundary-potential-app}]
\prooflabel{ex:sparse-boundary-potential-app}{proof:sparse-boundary-estimates}
We first prove contraction over each time interval of length $2(m+1)^2/J$, which yields the polynomial mixing bound, and then show that the gap vanishes with the system size. The mechanism is that hopping transports information from every site to the boundary baths. For the perturbed generator, we construct a localized fermionic mode whose occupation changes exponentially slowly. In each part, we give the bound first and then prove the estimates it needs. Finally, we verify primitivity for the potential in the statement.

\emph{Polynomial upper bound.} Since we only have dissipation on the boundary, the degrees of freedom in the bulk need to be transported to the boundary via the hopping Hamiltonian before they can dissipate, so this suggests a timescale growing at least linearly with $m$. We quantify this by showing that after time $T=\mathrm{poly}(m)$, the quantum channel $e^{T\LL^{(0)}}$ is strictly contractive outside of the steady-state subspace in Hilbert--Schmidt norm. Let $T=2(m+1)^2/J$ and $S_t=e^{t(\LL^{(0)})^*}$. For every state $\rho$,
\begin{align}
    \|e^{t\LL^{(0)}}(\rho)-\sigma\|_1&\overset{(1)}{=}\sup_{\substack{O=O^\dagger\\
    \|O\|\leq1}}\left|\Tr\left[Oe^{t\LL^{(0)}}(\rho-\sigma)\right]\right|,\notag\\
    &\overset{(2)}{=}\sup_{\substack{O=O^\dagger\\
    \|O\|\leq1}}\left|\Tr\left[(S_tO)(\rho-\sigma)\right]\right|,\notag\\
    &\overset{(3)}{\leq}\sup_{\substack{O=O^\dagger\\
    \|O\|\leq1}}\left\|S_t\bigl(O-\Tr(O)\sigma\bigr)\right\|_2\|\rho-\sigma\|_2,\notag\\
    &\overset{(4)}{\leq}2^{m/2}\sup_{\substack{O=O^\dagger,\,\Tr O=0\\
    \|O\|_2\leq1}}\|S_tO\|_2,\notag\\
    &\overset{(5)}{\leq}2^{m/2+1}\exp\left[-\frac{2Jt}{25(m+1)^7}\right],\quad\text{for }t\geq T,\label{eq:sparse-trace-bound}
\end{align}
where $(1)$ uses trace-norm duality and stationarity of $\sigma$, $(2)$ uses that $S_t=e^{t(\LL^{(0)})^*}$ is the Hilbert--Schmidt adjoint of $e^{t\LL^{(0)}}$, $(3)$ uses $S_t(\BI)=\BI$ and $\Tr(\rho-\sigma)=0$ to subtract the scalar part of $O$, followed by the Cauchy--Schwarz inequality, $(4)$ uses $\|\rho-\sigma\|_2^2=\Tr(\rho^2)-2^{-m}\leq1$ and $\|O-\Tr(O)\sigma\|_2^2=\Tr(O^2)-2^{-m}|\Tr O|^2\leq2^m\|O\|^2$ to reduce to traceless Hermitian operators of Hilbert--Schmidt norm at most one, and $(5)$ applies \cref{eq:sparse-HS-decay}, which we prove below. Requiring the last expression in \cref{eq:sparse-trace-bound} to be at most $\delta$ gives \cref{eq:sparse-uniform-time-bound}, since $25(m+1)^7\leq9(m+1)^9$ for $m\geq2$ and the stated mixing-time bound is at least $T$. This also proves primitivity of $\LL^{(0)}$, because every state converges to $\sigma$.

We prove the bound on $\|S_tO\|_2$, \cref{eq:sparse-HS-decay}, for traceless $O$ and $t\geq T$ by showing that it is a strict contraction at time $T$. Then we pass from strict contraction over one interval to decay at every time. Setting $b=8/(25(m+1)^5)$ and writing $t=kT+s$, where $k=\lfloor t/T\rfloor$ and $0\leq s<T$, yields
\begin{align}
    \|S_tO\|_2&\overset{(1)}{\leq}(1-b)^{k/2}\|O\|_2,\notag\\
    &\overset{(2)}{\leq}e^{b/2}e^{-bt/(2T)}\|O\|_2,\notag\\
    &\overset{(3)}{\leq}2\exp\left[-\frac{2Jt}{25(m+1)^7}\right]\|O\|_2,\label{eq:sparse-HS-decay}
\end{align}
where $(1)$ iterates \cref{eq:sparse-window-target} and uses contractivity for the remaining time $s$, $(2)$ uses $1-b\leq e^{-b}$ and $k\geq t/T-1$, and $(3)$ substitutes $b,T$ and uses $e^{b/2}<2$.

We now prove \cref{eq:sparse-window-target}. Define the Hamiltonian commutator and the dissipative part of the adjoint generator by
\begin{align}
    \mathcal A(O)&=i[H,O],\notag\\
    \mathcal B&=\gamma\sum_{s\in\{1,m\}}\bigl(\mathcal D[\sigma_s^-]^*+\mathcal D[\sigma_s^+]^*\bigr).\label{eq:sparse-A-B}
\end{align}
Then $(\LL^{(0)})^*=\mathcal A+\mathcal B$. Throughout this proof, $\mathcal B^{1/2}:=i(-\mathcal B)^{1/2}$, where $(-\mathcal B)^{1/2}$ is the positive semidefinite square root; hence $(\mathcal B^{1/2})^2=\mathcal B$ and $\|\mathcal B^{1/2}X\|_2=\|(-\mathcal B)^{1/2}X\|_2$. The equal-rate sum of raising and lowering dissipators is self-adjoint:
\begin{align}
    \mathcal B(O)&=\frac\gamma2\sum_{s\in\{1,m\}}(X_sOX_s+Y_sOY_s-2O),\notag\\
    -\langle O,\mathcal BO\rangle&=\frac\gamma4\sum_{s\in\{1,m\}}\bigl(\|[X_s,O]\|_2^2+\|[Y_s,O]\|_2^2\bigr),\qquad \|\mathcal B\|_{2\to2}\leq2J.\label{eq:sparse-dirichlet}
\end{align}
We used $\gamma=J/2$. Thus $\mathcal B$ is self-adjoint and negative semidefinite, while $\mathcal A$ is skew-adjoint, and $S_t=e^{t(\mathcal A+\mathcal B)}$. Differentiate the squared Hilbert--Schmidt norm:
\begin{align}
    \frac d{dt}\|S_tO\|_2^2&\overset{(1)}{=}2\Rea\langle S_tO,(\mathcal A+\mathcal B)S_tO\rangle,\notag\\
    &\overset{(2)}{=}2\langle S_tO,\mathcal BS_tO\rangle,\notag\\
    &\overset{(3)}{=}-2\|\mathcal B^{1/2}S_tO\|_2^2,\label{eq:sparse-energy-differential}
\end{align}
where $(1)$ uses the product rule and $\frac d{dt}S_tO=(\mathcal A+\mathcal B)S_tO$, while $(2)$ uses $\mathcal A^*=-\mathcal A$, so its quadratic form has zero real part, and $\mathcal B^*=\mathcal B$, and $(3)$ uses $-\mathcal B\succeq0$ and $\|\mathcal B^{1/2}X\|_2=\|(-\mathcal B)^{1/2}X\|_2$. Integrating from $0$ to $T$, with $S_0O=O$, gives
\begin{align}
    \|S_TO\|_2^2=\|O\|_2^2-2\int_0^T\|\mathcal B^{1/2}S_tO\|_2^2\,dt.\label{eq:sparse-energy-identity}
\end{align}
The derivative in \cref{eq:sparse-energy-differential} is nonpositive, so $S_t$ is a contraction. To establish strict contraction we will compare the evolution $S_t$ under $(\LL^{(0)})^*$ to the evolution $\mathcal U_t=e^{t\mathcal A}$ under the Hamiltonian only. Since $\LL^{(0)}(\BI)=0$ and $i[H,\BI]=0$, both $S_t,\mathcal U_t$ are trace-preserving, so they keep traceless $O$ in the traceless subspace. We now use \cref{eq:sparse-energy-identity} to show that every traceless operator loses a positive fraction, depending on $m$, of its Hilbert--Schmidt norm over each interval of length $T$:
\begin{align}
    \|S_TO\|_2^2&\overset{(1)}{\leq}\|O\|_2^2-\frac2{(1+2JT)^2}\int_0^T\|\mathcal B^{1/2}\mathcal U_tO\|_2^2\,dt,\notag\\
    &\overset{(2)}{\leq}\left(1-\frac8{(m+1)(1+4(m+1)^2)^2}\right)\|O\|_2^2,\notag\\
    &\overset{(3)}{\leq}\left(1-\frac8{25(m+1)^5}\right)\|O\|_2^2,\label{eq:sparse-window-target}
\end{align}
where $(1)$ is \cref{eq:sparse-observation-comparison}, $(2)$ is \cref{eq:sparse-observability} with $2JT=4(m+1)^2$, and $(3)$ uses $1+4(m+1)^2\leq5(m+1)^2$. We now establish these two estimates:
\begin{align}
    \int_0^T\|\mathcal B^{1/2}S_tO\|_2^2\,dt&\geq\frac1{(1+2JT)^2}\int_0^T\|\mathcal B^{1/2}\mathcal U_tO\|_2^2\,dt,\label{eq:sparse-observation-comparison}\\
    \int_0^T\|\mathcal B^{1/2}\mathcal U_tO\|_2^2\,dt&\geq\frac4{m+1}\|O\|_2^2,\qquad \Tr O=0.\label{eq:sparse-observability}
\end{align}
The first compares the dissipative trajectory with free evolution. The second quantifies how transport lets the boundary dissipation reach fluctuations at every site. That is, any initial nonzero traceless observable $O$, during evolution under the Hamiltonian over $[0,T]$, develops a nonzero component in the subspace that is not fixed by the dissipation, i.e. $(\mathrm{ker}\mathcal{B})^{\perp}$, where $\mathcal B$ has strictly negative spectrum, so $e^{t\mathcal B}$ strictly contracts this component for $t>0$. For example, if $O$ is initially supported only in the bulk, then $\mathcal{B}(O)=\mathcal B(\BI)O=0$ and it can't be dissipated at $t=0$, but \cref{eq:sparse-observability} guarantees a positive integrated dissipative form over $[0,T]$.

We begin by proving \cref{eq:sparse-observation-comparison}. Duhamel's identity gives
\begin{align}
    &\mathcal U_tO=S_tO-\int_0^t\mathcal U_{t-s}\mathcal BS_sO\,ds,\notag\\
    &\left(\int_0^T\|\mathcal B^{1/2}S_tO\|_2^2\,dt\right)^{1/2}\geq\left(\int_0^T\|\mathcal B^{1/2}\mathcal U_tO\|_2^2\,dt\right)^{1/2}-\left(\int_0^T\Big\|\mathcal B^{1/2}\int_0^t\mathcal U_{t-s}\mathcal BS_sO\,ds\Big\|_2^2dt\right)^{1/2},\label{eq:sparse-energy-duhamel}
\end{align}
where the second line applies $\mathcal B^{1/2}$ and the triangle inequality in $L^2([0,T])$.

We upper bound the squared second term:
\begin{align}
    &\int_0^T\left\|\mathcal B^{1/2}\int_0^t\mathcal U_{t-s}\mathcal BS_sO\,ds\right\|_2^2dt,\notag\\
    &\overset{(1)}{\leq}\int_0^T\Big(\int_0^t\|\mathcal B^{1/2}\mathcal U_{t-s}\mathcal B^{1/2}\|_{2\to2}^2ds\Big)\Big(\int_0^t\|\mathcal B^{1/2}S_sO\|_2^2\,ds\Big)\,dt,\notag\\
    &\overset{(2)}{\leq}(2J)^2\int_0^Tt\int_0^t\|\mathcal B^{1/2}S_sO\|_2^2\,ds\,dt,\notag\\
    &\overset{(3)}{\leq}(2JT)^2\int_0^T\|\mathcal B^{1/2}S_sO\|_2^2\,ds,\label{eq:sparse-duhamel-remainder}
\end{align}
where $(1)$ uses $\mathcal B=\mathcal B^{1/2}\mathcal B^{1/2}$ and Cauchy--Schwarz in $s$, $(2)$ uses submultiplicativity, $\|\mathcal U_{t-s}\|_{2\to2}=1$, and $\|\mathcal B^{1/2}\|_{2\to2}^2=\|\mathcal B\|_{2\to2}\leq2J$, and $(3)$ bounds $t$ and the length of the outer integral by $T$.

Taking square roots in \cref{eq:sparse-duhamel-remainder} and substituting into the inequality in \cref{eq:sparse-energy-duhamel}, then moving the resulting term to the left, gives
\begin{align}
    (1+2JT)\left(\int_0^T\|\mathcal B^{1/2}S_tO\|_2^2\,dt\right)^{1/2}\geq\left(\int_0^T\|\mathcal B^{1/2}\mathcal U_tO\|_2^2\,dt\right)^{1/2}.
\end{align}
Squaring and rearranging gives
\begin{align}
    \int_0^T\|\mathcal B^{1/2}S_tO\|_2^2\,dt\geq\frac1{(1+2JT)^2}\int_0^T\|\mathcal B^{1/2}\mathcal U_tO\|_2^2\,dt,\label{eq:sparse-comparison-proof}
\end{align}
which proves \cref{eq:sparse-observation-comparison}.

We now prove \cref{eq:sparse-observability}. For $1\leq j\leq m$, the Jordan--Wigner transformation defines the fermionic annihilation operators
\begin{align}
    f_j=\left(\prod_{r<j}(-Z_r)\right)\sigma_j^-,\qquad \{f_j,f_k\}=0,\qquad \{f_j,f_k^\dagger\}=\delta_{jk}\BI.\label{eq:sparse-fermions}
\end{align}
The anticommutation relations follow by moving the earlier site through the parity string, using $(-Z_j)\sigma_j^-=\sigma_j^-$ and $\sigma_j^-(-Z_j)=-\sigma_j^-$. The strings in a quadratic operator cancel outside the interval between its two indices:
\begin{align}
    f_j^\dagger f_k=\sigma_j^+\left(\prod_{r=j+1}^{k-1}(-Z_r)\right)\sigma_k^-\quad(j<k),\qquad f_j^\dagger f_j=n_j.\label{eq:sparse-string-support}
\end{align}
In particular, the hopping Hamiltonian is
\begin{align}
    H=\sum_{j,k=1}^m\mathbf h_{jk}f_j^\dagger f_k,\qquad \mathbf h_{jk}=-J(\delta_{j,k+1}+\delta_{j+1,k}).\label{eq:sparse-one-particle-H}
\end{align}
Here $\mathbf h$ acts on the $m$-dimensional single-particle space; $H$ acts on the full $2^m$-dimensional Hilbert space.

We use the Majorana notation of Appendix~\ref{app:quasifree-fermionic}, specifically \cref{eq:qf-parity-dressing,eq:dk_commutators,eq:qf-expectation}. For $1\leq j\leq m$ and $1\leq i\leq2m$, set
\begin{align}
    c_{2j-1}&=f_j+f_j^\dagger,\qquad c_{2j}=i(f_j-f_j^\dagger),\notag\\
    \Pi&=(-i)^m c_1\cdots c_{2m},\qquad \widetilde c_i=i\Pi c_i,\notag\\
    \delta_iO&=\frac{i}{2}\Pi[c_i,O],\qquad \mathcal E_i(O)=\frac12(O+c_iOc_i).\label{eq:sparse-reset-definitions}
\end{align}
The first line gives Hermitian Majoranas satisfying $\{c_i,c_j\}=2\delta_{ij}\BI$. The last line is the commutator representation of the deletion quasiderivations on ordered products of the dressed Majoranas. As in the quasifree section, $\Pi$ and $\widetilde c_i$ are unitary, and
\begin{align}
    (\id-\mathcal E_i)(O)&=\widetilde c_i\delta_iO,\notag\\
    \|\delta_iO\|_2&=\frac12\|[c_i,O]\|_2,\qquad \mathcal E_1\cdots\mathcal E_{2m}(O)=\Tr(O)\sigma.\label{eq:sparse-quasiderivation-properties}
\end{align}
The norm identity holds in Hilbert--Schmidt norm by the same unitary-invariance argument as for the operator norm. These formulas apply to all operators, without a parity restriction.

To diagonalize $H$, define the real sine matrix and the normal-mode annihilation operators by
\begin{align}
    \mathbf S_{jk}&=\sqrt{\frac2{m+1}}\sin\frac{jk\pi}{m+1},\qquad b_k=\sum_{j=1}^m\mathbf S_{jk}f_j,\notag\\
    \omega_k&=-2J\cos\frac{k\pi}{m+1},\qquad 1\leq k\leq m.\label{eq:sparse-sine-modes}
\end{align}
Sine orthogonality gives $\sum_j\mathbf S_{jk}\mathbf S_{jl}=\delta_{kl}$, so $\mathbf S^T\mathbf S=\BI$ and $f_j=\sum_k\mathbf S_{jk}b_k$. Extend the sine formula to the zero boundary values $\mathbf S_{0k}=\mathbf S_{m+1,k}=0$. Then
\begin{align}
    H&=\sum_{k,l=1}^m(\mathbf S^T\mathbf h\mathbf S)_{kl}b_k^\dagger b_l=\sum_{k=1}^m\omega_k b_k^\dagger b_k.\label{eq:sparse-diagonal-H}
\end{align}
The anticommutation relations give $[b_l^\dagger b_l,b_k]=-\delta_{lk}b_k$. Thus
\begin{align}
    \mathcal A(b_k)&=i[H,b_k]=-i\omega_kb_k,\notag\\
    \mathcal U_t(b_k)&=e^{-i\omega_kt}b_k,\qquad \mathcal U_t(b_k^\dagger)=e^{i\omega_kt}b_k^\dagger.\label{eq:sparse-mode-evolution}
\end{align}
Write $u_k=\mathbf S_{1k}$. Since $f_1=\sum_ku_kb_k$, unitary conjugation gives
\begin{align}
    \|[f_1,\mathcal U_tO]\|_2&\overset{(1)}{=}\|[\mathcal U_{-t}(f_1),O]\|_2\overset{(2)}{=}\left\|\sum_ku_ke^{i\omega_kt}[b_k,O]\right\|_2,\notag\\
    \sum_k\bigl(\|[b_k,O]\|_2^2+\|[b_k^\dagger,O]\|_2^2\bigr)&\overset{(3)}{=}\sum_j\bigl(\|[f_j,O]\|_2^2+\|[f_j^\dagger,O]\|_2^2\bigr),\notag\\
    &\overset{(4)}{=}\frac12\sum_{i=1}^{2m}\|[c_i,O]\|_2^2\overset{(5)}{=}2\sum_{i=1}^{2m}\|\delta_iO\|_2^2,\label{eq:sparse-mode-reset-norms}
\end{align}
where $(1)$ is unitary invariance and $(2)$ uses \cref{eq:sparse-mode-evolution}. For $(3)$, since $[b_k,O]=\sum_j\mathbf S_{jk}[f_j,O]$ and $\mathbf S^T\mathbf S=\BI$, the orthogonal transformation preserves the sum of the squared Hilbert--Schmidt norms. The same argument applies to the adjoint commutators. Step $(4)$ is the parallelogram identity with $f_j=(c_{2j-1}-ic_{2j})/2$, and $(5)$ is \cref{eq:sparse-quasiderivation-properties}. The first identity also holds for $f_1^\dagger$ and $b_k^\dagger$, with the frequency signs reversed.

Using $f_1=\sigma_1^-$, $X_1=f_1+f_1^\dagger$ and $Y_1=i(f_1-f_1^\dagger)$ we can now obtain
\begin{align}
    \int_0^T\|\mathcal B^{1/2}\mathcal U_tO\|_2^2\,dt &\overset{(1)}{\geq}\frac\gamma4\int_0^T\bigl(\|[X_1,\mathcal U_tO]\|_2^2+\|[Y_1,\mathcal U_tO]\|_2^2\bigr)dt,\notag\\
    &\overset{(2)}{=}\frac\gamma2\int_0^T\bigl(\|[f_1,\mathcal U_tO]\|_2^2+\|[f_1^\dagger,\mathcal U_tO]\|_2^2\bigr)dt,\notag\\
    &\overset{(3)}{\geq}\frac{2\gamma T}{(m+1)^3}\sum_k\bigl(\|[b_k,O]\|_2^2+\|[b_k^\dagger,O]\|_2^2\bigr),\notag\\
    &\overset{(4)}{=}\frac{4\gamma T}{(m+1)^3}\sum_{i=1}^{2m}\|\delta_iO\|_2^2,\notag\\
    &\overset{(5)}{\geq}\frac{4\gamma T}{(m+1)^3}\|O-\Tr(O)\sigma\|_2^2,\notag\\
    &\overset{(6)}{=}\frac4{m+1}\|O-\Tr(O)\sigma\|_2^2,\label{eq:sparse-observation-proof}
\end{align}
where $(1)$ writes $\|\mathcal B^{1/2}\mathcal U_tO\|_2^2=-\langle\mathcal U_tO,\mathcal B\mathcal U_tO\rangle$ using \cref{eq:sparse-dirichlet} and discards the nonnegative contribution from the bath at site $m$, $(2)$ uses $X_1=f_1+f_1^\dagger$, $Y_1=i(f_1-f_1^\dagger)$, and the parallelogram identity, $(3)$ applies \cref{eq:sparse-mode-integral-target} with $W_k=[b_k,O]$ and applies its reversed-frequency version with $W_k=[b_k^\dagger,O]$, $(4)$ uses \cref{eq:sparse-mode-reset-norms}, $(5)$ uses \cref{eq:sparse-HS-telescope}, and $(6)$ substitutes $\gamma=J/2$ and $T=2(m+1)^2/J$. For traceless $O$, this is \cref{eq:sparse-observability}. Both baths remain in the actual evolution. We now prove the two remaining bounds.

The maps $\mathcal E_i$ are commuting orthogonal projections on Hilbert--Schmidt space: the signs from anticommuting $c_i,c_j$ cancel in conjugation. Set $P_0=\id$ and $P_i=\prod_{j=1}^i\mathcal E_j$. By \cref{eq:sparse-quasiderivation-properties}, $P_{2m}$ is the projection onto scalar operators. Therefore
\begin{align}
    \|O-\Tr(O)\sigma\|_2^2 &\overset{(1)}{=}\langle O,(\id-P_{2m})O\rangle,\notag\\
    &\overset{(2)}{=}\sum_{i=1}^{2m}\langle O,P_{i-1}(\id-\mathcal E_i)O\rangle,\notag\\
    &\overset{(3)}{\leq}\sum_{i=1}^{2m}\left|\langle O,P_{i-1}(\id-\mathcal E_i)O\rangle\right|,\notag\\
    &\overset{(4)}{=}\sum_{i=1}^{2m}\left|\langle(\id-\mathcal E_i)O,P_{i-1}(\id-\mathcal E_i)O\rangle\right|,\notag\\
    &\overset{(5)}{\leq}\sum_{i=1}^{2m}\|(\id-\mathcal E_i)O\|_2^2,\notag\\
    &\overset{(6)}{=}\sum_{i=1}^{2m}\|\widetilde c_i\delta_iO\|_2^2 =\sum_{i=1}^{2m}\|\delta_iO\|_2^2,\label{eq:sparse-HS-telescope}
\end{align}
where $(1)$ uses that $\id-P_{2m}$ is an orthogonal projection, $(2)$ is the telescoping identity $\id-P_{2m}=\sum_iP_{i-1}(\id-\mathcal E_i)$, and $(3)$ upper bounds each summand by its absolute value. For $(4)$, set $Q_i=\id-\mathcal E_i$. Since $Q_i^*=Q_i$, $Q_i^2=Q_i$, and $Q_iP_{i-1}=P_{i-1}Q_i$, we have
\begin{align*}
    P_{i-1}Q_i=Q_iP_{i-1}Q_i, \qquad \langle O,P_{i-1}Q_iO\rangle =\langle Q_iO,P_{i-1}Q_iO\rangle.
\end{align*}
For $(5)$, Cauchy--Schwarz and the contractivity of the orthogonal projection $P_{i-1}$ give
\begin{align*}
    \left|\langle Q_iO,P_{i-1}Q_iO\rangle\right| \leq\|Q_iO\|_2\|P_{i-1}Q_iO\|_2 \leq\|Q_iO\|_2^2.
\end{align*}
Finally, $(6)$ uses \cref{eq:sparse-quasiderivation-properties} and unitarity of $\widetilde c_i$.

To prove \cref{eq:sparse-mode-integral-target}, let $(E_\alpha)_\alpha$ be an orthonormal basis of $\mathbb C^{2^m\times 2^m}$ with respect to the Hilbert--Schmidt inner product, let $W_1,\ldots,W_m\in\mathbb C^{2^m\times 2^m}$ be arbitrary matrices, and write $w_{\alpha k}=\langle E_\alpha,W_k\rangle$. Parseval's identity gives
\begin{align}
    \int_0^T\left\|\sum_ku_ke^{i\omega_kt}W_k\right\|_2^2\,dt =\sum_\alpha\int_0^T\left|\sum_ku_kw_{\alpha k}e^{i\omega_kt}\right|^2\,dt.\label{eq:sparse-mode-coordinate-parseval}
\end{align}

Define $\Delta=\min_{k\ne l}|\omega_k-\omega_l|$. For each $\alpha$, apply \cref{eq:sparse-ingham} with
\begin{align}
    N&=0,\qquad N'=m,\qquad \lambda_0=\omega_1-\Delta,\qquad \lambda_k=\omega_k\quad(1\leq k\leq m),\notag\\
    a_k&=u_kw_{\alpha k}e^{i\omega_kT/2},\qquad \gamma_{\mathrm I}=\Delta,\qquad T_{\mathrm I}=\frac T2,\qquad \varepsilon_{\mathrm I}=\frac{\Delta T}{2}-\pi.\label{eq:sparse-ingham-substitution}
\end{align}
The separation condition holds by the definition of $\Delta$, including $\lambda_1-\lambda_0=\Delta$. Moreover, \cref{eq:sparse-frequency-separation}, proved below, gives $T\Delta\geq24>2\pi$, so $\varepsilon_{\mathrm I}>0$ and $T_{\mathrm I}=(\pi+\varepsilon_{\mathrm I})/\gamma_{\mathrm I}$. Set
\begin{align}
    f_\alpha(s)=\sum_{k=1}^m a_ke^{-i\lambda_ks} =\sum_{k=1}^m u_kw_{\alpha k}e^{i\omega_k(T/2-s)}.
\end{align}
For these parameters, Ingham's constant in \cref{lem:sparse-ingham} below becomes
\begin{align}
    C_{\mathrm{Ing}}(\varepsilon_{\mathrm I}) &=\frac{\pi(\Delta T/2)^2}{2(\Delta T/2-\pi)(\Delta T/2+\pi)} =\frac{\pi}{2\left(1-4\pi^2/(T^2\Delta^2)\right)}.\label{eq:sparse-ingham-constant-substitution}
\end{align}
We now obtain
\begin{align}
    \int_0^T\left|\sum_ku_kw_{\alpha k}e^{i\omega_kt}\right|^2\,dt &\overset{(1)}{=}\int_{-T/2}^{T/2}|f_\alpha(s)|^2\,ds,\notag\\
    &\overset{(2)}{\geq}\frac{T}{C_{\mathrm{Ing}}(\varepsilon_{\mathrm I})}\sum_k|a_k|^2,\notag\\
    &\overset{(3)}{=}\frac{2T}{\pi}\left(1-\frac{4\pi^2}{T^2\Delta^2}\right) \sum_ku_k^2|w_{\alpha k}|^2,\label{eq:sparse-coordinate-ingham-bound}
\end{align}
where $(1)$ is the change of variables $t=T/2-s$, $(2)$ applies Ingham's scalar inequality (\cref{lem:sparse-ingham}), and $(3)$ uses $|a_k|=u_k|w_{\alpha k}|$ and \cref{eq:sparse-ingham-constant-substitution}.

Summing \cref{eq:sparse-coordinate-ingham-bound} over $\alpha$ and using $\sum_\alpha|w_{\alpha k}|^2=\|W_k\|_2^2$ in \cref{eq:sparse-mode-coordinate-parseval} gives
\begin{align}
    \int_0^T\left\|\sum_ku_ke^{i\omega_kt}W_k\right\|_2^2\,dt &\overset{(1)}{\geq}\frac{2T}{\pi}\left(1-\frac{4\pi^2}{T^2\Delta^2}\right)\sum_ku_k^2\|W_k\|_2^2,\notag\\
    &\overset{(2)}{\geq}\frac T2\sum_ku_k^2\|W_k\|_2^2,\notag\\
    &\overset{(3)}{\geq}\frac{4T}{(m+1)^3}\sum_k\|W_k\|_2^2,\label{eq:sparse-mode-integral-target}
\end{align}
where $(1)$ sums \cref{eq:sparse-coordinate-ingham-bound} over $\alpha$ and uses $\sum_\alpha|w_{\alpha k}|^2=\|W_k\|_2^2$ in \cref{eq:sparse-mode-coordinate-parseval}, $(2)$ uses $T\Delta\geq24$ and $\frac2\pi(1-\pi^2/144)>1/2$, and $(3)$ uses \cref{eq:sparse-boundary-weights}, proved below. Reversing all frequency signs and reordering the modes leaves $\Delta$ and the conclusion unchanged, giving the reversed-frequency version of \cref{eq:sparse-mode-integral-target}.

\Needspace{18\baselineskip}
\begin{lemma}[Ingham's scalar inequality, {\cite[Theorem~1]{ingham1936some}}]\label{lem:sparse-ingham}
Let $N<N'$ be integers, let $a_n\in\mathbb C$ for $N<n\leq N'$, and let $\lambda_N,\ldots,\lambda_{N'}$ be real frequencies separated as below, with $\gamma>0$. For $\varepsilon>0$,
\begin{align}
    f(t)&=\sum_{n=N+1}^{N'}a_ne^{-i\lambda_n t},\qquad \lambda_n-\lambda_{n-1}\geq\gamma>0\quad(N<n\leq N'),\notag\\
    T&=\frac{\pi+\varepsilon}{\gamma}>\frac\pi\gamma,\notag\\
    \sum_n|a_n|^2&\leq\frac{C_{\mathrm{Ing}}(\varepsilon)}{2T}\int_{-T}^T|f(t)|^2\,dt,\qquad C_{\mathrm{Ing}}(\varepsilon)=\frac{\pi(\pi+\varepsilon)^2}{2\varepsilon(2\pi+\varepsilon)}.\label{eq:sparse-ingham}
\end{align}
\end{lemma}

We now prove \cref{eq:sparse-boundary-weights,eq:sparse-frequency-separation}. For $1\leq k\leq m$,
\begin{align}
    \sin\frac{k\pi}{m+1}&\overset{(1)}{=}\sin\frac{\min\{k,m+1-k\}\pi}{m+1},\notag\\
    &\overset{(2)}{\geq}\frac{2\min\{k,m+1-k\}}{m+1}\overset{(3)}{\geq}\frac2{m+1},\notag\\
    u_k^2&=\frac2{m+1}\sin^2\frac{k\pi}{m+1}\geq\frac8{(m+1)^3},\label{eq:sparse-boundary-weights}
\end{align}
where $(1)$ uses $\sin(\pi-x)=\sin x$, $(2)$ uses $\sin x\geq2x/\pi$ for $0\leq x\leq\pi/2$, and $(3)$ uses $\min\{k,m+1-k\}\geq1$.

The frequencies $\omega_k=-2J\cos(k\pi/(m+1))$ are strictly increasing, so their minimum separation occurs between consecutive indices. The cosine-difference identity gives
\begin{align}
    \Delta&\overset{(1)}{=}\min_{1\leq k<m}4J\sin\frac{(2k+1)\pi}{2(m+1)}\sin\frac{\pi}{2(m+1)},\notag\\
    &\overset{(2)}{=}4J\sin\frac{3\pi}{2(m+1)}\sin\frac{\pi}{2(m+1)}\overset{(3)}{\geq}\frac{12J}{(m+1)^2},\label{eq:sparse-frequency-separation}
\end{align}
where $(1)$ uses the ordering and the cosine-difference identity. For $(2)$, the first sine has its smallest value at $k=1$ or $k=m-1$, by symmetry about $\pi/2$. For $(3)$, both displayed arguments lie in $[0,\pi/2]$ because $m\geq2$, so $\sin x\geq2x/\pi$ applies to each factor. These bounds complete \cref{eq:sparse-mode-integral-target,eq:sparse-observation-proof} and establish the polynomial upper bound.

\emph{Vanishing gap.} For a generator $\mathcal G$, define $\operatorname{gap}(\mathcal G)=\min\{-\Rea z:z\in\spec(\mathcal G),\ z\ne0\}$. We show that
\begin{align}
    \operatorname{gap}(\LL^{(0)})\leq\frac{2J}{m}.\label{eq:app-sparse-gap-claim}
\end{align}
It suffices to find an invariant subspace whose average eigenvalue decay rate is at most $2J/m$. Consider the traceless quadratic operators
\begin{align}
    F_{jk}=f_j^\dagger f_k-\frac12\delta_{jk}\BI,\qquad 1\leq j,k\leq m.
\end{align}
They are linearly independent: restricting a vanishing linear combination first to the fermionic vacuum makes its scalar term vanish, and restricting next to the single-particle subspace makes its coefficient matrix vanish. Their span therefore has dimension $m^2$. The quadratic Hamiltonian preserves this span, since
\begin{align}
    i[H,F_{jk}]=i\sum_{l=1}^m\bigl(\mathbf h_{lj}F_{lk}-\mathbf h_{kl}F_{jl}\bigr).
\end{align}
The boundary dissipators preserve it as well:
\begin{align}
    \gamma\bigl(\mathcal D[\sigma_s^-]^*+\mathcal D[\sigma_s^+]^*\bigr)(F_{jk})=-\gamma(\delta_{sj}+\delta_{sk})F_{jk},\qquad s\in\{1,m\}.\label{eq:sparse-quadratic-damping}
\end{align}
For $j=k=s$, the one-qubit dissipator sends $n_s-\BI/2$ to $-2\gamma(n_s-\BI/2)$. For $j\ne k$, the support formula \cref{eq:sparse-string-support} shows that a boundary bath either damps an endpoint factor $\sigma_s^\pm$ at rate $\gamma$, or acts trivially.

Let $z_1,\ldots,z_{m^2}$ be the eigenvalues of this restriction, with algebraic multiplicity. They are nonzero because the traceless subspace decays by \cref{eq:sparse-HS-decay}. The Hamiltonian part has zero real trace, so
\begin{align}
    \operatorname{gap}(\LL^{(0)})&\overset{(1)}{\leq}\min_a(-\Rea z_a)\overset{(2)}{\leq}\frac1{m^2}\sum_{a=1}^{m^2}(-\Rea z_a),\notag\\
    &\overset{(3)}{=}\frac\gamma{m^2}\sum_{j,k=1}^m\sum_{s\in\{1,m\}}(\delta_{sj}+\delta_{sk})\overset{(4)}{=}\frac{4\gamma}{m}=\frac{2J}{m},
\end{align}
where $(1)$ uses invariance and equality of the real parts of the spectra of a map and its adjoint, $(2)$ bounds the minimum by the average, and $(4)$ counts the two endpoint contributions for each of the two indices. For $(3)$, the trace of the restriction is the sum of the coefficients of $F_{jk}$ in $(\LL^{(0)})^*(F_{jk})$. The displayed Hamiltonian action and \cref{eq:sparse-quadratic-damping} therefore give
\begin{align*}
    \sum_{a=1}^{m^2}z_a
    &=\sum_{j,k=1}^m\left[i(\mathbf h_{jj}-\mathbf h_{kk})-\gamma\sum_{s\in\{1,m\}}(\delta_{sj}+\delta_{sk})\right]
    =-\gamma\sum_{j,k=1}^m\sum_{s\in\{1,m\}}(\delta_{sj}+\delta_{sk}),
\end{align*}
since the two Hamiltonian sums cancel. Taking minus the real part and dividing by $m^2$ gives $(3)$. This proves \cref{eq:app-sparse-gap-claim}.

\emph{Exponential lower bound.} Let $m=3\ell$ and let $V$ be the potential in \cref{eq:app-sparse-potential}. Assume $\ell\geq\max\{2,4J/\eps\}$. We will construct a Hermitian observable $Q$ such that
\begin{align}
    Q^2=\BI,\qquad \Tr Q=0,\qquad \|(\LL^{(V)})^*(Q)\|\leq8J\left(1+\frac\eps{2J}\right)^{-(\ell-1)}.\label{eq:sparse-witness-target}
\end{align}
Assume these properties for now; they will be verified using \cref{eq:sparse-mode-observable,eq:sparse-witness-proof}. They make $\rho_0=2^{-m}(\BI+Q)$ a density matrix with $\Tr(Q\rho_0)=1$, while $\Tr(Q\sigma)=0$. For $\rho_t=e^{t\LL^{(V)}}(\rho_0)$, set $x(t)=\Tr(Q\rho_t)$. Since $x(0)=1$, we obtain
\begin{align}
    \|\rho_t-\sigma\|_1&\overset{(1)}{\geq}|\Tr(Q(\rho_t-\sigma))|=|x(t)|,\notag\\
    &\overset{(2)}{\geq}1-8Jt\left(1+\frac\eps{2J}\right)^{-(\ell-1)},\label{eq:sparse-trace-lower}
\end{align}
where $(1)$ is trace-norm duality with $\|Q\|=1$, and $(2)$ integrates the derivative bound in \cref{eq:sparse-witness-differential}, which we prove below, from $x(0)=1$. The last expression is greater than $1/2$ for every $t<(16J)^{-1}(1+\eps/(2J))^{\ell-1}$, proving \cref{eq:app-sparse-main-lower}.

To prove the derivative bound used in \cref{eq:sparse-trace-lower}, differentiate the expectation:
\begin{align}
    x'(t)&\overset{(1)}{=}\Tr\bigl((\LL^{(V)})^*(Q)\rho_t\bigr),\notag\\
    &\overset{(2)}{\geq}-\|(\LL^{(V)})^*(Q)\|,\notag\\
    &\overset{(3)}{\geq}-8J\left(1+\frac\eps{2J}\right)^{-(\ell-1)},\label{eq:sparse-witness-differential}
\end{align}
where $(1)$ is the master equation and adjointness, $(2)$ bounds the expectation of a Hermitian observable in a state by its operator norm, and $(3)$ is the last estimate in \cref{eq:sparse-witness-target}. Integrating from $0$ to $t$ and using $x(0)=1$ gives $x(t)\geq1-8Jt(1+\eps/(2J))^{-(\ell-1)}$; together with $|x(t)|\geq x(t)$, this yields the second inequality in \cref{eq:sparse-trace-lower}.

We now construct $Q$ by considering a localized fermionic mode. Let $|1\rangle,\ldots,|m\rangle$ be the standard basis of the single-particle space, and set
\begin{align}
    \mathbf h^{(V)}&=\mathbf h+\eps\left(\sum_{j=1}^{\ell}|j\rangle\langle j|+\sum_{j=2\ell+1}^{3\ell}|j\rangle\langle j|\right),\notag\\
    H+V&=\sum_{j,k=1}^m\mathbf h^{(V)}_{jk}f_j^\dagger f_k.
\end{align}
We will find a real normalized single-particle vector $|v\rangle=\sum_jv_j|j\rangle$ and a real number $E_0$ such that
\begin{align}
    v_1=v_m=0,\qquad \mathbf h^{(V)}|v\rangle=E_0|v\rangle+|r\rangle,\qquad |r\rangle=-Jv_2|1\rangle-Jv_{m-1}|m\rangle,\label{eq:sparse-mode-residual}\\
    |v_2|+|v_{m-1}|\leq2\left(1+\frac\eps{2J}\right)^{-(\ell-1)}.\label{eq:sparse-boundary-decay}
\end{align}
The zero endpoint coefficients will make this mode's occupation invisible to both baths, while the residual $|r\rangle$ will control its Hamiltonian derivative. Assume \cref{eq:sparse-mode-residual} and \cref{eq:sparse-boundary-decay} for now; they will follow from the construction below and \cref{eq:sparse-boundary-decay-proof}.

Define the mode annihilation operator, its occupation, and the witness by
\begin{align}
    d=\sum_{j=1}^m v_jf_j,\qquad n_d=d^\dagger d,\qquad Q=2n_d-\BI.\label{eq:sparse-mode-observable}
\end{align}
Normalization of $|v\rangle$ and the anticommutation relations imply $d^2=0$ and $dd^\dagger=\BI-n_d$. Consequently,
\begin{align}
    n_d^2&=d^\dagger(\BI-n_d)d=n_d,\notag\\
    \Tr n_d&=\Tr(dd^\dagger)=2^m-\Tr n_d.
\end{align}
Thus $\Tr n_d=2^{m-1}$, so $n_d$ is a rank-$2^{m-1}$ projection and $Q$ satisfies $Q^2=\BI$ and $\Tr Q=0$. The initial state above is $\rho_0=2^{1-m}n_d$, the maximally mixed state with this fermionic mode occupied. The same projection identity gives $\|d\|=1$.

For $r_j=\langle j|r\rangle$, write $d_r=\sum_jr_jf_j$. Applying the same argument to the normalized coefficient vector of $d_r$ gives $\|d_r\|=\||r\rangle\|_2$, with the zero vector treated trivially. Since $v_1=v_m=0$, the expansion $n_d=\sum_{j,k}v_jv_kf_j^\dagger f_k$ acts trivially on both endpoint sites by \cref{eq:sparse-string-support}. Each endpoint dissipator therefore annihilates $Q$. For the Hamiltonian part, the identity $[f_j^\dagger f_k,f_l]=-\delta_{jl}f_k$ gives
\begin{align}
    [H+V,d]&\overset{(1)}{=}-\sum_{j,k=1}^m\mathbf h^{(V)}_{jk}v_jf_k\overset{(2)}{=}-E_0d-d_r,\notag\\ [H+V,n_d]
    &\overset{(3)}{=}(E_0d^\dagger+d_r^\dagger)d-d^\dagger(E_0d+d_r)=d_r^\dagger d-d^\dagger d_r,\label{eq:sparse-mode-commutator}
\end{align}
where $(1)$ expands the quadratic Hamiltonian, $(2)$ uses real symmetry of $\mathbf h^{(V)}$ and \cref{eq:sparse-mode-residual}, and $(3)$ uses the product rule for the commutator and the adjoint of the first identity. It follows that
\begin{align}
    \|(\LL^{(V)})^*(Q)\|&\overset{(1)}{=}2\|d_r^\dagger d-d^\dagger d_r\|,\notag\\
    &\overset{(2)}{\leq}4\|d_r\|\|d\|=4\||r\rangle\|_2,\notag\\
    &\overset{(3)}{\leq}4J(|v_2|+|v_{m-1}|),\notag\\
    &\overset{(4)}{\leq}8J\left(1+\frac\eps{2J}\right)^{-(\ell-1)},\label{eq:sparse-witness-proof}
\end{align}
where $(1)$ uses the vanishing bath terms and \cref{eq:sparse-mode-commutator}, $(2)$ is the triangle inequality, $(3)$ uses the residual in \cref{eq:sparse-mode-residual}, and $(4)$ is \cref{eq:sparse-boundary-decay}. This proves \cref{eq:sparse-witness-target}, subject to constructing $|v\rangle$.

We now construct $|v\rangle$. Let $\mathbf h_D^{(V)}$ be the restriction of $\mathbf h^{(V)}$ to sites $2,\ldots,m-1$. Choose a real normalized ground-state vector $|v\rangle$ with energy $E_0$, and extend it by $v_1=v_m=0$. The eigenvalue equation holds at every interior site, while the endpoint rows contain only the adjacent hopping term because $v_1=v_m=0$:
\begin{align*}
    \langle j|\mathbf h^{(V)}|v\rangle&=E_0v_j\qquad(2\leq j\leq m-1),\\
    \langle1|\mathbf h^{(V)}|v\rangle&=-Jv_2,\qquad \langle m|\mathbf h^{(V)}|v\rangle=-Jv_{m-1}.
\end{align*}
Subtracting $E_0|v\rangle$ gives precisely the residual in \cref{eq:sparse-mode-residual}. We now prove the small boundary-amplitude estimate.

We may choose every interior coefficient $v_j$ positive. Replacing a real energy minimizer by its componentwise absolute values cannot increase its energy, because all nonzero off-diagonal entries are negative. The resulting nonnegative minimizer remains an eigenvector. If an interior coefficient vanished, its eigenvalue equation would force its neighbors to vanish, and then all coefficients would vanish.

The normalized trial vector $|\phi\rangle:=\ell^{-1/2}\sum_{j=\ell+1}^{2\ell}|j\rangle$ is supported on the middle third, where the potential vanishes. The variational principle gives
\begin{align}
    E_0&=\min_{\|w\|_2=1}\langle w|\mathbf h_D^{(V)}|w\rangle
    \overset{(1)}{\leq}\langle\phi|\mathbf h_D^{(V)}|\phi\rangle
    =-\frac{J}{\ell}\sum_{j=\ell+1}^{2\ell-1}(1+1),\notag\\
    &=-\frac{2J(\ell-1)}{\ell}=-2J+\frac{2J}{\ell}\overset{(2)}{\leq}-2J+\frac\eps2,\label{eq:sparse-Ebound}
\end{align}
Here $(1)$ evaluates the minimum on $|\phi\rangle$; each of its $\ell-1$ internal bonds contributes twice, and $(2)$ uses $\ell\geq4J/\eps$. On the left barrier, the eigenvalue equation gives
\begin{align}
    v_{j+1}=\frac{\eps-E_0}{J}v_j-v_{j-1},\qquad 2\leq j\leq\ell,\qquad \frac{\eps-E_0}{J}\geq2+\frac\eps{2J}.\label{eq:sparse-barrier-recurrence}
\end{align}
Starting from $v_2>v_1=0$, suppose inductively that $v_j\geq v_{j-1}$. Then
\begin{align}
    v_{j+1}&\overset{(1)}{\geq}\left(2+\frac\eps{2J}\right)v_j-v_{j-1}\overset{(2)}{\geq}\left(1+\frac\eps{2J}\right)v_j\geq v_j,
\end{align}
where $(1)$ is \cref{eq:sparse-barrier-recurrence} and $(2)$ is the induction hypothesis. Iterating for $j=2,\ldots,\ell$ gives
\begin{align}
    v_2&\leq\left(1+\frac\eps{2J}\right)^{-(\ell-1)}v_{\ell+1}\leq\left(1+\frac\eps{2J}\right)^{-(\ell-1)},\notag\\
    |v_2|+|v_{m-1}|&\leq2\left(1+\frac\eps{2J}\right)^{-(\ell-1)}.\label{eq:sparse-boundary-decay-proof}
\end{align}
The first line uses normalization, which implies $v_{\ell+1}\leq1$. The second follows by applying the same induction from the right endpoint, starting at $v_m=0$. This proves \cref{eq:sparse-boundary-decay} and completes the exponential lower bound.

\Needspace{9\baselineskip}
\emph{Primitivity of the perturbed generator.} It remains to prove that the specific generator $\LL^{(V)}$ in \cref{eq:app-sparse-potential} converges to the same unique stationary state. Equal raising and lowering rates give $\LL^{(V)}(\sigma)=0$. Its Hamiltonian commutator is skew-adjoint, so the energy calculation \cref{eq:sparse-energy-differential} applies with $H+V$ and again gives Hilbert--Schmidt contractivity. It suffices to show that every eigenoperator with an imaginary eigenvalue is scalar.

Let $\omega\in\mathbb R$ and suppose $O\ne0$ satisfies $(\LL^{(V)})^*(O)=i\omega O$. Taking the real part of its Hilbert--Schmidt inner product with $O$ gives
\begin{align}
    0=\Rea\langle O,(\LL^{(V)})^* O\rangle=-\frac\gamma4\sum_{s\in\{1,m\}}\bigl(\|[X_s,O]\|_2^2+\|[Y_s,O]\|_2^2\bigr).
\end{align}
Every commutator in this sum vanishes. Thus the dissipative part annihilates $O$, $[H+V,O]=\omega O$, and $O=\BI_1\otimes O_{2,\ldots,m}$. Suppose inductively that $O$ acts trivially on sites $1,\ldots,r$, with $r<m$. In $[H+V,O]=\omega O$, the only terms with nonidentity dependence on site $r$ are
\begin{align}
    -J\sigma_r^+\otimes[\sigma_{r+1}^-,O_{r+1,\ldots,m}]-J\sigma_r^-\otimes[\sigma_{r+1}^+,O_{r+1,\ldots,m}].
\end{align}
The on-site potential contributes only terms proportional to the identity on site $r$, as do the right-hand side and the remaining bonds. Linear independence of $\BI_r,\sigma_r^+,\sigma_r^-$ therefore forces both displayed commutators to vanish. Since $\sigma_{r+1}^\pm$ generate the one-qubit algebra, $O$ also acts trivially on site $r+1$. Induction gives $O\in\C\BI$ and $\omega=0$.

A nontrivial Jordan block on the imaginary axis would cause polynomial growth, contradicting contractivity. Hence zero is a simple eigenvalue and every other eigenvalue has negative real part. Trace preservation and stationarity of the identity identify the limiting Heisenberg projection as $O\mapsto\Tr(O)\sigma$. This proves primitivity of $\LL^{(V)}$ and completes the proof.
\end{proof}

\begin{proof}[Proof of \cref{lem:sparse-boundary-potential}]
\prooflabel{lem:sparse-boundary-potential}{proof:sparse-boundary-potential}
The unperturbed generators in \cref{ex:sparse-boundary-potential-app} are restrictions of the fixed nearest-neighbor hopping interaction, with equal-rate endpoint baths as boundary conditions. Evaluating this family at physical chain length $m=3\lfloor a\log N/3\rfloor=\Theta(\log N)$ gives the generators $\LL_N$ in \cref{lem:sparse-boundary-potential}. For each $N$, take $\KK_N=-i[V_N,\cdot]$ with $V_N$ from \cref{eq:sparse-potential}. Each on-site term has induced trace norm at most $2$, so their local strength is uniformly bounded, and this choice is independent of $\eps$. At fixed accuracy, \cref{eq:sparse-uniform-time-bound} gives an $O((\log N)^{10})$ unperturbed mixing time. For every fixed $\eps>0$, the condition $\ell\geq\max\{2,4J/\eps\}$ holds for sufficiently large $N$, and \cref{eq:app-sparse-main-lower} gives the lower bound $N^{c_\eps+o(1)}$, with $c_\eps=\frac a3\log(1+\eps/(2J))>0$. Both generators are primitive with the same maximally mixed stationary state by the preceding example. These are the asserted polylogarithmic and polynomial bounds in the auxiliary parameter $N$.
\end{proof}

\partbibliography

\addtocontents{toc}{\protect\setcounter{tocdepth}{2}}
\hypersetup{bookmarksdepth=2}

\paperpart{appendices-dissipative_covers}{1}
\section{Strong bulk dissipation}\label{app:dissipative-covers}
Throughout this appendix, we consider Lindbladians $\LL$ and $\GG=\LL+V$ on a lattice $\Lambda$. This will allow us to study rapid mixing of $\LL$ itself, or establish rapid mixing of $\GG$ by comparing it with a simpler Lindbladian $\LL$ that has a bulk-dissipative cover and satisfies the influence condition. The perturbation $V$ need not itself be a Lindbladian, but we assume that it is locally trace-annihilating: writing $V=\sum_{B\in\mathcal S_V}V_B$, we require
\begin{align}
    V_B^*(\BI_B\otimes Y_{B^c}) &=0 \qquad \text{for every $Y_{B^c}$.} \label{eq:app-perturbation-locally-trace-preserving}
\end{align}

We begin with a technical extension of the construction in \cref{sec:dissipative-covers}. There the cover generators are individual terms of the displayed decomposition of $\LL$; here each $\MM_\alpha$ may instead be a sum of terms of $\LL$. Fix a decomposition $\LL=\sum_{B\in\mathcal S_{\LL}}\LL_B$ into local Lindblad generators. For each $\alpha=1,\ldots,M$, choose a subfamily $\mathcal S_{\MM_\alpha}\subseteq\mathcal S_{\LL}$ and set
\begin{align}
    \MM_\alpha&:=\sum_{B\in\mathcal S_{\MM_\alpha}}\LL_B,\qquad
    A_\alpha:=\bigcup_{B\in\mathcal S_{\MM_\alpha}}B.
\end{align}
We define the bulk of $\MM_\alpha$ to be the largest region on which every stationary observable acts trivially, as illustrated in \cref{fig:bulk-dissipative-cover}b):
\begin{align}
    A_\alpha^\circ
    &:=\left\{x\in A_\alpha:\ker\MM_\alpha^*\subseteq
    \BI_x\otimes\mathcal B(\mathcal H_{x^c})\right\}. \label{eq:app-bulk-dissipation}
\end{align}
Equivalently,
\begin{align}
    \ker\MM_\alpha^*&\subseteq
    \BI_{A_\alpha^\circ}\otimes
    \mathcal B(\mathcal H_{(A_\alpha^\circ)^c}).
\end{align}

\begin{definition}[Bulk-dissipative cover]
\label{def:app-bulk-dissipative-cover}
A collection $\mathfrak U := \left\{ (A_\alpha^\circ,A_\alpha,\MM_\alpha): \alpha=1,\ldots,M \right\}$ constructed from the fixed decomposition above is a bulk-dissipative cover on $\Lambda$ with local mixing time $t_{\mathrm{loc}}>0$ if:

\begin{enumerate}
\item The bulk regions cover the lattice:
\begin{align}
    \bigcup_{\alpha=1}^M A_\alpha^\circ &= \Lambda.
\end{align}

\item Every $\MM_\alpha$ has a well-defined infinite-time limit and, uniformly in $\alpha$, mixes locally by time $t_{\mathrm{loc}}$:
\begin{align}
    ||e^{t_{\mathrm{loc}}\MM_\alpha}-P_{\MM_\alpha}||_\diamond &\leq\frac12, \label{eq:app-cover-local-mixing}
\end{align}
\end{enumerate}
\end{definition}

\begin{figure}[!t]
\centering
\begin{minipage}[t]{0.49\linewidth}
\textbf{a)}\par\smallskip
\includegraphics[width=\linewidth]{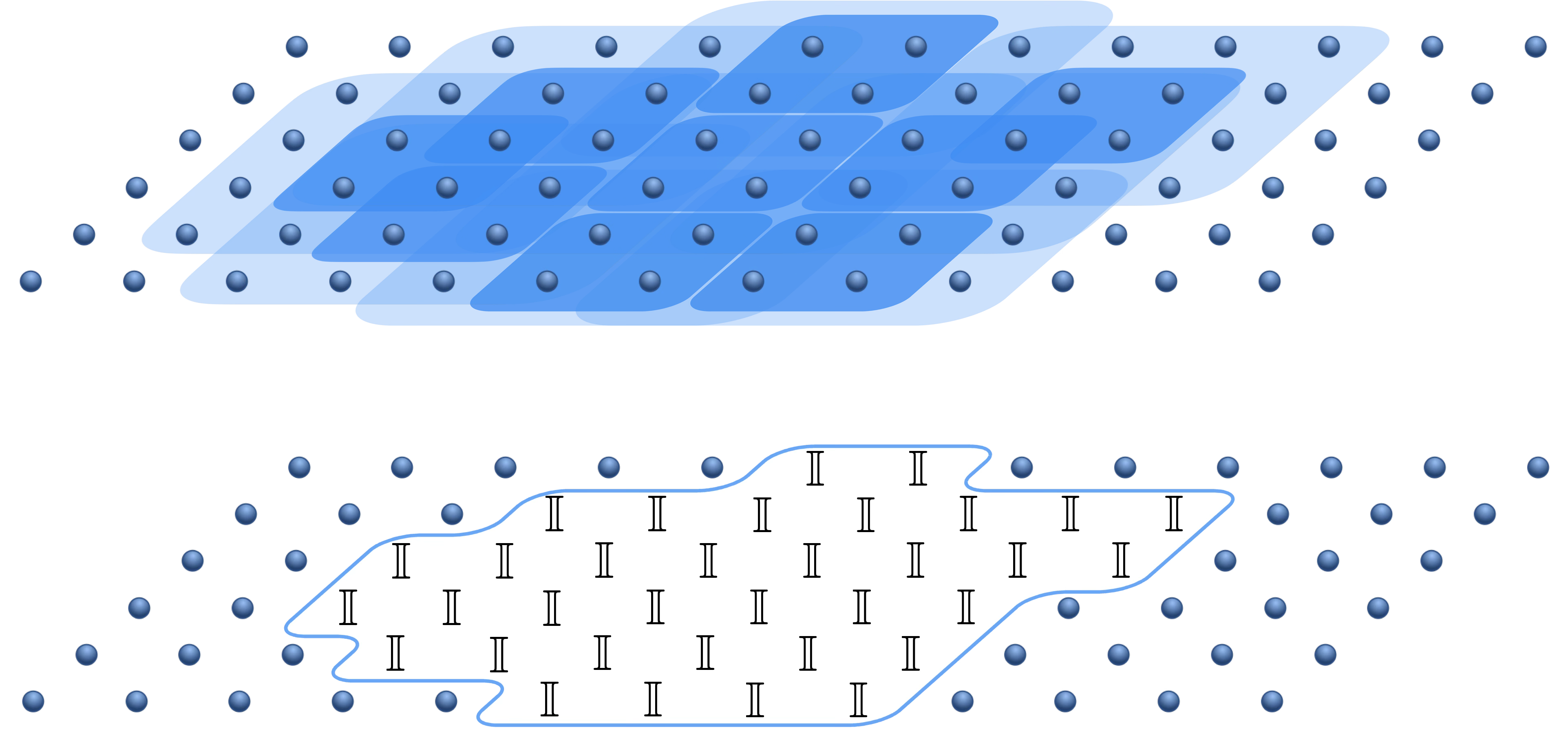}
\end{minipage}\hfill
\begin{minipage}[t]{0.49\linewidth}
\textbf{b)}\par\smallskip
\includegraphics[width=\linewidth]{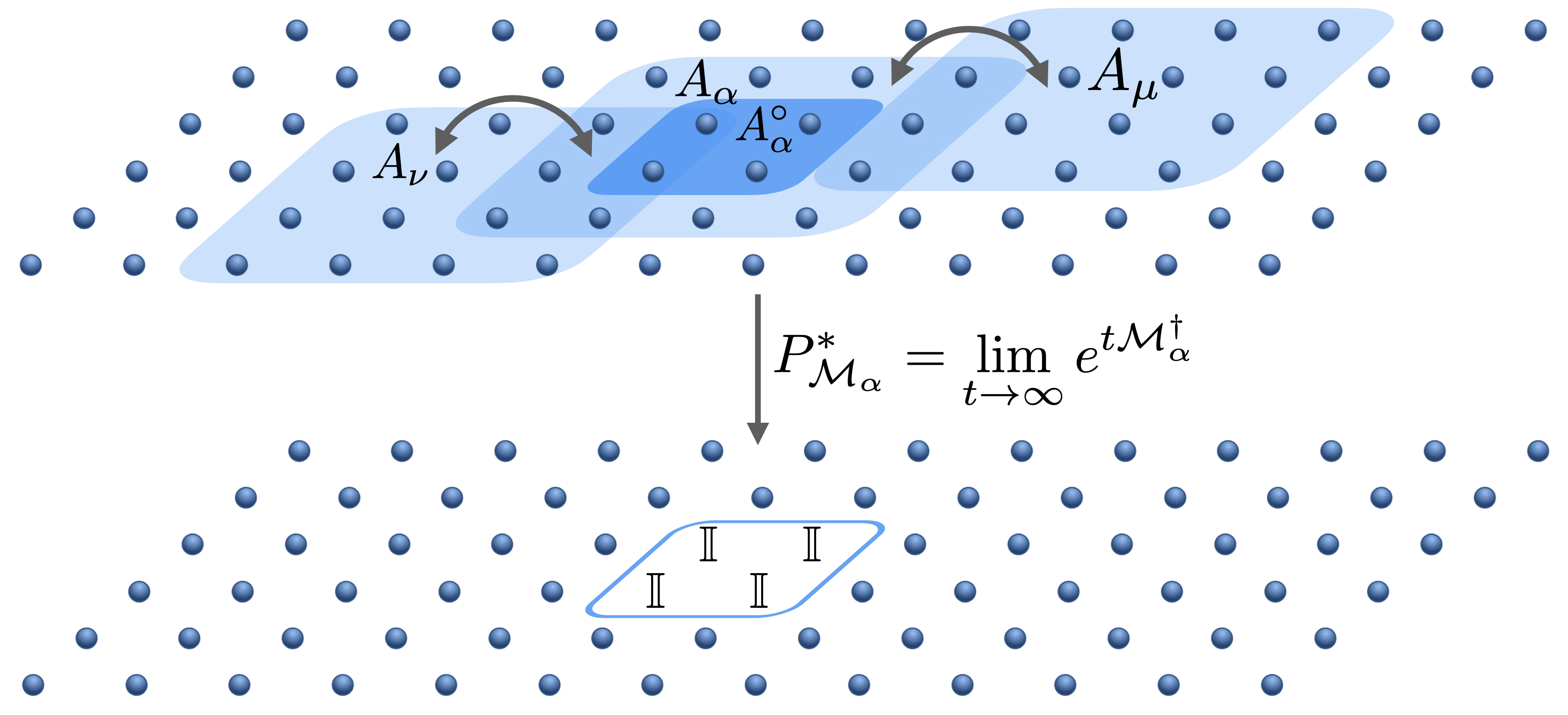}
\end{minipage}
\caption{A bulk-dissipative cover and the local-to-global mixing mechanism. \textbf{a)} The bulk regions cover the lattice. Auxiliary reset maps replace an observable's action on each $A_\alpha^\circ$ by the identity while preserving the identity on regions already treated. Composing these maps therefore leaves only a scalar multiple of the identity everywhere. \textbf{b)} Each local generator $\MM_\alpha$ is supported on $A_\alpha$ (top) and drives observables toward operators that act as the identity on its bulk $A_\alpha^\circ$ (bottom). When this local dissipation dominates the influence of the remaining interactions, as quantified by \cref{app:eq:strong-bulk-stability-condition}, these local decay estimates control convergence of the full system to its unique stationary state.}
\label{fig:bulk-dissipative-cover}
\end{figure}

We quantify the size of the enlarged regions and their incidence by
\begin{align}
    a &:= \max_\alpha|A_\alpha|, \quad \mathfrak d := \max_{x\in\Lambda} \#\left\{ \alpha\in\{1,\ldots,M\}:x\in A_\alpha \right\}.
\end{align}
Equivalently, we say that $\LL$ contains the cover $\mathfrak U$. Thus $\MM_\alpha$ need not contain every term supported in $A_\alpha$. We keep the decomposition and subfamilies $\mathcal S_{\MM_\alpha}$ fixed throughout. The same bulk-dissipative cover used for $\LL$ is thus available for $\GG=\LL+V$, and we only need to control the incoming influence of $V$ relative to $\mathfrak U$.

Define the dual projections and the associated quasiderivations by
\begin{align}
    P_{\MM_\alpha}^*&:=\lim_{t\to\infty}e^{t\MM_\alpha^*}, & \delta_\alpha&:=\operatorname{id}-P_{\MM_\alpha}^*.
\end{align}
Both $P_{\MM_\alpha}^*$ and $\delta_\alpha$ are projectors: $(P_{\MM_\alpha}^*)^2=P_{\MM_\alpha}^*$ and $\delta_\alpha^2=\delta_\alpha$. The map $\delta_\alpha$ removes the component of an observable in the steady-state subspace of $\MM_\alpha^*$. We therefore say that $\delta_\alpha$ measures the \textit{fluctuations} of an observable relative to the local bulk generator $\MM_\alpha$.

The definition of $t_{\mathrm{loc}}$ gives a decay estimate that we will use repeatedly; it quantifies the local relaxation depicted in \cref{fig:bulk-dissipative-cover}b). The duality between the diamond norm and the completely bounded norm, together with \cref{eq:app-cover-local-mixing}, implies
\begin{align}
    \left\|e^{t_{\mathrm{loc}}\MM_\alpha^*}-P_{\MM_\alpha}^*\right\|_{\mathrm{cb}}
    &=\left\|e^{t_{\mathrm{loc}}\MM_\alpha}-P_{\MM_\alpha}\right\|_\diamond\leq\frac12. \label{eq:app-one-step-local-contraction}
\end{align}
For $t=nt_{\mathrm{loc}}+r$, where $n=\lfloor t/t_{\mathrm{loc}}\rfloor$ and $0\leq r<t_{\mathrm{loc}}$, the identities $P_{\MM_\alpha}^*\delta_\alpha=0$ and $e^{s\MM_\alpha^*}P_{\MM_\alpha}^*=P_{\MM_\alpha}^*$ give
\begin{align}
    e^{t\MM_\alpha^*}\delta_\alpha
    &=e^{r\MM_\alpha^*}
      \left(e^{t_{\mathrm{loc}}\MM_\alpha^*}-P_{\MM_\alpha}^*\right)^n
      \delta_\alpha.
\end{align}
The remaining evolution $e^{r\MM_\alpha^*}$ is a unital contraction. Hence, with
\begin{align}
    g(t)&:=2^{-\lfloor t/t_{\mathrm{loc}}\rfloor}. \label{eq:app-local-mixing-kernel}
\end{align}
we obtain
\begin{align}
    \left\|e^{t\MM_\alpha^*}\delta_\alpha O\right\|
    &\leq g(t)\left\|\delta_\alpha O\right\|. \label{eq:app-iterated-local-contraction}
\end{align}

It will be useful to introduce a quasiderivation associated to a reset map. Choose an arbitrary state $\tau_\alpha$ on $A_\alpha^\circ$ and define
\begin{align}
    \mathcal R_\alpha(X_{A_\alpha^\circ}\otimes Y_{(A_\alpha^\circ)^c}) &:= \Tr(\tau_\alpha X_{A_\alpha^\circ})\BI_{A_\alpha^\circ}\otimes Y_{(A_\alpha^\circ)^c}, & \eta_\alpha&:=\operatorname{id}-\mathcal R_\alpha.
\end{align}
The map $\eta_\alpha$ vanishes on observables that act trivially on $A_\alpha^\circ$:
\begin{align}
    \eta_\alpha(\BI_{A_\alpha^\circ}\otimes Y_{(A_\alpha^\circ)^c}) &=0.
\end{align}
By \cref{eq:app-bulk-dissipation},
\begin{align}
    \operatorname{Ran}P_{\MM_\alpha}^* &= \ker\MM_\alpha^* \subseteq \BI_{A_\alpha^\circ}\otimes \mathcal B(\mathcal H_{(A_\alpha^\circ)^c}),
\end{align}
and hence $\mathcal R_\alpha P_{\MM_\alpha}^*=P_{\MM_\alpha}^*$. Therefore
\begin{align}
    \eta_\alpha\delta_\alpha &= (\operatorname{id}-\mathcal R_\alpha) (\operatorname{id}-P_{\MM_\alpha}^*), \notag\\
    &= \operatorname{id}-P_{\MM_\alpha}^*-\mathcal R_\alpha +\mathcal R_\alpha P_{\MM_\alpha}^*, \notag\\
    &= \operatorname{id}-\mathcal R_\alpha =\eta_\alpha. \label{eq:app-eta-absorbs-delta}
\end{align}
Moreover, $\mathcal R_\alpha$ is a unital contraction, so $||\mathcal R_\alpha||_{\infty\to\infty}=1$. Using \cref{eq:app-eta-absorbs-delta},
\begin{align}
    ||\eta_\alpha(O)|| &= ||\eta_\alpha\delta_\alpha(O)||, \notag\\
    &\leq \bigl( ||\operatorname{id}||_{\infty\to\infty} +||\mathcal R_\alpha||_{\infty\to\infty} \bigr) ||\delta_\alpha(O)||, \notag\\
    &\leq 2||\delta_\alpha(O)||. \label{eq:app-eta-controlled-by-delta}
\end{align}

The reason for introducing $\eta_\alpha$ is that composing $\mathcal{R}_\alpha$ yields an observable that is proportional to the identity in their joint support, as we now show. This property will be used at several points in the proof later. Suppose that $\alpha_1,\ldots,\alpha_m$ is any ordered collection of reset labels and define
\begin{align}
    A_{1,\ldots,j}^\circ &:= \bigcup_{i=1}^jA_{\alpha_i}^\circ, & Q_j &:= \mathcal R_{\alpha_1}\circ\cdots\circ\mathcal R_{\alpha_j}.
\end{align}
Directly from the definition of the reset, for every $B\subseteq\Lambda$,
\begin{align}
    O\in \BI_B\otimes \mathcal B(\mathcal H_{B^c}) \quad\Longrightarrow\quad \mathcal R_{\alpha_j}(O) \in \BI_{B\cup A_{\alpha_j}^\circ}\otimes \mathcal B(\mathcal H_{(B\cup A_{\alpha_j}^\circ)^c}). \label{eq:app-one-reset-enlarges-trivial-region}
\end{align}
The reset leaves the identity on $B\setminus A_{\alpha_j}^\circ$ unchanged and replaces the action on $A_{\alpha_j}^\circ$ by the identity. Iterating \cref{eq:app-one-reset-enlarges-trivial-region} yields
\begin{align}
    \operatorname{Ran}Q_j &\subseteq \BI_{A_{1,\ldots,j}^\circ}\otimes \mathcal B(\mathcal H_{(A_{1,\ldots,j}^\circ)^c}). \label{eq:app-reset-composition-range}
\end{align}
In particular, if the reset regions cover the full lattice, then their composition maps every observable to a scalar multiple of the identity. This process is illustrated in \cref{fig:bulk-dissipative-cover}a).

Given a bulk-dissipative cover $\mathfrak U$, define the oscillator norm using the $\delta_\alpha$ quasiderivations by
\begin{align}
    \osc{O}_{\mathfrak U} := \sum_{\alpha=1}^M||\delta_\alpha(O)||.
\end{align}

\begin{lemma}
\label{lem:app-trace-oscillator}
Let $\GG$ be a Lindbladian on $\Lambda$, let $\sigma$ be a fixed state of $\GG$, and let $\mathfrak U$ be a bulk-dissipative cover on $\Lambda$. Then
\begin{align}
    ||e^{t\GG}(\rho)-\sigma||_1 &\leq 2||\rho-\sigma||_1\sup_{||O||\leq1} \osc{e^{t\GG^*}(O)}_{\mathfrak U}. \label{eq:app-trace-oscillator-bound}
\end{align}
\end{lemma}

\begin{proof}[Proof of \cref{lem:app-trace-oscillator}]
\prooflabel{lem:app-trace-oscillator}{proof:app-trace-oscillator}
Consider the time evolution:
\begin{align}
    ||e^{t\GG}(\rho)-\sigma||_1 &\overset{(1)}{=} \sup_{O=O^\dagger,\,||O||\leq1}\Tr(Oe^{t\GG}(\rho-\sigma)),\notag\\
    &\overset{(2)}{=} \sup_{O=O^\dagger,\,||O||\leq1}\Tr(e^{t\GG^*}(O)(\rho-\sigma)),\notag\\
    &\overset{(3)}{=} \sup_{O=O^\dagger,\,||O||\leq1}\inf_{c\in\R} \Tr((e^{t\GG^*}(O)-c\BI)(\rho-\sigma)),\notag\\
    &\overset{(4)}{\leq} ||\rho-\sigma||_1\sup_{O=O^\dagger,\,||O||\leq1}\inf_{c\in\R} ||e^{t\GG^*}(O)-c\BI||,
\end{align}
where $(1)$ uses the variational definition of the trace norm, $(2)$ that $\GG^*$ is the Hilbert--Schmidt adjoint, $(3)$ uses that $\Tr(c\BI(\rho-\sigma))=c\Tr(\rho-\sigma)=0$, and $(4)$ uses H\"older's inequality.

We upper bound $q(O):=\inf_{c\in\R}||O-c\BI||$ using a telescoping argument over the resets. Order the cover labels as $\alpha_1,\ldots,\alpha_M$. Since every site belongs to some $A_\alpha^\circ$, \cref{eq:app-reset-composition-range} gives
\begin{align}
    \mathcal R_{\alpha_1}\circ\cdots\circ\mathcal R_{\alpha_M}(O) &=c(O,\alpha_1,\ldots,\alpha_M)\BI, \label{eq:app-R-product-scalar}
\end{align}
where the proportionality constant depends on the ordering of $\alpha_1,...,\alpha_M$. This yields the telescoping sum
\begin{align}
    \operatorname{id}-\mathcal R_{\alpha_1}\circ\cdots\circ\mathcal R_{\alpha_M} &= \sum_{i=1}^M \mathcal R_{\alpha_1}\circ\cdots\circ\mathcal R_{\alpha_{i-1}}\circ(\operatorname{id}-\mathcal{R}_{\alpha_i}), \label{eq:app-R-telescope}
\end{align}
which allows us to obtain the bound
\begin{align}
    q(O) &\leq||O-c(O,\alpha_1,\ldots,\alpha_M)\BI||,\notag\\
    &\overset{(1)}{\leq} \sum_{i=1}^M ||\mathcal R_{\alpha_1}\circ\cdots\circ \mathcal R_{\alpha_{i-1}}\circ\eta_{\alpha_i}(O)||,\notag\\
    &\overset{(2)}{\leq} \sum_{i=1}^M||\eta_{\alpha_i}(O)||, \notag\\
    &\overset{(3)}{\leq} 2\sum_{i=1}^M||\delta_{\alpha_i}(O)|| =2\osc{O}_{\mathfrak U},
\end{align}
where $(1)$ uses \cref{eq:app-R-telescope} and the triangle inequality, $(2)$ that the resets are unital contractions $||\mathcal{R}_\alpha||_{\infty\to\infty}\leq 1$, and $(3)$ uses \cref{eq:app-eta-controlled-by-delta}. Applying this estimate to $e^{t\GG^*}(O)$ proves \cref{eq:app-trace-oscillator-bound}.
\end{proof}

We will study the decay of $\osc{e^{t\GG^*}(O)}_{\mathfrak U}$ by comparing the decay given by the bulk-dissipative cover to the interactions between the local bulks and other regions. In order to quantify the strength of these interactions we need to make some definitions. First, consider a map $W_B$ that is locally trace-annihilating in $B$. Choose an ordered list $\alpha_{B,1},\ldots,\alpha_{B,M_B}$ without repetition such that
\begin{align}
    B &\subseteq \bigcup_{i=1}^{M_B}A_{\alpha_{B,i}}^\circ, \qquad A_{\alpha_{B,i}}^\circ\cap B\neq\varnothing. \label{eq:app-bulk-list-covers-B}
\end{align}
Such a list may be chosen with $M_B\leq|B|$ by choosing one bulk region for each site of $B$ and removing repetitions. Fix these lists throughout. In particular,
\begin{align}
    \sum_{i=1}^{M_B}\mathbf 1[\alpha_{B,i}=\nu] &\leq \min\{1,|A_\nu^\circ\cap B|\}. \label{eq:app-list-multiplicity}
\end{align}
Define
\begin{align}
    Q_B &:= \mathcal R_{\alpha_{B,1}}\circ\cdots\circ \mathcal R_{\alpha_{B,M_B}}.
\end{align}
Since the bulk regions in the list cover $B$, \cref{eq:app-reset-composition-range} shows that the range of $Q_B$ is contained in the observables that act trivially on $B$. Consequently,
\begin{align}
    W_B^* Q_B &=0. \label{eq:app-TQ-zero}
\end{align}
We will use the telescoping sum
\begin{align}
    \operatorname{id}-Q_B &= \sum_{i=1}^{M_B} \mathcal R_{\alpha_{B,1}}\circ\cdots\circ \mathcal R_{\alpha_{B,i-1}}\circ\eta_{\alpha_{B,i}}, \label{eq:app-reset-telescope}
\end{align}
where an empty product of reset maps is understood to be the identity map. It gives
\begin{align}
    W_B^* X &\overset{(1)}{=} W_B^*(\operatorname{id}-Q_B)X, \notag\\
    &\overset{(2)}{=} \sum_{i=1}^{M_B} W_B^* \mathcal R_{\alpha_{B,1}}\circ\cdots\circ \mathcal R_{\alpha_{B,i-1}} \circ \eta_{\alpha_{B,i}}X, \label{eq:app-bulk-single-telescope}
\end{align}
where $(1)$ uses \cref{eq:app-TQ-zero}, and $(2)$ uses \cref{eq:app-reset-telescope}. Since the resets are contractions, \cref{eq:app-eta-controlled-by-delta} gives
\begin{align}
    ||\delta_\alpha W_B^* X|| &\leq 2||\delta_\alpha W_B^*||_{\infty\to\infty} \sum_{i=1}^{M_B}||\delta_{\alpha_{B,i}}X||. \label{eq:app-single-telescope-source-bound}
\end{align}

Define
\begin{align}
    \kappa_{\LL}^{\sharp} &:= 2\max_{\nu\in\{1,\ldots,M\}} \sum_{\alpha=1}^M \sum_{\substack{B\in\mathcal S_{\LL}\setminus\mathcal S_{\MM_\alpha}\\ [\MM_\alpha,\LL_B]
    \neq0}} \sum_{i=1}^{M_B} \mathbf 1[\alpha_{B,i}=\nu] ||\delta_\alpha\LL_B^*||_{\infty\to\infty}, \label{eq:app-bulk-incoming-L}\\
    \kappa_V^{\mathrm{bulk}} &:= 2\max_{\nu\in\{1,\ldots,M\}} \sum_{\alpha=1}^M \sum_{B\in\mathcal S_V} \sum_{i=1}^{M_B} \mathbf 1[\alpha_{B,i}=\nu] ||\delta_\alpha V_B^*||_{\infty\to\infty}. \label{eq:app-bulk-incoming-V}
\end{align}
We assume that $V$ has finite local strength $s_V:=\max_{x\in\Lambda}\sum_{B\ni x}\|V_B\|_{1\to1}$, without requiring a bound on its support sizes. The following matrix measures the strength of the couplings between different terms in the cover:
\begin{align}
    K_{\alpha\nu} &:= 2\sum_{\substack{B\in\mathcal S_{\LL}\setminus\mathcal S_{\MM_\alpha}\\ [\MM_\alpha,\LL_B]
    \neq0}}\sum_{i=1}^{M_B}\mathbf 1[\alpha_{B,i}=\nu]||\delta_\alpha\LL_B^*||_{\infty\to\infty} +2\sum_{B\in\mathcal S_V}\sum_{i=1}^{M_B}\mathbf 1[\alpha_{B,i}=\nu]||\delta_\alpha V_B^*||_{\infty\to\infty}, \label{eq:app-bulk-comparison-matrix}\\
    \kappa_{\GG}^{\mathrm{bulk}} &:= \max_{1\leq\nu\leq M}\sum_{\alpha=1}^M K_{\alpha\nu}. \label{eq:app-general-influence-load}
\end{align}
These quantities are relative to the fixed cover, decomposition, and covering lists. The superscript ``bulk'' distinguishes this direct coefficient from the count-based coefficient in \cref{eq:main-local-influence}: the direct matrix charges every perturbation term, including those commuting with the cover. When $V=0$, $\kappa_{\LL}^{\mathrm{bulk}}=\kappa_{\LL}^{\sharp}$. For the individual-term covers of the main text, the refined reference coefficient satisfies $\kappa_{\LL}^{\sharp}\leq\kappa_{\LL}$, as shown below in \cref{eq:app-sharp-main-comparison}. In particular, $\kappa_{\GG}^{\mathrm{bulk}} \leq \kappa_{\LL}^{\sharp}+\kappa_V^{\mathrm{bulk}}$.

For the reference Lindbladian, we also define the cover-coupling strength by
\begin{align}
    n_{\mathfrak U}(B) &:= \#\left\{\alpha:B\notin\mathcal S_{\MM_\alpha},\ [\MM_\alpha,\LL_B]\neq0\right\}, \notag\\
    s_{\LL,\mathfrak U} &:= \max_{x\in\Lambda}\sum_{\substack{B\in\mathcal S_{\LL}\\
    x\in B}}n_{\mathfrak U}(B)||\LL_B||_{1\to1}. \label{eq:app-cover-coupling-strength}
\end{align}
While $s_{\LL}$ measures the total interaction strength per site of the Lindbladian $\LL$, only those terms that don't commute with the cover will be detrimental to the local-to-global rapid mixing proof. Thus $s_{\LL,\mathfrak U}$ measures this detrimental interaction strength per site, weighting each term by the number $n_{\mathfrak U}(B)$ of bulk generators it neither belongs to nor commutes with. In particular, $s_{\LL,\mathfrak U}=0$ if every term either belongs to or commutes with each bulk generator.

\begin{proposition}[Bulk-cover rapid mixing with constant decay rate and stability]
\label{prop:app-rapid-mixing-bulk-cover}
Let $\mathfrak U$ be a bulk-dissipative cover with local mixing time $t_{\mathrm{loc}}$ and $M$ terms, and let $\LL$ be a Lindbladian that contains $\mathfrak U$. Let $V$ be a locally trace-annihilating map with finite local strength $s_V$, and suppose that $\GG=\LL+V$ is a Lindbladian. Define
\begin{align}
    \Delta_{\GG} &:= \frac12-\kappa_{\GG}^{\mathrm{bulk}}t_{\mathrm{loc}}.
\end{align}
If
\begin{align}
    \Delta_{\GG} &>0, \label{eq:app-general-bulk-condition}
\end{align}
then $\GG$ has a unique fixed state $\sigma$ and, for every state $\rho$,
\begin{align}
    ||e^{t\GG}(\rho)-\sigma||_1 &\leq 4D_{\GG}M e^{-\gamma_{\GG}t}||\rho-\sigma||_1, \label{eq:app-general-bulk-mixing-bound}
\end{align}
where
\begin{align}
    \gamma_{\GG} &:= \frac{\Delta_{\GG}}{2t_{\mathrm{loc}}}>0, & D_{\GG} &:= e^{\Delta_{\GG}/2}\frac{\frac12+\kappa_{\GG}^{\mathrm{bulk}}t_{\mathrm{loc}}}{\Delta_{\GG}}>0. \label{eq:app-general-bulk-explicit-constants}
\end{align}
For a family $(\GG_N)_N$ with $\GG=\GG_N$ at each size, fixed $t_{\mathrm{loc}}$ and $\kappa_{\GG}^{\mathrm{bulk}}$ independent of the system size, the strict inequality gives a decay rate $\Omega(1)$ and hence rapid mixing with constant decay rate when $M$ grows polynomially. More generally, if $t_{\mathrm{loc}}=O(\operatorname{polylog}|\Lambda|)$ and $\Delta_\GG$ is bounded below by a positive constant, the same estimate gives rapid mixing.
\end{proposition}

\begin{proof}[Proof of \cref{prop:app-rapid-mixing-bulk-cover}]
\prooflabel{prop:app-rapid-mixing-bulk-cover}{proof:app-rapid-mixing-bulk-cover}
The bound $\|\delta_\alpha\|_{\infty\to\infty}\leq2$ and \cref{eq:app-list-multiplicity} give $\kappa_{\LL}^{\sharp}\leq4a s_{\LL,\mathfrak U}$ and $\kappa_V^{\mathrm{bulk}}\leq4Ma s_V$.

In order to bound the decay of the oscillator norm, we estimate each term in its defining sum. Note that $\delta_\alpha(O) =0$ implies that $O$ is in the kernel of $\mathcal{M}_\alpha^*$. We will use that certain terms in $\mathcal{G}$ do not introduce fluctuations into the support $A_{\alpha}$, i.e. there are terms $\mathcal{C}_\alpha$ such that $[\mathcal{C}_\alpha^*,\delta_\alpha]=0$. Using the fixed decomposition of $\LL$ that contains $\mathfrak U$, for each $\alpha$, split the terms of $\LL-\MM_\alpha$ into those that commute with $\MM_\alpha$ and those that do not:
\begin{align}
    \LL_{\alpha,c} &:= \sum_{\substack{B\in\mathcal S_{\LL}\setminus\mathcal S_{\MM_\alpha}\\ [\MM_\alpha,\LL_B]
    =0}}\LL_B, & \LL_{\alpha,\not c} &:= \sum_{\substack{B\in\mathcal S_{\LL}\setminus\mathcal S_{\MM_\alpha}\\ [\MM_\alpha,\LL_B]
    \neq0}}\LL_B.
\end{align}
Thus for each $\alpha$ we obtain a decomposition of $\GG^*$:
\begin{align}
    \GG^* &= \MM_\alpha^*+\LL_{\alpha,c}^*+\LL_{\alpha,\not c}^*+V^*.
\end{align}
Since $P_{\MM_\alpha}^*=\lim_{t\to\infty}e^{t\MM_\alpha^*}$, we have $[\MM_\alpha^*,\delta_\alpha]=0$. Moreover, $[\MM_\alpha,\LL_{\alpha,c}]=0$, so $[\LL_{\alpha,c}^*,\delta_\alpha]=0$ under composition. Let $S_t^\alpha:=e^{t(\MM_\alpha^*+\LL_{\alpha,c}^*)}$ and define $O_t:=e^{t\GG^*}(O)$. Then
\begin{align}
    \frac{d}{ds}\Bigl(S_{t-s}^\alpha\delta_\alpha O_s\Bigr) &= -S_{t-s}^\alpha(\MM_\alpha^*+\LL_{\alpha,c}^*)\delta_\alpha O_s+S_{t-s}^\alpha\delta_\alpha\GG^* O_s, \notag\\
    &= S_{t-s}^\alpha\delta_\alpha(\LL_{\alpha,\not c}^*+V^*)O_s.
\end{align}
Integrating from $s=0$ to $s=t$ yields
\begin{align}
    \delta_\alpha O_t &= S_t^\alpha\delta_\alpha O+\int_0^t S_{t-s}^\alpha\delta_\alpha(\LL_{\alpha,\not c}^*+V^*)O_s\,ds. \label{eq:app-bulk-duhamel}
\end{align}
Because $\MM_\alpha$ and $\LL_{\alpha,c}$ commute, $S_t^\alpha=e^{t\LL_{\alpha,c}^*}e^{t\MM_\alpha^*}$, and therefore
\begin{align}
    ||S_t^\alpha\delta_\alpha O|| &\overset{(1)}{\leq} ||e^{t\MM_\alpha^*}\delta_\alpha O||, \notag\\
    &\overset{(2)}{\leq} g(t)||\delta_\alpha O||,
\end{align}
where $(1)$ uses that $e^{t\LL_{\alpha,c}^*}$ is a unital contraction, and $(2)$ is \cref{eq:app-iterated-local-contraction}. For the integrand in \cref{eq:app-bulk-duhamel}, set
\begin{align}
    X_{\alpha,s}&:=\delta_\alpha(\LL_{\alpha,\not c}^*+V^*)O_s.
\end{align}
Since $\delta_\alpha^2=\delta_\alpha$, we have $X_{\alpha,s}=\delta_\alpha X_{\alpha,s}$. Commutation of $\MM_\alpha$ and $\LL_{\alpha,c}$ and contractivity of $e^{u\LL_{\alpha,c}^*}$ therefore give, for $0\leq s\leq t$,
\begin{align}
    \left\|S_{t-s}^\alpha X_{\alpha,s}\right\|
    &\leq\left\|e^{(t-s)\MM_\alpha^*}\delta_\alpha X_{\alpha,s}\right\|
    \leq g(t-s)\left\|X_{\alpha,s}\right\|, \label{eq:app-bulk-integral-kernel}
\end{align}
where the second inequality again uses \cref{eq:app-iterated-local-contraction}. Applying \cref{eq:app-single-telescope-source-bound} term by term gives
\begin{align}
    ||\delta_\alpha(\LL_{\alpha,\not c}^*+V^*)O_s||  &\leq 2\sum_{\substack{B\in\mathcal S_{\LL}\setminus\mathcal S_{\MM_\alpha}\\ [\MM_\alpha,\LL_B]
    \neq0}}\sum_{i=1}^{M_B}||\delta_\alpha\LL_B^*||_{\infty\to\infty}||\delta_{\alpha_{B,i}}O_s||+2\sum_{B\in\mathcal S_V}\sum_{i=1}^{M_B}||\delta_\alpha V_B^*||_{\infty\to\infty}||\delta_{\alpha_{B,i}}O_s||,\notag\\
    &= \sum_{\nu=1}^M K_{\alpha\nu} ||\delta_\nu O_s||.\label{eq:app-bulk-source-bound}
\end{align}
Writing $f_\mu(t):=||\delta_\mu O_t||$, $\mu=1,\ldots,M$, and
\begin{align}
    \mathbf f(t) &:= (f_1(t),\ldots,f_M(t))^\top,
\end{align}
the bound for $||\delta_\alpha O_t||$ takes the vector Volterra form
\begin{align}
    \mathbf f(t) &\leq g(t)\mathbf f(0)+K\int_0^t g(t-s)\mathbf f(s)\,ds. \label{eq:app-direct-vector-volterra}
\end{align}
Summing over components yields
\begin{align}
    \osc{O_t}_{\mathfrak U} &\leq g(t)\osc{O}_{\mathfrak U}+\kappa_{\GG}^{\mathrm{bulk}}\int_0^t g(t-s)\osc{O_s}_{\mathfrak U}\,ds. \label{eq:app-direct-scalar-volterra}
\end{align}
Applying \cref{lem:app-volterra-decay} with $b=\kappa_{\GG}^{\mathrm{bulk}}$ gives
\begin{align}
    \osc{O_t}_{\mathfrak U} &\leq D_{\GG}e^{-\gamma_{\GG}t}\osc{O}_{\mathfrak U}.
\end{align}
Since $||\delta_\alpha(O)||\leq2||O||$, we have
\begin{align}
    \osc{O}_{\mathfrak U} &\leq 2M||O||.
\end{align}
Combining this with \cref{lem:app-trace-oscillator} proves \cref{eq:app-general-bulk-mixing-bound}. Applying the same estimate to two stationary states proves uniqueness.
\end{proof}

It remains to show exponential decay for the staircase kernel obtained from the local mixing time.

\begin{lemma}[Exponential decay for local-mixing-time Volterra inequalities]
\label{lem:app-volterra-decay}
Let \(T>0\), \(b\geq0\), and
\begin{align}
    g(t)&:=2^{-\lfloor t/T\rfloor}.
\end{align}
Suppose \(S:[0,\infty)\to[0,\infty)\) is locally bounded and measurable and satisfies
\begin{align}
    S(t)&\leq g(t)S(0)+b\int_0^t g(t-s)S(s)\,ds \label{eq:app-volterra-S}
\end{align}
for all \(t\geq0\). Define
\begin{align}
    \Delta&:=\frac12-b T.
\end{align}
If \(\Delta>0\), then
\begin{align}
    S(t)&\leq De^{-\gamma t}S(0),\qquad t\geq0,
\end{align}
where
\begin{align}
    \gamma&:=\frac{\Delta}{2T},\qquad
    D:=e^{\Delta/2}\frac{\frac12+b T}{\Delta}. \label{eq:app-volterra-constants}
\end{align}
\end{lemma}

\begin{proof}[Proof of \cref{lem:app-volterra-decay}]
\prooflabel{lem:app-volterra-decay}{proof:app-volterra-decay}
Set
\begin{align}
    g_\gamma(t)&:=e^{\gamma t}g(t),&
    S_\gamma(t)&:=e^{\gamma t}S(t),&
    G_\gamma&:=\int_0^\infty g_\gamma(t)\,dt,&
    M_\gamma&:=\sup_{t\geq0}g_\gamma(t).
\end{align}
Multiplying \cref{eq:app-volterra-S} by \(e^{\gamma t}\) gives
\begin{align}
    S_\gamma(t)&\leq g_\gamma(t)S(0)
    +b\int_0^t g_\gamma(t-s)S_\gamma(s)\,ds. \label{eq:app-weighted-volterra}
\end{align}
For \(R\geq0\), let \(Q_\gamma(R):=\sup_{0\leq t\leq R}S_\gamma(t)\). For every \(t\in[0,R]\),
\begin{align}
    S_\gamma(t)&\leq M_\gamma S(0)+b G_\gamma Q_\gamma(R).
\end{align}
Taking the supremum and rearranging, whenever \(b G_\gamma<1\),
\begin{align}
    S(t)&\leq\frac{M_\gamma}{1-b G_\gamma}e^{-\gamma t}S(0). \label{eq:app-volterra-decay-intermediate}
\end{align}

It remains to bound the two quantities in this estimate. Put
\begin{align}
    x&:=\gamma T=\frac{\Delta}{2}.
\end{align}
Since \(0<x\leq1/4<\log2\), summing over the intervals
\([nT,(n+1)T)\) gives
\begin{align}
    G_\gamma
    &=\sum_{n=0}^\infty2^{-n}\int_{nT}^{(n+1)T}e^{\gamma t}\,dt \notag\\
    &=T\frac{e^x-1}{x\left(1-\frac12e^x\right)}. \label{eq:app-staircase-exponential-moment}
\end{align}
For \(0\leq x<1\), the inequalities
\begin{align}
    e^x&\leq\frac{1}{1-x},\quad
    \frac{e^x-1}{x}\leq\frac{1}{1-x}
\end{align}
imply
\begin{align}
    1-\frac12e^x&\geq\frac{\frac12-x}{1-x},\quad
    G_\gamma\leq\frac{T}{\frac12-x}
    =\frac{2T}{\frac12+b T}.
\end{align}
Consequently,
\begin{align}
    1-b G_\gamma
    &\geq\frac{\frac12-b T}{\frac12+b T}
    =\frac{\Delta}{\frac12+b T}>0. \label{eq:app-staircase-resolvent-margin}
\end{align}
Moreover, for \(nT\leq t<(n+1)T\),
\begin{align}
    g_\gamma(t)&\leq e^x\left(\frac{e^x}{2}\right)^n,
\end{align}
and therefore \(M_\gamma\leq e^x=e^{\Delta/2}\). Substituting these bounds into
\cref{eq:app-volterra-decay-intermediate} yields
\begin{align}
    \frac{M_\gamma}{1-b G_\gamma}
    &\leq e^{\Delta/2}\frac{\frac12+b T}{\Delta}=D.
\end{align}
This proves the claim.
\end{proof}

\subsection{Local perturbations}

The direct estimate obtained from \cref{eq:app-bulk-comparison-matrix} applies to arbitrary summable perturbations. Below we obtain a simpler bound when $V=\sum_{B\in\mathcal S_V}V_B$ is local and the terms supported away from each $A_\alpha$ can be included in a contractive reference evolution. For every $\alpha$, define
\begin{align}
    \LL_{\alpha,c} &:= \sum_{\substack{B\in\mathcal S_{\LL}\setminus\mathcal S_{\MM_\alpha}\\ [\MM_\alpha,\LL_B]
    =0}}\LL_B, & \LL_{\alpha,\not c} &:= \sum_{\substack{B\in\mathcal S_{\LL}\setminus\mathcal S_{\MM_\alpha}\\ [\MM_\alpha,\LL_B]
    \neq0}}\LL_B, \notag\\
    V_{\alpha,\mathrm{far}} &:= \sum_{\substack{B\in\mathcal S_V\\
    B\cap A_\alpha=\varnothing}}V_B, & \mathcal C_\alpha &:= \LL_{\alpha,c}+V_{\alpha,\mathrm{far}}. \label{eq:app-local-perturbation-reference-generator}
\end{align}
Since every term of $V_{\alpha,\mathrm{far}}$ is supported away from $A_\alpha$,
\begin{align}
    [\MM_\alpha,\mathcal C_\alpha] &=0.
\end{align}

We choose $V_{\alpha,\mathrm{far}}$ by disjoint support so that, in the applications below, $\mathcal C_\alpha$ is a Lindblad generator and $e^{t\mathcal C_\alpha^*}$ is a contraction. Commutation with $\MM_\alpha$ alone would not ensure this, since the terms $V_B$ need not be Lindblad generators.

\begin{lemma}[Per-site influence estimate]
\label{lem:app-per-site-influence}
Suppose that $\mathcal C_\alpha$ generates a unital contraction for every $\alpha$:
\begin{align}
    ||e^{t\mathcal C_\alpha^*}||_{\infty\to\infty} &\leq1, \qquad t\geq0. \label{eq:app-local-reference-contraction}
\end{align}
In particular, this holds if every $\mathcal C_\alpha$ is a Lindblad generator. Define
\begin{align}
    \zeta(V) &:= \max_{x\in\Lambda} \sum_{\substack{B\in\mathcal S_V\\
    x\in B}} |B|\,||V_B||_{1\to1}.
\end{align}
Then \cref{eq:app-scalar-volterra} holds with the local coefficient $\kappa_{\GG}^{\mathrm{loc}}$ defined in \cref{eq:app-local-influence}, satisfying
\begin{align}
    \kappa_{\GG}^{\mathrm{loc}} &\leq \kappa_{\LL}^{\sharp}+4a\mathfrak d\,\zeta(V). \label{eq:app-local-perturbation-first-moment-load}
\end{align}
If the summands of $V$ act on at most $r$ sites and
\begin{align}
    s_V &:= \max_{x\in\Lambda} \sum_{\substack{B\in\mathcal S_V\\
    x\in B}} ||V_B||_{1\to1},
\end{align}
then
\begin{align}
    \kappa_{\GG}^{\mathrm{loc}} &\leq \kappa_{\LL}^{\sharp}+4a\mathfrak d\,rs_V. \label{eq:app-local-perturbation-bounded-body-load}
\end{align}
Moreover, if $\LL$ is $l$-local with interaction strength bounded by $s_{\LL}$, then
\begin{align}
    \kappa_{\LL}^{\sharp} &\leq 4a\,s_{\LL,\mathfrak U}\leq 4a\mathfrak d\,ls_{\LL}. \label{eq:app-local-perturbation-base-load}
\end{align}
Consequently,
\begin{align}
    \kappa_{\GG}^{\mathrm{loc}} &\leq 4a\bigl(s_{\LL,\mathfrak U}+\mathfrak d\,rs_V\bigr). \label{eq:app-bounded-body-load}
\end{align}
The rapid-mixing criterion and bounds of \cref{prop:app-rapid-mixing-bulk-cover} therefore also hold with $\kappa_{\GG}^{\mathrm{bulk}}$ replaced by $\kappa_{\GG}^{\mathrm{loc}}$.
\end{lemma}

\begin{proof}[Proof of \cref{lem:app-per-site-influence}]
\prooflabel{lem:app-per-site-influence}{proof:app-per-site-influence}
By the definition of $\mathcal C_\alpha$,
\begin{align}
    \GG-\MM_\alpha-\mathcal C_\alpha &= \LL_{\alpha,\not c}+\sum_{\substack{B\in\mathcal S_V\\
    B\cap A_\alpha\neq\varnothing}} V_B. \label{eq:app-local-perturbation-source-generator}
\end{align}
Let $S_t^\alpha:=e^{t(\MM_\alpha^*+\mathcal C_\alpha^*)}$ and define $O_t:=e^{t\GG^*}(O)$. Since $[\MM_\alpha^*,\delta_\alpha]=[\mathcal C_\alpha^*,\delta_\alpha]=0$ under composition, we have
\begin{align}
    \frac{d}{ds}\Bigl(S_{t-s}^\alpha\delta_\alpha O_s\Bigr) &= -S_{t-s}^\alpha(\MM_\alpha^*+\mathcal C_\alpha^*)\delta_\alpha O_s+S_{t-s}^\alpha\delta_\alpha\GG^* O_s, \notag\\
    &= S_{t-s}^\alpha\delta_\alpha(\GG^*-\MM_\alpha^*-\mathcal C_\alpha^*)O_s.
\end{align}
Integrating from $s=0$ to $s=t$ yields
\begin{align}
    \delta_\alpha O_t &= S_t^\alpha\delta_\alpha O+\int_0^t S_{t-s}^\alpha\delta_\alpha(\GG^*-\MM_\alpha^*-\mathcal C_\alpha^*)O_s\,ds. \label{eq:app-local-perturbation-duhamel}
\end{align}
Because $\MM_\alpha$ and $\mathcal C_\alpha$ commute, $S_t^\alpha=e^{t\mathcal C_\alpha^*}e^{t\MM_\alpha^*}$, and therefore
\begin{align}
    ||S_t^\alpha\delta_\alpha O|| &\overset{(1)}{\leq} ||e^{t\MM_\alpha^*}\delta_\alpha O||, \notag\\
    &\overset{(2)}{\leq} g(t)||\delta_\alpha O||,
\end{align}
where $(1)$ uses \cref{eq:app-local-reference-contraction}, and $(2)$ uses \cref{eq:app-iterated-local-contraction}. For the integral term, define
\begin{align}
    X_{\alpha,s}&:=\delta_\alpha(\GG^*-\MM_\alpha^*-\mathcal C_\alpha^*)O_s.
\end{align}
The projector identity $\delta_\alpha^2=\delta_\alpha$ gives $X_{\alpha,s}=\delta_\alpha X_{\alpha,s}$. Thus commutation of $\MM_\alpha$ and $\mathcal C_\alpha$, \cref{eq:app-local-reference-contraction}, and \cref{eq:app-iterated-local-contraction} give
\begin{align}
    \left\|S_{t-s}^\alpha X_{\alpha,s}\right\|
    &\leq\left\|e^{(t-s)\MM_\alpha^*}\delta_\alpha X_{\alpha,s}\right\|
    \leq g(t-s)\left\|X_{\alpha,s}\right\|. \label{eq:app-local-integral-kernel}
\end{align}

Applying \cref{eq:app-single-telescope-source-bound} term by term to \cref{eq:app-local-perturbation-source-generator} gives
\begin{align}
    &||\delta_\alpha(\GG^*-\MM_\alpha^*-\mathcal C_\alpha^*)O_s|| \notag\\
    &\quad\leq 2\sum_{\substack{B\in\mathcal S_{\LL}\setminus\mathcal S_{\MM_\alpha}\\ [\MM_\alpha,\LL_B]
    \neq0}}\sum_{i=1}^{M_B}||\delta_\alpha\LL_B^*||_{\infty\to\infty}||\delta_{\alpha_{B,i}}O_s||+2\sum_{\substack{B\in\mathcal S_V\\
    B\cap A_\alpha\neq\varnothing}}\sum_{i=1}^{M_B}||\delta_\alpha V_B^*||_{\infty\to\infty}||\delta_{\alpha_{B,i}}O_s||. \label{eq:app-local-perturbation-source-bound}
\end{align}
Writing $f_\mu(t):=||\delta_\mu O_t||$, $\mu=1,\ldots,M$, define
\begin{align}
    K^{\mathrm{loc}}_{\alpha\nu} &:= 2\sum_{\substack{B\in\mathcal S_{\LL}\setminus\mathcal S_{\MM_\alpha}\\ [\MM_\alpha,\LL_B]
    \neq0}}\sum_{i=1}^{M_B}\mathbf 1[\alpha_{B,i}=\nu]||\delta_\alpha\LL_B^*||_{\infty\to\infty} +2\sum_{\substack{B\in\mathcal S_V\\
    B\cap A_\alpha\neq\varnothing}}\sum_{i=1}^{M_B}\mathbf 1[\alpha_{B,i}=\nu]||\delta_\alpha V_B^*||_{\infty\to\infty}, \notag\\
    \mathbf f(t) &:= (f_1(t),\ldots,f_M(t))^\top.
\end{align}
Set
\begin{align}
    \kappa_{\GG}^{\mathrm{loc}} &:= \max_{1\leq\nu\leq M}\sum_{\alpha=1}^M K^{\mathrm{loc}}_{\alpha\nu}. \label{eq:app-local-influence}
\end{align}
The superscript ``loc'' distinguishes this sharper coefficient from the general incoming influence $\kappa_{\GG}^{\mathrm{bulk}}$ in \cref{eq:app-general-influence-load}. The general comparison matrix charges every summable term $V_B$. In the present local-perturbation decomposition, terms with $B\cap A_\alpha=\varnothing$ have instead been absorbed into the contractive reference generator $\mathcal C_\alpha$, so $K^{\mathrm{loc}}$ charges only the remaining terms with $B\cap A_\alpha\neq\varnothing$. Thus $\kappa_{\GG}^{\mathrm{loc}}$ is the exact maximum column sum for this local decomposition, and $\kappa_{\GG}^{\mathrm{loc}}\leq\kappa_{\GG}^{\mathrm{bulk}}$ for the same split and covering lists. At $V=0$, both matrices coincide and $\kappa_{\LL}^{\mathrm{loc}}=\kappa_{\LL}^{\sharp}$. For the individual-term covers of the main text, $\kappa_{\LL}^{\mathrm{loc}}\leq\kappa_{\LL}$, with $\kappa_{\LL}$ defined in \cref{eq:main-local-influence}.
The bound for $||\delta_\alpha O_t||$ takes the vector Volterra form
\begin{align}
    \mathbf f(t) &\leq g(t)\mathbf f(0)+K^{\mathrm{loc}}\int_0^t g(t-s)\mathbf f(s)\,ds. \label{eq:app-vector-volterra}
\end{align}
Summing over components yields
\begin{align}
    \osc{O_t}_{\mathfrak U} &\leq g(t)\osc{O}_{\mathfrak U}+\kappa_{\GG}^{\mathrm{loc}}\int_0^t g(t-s)\osc{O_s}_{\mathfrak U}\,ds. \label{eq:app-scalar-volterra}
\end{align}

The base-generator contribution to $\kappa_{\GG}^{\mathrm{loc}}$ is bounded by $\kappa_{\LL}^{\sharp}$. Since $||\delta_\alpha||_{\infty\to\infty}\leq2$, the perturbation contribution is bounded by the maximum column sum of
\begin{align}
    K^V_{\alpha\nu} &:= 4\sum_{\substack{B\in\mathcal S_V\\
    B\cap A_\alpha\neq\varnothing}} \sum_{i=1}^{M_B} \mathbf 1[\alpha_{B,i}=\nu]||V_B||_{1\to1}. \label{eq:app-local-perturbation-comparison-matrix}
\end{align}
For fixed $\nu$,
\begin{align}
    \sum_{\alpha=1}^M K^V_{\alpha\nu} &\overset{(1)}{\leq} 4 \sum_{B\in\mathcal S_V}\left(\sum_{i=1}^{M_B}\mathbf 1[\alpha_{B,i}=\nu]\right)\#\left\{\alpha:B\cap A_\alpha\neq\varnothing\right\} ||V_B||_{1\to1}, \notag\\
    &\overset{(2)}{\leq} 4\mathfrak d\sum_{B\in\mathcal S_V}|A_\nu^\circ\cap B||B|\,||V_B||_{1\to1}, \notag\\
    &\overset{(3)}{\leq} 4\mathfrak d\sum_{x\in A_\nu^\circ}\sum_{B\ni x}|B|\,||V_B||_{1\to1}, \notag\\
    &\overset{(4)}{\leq}4a\mathfrak d\,\zeta(V), \label{eq:app-local-perturbation-column-sum}
\end{align}
where $(1)$ sums the comparison matrix over $\alpha$, $(2)$ uses $\sum_{i=1}^{M_B}\mathbf 1[\alpha_{B,i}=\nu]\leq|A_\nu^\circ\cap B|$ and
\begin{align}
    \#\left\{\alpha:B\cap A_\alpha\neq\varnothing\right\}&\leq \sum_{x\in B}\#\left\{\alpha:x\in A_\alpha\right\} \leq \mathfrak d|B|,
\end{align}
$(3)$ writes $|A_\nu^\circ\cap B|$ as a sum over sites, and $(4)$ uses $|A_\nu^\circ|\leq|A_\nu|\leq a$. This proves \cref{eq:app-local-perturbation-first-moment-load}. Since $\zeta(V)\leq rs_V$, it also proves \cref{eq:app-local-perturbation-bounded-body-load}.

For the base-generator estimate, use $||\delta_\alpha\LL_B^*||_{\infty\to\infty}\leq2||\LL_B||_{1\to1}$ in the first sum of \cref{eq:app-local-perturbation-source-bound} and take
\begin{align}
    K^{\LL}_{\alpha\nu} &:= 4\sum_{\substack{B\in\mathcal S_{\LL}\setminus\mathcal S_{\MM_\alpha}\\ [\MM_\alpha,\LL_B]
    \neq0}} \sum_{i=1}^{M_B} \mathbf 1[\alpha_{B,i}=\nu]||\LL_B||_{1\to1}. \label{eq:app-local-perturbation-base-comparison-matrix}
\end{align}
Exchanging the sums and using \cref{eq:app-list-multiplicity} gives, for every $\nu$,
\begin{align}
    \sum_{\alpha=1}^M K^{\LL}_{\alpha\nu} &=4\sum_{B\in\mathcal S_{\LL}}\left(\sum_{i=1}^{M_B}\mathbf 1[\alpha_{B,i}=\nu]\right)n_{\mathfrak U}(B)||\LL_B||_{1\to1}, \notag\\
    &\leq4\sum_{x\in A_\nu^\circ}\sum_{\substack{B\in\mathcal S_{\LL}\\
    x\in B}}n_{\mathfrak U}(B)||\LL_B||_{1\to1}, \notag\\
    &\leq4|A_\nu^\circ|s_{\LL,\mathfrak U}\leq4a\,s_{\LL,\mathfrak U}. \label{eq:app-base-column-sum-sharp}
\end{align}
Noncommutation implies $B\cap A_\alpha\neq\varnothing$, so $n_{\mathfrak U}(B)\leq\mathfrak d|B|\leq\mathfrak d\,l$ and $s_{\LL,\mathfrak U}\leq\mathfrak d\,ls_{\LL}$. Hence
\begin{align}
    \kappa_{\LL}^{\sharp} &\leq \max_\nu\sum_{\alpha=1}^M K^{\LL}_{\alpha\nu}\leq4a\,s_{\LL,\mathfrak U}\leq 4a\mathfrak d\,ls_{\LL}.
\end{align}
This proves \cref{eq:app-local-perturbation-base-load}. Adding this comparison matrix to \cref{eq:app-local-perturbation-comparison-matrix} proves \cref{eq:app-bounded-body-load}. Applying \cref{lem:app-volterra-decay} with $b=\kappa_{\GG}^{\mathrm{loc}}$ to \cref{eq:app-scalar-volterra}, followed by \cref{lem:app-trace-oscillator}, gives the stated rapid-mixing criterion and bounds.
\end{proof}

\begin{proposition}[General form of \cref{prop:stability-strong-bulk-dissipation}: Stability of strong bulk dissipation]
\label{app:prop:stability-strong-bulk-dissipation}
Let $\LL$ be a local Lindbladian containing a bulk-dissipative cover $\mathfrak U$ with local mixing time $t_{\mathrm{loc}}$ and cover coupling strength $s_{\LL,\mathfrak U}$. Let $\KK$ be a local Lindbladian with cover coupling strength $s_{\KK,\mathfrak U}$. For $\varepsilon\geq0$, use the coefficients $\kappa_{\LL}^{\sharp}$ and $\kappa_{\LL+\varepsilon\KK}^{\mathrm{loc}}$ from \cref{eq:app-bulk-incoming-L,eq:app-local-influence}, respectively, with the commuting perturbation terms included in $\mathcal C_\alpha$. Suppose that
\begin{align}
    \kappa_{\LL}^{\sharp}t_{\mathrm{loc}}&<\frac12. \label{app:eq:strong-bulk-unperturbed-condition}
\end{align}
Then $\LL$ has a unique fixed point. If, in addition,
\begin{align}
    \kappa_{\LL+\varepsilon\KK}^{\mathrm{loc}} t_{\mathrm{loc}}&<\frac12, \label{app:eq:strong-bulk-stability-condition}
\end{align}
then $\LL+\varepsilon\KK$ has a unique fixed point. For families $(\LL_N)_N$ and $(\KK_N)_N$ with $\LL=\LL_N$, $\KK=\KK_N$ and $M=\operatorname{poly}(|\Lambda|)$, either $(\LL_N)_N$ or $(\LL_N+\varepsilon\KK_N)_N$ is rapidly mixing if its corresponding margin, $1/2-\kappa_{\LL}^{\sharp}t_{\mathrm{loc}}$ or $1/2-\kappa_{\LL+\varepsilon\KK}^{\mathrm{loc}}t_{\mathrm{loc}}$, is bounded below by a positive constant and $t_{\mathrm{loc}}=O(\operatorname{polylog}|\Lambda|)$. The decay rate is constant when $t_{\mathrm{loc}}=O(1)$.
\end{proposition}

The main text restricts the cover generators to individual terms of $\LL$ and uses the explicit coefficient $\kappa_{\LL}$ in \cref{eq:main-local-influence}. For these covers, $n_{\mathfrak U}(B)=n_{\LL}(B)$. Each covering list has no repetitions and contains $\nu$ only if $B\cap A_\nu^\circ\neq\varnothing$. Hence $\|\delta_\alpha\LL_B^*\|_{\infty\to\infty}\leq2\|\LL_B\|_{1\to1}$ and \cref{eq:app-local-perturbation-base-comparison-matrix} give
\begin{align}
    \kappa_{\LL}^{\sharp}
    &\leq4\max_\nu\sum_{\substack{B\in\mathcal S_{\LL}\\ B\cap A_\nu^\circ\neq\varnothing}}n_{\LL}(B)\|\LL_B\|_{1\to1}
    =\kappa_{\LL}. \label{eq:app-sharp-main-comparison}
\end{align}
For the Lindbladian perturbation, include all terms commuting with $\MM_\alpha$ in $\mathcal C_\alpha$, as in the proof below. The same comparison gives $\kappa_{\LL+\varepsilon\KK}^{\mathrm{loc}}\leq\kappa_{\LL+\varepsilon\KK}\leq\kappa_{\LL}+\varepsilon\kappa_{\KK}$, with the main-text coefficients computed for the same fixed cover and decomposition. The increment bound
\begin{align}
    \kappa_{\LL+\varepsilon\KK}&\leq\kappa_{\LL}+4a\mathfrak d\,\varepsilon k s_{\KK}
\end{align}
is used there only to turn an assumed positive margin for $\kappa_{\LL}t_{\mathrm{loc}}<1/2$ into the explicit perturbation threshold in \cref{eq:main-constant-perturbation-threshold}. The separate estimate $\kappa_{\LL}\leq4a\mathfrak d\,l s_{\LL}$ records a size-independent bound from the local parameters; it is not used in the stability condition.

\begin{proof}[Proof of \cref{prop:stability-strong-bulk-dissipation}]
\prooflabel{prop:stability-strong-bulk-dissipation}{proof:stability-strong-bulk-dissipation}
We can write $\mathcal{G}=\LL + \varepsilon \KK = \LL +V$, with $V = \sum_{B\in \mathcal{S}_\GG} V_B$, padding with possibly zero terms. For every $\alpha$ the map
\begin{align}
    \mathcal C_\alpha &= \LL_{\alpha,c}+\varepsilon\sum_{\substack{B\in\mathcal S_{\KK}\\
    [\MM_\alpha,\KK_B]=0}}\KK_B
\end{align}
is a Lindblad generator that commutes with $\MM_\alpha$. Applying the termwise comparison-matrix bound \cref{eq:app-local-perturbation-comparison-matrix} to $V=\varepsilon\KK$, the perturbation part of every column satisfies
\begin{align}
    \sum_{\alpha=1}^M K^{\KK}_{\alpha\nu}
    &=4\varepsilon\sum_{B\in\mathcal S_{\KK}}\left(\sum_{i=1}^{M_B}\mathbf 1[\alpha_{B,i}=\nu]\right)\#\{\alpha:[\MM_\alpha,\KK_B]\neq0\}||\KK_B||_{1\to1}\notag\\
    &\leq4a\varepsilon s_{\KK,\mathfrak U}. \label{eq:app-strong-bulk-perturbation-column}
\end{align}
Together with \cref{eq:app-local-perturbation-base-load}, this gives
\begin{align}
    \kappa_{\LL+\varepsilon\KK}^{\mathrm{loc}} &\leq 4a\bigl(s_{\LL,\mathfrak U}+\varepsilon s_{\KK,\mathfrak U}\bigr).
\end{align}
The two hypotheses are precisely the strict criterion in \cref{eq:app-general-bulk-condition} for the exact influence parameters $\kappa_{\LL}^{\sharp}$ and $\kappa_{\LL+\varepsilon\KK}^{\mathrm{loc}}$, respectively. For either coefficient, \cref{eq:app-scalar-volterra} therefore satisfies the hypotheses of \cref{lem:app-volterra-decay}; its oscillator decay, combined with \cref{lem:app-trace-oscillator}, proves the two conclusions and the stated family scalings. For the individual-term covers, the bounds $\kappa_{\LL}^{\sharp}\leq\kappa_{\LL}$ and $\kappa_{\LL+\varepsilon\KK}^{\mathrm{loc}}\leq\kappa_{\LL+\varepsilon\KK}$ prove the main-text statement as well. Indeed, choose a common upper bound $T$ on the local mixing times. The threshold $\varepsilon_*=\Delta/(8a\mathfrak d\,k s_{\KK}T)$ preserves a margin at least $\Delta/2$ for every perturbation obeying the fixed strength bound and every $0\leq\varepsilon\leq\varepsilon_*$. For $s_{\KK}=0$, any positive threshold works.
\end{proof}

\begin{corollary}[General form of \cref{cor:commuting-strong-bulk-dissipation}: Commuting strong bulk dissipation]
\label{app:cor:commuting-strong-bulk-dissipation}
Let $\LL=\sum_{B\in\mathcal S_{\LL}}\LL_B$ be a local Lindbladian containing a bulk-dissipative cover $\mathfrak U$ with local mixing time $t_{\mathrm{loc}}$. Suppose additionally that its local terms commute:
\begin{align}
    [\LL_B,\LL_{B'}]&=0\qquad\text{for every }B,B'\in\mathcal S_{\LL}.
\end{align}
Let $\KK$ be a local Lindbladian with cover coupling strength $s_{\KK,\mathfrak U}$. For $\varepsilon\geq0$, let $\kappa_\varepsilon^{\mathrm{com}}:=\kappa_{\LL+\varepsilon\KK}^{\mathrm{loc}}$. Then
\begin{align}
    \kappa_\varepsilon^{\mathrm{com}}&\leq4a\varepsilon s_{\KK,\mathfrak U}.
\end{align}
If
\begin{align}
    \kappa_\varepsilon^{\mathrm{com}}t_{\mathrm{loc}}&<\frac12, \label{app:eq:commuting-strong-bulk-stability-condition}
\end{align}
then $\LL+\varepsilon\KK$ has a unique fixed point. For families $(\LL_N)_N$ and $(\KK_N)_N$ with $\LL=\LL_N$, $\KK=\KK_N$ and $M=\operatorname{poly}(|\Lambda|)$, a size-independent positive margin and $t_{\mathrm{loc}}=O(\operatorname{polylog}|\Lambda|)$ imply rapid mixing of $(\LL_N+\varepsilon\KK_N)_N$; the decay rate is constant when $t_{\mathrm{loc}}=O(1)$.
\end{corollary}

\begin{proof}[Proof of \cref{cor:commuting-strong-bulk-dissipation}]
\prooflabel{cor:commuting-strong-bulk-dissipation}{proof:commuting-strong-bulk-dissipation}
Pairwise commutativity gives $\kappa_{\LL}^{\sharp}=0$, while \cref{eq:app-strong-bulk-perturbation-column} gives $\kappa_\varepsilon^{\mathrm{com}}\leq4a\varepsilon s_{\KK,\mathfrak U}$. The result follows directly from \cref{app:prop:stability-strong-bulk-dissipation}. For individual-term covers, $n_{\LL}(B)=0$ for every $B$, so $\kappa_{\LL}=0$ and \cref{prop:stability-strong-bulk-dissipation} gives the main-text corollary. In particular, when $M=\operatorname{poly}(|\Lambda|)$ and $t_{\mathrm{loc}}=O(\operatorname{polylog}|\Lambda|)$, the unperturbed commuting generator is rapidly mixing without any restriction on the growth of $a$.
\end{proof}

We can now apply the results above to several perturbations of interest. In each application, the cover $\mathfrak U$ and its quasiderivations remain fixed. Successive perturbations are combined into their total difference from $\LL$; the triangle inequality bounds the incoming influence and per-site strengths of this difference by the sums of the individual bounds. The same cover therefore continues to give rapid mixing with constant decay rate whenever the total influence satisfies the corresponding mixing condition.

\begin{corollary}[Difference of local Lindbladians]
\label{cor:app-difference}\resulttoc{cor:app-difference}{Difference of local Lindbladians}
Let $(\LL_N)_N$ and $(\GG_N)_N$ be local Lindbladian families, and write $\LL=\LL_N$ and $\GG=\GG_N$ at each size. After adding zero terms so that their local decompositions have the same supports, write
\begin{align}
    \LL &=\sum_{B\in\mathcal S}\LL_B, & \GG&=\sum_{B\in\mathcal S}\GG_B, & V&:=\GG-\LL=\sum_{B\in\mathcal S}V_B, & V_B&:=\GG_B-\LL_B.
\end{align}
Suppose that $\LL$ is $l$-local with interaction strength bounded by $s_{\LL}$ and contains a bulk-dissipative cover $\mathfrak U$ with local mixing time $t_{\mathrm{loc}}$ in the displayed decomposition, and that $\GG$ is $k$-local with interaction strength bounded by $s_{\GG}$. Assume that $t_{\mathrm{loc}},a,\mathfrak d$ are bounded uniformly in the system size. We use the same cover to study $\GG=\LL+V$. Let $r:=\max\{l,k\}$ and define
\begin{align}
    s_V &:= \max_{x\in\Lambda} \sum_{B\ni x}||V_B||_{1\to1}.
\end{align}
Then
\begin{align}
    \kappa_{\GG}^{\mathrm{loc}} &\leq\kappa_{\LL}^{\sharp}+4a\mathfrak d\,r s_V. \label{eq:app-difference-sharp-load}
\end{align}
Moreover,
\begin{align}
    \kappa_{\LL}^{\sharp} &\leq 4a\,s_{\LL,\mathfrak U}, & s_V &\leq s_{\LL}+s_{\GG}, \notag\\
    \kappa_{\GG}^{\mathrm{loc}} &\leq4a\bigl(s_{\LL,\mathfrak U}+\mathfrak d\,r(s_{\LL}+s_{\GG})\bigr). && \label{eq:app-difference-coarse-load}
\end{align}
If
\begin{align}
    \kappa_{\GG}^{\mathrm{loc}}t_{\mathrm{loc}}&\leq\frac12-\Delta,
\end{align}
for some $\Delta>0$ independent of the system size, then $(\GG_N)_N$ has a unique fixed point at each size and is rapidly mixing with constant decay rate.
\end{corollary}

\begin{proof}[Proof of \cref{cor:app-difference}]
\prooflabel{cor:app-difference}{proof:app-difference}
Since each local Lindbladian is trace-annihilating on its support,
\begin{align}
    V_B^*(\BI_B\otimes O_{B^c}) &= \GG_B^*(\BI_B\otimes O_{B^c})-\LL_B^*(\BI_B\otimes O_{B^c})=0.
\end{align}
Thus $V$ is locally trace-annihilating. For every $\alpha$, the reference generator in \cref{eq:app-local-perturbation-reference-generator} is
\begin{align}
    \mathcal C_\alpha &= \sum_{\substack{B\in\mathcal S\setminus\mathcal S_{\MM_\alpha}\\
    B\cap A_\alpha\neq\varnothing\\ [\MM_\alpha,\LL_B]
    =0}}\LL_B+\sum_{\substack{B\in\mathcal S\\
    B\cap A_\alpha=\varnothing}}\GG_B.
\end{align}
This is a Lindblad generator. Therefore \cref{lem:app-per-site-influence} gives \cref{eq:app-scalar-volterra} and \cref{eq:app-difference-sharp-load}. \cref{eq:app-local-perturbation-base-load} gives the bound on $\kappa_{\LL}^{\sharp}$, while the triangle inequality gives
\begin{align}
    s_V &\leq \max_{x\in\Lambda}\sum_{B\ni x}\bigl(||\GG_B||_{1\to1}+||\LL_B||_{1\to1}\bigr)\leq s_{\GG}+s_{\LL}.
\end{align}
These two estimates give \cref{eq:app-difference-coarse-load}. The conclusion follows by applying \cref{lem:app-volterra-decay,lem:app-trace-oscillator} to \cref{eq:app-scalar-volterra} with the exact coefficient $\kappa_{\GG}^{\mathrm{loc}}$. The assumed margin gives a decay rate at least $\Delta/(2t_{\mathrm{loc}})$ and a uniformly bounded Volterra prefactor. Since $M\leq\mathfrak d|\Lambda|$, the trace-norm prefactor grows at most linearly with the system size.
\end{proof}

Perturbations of jump operators have also been controlled in proofs of rapid mixing for Gibbs samplers at high temperature \cites[Appendix~C]{RouzeFrancaAlhambra2024b}[Appendix~A]{SmidMeisterBertaBondesan2025rm}, as well as in ground-state preparation for weak perturbations of noninteracting fermionic Hamiltonians \cite[Appendix~I]{ZhanDingHuhnGrayPreskillChanLin2025}. In those works, the estimates concern the filtered jump operators associated with their specific state-preparation constructions.

\begin{corollary}[Perturbation of Hamiltonians and jump operators]
\label{cor:app-perturbed-hamiltonians-jumps}\resulttoc{cor:app-perturbed-hamiltonians-jumps}{Perturbation of Hamiltonians and jump operators}
Consider the Lindbladian families $(\GG_{\varepsilon,N})_N$, with $\GG_\varepsilon=\GG_{\varepsilon,N}$ at each size and
\begin{align}
    \GG_\varepsilon(\rho) &:= -i[H_0+\varepsilon H_1,\rho] +\sum_{B,\mu} \Diss{L_{B,\mu,0}+\varepsilon L_{B,\mu,1}}(\rho), \qquad H_j=\sum_Xh_{X,j},
\end{align}
where $\varepsilon\in\mathbb R$, $\mu$ labels the different jump operators supported on $B$, $h_{X,j}=h_{X,j}^\dagger$ for $j=0,1$, and all displayed local terms have supports of size at most $r$. Let $\LL:=\GG_0$, with the local decomposition given by the displayed Hamiltonian terms and dissipators, and suppose that this decomposition contains a bulk-dissipative cover $\mathfrak U$ with local mixing time $t_{\mathrm{loc}}$. Assume that $t_{\mathrm{loc}},a,\mathfrak d$ are bounded uniformly in the system size. Define
\begin{align}
    s_1 &:= \max_{x\in\Lambda} \left( 2\sum_{X\ni x}||h_{X,1}|| +4\sum_{\substack{B\ni x\\
    \mu}} ||L_{B,\mu,0}||||L_{B,\mu,1}|| \right), \notag\\
    s_2 &:= 2\max_{x\in\Lambda} \sum_{\substack{B\ni x\\
    \mu}} ||L_{B,\mu,1}||^2.
\end{align}
The same cover is used for every $\varepsilon$, and \cref{eq:app-local-perturbation-base-load} gives $\kappa_{\LL}^{\sharp}\leq4a\,s_{\LL,\mathfrak U}$. If
\begin{align}
    \bigl(\kappa_{\LL}^{\sharp}+4a\mathfrak d\,r(|\varepsilon|s_1+\varepsilon^2s_2)\bigr)t_{\mathrm{loc}}&\leq\frac12-\Delta,
\end{align}
for some $\Delta>0$ independent of the system size, then $(\GG_{\varepsilon,N})_N$ has a unique fixed point at each size and is rapidly mixing with constant decay rate. No tracelessness assumption on the jump operators is required.
\end{corollary}

\begin{proof}[Proof of \cref{cor:app-perturbed-hamiltonians-jumps}]
\prooflabel{cor:app-perturbed-hamiltonians-jumps}{proof:app-perturbed-hamiltonians-jumps}
For operators $L_0,L_1$ supported on $B$, define
\begin{align}
    \mathcal C[L_0,L_1](\rho) &:= L_0\rho L_1^\dagger+L_1\rho L_0^\dagger-\frac12\{L_0^\dagger L_1+L_1^\dagger L_0,\rho\}.
\end{align}
Expanding the dissipator gives
\begin{align}
    \Diss{L_0+\varepsilon L_1} &= \Diss{L_0}+\varepsilon\mathcal C[L_0,L_1]+\varepsilon^2\Diss{L_1}. \label{eq:app-jump-expansion}
\end{align}
If $X=\BI_B\otimes O_{B^c}$, then $X$ commutes with $L_0$ and $L_1$, and hence
\begin{align}
    \mathcal C[L_0,L_1]^*(X) &= L_0^\dagger XL_1+L_1^\dagger XL_0-\frac12\{L_0^\dagger L_1+L_1^\dagger L_0,X\}=0.
\end{align}
Thus the cross term is locally trace-annihilating. Moreover,
\begin{align}
    ||\mathcal C[L_0,L_1]||_{1\to1} &\leq 4||L_0||||L_1||, \notag\\
    ||\Diss{L_1}||_{1\to1} &\leq 2||L_1||^2, \notag\\
    ||-i[h,\,\cdot\,]||_{1\to1} &\leq 2||h||, \label{eq:app-jump-local-norms}
\end{align}
where the estimates follow from $||A\rho B||_1\leq||A||||B||||\rho||_1$ and the triangle inequality. It follows that $\GG_\varepsilon-\LL$ is locally trace-annihilating and has local interaction strength at most $|\varepsilon|s_1+\varepsilon^2s_2$.

Align the displayed Hamiltonian terms and dissipators, adding zero terms when necessary, and write
\begin{align}
    \LL &= \sum_{B\in\mathcal S}\LL_B, & \GG_\varepsilon &= \sum_{B\in\mathcal S}\GG_{\varepsilon,B}, & V_B&:=\GG_{\varepsilon,B}-\LL_B.
\end{align}
For every $\alpha$, the reference generator in \cref{eq:app-local-perturbation-reference-generator} is
\begin{align}
    \mathcal C_\alpha &= \sum_{\substack{B\in\mathcal S\setminus\mathcal S_{\MM_\alpha}\\
    B\cap A_\alpha\neq\varnothing\\ [\MM_\alpha,\LL_B]
    =0}}\LL_B+\sum_{\substack{B\in\mathcal S\\
    B\cap A_\alpha=\varnothing}}\GG_{\varepsilon,B}.
\end{align}
Every term in this expression is a Lindblad generator. Hence \cref{lem:app-per-site-influence} gives $\kappa_{\GG_\varepsilon}^{\mathrm{loc}}\leq \kappa_{\LL}^{\sharp}+4a\mathfrak d\,r(|\varepsilon|s_1+\varepsilon^2s_2)$. The assumed margin and \cref{lem:app-volterra-decay,lem:app-trace-oscillator} give a decay rate at least $\Delta/(2t_{\mathrm{loc}})$ and a trace-norm prefactor of order $M\leq\mathfrak d|\Lambda|$, proving the stated uniform mixing bound. The calculation also shows directly that no tracelessness assumption is needed.
\end{proof}

\begin{corollary}[Stability of the Gibbs sampler in \cite{RouzeFrancaAlhambra2024b}]
\label{cor:app-rfa-gibbs-sampler-stability}\resulttoc{cor:app-rfa-gibbs-sampler-stability}{Stability of the Gibbs sampler}
Let $\Lambda$ be a $D$-dimensional lattice of qubits and let $H=\sum_Xh_X$ have finite interaction range and be $(k,l)$-local: each term acts on at most $k$ sites, and each site meets at most $l$ terms. Assume $||h_X||\leq h$, with the range and $D,h,k,l$ bounded uniformly in the system size. Write $\LL^{(\beta)}=\LL_N^{(\beta)}$ for the family $(\LL_N^{(\beta)})_N$ of Gibbs samplers in \cite{RouzeFrancaAlhambra2024b} with single-qubit Pauli bath couplings. The infinite-temperature generator $\LL^{(0)}$ contains the singleton bulk-dissipative cover $\mathfrak U_0$ defined in \cref{eq:app-rfa-bulk-cover}, with local mixing time $t_{\mathrm{loc}}^{(0)}=\sqrt{2}e^{1/4}\log4$ and $a=\mathfrak d=1$. There exists a size-independent $\beta_*>0$ such that, for every fixed locality $r$ and interaction-strength bound $s_{\KK}$, there is a size-independent $\varepsilon_*>0$ with the following property: for every $r$-local Lindbladian family $(\KK_N)_N$ with interaction strength at most $s_{\KK}$, $(\LL_N^{(\beta)}+\varepsilon\KK_N)_N$ has a unique fixed point at each size and is rapidly mixing with constant decay rate whenever $0\leq\beta<\beta_*$ and $0\leq\varepsilon<\varepsilon_*$.
\label{eq:app-rfa-stability-radius}
\end{corollary}

\begin{proof}[Proof of \cref{cor:app-rfa-gibbs-sampler-stability}]
\prooflabel{cor:app-rfa-gibbs-sampler-stability}{proof:app-rfa-gibbs-sampler-stability}
Write $\KK=\KK_N$ and define
\begin{align}
    \lambda_0 &:= \frac{1}{\sqrt{2}e^{1/4}},\qquad t_{\mathrm{loc}}^{(0)}:=\lambda_0^{-1}\log4.
\end{align}
Define the cover
\begin{align}
    \mathfrak U_0 &:= \left\{\bigl(\{a\},\{a\},\MM_a\bigr):a\in\Lambda\right\}, & \MM_a^*(O) &:= \lambda_0\left(\frac12\BI_a\otimes\Tr_a(O)-O\right). \label{eq:app-rfa-bulk-cover}
\end{align}

We first verify that $\mathfrak U_0$ is a bulk-dissipative cover with local mixing time $t_{\mathrm{loc}}^{(0)}$ and that $\LL^{(0)}$ contains it. Using $ \mathcal R_a(O) := \frac12\BI_a\otimes\Tr_a(O), \MM_a^* = \lambda_0(\mathcal R_a-\operatorname{id})$, at infinite temperature we have
\begin{align}
    \LL^{(0)} &= \sum_{a\in\Lambda}\MM_a.
\end{align}
The asymptotic projection of $\MM_a^*$ is $P_{\MM_a}^*=\mathcal R_a$, and
\begin{align}
    e^{t\MM_a^*}-P_{\MM_a}^* &=e^{-\lambda_0t}\left(\operatorname{id}-\mathcal R_a\right).
\end{align}
Therefore
\begin{align}
    ||e^{t\MM_a}-P_{\MM_a}||_\diamond &\leq2e^{-\lambda_0t}.
\end{align}
Moreover,
\begin{align}
    \ker\MM_a^* &=\BI_a\otimes\mathcal B(\mathcal H_{a^c}).
\end{align}
At $t=t_{\mathrm{loc}}^{(0)}$, this diamond-norm bound is at most $1/2$. The regions $A_a=A_a^\circ=\{a\}$ cover the lattice, each $\MM_a$ is a term of $\LL^{(0)}$, and their size and incidence are both one. Hence $\mathfrak U_0$ is the asserted bulk-dissipative cover and $\LL^{(0)}$ contains it.

Let $\kappa_\beta$ denote the size-independent coefficient in the commutator and diagonal-difference bounds below. We now combine the commutator estimate of \cite[Appendix~B.1]{RouzeFrancaAlhambra2024b} with the local perturbation estimate. Write $\LL^{(\beta)}=\sum_b\LL_b^{(\beta)}$, where each bath-site term $\LL_b^{(\beta)}$ is a Lindbladian, and put $\delta_a=\operatorname{id}-\mathcal R_a$. Define
\begin{align}
    \mathcal F_a^{(\beta)} &:=\delta_a((\LL_a^{(\beta)})^*-\MM_a^*) +\sum_{b\neq a}[\delta_a,(\LL_b^{(\beta)})^*].
\end{align}
In this work \cite[Eq.~(25)]{RouzeFrancaAlhambra2024b}, the commutators and the diagonal difference are bounded separately. Summing those estimates gives, for every observable $O$,
\begin{align}
    \sum_a||\mathcal F_a^{(\beta)}(O)|| &\leq\kappa_\beta\sum_c||\delta_c O||. \label{eq:app-rfa-source-bound}
\end{align}
For the perturbation, the singleton reset telescope used in \cref{lem:app-per-site-influence} gives
\begin{align}
    \sum_a\left\|\sum_{B:a\in B}\delta_a\KK_B^*(O)\right\| &\overset{(1)}{\leq}4\sum_B |B|\,||\KK_B||_{1\to1}\sum_{c\in B}||\delta_c O||,\notag\\
    &\overset{(2)}{\leq}4rs_{\KK}\sum_c||\delta_c O||, \label{eq:app-rfa-local-bound}
\end{align}
where $(1)$ uses $\KK_B^*(\BI)=0$, $||\delta_a||\leq2$, and the reset telescoping sum \cref{eq:app-reset-telescope} on $B$, and $(2)$ exchanges the sums and uses $|B|\leq r$ and $\sum_{B\ni c}||\KK_B||_{1\to1}\leq s_{\KK}$.
Set $\GG:=\LL^{(\beta)}+\varepsilon\KK$ and $O_t:=e^{t\GG^*}(O)$. For each site $a$, define the Lindblad generator
\begin{align}
    \mathcal C_a&:=\sum_{b\neq a}\LL_b^{(\beta)}+\varepsilon\sum_{B:a\notin B}\KK_B.
\end{align}
Differentiating and grouping the terms gives
\begin{align}
    \frac{d}{dt}(\delta_a O_t)
    &=\delta_a\bigl((\LL^{(\beta)})^*+\varepsilon\KK^*\bigr)(O_t),\notag\\
    &\overset{(1)}{=}\left(\sum_{b\neq a}(\LL_b^{(\beta)})^*-\lambda_0\operatorname{id}\right)(\delta_a O_t)
      +\mathcal F_a^{(\beta)}(O_t)
      +\varepsilon\sum_{B:a\notin B}\KK_B^*(\delta_a O_t)
      +\varepsilon\sum_{B:a\in B}\delta_a\KK_B^*(O_t),\notag\\
    &\overset{(2)}{=}(\mathcal C_a^*-\lambda_0\operatorname{id})(\delta_a O_t)
      +\mathcal F_a^{(\beta)}(O_t)
      +\varepsilon\sum_{B:a\in B}\delta_a\KK_B^*(O_t). \label{eq:app-rfa-damped-source}
\end{align}
Here $(1)$ uses $\delta_a(\LL_b^{(\beta)})^*=(\LL_b^{(\beta)})^*\delta_a+[\delta_a,(\LL_b^{(\beta)})^*]$, $\delta_a\MM_a^*=-\lambda_0\delta_a$, and the definition of $\mathcal F_a^{(\beta)}$; perturbation terms supported away from $a$ commute with $\delta_a$. Step $(2)$ is the definition of $\mathcal C_a$.
The scalar damping commutes with $\mathcal C_a^*$, whose semigroup is a unital contraction. Integrating \cref{eq:app-rfa-damped-source} gives
\begin{align}
    \delta_a O_t
    &=e^{-\lambda_0t}e^{t\mathcal C_a^*}\delta_a O
      +\int_0^t e^{-\lambda_0(t-s)}e^{(t-s)\mathcal C_a^*}
      \left(\mathcal F_a^{(\beta)}(O_s)
      +\varepsilon\sum_{B:a\in B}\delta_a\KK_B^*(O_s)\right)ds.
\end{align}
Taking norms, summing over $a$, and using \cref{eq:app-rfa-source-bound,eq:app-rfa-local-bound}, we obtain, for $S(t):=\sum_a\|\delta_a O_t\|$ and $b:=\kappa_\beta+4\varepsilon r s_{\KK}$,
\begin{align}
    S(t)&\leq e^{-\lambda_0t}S(0)+b\int_0^t e^{-\lambda_0(t-s)}S(s)\,ds. \label{eq:app-rfa-volterra}
\end{align}
A sufficient condition for a positive decay rate is
\begin{align}
    \left(\kappa_\beta+4\varepsilon r s_{\KK}\right)t_{\mathrm{loc}}^{(0)} &<\frac12. \label{eq:app-rfa-stability-condition}
\end{align}
Since $e^{-\lambda_0t}\leq2^{-\lfloor t/t_{\mathrm{loc}}^{(0)}\rfloor}$, \cref{lem:app-volterra-decay} applies under \cref{eq:app-rfa-stability-condition}. Combining its oscillator decay with \cref{lem:app-trace-oscillator}, for the reset cover $\mathfrak U_0$ with $M=|\Lambda|$, proves uniqueness and the stated mixing bound. For fixed $\beta,r,s_{\KK},\varepsilon$ satisfying \cref{eq:app-rfa-stability-condition}, the coefficient $b$ and the local mixing time are uniform in the system size, so the decay rate is constant.

The high-temperature estimates of \cite{RouzeFrancaAlhambra2024b} give $\beta_*>0$ and a uniform margin $\Delta>0$ such that $\kappa_\beta t_{\mathrm{loc}}^{(0)}\leq1/2-\Delta$ for $0\leq\beta<\beta_*$. Choosing $\varepsilon_*>0$ so that $4\varepsilon_*r s_{\KK}t_{\mathrm{loc}}^{(0)}\leq\Delta/2$ retains a uniform positive margin in \cref{eq:app-rfa-stability-condition}, proving the corollary and \cref{cor:main-gibbs-stability}.
\end{proof}

We next apply the cover criterion to the parent Lindbladian for simple injective matrix product density operators (MPDOs) at a renormalization fixed point constructed in \cite{liu2026parent}.

\begin{corollary}[Stability of simple injective MPDO parent Lindbladians]
\label{cor:app-mpdo-parent-stability}\resulttoc{cor:app-mpdo-parent-stability}{Stability of simple injective MPDO parent Lindbladians}
Consider the parent Lindbladian family $(\LL_N)_N$ of the simple injective renormalization-fixed-point construction in \cite[Theorem~III.9 and Section~V~A]{liu2026parent}, on a periodic chain of $N\geq5$ sites. Fix the local tensor data independently of $N$, with sites understood after any fixed blocking. Work on the complete canonical physical space with positive sector weights $a_u b_v>0$. Each $\LL_N$ contains the three-site bulk-dissipative cover $\mathfrak U_{\mathrm{MPDO}}$ defined in \cref{eq:app-mpdo-cover}, with local mixing time $t_{\mathrm{loc}}=\log4$ and $a=\mathfrak d=3$. For every fixed locality $k$ and per-site interaction-strength bound $s_{\KK}$, there is a size-independent $\varepsilon_*>0$ such that, for every $k$-local Lindbladian family $(\KK_N)_N$ with interaction strength at most $s_{\KK}$, $(\LL_N+\varepsilon\KK_N)_N$ has a unique stationary state at each size and is rapidly mixing with constant decay rate for all $0\leq\varepsilon\leq\varepsilon_*$. At $\varepsilon=0$, the stationary state is the specified MPDO.
\end{corollary}

\begin{proof}[Proof of \cref{cor:app-mpdo-parent-stability}]
\prooflabel{cor:app-mpdo-parent-stability}{proof:app-mpdo-parent-stability}
The complete canonical physical space and bond data are
\begin{align}
    \HH&=\bigoplus_u\bigl(\HH_l^{(u)}\otimes\HH_r^{(u)}\bigr), \qquad \eta_{u,v}\succeq0,\qquad \Tr\eta_{u,v}=a_u b_v,\qquad \sum_u a_u b_u=1, \label{eq:app-mpdo-canonical-data}
\end{align}
where $\eta_{u,v}$ acts on $\HH_r^{(u)}\otimes\HH_l^{(v)}$. Let $\mathcal E_i$ be the three-site channel of Eq.~(55) of that reference, translated to sites $i,i+1,i+2$, with indices taken modulo $N$, and set
\begin{align}
    \LL=\LL_N&=\sum_{i=1}^N\LL_i,\qquad \LL_i=\mathcal E_i-\operatorname{id}.
\end{align}
Define
\begin{align}
    \mathfrak U_{\mathrm{MPDO}} &:=\left\{\bigl(A_i^\circ,\{i,i+1,i+2\},\LL_i\bigr):i=1,\ldots,N\right\} \label{eq:app-mpdo-cover}
\end{align}
and verify below that it is a bulk-dissipative cover with local mixing time $t_{\mathrm{loc}}=\log4$ and $a=\mathfrak d=3$, contained in $\LL$. Here $A_i^\circ$ is the bulk of $\LL_i$ defined in \cref{eq:app-bulk-dissipation}; it contains the middle site $i+1$.

We verify the cover hypotheses and apply the local-mixing-time criterion. On a patch of sites $1,2,3$, define the sector maps $Q_u:\HH\to\HH_l^{(u)}\otimes\HH_r^{(u)}$ by
\begin{align*}
    Q_u\left(\bigoplus_v\psi_v\right)&:=\psi_u,\qquad \psi_v\in\HH_l^{(v)}\otimes\HH_r^{(v)}.
\end{align*}
They satisfy $Q_uQ_v^\dagger=0$ for $u\neq v$, $Q_uQ_u^\dagger=\BI_{\HH_l^{(u)}\otimes\HH_r^{(u)}}$, and $\sum_uQ_u^\dagger Q_u=\BI_\HH$. Set
\begin{align}
    R_{u,v}(X) &:=\Tr_{1r,2,3l}\!\left[(Q_u\otimes\BI_2\otimes Q_v)X(Q_u^\dagger\otimes\BI_2\otimes Q_v^\dagger)\right],\notag\\
    \Omega_{u,v}&:=\bigoplus_w\bigl(\eta_{u,w}\otimes\eta_{w,v}\bigr). \label{eq:app-mpdo-channel-factors}
\end{align}
The channel of Eq.~(55) of \cite{liu2026parent} is
\begin{align}
    \mathcal E(X)&=\bigoplus_{u,v}\left[R_{u,v}(X)\otimes\frac{\Omega_{u,v}}{a_u b_v}\right], \label{eq:app-mpdo-channel}
\end{align}
with tensor factors placed in physical order: $R_{u,v}$ retains $1l,3r$, while $\Omega_{u,v}$ occupies $1r,2,3l$. The bond normalization gives
\begin{align}
    \Tr\Omega_{u,v} &\overset{(1)}{=}\sum_w(a_u b_w)(a_w b_v) \overset{(2)}{=}a_u b_v, \label{eq:app-mpdo-channel-normalization}
\end{align}
where $(1)$ uses the bond traces and $(2)$ uses $\sum_w a_w b_w=1$. Thus the channel measures the complete outer-sector decomposition, traces the selected factors, and inserts a normalized state. Orthogonality of the sectors then gives
\begin{align}
    R_{u,v}(\mathcal E(X))&=R_{u,v}(X)\frac{\Tr\Omega_{u,v}}{a_u b_v}=R_{u,v}(X), \qquad \mathcal E^2(X)=\mathcal E(X). \label{eq:app-mpdo-idempotence}
\end{align}
Consequently,
\begin{align}
    e^{t\LL_i}&=\mathcal E_i+e^{-t}(\operatorname{id}-\mathcal E_i),\notag\\
    \|e^{t\LL_i}-\mathcal E_i\|_\diamond&\leq2e^{-t}. \label{eq:app-mpdo-local-decay}
\end{align}
Each $\LL_i$ is a Lindbladian: if $F_\nu$ are Kraus operators of $\mathcal E_i$, their completeness relation gives $\mathcal E_i-\operatorname{id}=\sum_\nu\Diss{F_\nu}$.

Since the sector maps act only on the outer sites, $R_{u,v}(X)$ depends on $X$ only through $\Tr_2X$. Hence $\mathcal E=\Phi\circ\Tr_2$, where $\Phi$ performs the same outer-sector measurement and normalized-state insertion. Taking adjoints gives
\begin{align}
    \ker(\mathcal E^*-\operatorname{id}) &\overset{(1)}{=}\operatorname{Ran}\mathcal E^* \overset{(2)}{\subseteq}\BI_2\otimes\BB(\HH_{\{1,3\}}), \label{eq:app-mpdo-bulk-erasure}
\end{align}
where $(1)$ uses idempotence and $(2)$ uses that $(\Tr_2)^*$ inserts the identity at site $2$. This proves the bulk condition for the entire middle site, including its sector label, and remains true after tensoring with the identity outside the patch.

At $t=\log4$, \cref{eq:app-mpdo-local-decay} is at most $1/2$. Since $i+1\in A_i^\circ$ for every $i$, the bulk regions cover the chain. Each enlarged region has three sites, and each site belongs to three enlarged regions. Thus \cref{eq:app-mpdo-local-decay,eq:app-mpdo-bulk-erasure} give the claimed cover, with $M=N$, $t_{\mathrm{loc}}=\log4$, and $a=\mathfrak d=3$. The local terms commute by Section~V~A, Eqs.~(57)--(58), of \cite{liu2026parent}; triples sharing only one endpoint act on its opposite left and right factors within the same diagonal sector.
Write $\KK=\KK_N$. For every perturbation term $\KK_B$, $n_{\mathfrak U_{\mathrm{MPDO}}}(B)\leq\mathfrak d|B|\leq3k$, and hence $s_{\KK,\mathfrak U_{\mathrm{MPDO}}}\leq3ks_{\KK}$. The commuting incoming influence therefore satisfies
\begin{align}
    \kappa_\varepsilon^{\mathrm{com}}\leq4a\varepsilon s_{\KK,\mathfrak U_{\mathrm{MPDO}}}&\leq36\varepsilon k s_{\KK}.
\end{align}
A sufficient condition is therefore
\begin{align}
    36\varepsilon k s_{\KK}\,t_{\mathrm{loc}}&<\frac12,\qquad t_{\mathrm{loc}}=\log4. \label{eq:app-mpdo-stability-condition}
\end{align}
This implies $\kappa_\varepsilon^{\mathrm{com}}t_{\mathrm{loc}}<1/2$, and \cref{cor:commuting-strong-bulk-dissipation} proves uniqueness and rapid mixing with constant decay rate. At $\varepsilon=0$, the parent construction fixes the specified MPDO, which is therefore the unique stationary state. For fixed $k$ and $s_{\KK}>0$, choose $\varepsilon_*=(144k s_{\KK}\log4)^{-1}$. Then the left-hand side of \cref{eq:app-mpdo-stability-condition} is at most $1/4$ for every permitted perturbation and every $0\leq\varepsilon\leq\varepsilon_*$. When $s_{\KK}=0$, the perturbation vanishes and any positive threshold works. This proves the corollary and \cref{cor:main-mpdo-stability}.
\end{proof}

\subsection{Strong local noise and preparation of macroscopic order}

\begin{proof}[Proof of \cref{cor:main-local-noise}]
\prooflabel{cor:main-local-noise}{proof:main-local-noise}
Let $\LL=\LL_N=\sum_B\LL_B$ be a geometrically local Lindbladian on the fixed-dimensional lattice, with $\zeta(\LL)=\max_x\sum_{B\ni x}|B|\,\|\LL_B\|_{1\to1}$ bounded independently of the system size. Let $\mathcal S_{\NN}$ be a family of bounded, possibly overlapping regions of size at most $a$ that cover the lattice, with each site contained in at most $\mathfrak d$ regions, and write $\NN=\NN_N=\sum_{B\in\mathcal S_{\NN}}\NN_B$. Assume that $\NN$ is geometrically local and has uniformly bounded per-site interaction strength. The setting is illustrated in \cref{fig:local-noise-obstructs-order}a). Each $\NN_B$ may be any locally relaxing Lindbladian with a unique stationary state on its support $B$. Thus its whole support is a bulk. Suppose uniformly in $B$ that
\begin{align}
    \left\lVert e^{t_{\mathrm{noise}}\NN_B}-P_{\NN_B}\right\rVert_\diamond&\leq\frac12.
\end{align}
Let $\kappa_{\NN}^{\sharp}$ be the reference coefficient in \cref{eq:app-bulk-incoming-L}, with $\LL=\NN$ and this noise cover. Thus it is the same refinement as $\kappa_{\LL}^{\sharp}$, and \cref{eq:app-sharp-main-comparison} gives $\kappa_{\NN}^{\sharp}\leq\kappa_{\NN}$, where $\kappa_{\NN}$ is the main-text coefficient. Consequently, the main-text assumption implies $1/2-\kappa_{\NN}^{\sharp}t_{\mathrm{noise}}\geq\Delta>0$. Scaling the noise by $\varepsilon$ changes this rate to $\varepsilon\kappa_{\NN}^{\sharp}$ and the local mixing time to $t_{\mathrm{noise}}/\varepsilon$. Let $\Delta_{\NN}:=1/2-\kappa_{\NN}^{\sharp}t_{\mathrm{noise}}$.
By \cref{lem:app-per-site-influence}, the bound $\kappa_{\LL+\varepsilon\NN}^{\mathrm{loc}}\leq\varepsilon\kappa_{\NN}^{\sharp}+4a\mathfrak d\,\zeta(\LL)$ shows that, if $\Delta_{\NN}>0$, the sum $\LL+\varepsilon\NN$ is rapidly mixing whenever
\begin{align}
    \varepsilon&>\frac{4a\mathfrak d\,\zeta(\LL)t_{\mathrm{noise}}}{\Delta_{\NN}}. \label{eq:local-noise-order-threshold}
\end{align}
For a $k$-local $\LL$ with interaction strength bounded by $s_{\LL}$, one may use $\zeta(\LL)\leq k s_{\LL}$. If the $\NN_B$ commute, then $\kappa_{\NN}^{\sharp}=0$, so the mixing estimate imposes no additional smallness condition on the interaction strength of $\NN$, and a sufficient condition becomes
\begin{align}
    \varepsilon&>8a\mathfrak d\,\zeta(\LL)t_{\mathrm{noise}}. \label{eq:commuting-local-noise-order-threshold}
\end{align}
To obtain a single threshold for the main-text family, choose uniform upper bounds $Z\geq\zeta(\LL_N)$ and $T\geq t_{\mathrm{noise}}$, and take $\varepsilon_*:=4a\mathfrak d ZT/\Delta$. Every fixed $\varepsilon>\varepsilon_*$ leaves a margin at least $\Delta(1-\varepsilon_*/\varepsilon)>0$. The cover has $M=O(|\Lambda|)$ and local mixing time $O(1)$, proving the corollary and giving the rapid relaxation illustrated in \cref{fig:local-noise-obstructs-order}b).

For every fixed $\varepsilon$, $\LL_N+\varepsilon\NN_N$ is geometrically local with uniformly bounded per-site interaction strength. Preparing a state with nonvanishing connected correlations between bounded-support observables separated by distance $R$ from a product state requires time $\Omega(R)$ at sufficiently small fixed trace-norm error \cite{BravyiHastingsVerstraete2006,Poulin2010}. The same lower bound holds for topological target states covered by \cite{KoenigPastawski2014}, with separation scale $R$. For $R$ proportional to the lattice diameter, these bounds are incompatible with the $O(\log|\Lambda|)$ preparation time supplied by rapid mixing on a fixed-dimensional lattice.

\end{proof}

\begin{figure}[!t]
\centering
\begin{minipage}[t]{0.392\linewidth}
a)\par\smallskip
\centering
\includegraphics[width=\linewidth]{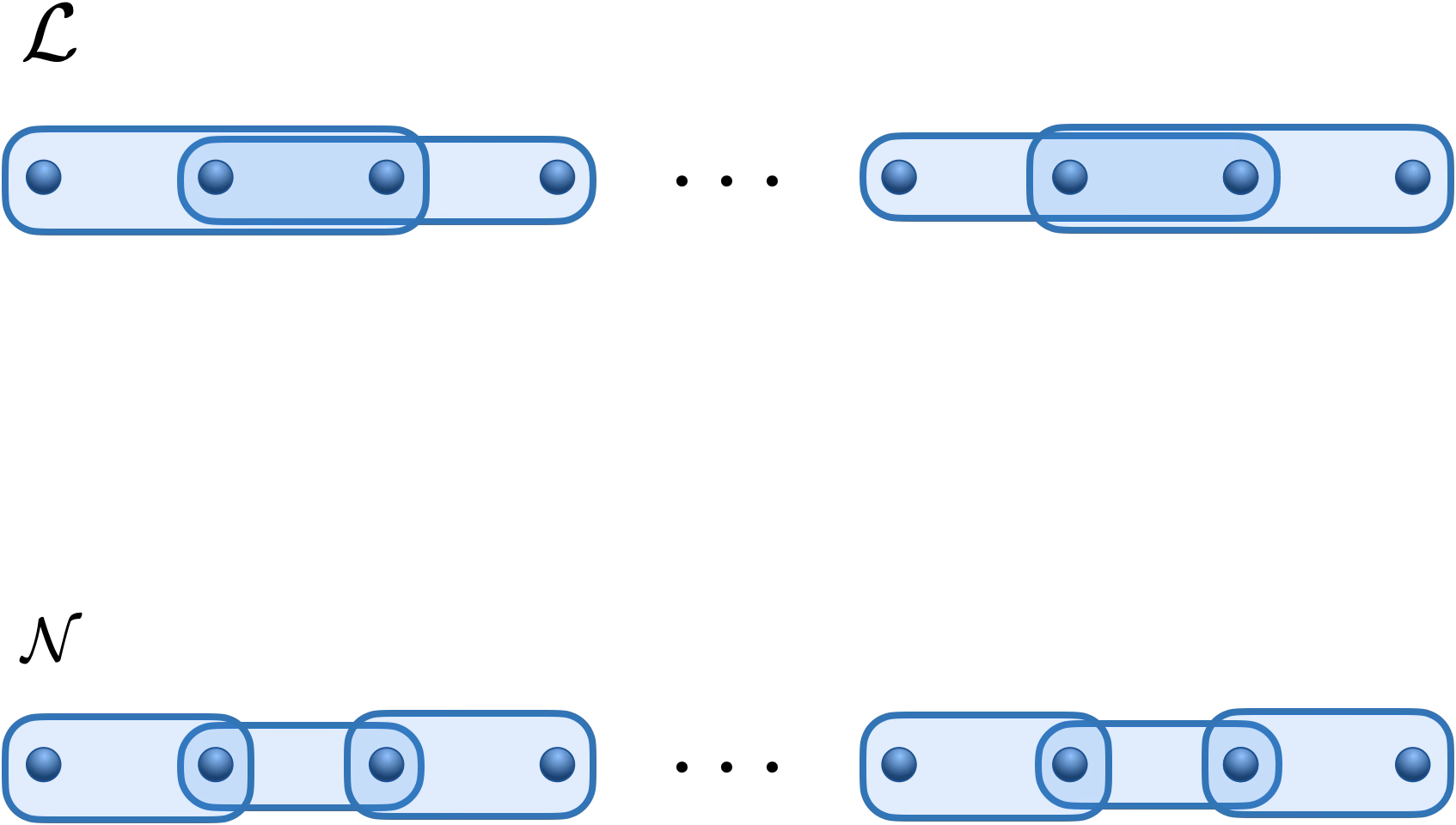}
\end{minipage}\hspace{0.02\linewidth}
\begin{minipage}[t]{0.392\linewidth}
b)\par\smallskip
\centering
\includegraphics[width=\linewidth]{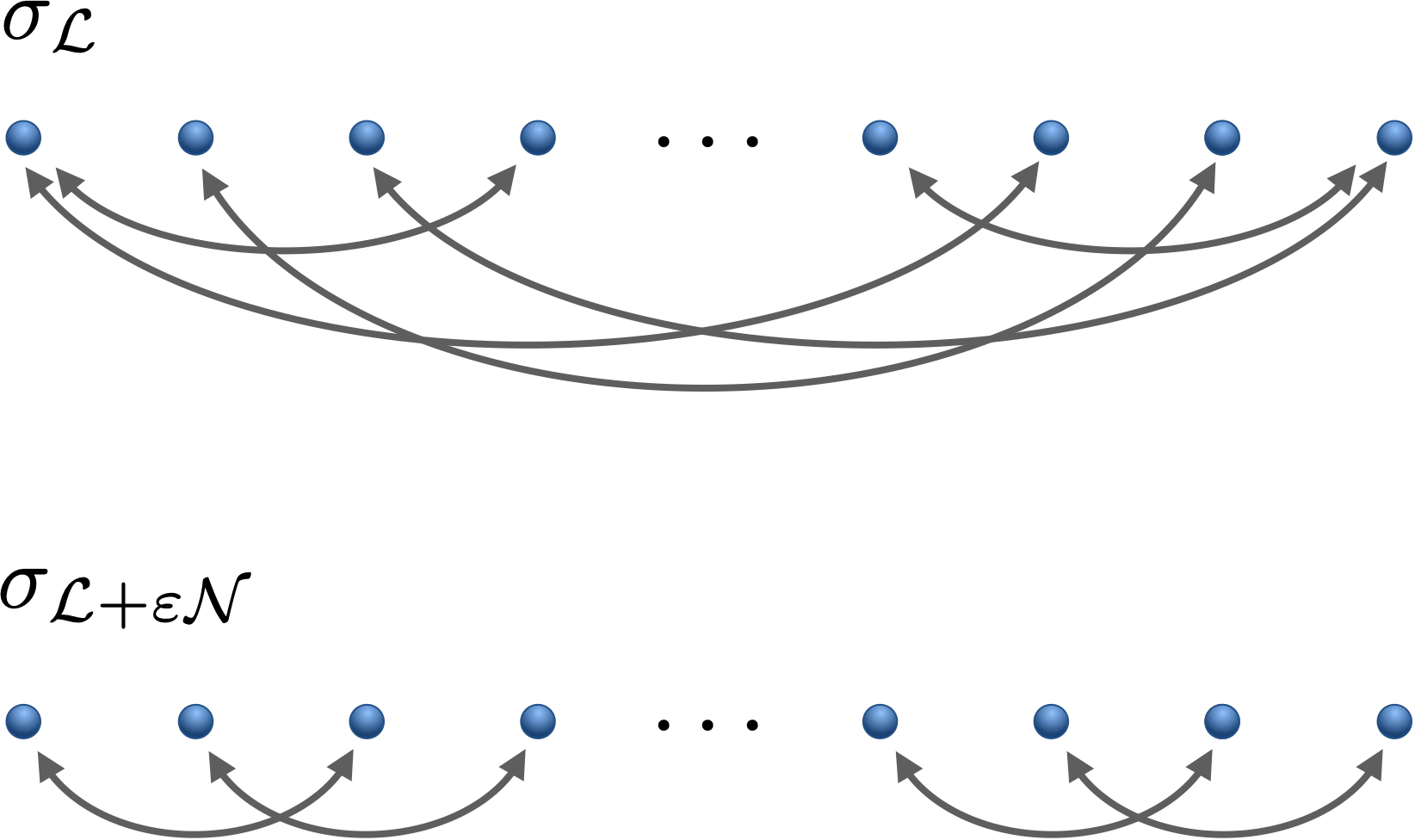}
\end{minipage}
\caption{Local noise obstructs dissipative preparation of macroscopic order. a) The geometrically local Lindbladian $\LL$ aims to prepare a state $\sigma_\LL$ with nonvanishing connected correlations between bounded-support observables at macroscopic separation. Both $\LL$ and the geometrically local noise $\NN$ have uniformly bounded per-site interaction strength. In the one-dimensional setting illustrated, preparation from a product state to sufficiently small fixed trace-norm error requires time $\Omega(N)$. The noise terms act on bounded, possibly overlapping regions that cover the lattice, each with a unique local fixed point and local mixing time $t_{\mathrm{noise}}$. b) Suppose the refined noise influence satisfies $\kappa_{\NN}^{\sharp}t_{\mathrm{noise}}\leq1/2-\Delta$ uniformly for some $\Delta>0$. With $k$ bounding the support sizes of the terms of $\LL$, and $s_\LL,T$ uniform upper bounds on its interaction strength and the noise mixing time, the threshold $\varepsilon_*:=4a\mathfrak d\,k s_\LL T/\Delta$ ensures rapid mixing with constant decay rate for every fixed $\varepsilon>\varepsilon_*$. Its $O(\log N)$ preparation time excludes the macroscopic correlations in a). If the noise terms commute, one may take $\Delta=1/2$.}
\label{fig:local-noise-obstructs-order}
\end{figure}

\partbibliography

\addtocontents{toc}{\protect\setcounter{tocdepth}{1}}
\hypersetup{bookmarksdepth=1}

\paperpart{appendices-quasifree_fermionic}{1}
\section{Quasifree fermionic Lindbladians}
\label{app:quasifree-fermionic}
Consider $N\geq1$ fermionic modes written as $m=2N$ Majoranas $c_1,\ldots,c_m$ satisfying $c_i^\dagger=c_i$ and $\{c_i,c_j\}=2\delta_{ij}\BI$. These generate the Clifford algebra $\mathcal A_m$ on a $2^N$-dimensional Hilbert space.

\begin{definition}[Quasifree fermionic Lindbladian]
The Lindbladian
\begin{align*}
    \LL(\rho)=-i[H,\rho]+\sum_\mu\left(2L_\mu\rho L_\mu^\dagger-\{L_\mu^\dagger L_\mu,\rho\}\right)
\end{align*}
is called \emph{quasifree} if $H$ is quadratic in the Majorana operators and $L_\mu$ are linear:
\begin{align*}
    H=\frac{i}{4}\sum_{j,k=1}^m \mathbf{H}_{j,k}c_jc_k,\qquad L_\mu=\sum_{j=1}^m l_{\mu,j}c_j,
\end{align*}
where $\mathbf{H}=-\mathbf{H}^T\in\mathbb R^{m\times m}$.
\end{definition}
The dynamics and spectrum of this family of Lindbladians can be studied efficiently \cite{Prosen2008}. We follow the presentation of \cite{bravyi2011classical}. The Hamiltonian operator is denoted by $H$, and its coefficient matrix by $\mathbf{H}$. Define the positive semidefinite Hermitian matrix $\mathbf{M}$ and the real matrix $\mathbf{X}$ as
\begin{align}
    \mathbf{M}_{i,j}:=\sum_\mu l_{\mu,i}l_{\mu,j}^*,\qquad \mathbf{X}:=-\mathbf{H}-2(\mathbf{M}+\mathbf{M}^T)=-\mathbf{H}-4\operatorname{Re}(\mathbf{M}). \label{app:eq:quasifree_X}
\end{align}
Define the parity operator and the parity-dressed Majoranas by
\begin{align}
    \Pi:=(-i)^N c_1\cdots c_m,\qquad \widetilde c_i:=i\Pi c_i. \label{eq:qf-parity-dressing}
\end{align}
Reversing the product defining $\Pi$ uses $m(m-1)/2$ anticommutations, while moving $c_i$ through its other $m-1$ factors changes its sign. Since $m=2N$, these give $\Pi^\dagger=\Pi$, $\Pi^2=\BI$, and $\Pi c_i\Pi=-c_i$. Consequently,
\begin{align}
    \widetilde c_i^\dagger&=-ic_i\Pi=\widetilde c_i,\notag\\
    \widetilde c_i\widetilde c_j&=-\Pi c_i\Pi c_j=c_ic_j, \label{eq:qf-dressed-products}
\end{align}
Thus the $\widetilde c_i$ satisfy the same Majorana relations. Pairing successive factors in their product also gives $(-i)^N\widetilde c_1\cdots\widetilde c_m=\Pi$, and $c_i=-i\Pi\widetilde c_i$, so they generate the same algebra. An operator $O$ is \emph{even} if $[O,\Pi]=0$, equivalently $\Pi O\Pi=O$; it is \emph{odd} if $\Pi O\Pi=-O$. Equivalently, its Majorana expansion contains only even or only odd degrees, respectively.

We write products of the dressed Majoranas in increasing label order: for $S=\{j_1,\ldots,j_l\}\subseteq[m]$, with $j_1<\cdots<j_l$, define $\widetilde c_S:=\widetilde c_{j_1}\cdots\widetilde c_{j_l}$ and $\widetilde c_\emptyset:=\BI$. Define the expectation maps and quasiderivations for label $k$ by linear extension of
\begin{align}
    \mathcal E_k(\widetilde c_S)&:=\begin{cases}0,&k\in S,\\
    \widetilde c_S,&k\notin S,\end{cases} &\delta_k(\widetilde c_S)&:=\begin{cases}(-1)^{r-1}\widetilde c_{S\setminus\{k\}},&k=j_r\in S,\\
    0,&k\notin S.\end{cases} \label{eq:qf-quasiderivation-definition}
\end{align}
They satisfy $\mathcal E_k(\BI)=\BI$, $\delta_k(\BI)=0$, $[\mathcal E_i,\mathcal E_j]=0$ and $\{\delta_i,\delta_j\}=0$. Another useful identity is:
\begin{align}
    (\operatorname{id}-\mathcal E_k)(\widetilde c_S)&= \begin{cases}\widetilde c_S,&k\in S,\\
    0,&k\notin S,\end{cases} \overset{(1)}{=} \widetilde c_k \delta_k(\widetilde c_S), \notag\\
    (\operatorname{id}-\mathcal E_k)(O)&\overset{(2)}{=}\widetilde c_k\delta_kO,\label{eq:id-Ef-cdelta}
\end{align}
where $(1)$ uses that the $(-1)^{r-1}$ is cancelled by moving $\widetilde c_k$ past $r-1$ terms and $(2)$ extends the result to arbitrary observables $O$ by linearity, so we can also write $\operatorname{id}-\mathcal{E}_k = \widetilde c_k \delta_k$. We will also use the following expressions for an observable $O$ of parity $p\in \{+1,-1\}$:
\begin{align}
    \widetilde c_k \mathcal{E}_k(O)&\overset{(1)}{=}\frac12 (\widetilde c_k O + pO \widetilde c_k), \label{eq:ckEk-p}\\
    \delta_k(O)&\overset{(2)}{=}\frac12 (\widetilde c_k O - pO \widetilde c_k),\label{eq:dk-p}
\end{align}
where $(1)$ uses that the right hand side is $0$ for monomials containing $k$ and left multiplication by $\widetilde c_k$ otherwise, and $(2)$ uses \cref{eq:id-Ef-cdelta}. The Majorana quasiderivations are analogous to differentiation of Grassmann variables \cite[Eqs.~(7)--(8)]{BravyiLagrangian}. For an operator $O$ of parity $p\in\{+1,-1\}$:
\begin{align}
    \delta_kO &\overset{(1)}{=}\frac12\left(\widetilde c_kO-pO\widetilde c_k\right),\notag\\
    &\overset{(2)}{=}\frac{i}{2}\left(\Pi c_kO-pO\Pi c_k\right),\notag\\
    &\overset{(3)}{=}\frac{i}{2}\Pi[c_k,O], \label{eq:dk_commutators}
\end{align}
where $(1)$ uses \cref{eq:dk-p}, $(2)$ uses $\widetilde c_k = i\Pi c_k$, and $(3)$ uses $pO\Pi=\Pi O$. Every operator is the sum of its even and odd parts $(O\pm\Pi O\Pi)/2$, so the last expression holds for arbitrary $O$ by linearity. It follows that
\begin{align}
    \mathcal E_k(O)&\overset{(1)}{=}O-\widetilde c_k\delta_kO\overset{(2)}{=}\frac12(O+c_kOc_k),\label{eq:qf-expectation}\\
    \|\delta_kO\|&\overset{(3)}{=}\frac12\|[c_k,O]\|\overset{(4)}{\leq}\|O\|, \label{eq:qf-expectation-contraction}
\end{align}
where $(1)$ uses \cref{eq:id-Ef-cdelta}, $(2)$ substitutes \cref{eq:dk_commutators} and uses $\Pi c_k\Pi=-c_k$, $(3)$ uses unitarity of $i\Pi$, and $(4)$ is the triangle inequality and unitarity of $c_k$. In particular, $\mathcal E_k$ is an average of two unitary conjugations, so it is trace preserving and $\|\mathcal E_k\|_{\infty\to\infty}\leq1$ on the full algebra. Its composition over all labels removes every nonempty monomial and preserves the trace:
\begin{align}
    \mathcal E_1\circ\cdots\circ\mathcal E_m(O)=2^{-N}\operatorname{Tr}(O)\BI. \label{eq:qf-full-expectation}
\end{align}

For $q(O):=\inf_{z\in\mathbb C}\|O-z\BI\|$ defined using the operator norm, this equation yields the telescoping identity
\begin{align}
    q(O)&\overset{(1)}{\leq}\|O-2^{-N}\operatorname{Tr}(O)\BI\|,\notag\\
    &\overset{(2)}{=}\left\|\sum_{i=1}^m\mathcal E_1\cdots\mathcal E_{i-1}(\operatorname{id}-\mathcal E_i)(O)\right\|,\notag\\
    &\overset{(3)}{\leq}\sum_{i=1}^m\|(\operatorname{id}-\mathcal E_i)(O)\|,\notag\\
    &\overset{(4)}{\leq}\sum_{i=1}^m\|\widetilde c_i\delta_iO\|,\notag\\
    &\overset{(5)}{=}\sum_{i=1}^m\|\delta_iO\|, \label{eq:qf-telescoping}
\end{align}
where $(1)$ chooses $z=2^{-N}\operatorname{Tr}(O)$, $(2)$ is a telescoping sum with \cref{eq:qf-full-expectation} on one endpoint, $(3)$ uses the triangle inequality and $||\mathcal{E}_k||_{\infty\to\infty}\leq 1$, $(4)$ uses \cref{eq:id-Ef-cdelta} and $(5)$ uses unitarity of $\widetilde c_i$. This motivates the Majorana oscillator seminorm
\begin{align}
    \oscm{O}:=\sum_{i=1}^m\|\delta_iO\|. \label{eq:qf-oscillator-definition}
\end{align}
\cref{eq:qf-expectation-contraction,eq:qf-telescoping} give $q(O)\leq\oscm{O}\leq m\|O\|$.

We consider the Heisenberg evolution of an observable $O$ under $\LL^*$, $T_t(O) = e^{t\LL^*}(O)$. We will use the notation $T_t O$ for clarity, and composition is understood. We want to study how $\oscm{T_tO}$ decays with time. We first show that the Majorana quasiderivations have a closed first-order evolution.

\begin{lemma}[Adjoint identities for Majorana quasiderivations]
\label{lem:qf-adjoint-identities}
Let $\LL^*$ be the Hilbert--Schmidt adjoint of $\LL$. For every operator $O\in \mathcal{A}_m$,
\begin{align}
    \LL^*(O)&=\sum_{i,j=1}^m\mathbf{X}_{i,j}\widetilde c_i\delta_jO+\sum_{i,j=1}^m(\mathbf{H}_{i,j}+4\mathbf{M}_{i,j})\delta_i\delta_jO, \label{eq:qf_adjoint_quasiderivations}\\
    \delta_i\LL^* O&=\LL^*\delta_iO+\sum_{r=1}^m\mathbf{X}_{i,r}\delta_rO. \label{eq:one_quasiderivation}
\end{align}
\end{lemma}

\begin{proof}[Proof of \cref{lem:qf-adjoint-identities}]
\prooflabel{lem:qf-adjoint-identities}{proof:qf-adjoint-identities}
We first prove \cref{eq:qf_adjoint_quasiderivations} for an operator $O$ of parity $p\in\{+1,-1\}$. \cref{eq:dk_commutators} gives
\begin{align}
    O\widetilde c_j=p\widetilde c_jO-2p\delta_jO. \label{eq:right_mult_parity}
\end{align}
Since $\delta_i$ removes one Majorana factor, $\delta_iO$ has parity $-p$. The relation $\delta_j\delta_i=-\delta_i\delta_j$ follows directly from \cref{eq:qf-quasiderivation-definition}: when both distinct labels occur, removing the earlier factor shifts the position of the later one by one, reversing the sign between the two orders. If either label is absent, or the labels agree, both compositions vanish. Applying \cref{eq:right_mult_parity} to $\delta_iO$ therefore yields
\begin{align}
    (\delta_iO)\widetilde c_j&\overset{(1)}{=}-p\widetilde c_j\delta_iO+2p\delta_j\delta_iO,\notag\\
    &\overset{(2)}{=}-p\widetilde c_j\delta_iO-2p\delta_i\delta_jO, \label{eq:right_mult_derivative}
\end{align}
where $(1)$ uses its parity $-p$ and $(2)$ uses anticommutation of the quasiderivations. Using these two identities, we obtain
\begin{align}
    O\widetilde c_i\widetilde c_j&\overset{(1)}{=}(p\widetilde c_iO-2p\delta_iO)\widetilde c_j,\notag\\
    &\overset{(2)}{=}p\widetilde c_i(p\widetilde c_jO-2p\delta_jO)-2p(-p\widetilde c_j\delta_iO-2p\delta_i\delta_jO),\notag\\
    &\overset{(3)}{=}\widetilde c_i\widetilde c_jO-2\widetilde c_i\delta_jO+2\widetilde c_j\delta_iO+4\delta_i\delta_jO, \label{eq:qf-right-quadratic}
\end{align}
where $(1)$ uses \cref{eq:right_mult_parity}, $(2)$ uses \cref{eq:right_mult_parity} and \cref{eq:right_mult_derivative}, and $(3)$ expands and uses $p^2=1$. Subtracting \cref{eq:qf-right-quadratic} from $\widetilde c_i\widetilde c_jO$ gives
\begin{align}
    [\widetilde c_i\widetilde c_j,O]=2\widetilde c_i\delta_jO-2\widetilde c_j\delta_iO-4\delta_i\delta_jO. \label{eq:quadratic_commutator}
\end{align}
Since $\widetilde c_i \widetilde c_j = c_i c_j$, the Hamiltonian has the same coefficients in the dressed Majoranas. Consequently,
\begin{align}
    i[H,O]&\overset{(1)}{=}-\frac14\sum_{i,j=1}^m\mathbf{H}_{i,j}[\widetilde c_i\widetilde c_j,O],\notag\\
    &\overset{(2)}{=}-\frac12\sum_{i,j=1}^m\mathbf{H}_{i,j}\widetilde c_i\delta_jO+\frac12\sum_{i,j=1}^m\mathbf{H}_{i,j}\widetilde c_j\delta_iO+\sum_{i,j=1}^m\mathbf{H}_{i,j}\delta_i\delta_jO,\notag\\
    &\overset{(3)}{=}-\sum_{i,j=1}^m\mathbf{H}_{i,j}\widetilde c_i\delta_jO+\sum_{i,j=1}^m\mathbf{H}_{i,j}\delta_i\delta_jO, \label{eq:hamiltonian_quasiderivations}
\end{align}
where $(1)$ substitutes the Hamiltonian, $(2)$ uses \cref{eq:quadratic_commutator}, and $(3)$ exchanges $i,j$ in the second sum and uses $\mathbf{H}_{j,i}=-\mathbf{H}_{i,j}$.

We next consider the dissipative part. Since $\widetilde c_i = i\Pi c_i$ and $\Pi O\Pi=pO$, moving $\Pi$ through $\widetilde c_i$ gives
\begin{align}
    p\widetilde c_iO\widetilde c_j=c_iOc_j,\qquad c_ic_j=\widetilde c_i\widetilde c_j. \label{eq:qf-recycling-parity}
\end{align}
\cref{eq:right_mult_parity} also gives $p\widetilde c_iO\widetilde c_j=\widetilde c_i\widetilde c_jO-2\widetilde c_i\delta_jO$. Thus
\begin{align}
    2p\widetilde c_iO\widetilde c_j-\{\widetilde c_i\widetilde c_j,O\} &\overset{(1)}{=}2\widetilde c_i\widetilde c_jO-4\widetilde c_i\delta_jO-\widetilde c_i\widetilde c_jO-O\widetilde c_i\widetilde c_j,\notag\\
    &\overset{(2)}{=}-2\widetilde c_i\delta_jO-2\widetilde c_j\delta_iO-4\delta_i\delta_jO, \label{eq:dissipator_quasiderivations}
\end{align}
where $(1)$ uses the preceding identity and expands the anticommutator, and $(2)$ substitutes \cref{eq:qf-right-quadratic} and collects terms. Note that the right hand side does not depend on the parity $p$, so we can use it for an operator $O$ of either parity. Since $\mathbf{M}_{i,j}=\sum_\mu l_{\mu,i}l_{\mu,j}^*$, we obtain
\begin{align}
    \sum_\mu\left(2L_\mu^\dagger OL_\mu-\{L_\mu^\dagger L_\mu,O\}\right) &\overset{(1)}{=}\sum_{i,j=1}^m\mathbf{M}_{j,i}\left(2p\widetilde c_iO\widetilde c_j-\{\widetilde c_i\widetilde c_j,O\}\right),\notag\\
    &\overset{(2)}{=}-2\sum_{i,j=1}^m\mathbf{M}_{j,i}\widetilde c_i\delta_jO-2\sum_{i,j=1}^m\mathbf{M}_{j,i}\widetilde c_j\delta_iO-4\sum_{i,j=1}^m\mathbf{M}_{j,i}\delta_i\delta_jO,\notag\\
    &\overset{(3)}{=}-2\sum_{i,j=1}^m(\mathbf{M}_{i,j}+\mathbf{M}_{j,i})\widetilde c_i\delta_jO+4\sum_{i,j=1}^m\mathbf{M}_{i,j}\delta_i\delta_jO, \label{eq:dissipative_part_quasiderivations}
\end{align}
where $(1)$ expands the jumps and uses \cref{eq:qf-recycling-parity}, $(2)$ uses \cref{eq:dissipator_quasiderivations}, and $(3)$ exchanges $i,j$ in the second and third sums and uses $\delta_j\delta_i=-\delta_i\delta_j$. Adding \cref{eq:hamiltonian_quasiderivations} and \cref{eq:dissipative_part_quasiderivations} gives
\begin{align}
    \LL^*(O)&\overset{(1)}{=}-\sum_{i,j=1}^m\mathbf{H}_{i,j}\widetilde c_i\delta_jO+\sum_{i,j=1}^m\mathbf{H}_{i,j}\delta_i\delta_jO-2\sum_{i,j=1}^m(\mathbf{M}_{i,j}+\mathbf{M}_{j,i})\widetilde c_i\delta_jO+4\sum_{i,j=1}^m\mathbf{M}_{i,j}\delta_i\delta_jO,\notag\\
    &\overset{(2)}{=}-\sum_{i,j=1}^m\bigl(\mathbf{H}_{i,j}+2\mathbf{M}_{i,j}+2\mathbf{M}_{j,i}\bigr)\widetilde c_i\delta_jO+\sum_{i,j=1}^m(\mathbf{H}_{i,j}+4\mathbf{M}_{i,j})\delta_i\delta_jO,\notag\\
    &\overset{(3)}{=}\sum_{i,j=1}^m\mathbf{X}_{i,j}\widetilde c_i\delta_jO+\sum_{i,j=1}^m(\mathbf{H}_{i,j}+4\mathbf{M}_{i,j})\delta_i\delta_jO, \label{eq:qf_adjoint_quasiderivations_derived}
\end{align}
where $(1)$ substitutes \cref{eq:hamiltonian_quasiderivations} into $i[H,O]$ and \cref{eq:dissipative_part_quasiderivations} into the dissipative part, $(2)$ collects the two first-order sums, and $(3)$ uses $\mathbf{X}_{i,j}=-\mathbf{H}_{i,j}-2(\mathbf{M}_{i,j}+\mathbf{M}_{j,i})$ from \cref{app:eq:quasifree_X}. This is \cref{eq:qf_adjoint_quasiderivations} on each parity sector, hence on every operator by linearity.

We now derive the second identity. Directly from \cref{eq:qf-quasiderivation-definition},
\begin{align}
    \delta_i\delta_j=-\delta_j\delta_i,\qquad \delta_i(\widetilde c_pO)=\delta_{i,p}O-\widetilde c_p\delta_iO, \label{eq:delta_relations}
\end{align}
where in the second equation we used that moving $\delta_i$ past the first factor contributes the minus sign, with the additional term $O$ when $i=p$. Applying \cref{eq:delta_relations} to \cref{eq:qf_adjoint_quasiderivations_derived} gives
\begin{align}
    \delta_i\LL^* O&\overset{(1)}{=}\sum_{p,q=1}^m\mathbf{X}_{p,q}\delta_i(\widetilde c_p\delta_qO)+\sum_{p,q=1}^m(\mathbf{H}_{p,q}+4\mathbf{M}_{p,q})\delta_i\delta_p\delta_qO,\notag\\
    &\overset{(2)}{=}\sum_{q=1}^m\mathbf{X}_{i,q}\delta_qO+\sum_{p,q=1}^m\mathbf{X}_{p,q}\widetilde c_p\delta_q\delta_iO+\sum_{p,q=1}^m(\mathbf{H}_{p,q}+4\mathbf{M}_{p,q})\delta_p\delta_q\delta_iO,\notag\\
    &\overset{(3)}{=}\LL^*\delta_iO+\sum_{q=1}^m\mathbf{X}_{i,q}\delta_qO, \notag
\end{align}
where $(1)$ substitutes \cref{eq:qf_adjoint_quasiderivations_derived}, $(2)$ uses both relations in \cref{eq:delta_relations}, and $(3)$ applies \cref{eq:qf_adjoint_quasiderivations_derived} to $\delta_iO$, which is allowed because that identity holds on both parity sectors. We have therefore proved, for every $O$,
\begin{align*}
    \LL^*(O)&=\sum_{i,j=1}^m\mathbf{X}_{i,j}\widetilde c_i\delta_jO+\sum_{i,j=1}^m(\mathbf{H}_{i,j}+4\mathbf{M}_{i,j})\delta_i\delta_jO,\\
    \delta_i\LL^* O&=\LL^*\delta_iO+\sum_{r=1}^m\mathbf{X}_{i,r}\delta_rO.
\end{align*}
\end{proof}

\begin{lemma}[Closed evolution of quasiderivations]
\label{lem:qf-closed-evolution}
Let $T_t=e^{t\LL^*}$. For every operator $O\in \mathcal{A}_m$,
\begin{align}
    \delta_iT_tO=\sum_{p=1}^m(e^{t\mathbf{X}})_{i,p}T_t(\delta_pO). \label{eq:X_evolution}
\end{align}
\end{lemma}

\begin{proof}[Proof of \cref{lem:qf-closed-evolution}]
\prooflabel{lem:qf-closed-evolution}{proof:qf-closed-evolution}
Let $G_i(t)=T_{-t}\delta_iT_tO$. Here $T_{-t}=e^{-t\LL^*}$ is used only as an invertible linear map, even though it is not a quantum channel. Differentiating gives
\begin{align}
    \frac{d}{dt}G_i(t)&\overset{(1)}{=}-T_{-t}\LL^*\delta_iT_tO+T_{-t}\delta_i\LL^* T_tO,\notag\\
    &\overset{(2)}{=}\sum_{r=1}^m\mathbf{X}_{i,r}T_{-t}\delta_rT_tO,\notag\\
    &\overset{(3)}{=}\sum_{r=1}^m\mathbf{X}_{i,r}G_r(t),\qquad G_i(0)=\delta_iO, \label{eq:G_ode}
\end{align}
where $(1)$ is the product rule, $(2)$ uses \cref{eq:one_quasiderivation}, and $(3)$ is the definition of $G_r$. We verify the solution explicitly by setting
\begin{align*}
    \widetilde G_i(t):=\sum_{p=1}^m(e^{t\mathbf{X}})_{i,p}\delta_pO.
\end{align*}
The derivative of the matrix exponential gives
\begin{align*}
    \frac{d}{dt}\widetilde G_i(t)&\overset{(1)}{=}\sum_{r,p=1}^m\mathbf{X}_{i,r}(e^{t\mathbf{X}})_{r,p}\delta_pO\overset{(2)}{=}\sum_{r=1}^m\mathbf{X}_{i,r}\widetilde G_r(t),\\
    \widetilde G_i(0)&\overset{(3)}{=}\sum_{p=1}^m\delta_{i,p}\delta_pO=\delta_iO,
\end{align*}
where $(1)$ uses $\frac{d}{dt}e^{t\mathbf{X}}=\mathbf{X}e^{t\mathbf{X}}$, $(2)$ is the definition of $\widetilde G_r$, $(3)$ uses $e^{0\mathbf{X}}=\BI$. Thus $G$ and $\widetilde G$ solve \cref{eq:G_ode} with the same initial value. Uniqueness gives $G_i(t)=\widetilde G_i(t)$. Applying $T_t$ and using $T_tT_{-t}=\operatorname{id}$ yields
\begin{align*}
    \delta_iT_tO=\sum_{p=1}^m(e^{t\mathbf{X}})_{i,p}T_t(\delta_pO).
\end{align*}
\end{proof}

\begin{lemma}[Oscillator and trace-norm estimate]
\label{lem:qf-oscillator-mixing}
Let $\LL$ be a quasifree Lindbladian and $\sigma$ a state satisfying $\LL(\sigma)=0$. Then for every state $\rho$ and every $t\geq0$,
\begin{align}
    \|e^{t\LL}(\rho)-\sigma\|_1\leq m\|e^{t\mathbf{X}}\| \cdot||\rho-\sigma||_1. \label{eq:qf_general_mixing_bound}
\end{align}
\end{lemma}

\begin{proof}[Proof of \cref{lem:qf-oscillator-mixing}]
\prooflabel{lem:qf-oscillator-mixing}{proof:qf-oscillator-mixing}
For a real vector $v\in\mathbb R^m$, write $c(v):=\sum_pv_pc_p$ and $\delta(v):=\sum_pv_p\delta_p$. The Majorana relations give $c(v)^2=\|v\|_2^2\BI$, so $\|c(v)\|=\|v\|_2$. By \cref{eq:dk_commutators},
\begin{align*}
    \delta(v)O&=\frac{i}{2}\Pi[c(v),O],&
    \|\delta(v)O\|&\leq\|v\|_2\|O\|.
\end{align*}
Put $A:=e^{t\mathbf X}$ and let $e_i$ be the $i$-th coordinate vector. Combining the quasiderivations in \cref{eq:X_evolution} before taking norms gives
\begin{align*}
    \delta_iT_tO&=T_t\bigl(\delta(A^Te_i)O\bigr),\\
    \|\delta_iT_tO\|&\leq\|A^Te_i\|_2\|O\|\leq\|A\|\|O\|.
\end{align*}
Here $A$ is real, and the first inequality uses contractivity of the unital completely positive map $T_t$. The second uses $\|A^Te_i\|_2\leq\|A^T\|\|e_i\|_2=\|A\|$, since $\|e_i\|_2=1$ and the Euclidean operator norm is invariant under transposition. Summing over $i$ yields
\begin{align}
    \oscm{T_tO}\leq m\|e^{t\mathbf X}\|\|O\|.
    \label{eq:qf-oscillator-contraction}
\end{align}

This allows us to bound:
\begin{align}
    \|e^{t\LL}(\rho)-\sigma\|_1 &\overset{(1)}{=}\sup_{\|O\|\leq1}\left|\operatorname{Tr}\left[O\left(e^{t\LL}(\rho)-\sigma\right)\right]\right|,\notag\\
    &\overset{(2)}{=}\sup_{\|O\|\leq1}\left|\operatorname{Tr}\left[(T_tO)(\rho-\sigma)\right]\right|,\notag\\
    &\overset{(3)}{=}\sup_{\|O\|\leq1}\inf_{z\in\mathbb C}\left|\operatorname{Tr}\left[(T_tO-z\BI)(\rho-\sigma)\right]\right|,\notag\\
    &\overset{(4)}{\leq}\|\rho -\sigma \|_1\sup_{\|O\|\leq1}\inf_{z\in\mathbb C}\left\|T_tO-z\BI\right\|, \label{eq:qf-trace-duality}
\end{align}
where $(1)$ uses the variational definition of the trace norm $\|A\|_1=\sup_{\|O\|\leq1}\left|\operatorname{Tr}(OA)\right|$, $(2)$ uses adjointness and stationarity of $\sigma$, $(3)$ uses $\operatorname{Tr}(\rho-\sigma)=0$, and $(4)$ uses H\"older's inequality.

Finally, using \cref{eq:qf-telescoping} and then \cref{eq:qf-oscillator-contraction} we have:
\begin{align}
    \inf_{z\in\mathbb C}\left\|T_tO-z\BI\right\|&=q(T_tO)\leq \oscm{T_tO}\leq m\|e^{t\mathbf{X}}\|\|O\|,
\end{align}

which yields the desired inequality
\begin{align*}
    \|e^{t\LL}(\rho)-\sigma\|_1\leq m\|e^{t\mathbf{X}}\|\cdot ||\rho-\sigma||_1.
\end{align*}
\end{proof}

\begin{proposition}[Rapid mixing of quasifree fermionic Lindbladians]
\label{prop:quasifree_rapid_mixing}\resulttoc{prop:quasifree_rapid_mixing}{Rapid mixing of quasifree fermionic Lindbladians}
Let $\LL$ be a quasifree fermionic Lindbladian on $m=2N$ Majoranas, and let $\mathbf{M},\mathbf{X}$ be defined in \cref{app:eq:quasifree_X}. Let $\sigma$ be a state satisfying $\LL(\sigma)=0$. Then for every state $\rho$ and every $t\geq 0$,
\begin{align*}
    \|e^{t\LL}(\rho)-\sigma\|_1\leq m\|e^{t\mathbf{X}}\|\cdot||\rho-\sigma||_1.
\end{align*}
If $\mathbf{X}=SJS^{-1}$ is a Jordan decomposition, the largest Jordan block has size $L$, and
\begin{align*}
    \gamma(\mathbf{X}):=\min_{z\in\operatorname{spec}(\mathbf{X})}-\operatorname{Re}(z)>0,
\end{align*}
then
\begin{align}
    \|e^{t\LL}(\rho)-\sigma\|_1\leq m\kappa(S)e^{-t\gamma(\mathbf{X})}\left(\sum_{j=0}^{L-1}\frac{t^j}{j!}\right)\|\rho-\sigma\|_1,\qquad \kappa(S):=\|S\|\|S^{-1}\|. \notag
\end{align}
The following bound also holds:
\begin{align}
    \|e^{t\LL}(\rho)-\sigma\|_1\leq me^{-4t\lambda_{\min}(\operatorname{Re}(\mathbf{M}))}\|\rho-\sigma\|_1, \notag
\end{align}
Thus, for a family $(\LL_N)_N$ with $\LL=\LL_N$ at each size, a uniform bound $\operatorname{Re}(\mathbf{M})\succeq\mu\BI$ with $\mu>0$ gives rapid mixing with constant decay rate $4\mu$.

In particular, if $\gamma(\mathbf{X})=\Omega((\log N)^{-r})$ for some fixed $r\geq0$, $\kappa(S)=\operatorname{poly}(N)$, and $L=O(1)$, then $(\LL_N)_N$ is rapidly mixing. For $\gamma(\mathbf X)=\Omega(1)$, constant-rate rapid mixing also holds with $L=O(\log(N+1))$ and $\kappa(S)=\operatorname{poly}(N)$. The same conclusions hold if $\lambda_{\min}(\operatorname{Re}(\mathbf{M}))$ obeys the corresponding inverse-polylogarithmic or constant lower bound. Each condition also implies uniqueness of the stationary state.
\end{proposition}

\begin{proof}[Proof of \cref{prop:quasifree_rapid_mixing}]
\prooflabel{prop:quasifree_rapid_mixing}{proof:quasifree_rapid_mixing}
A stationary state exists because $\LL$ is a finite-dimensional Lindbladian. The first estimate is \cref{lem:qf-oscillator-mixing}. We bound $\|e^{t\mathbf{X}}\|$ in two ways. For a Jordan block $z\BI+N_r$ of size $r\leq L$, the matrix $N_r$ has ones on the first superdiagonal and $N_r^r=0$. Its powers have norm at most one, so
\begin{align}
    \|e^{t(z\BI+N_r)}\|&\overset{(1)}{\leq}e^{t\operatorname{Re}(z)}\sum_{q=0}^{r-1}\frac{t^q}{q!},\notag\\
    &\overset{(2)}{\leq}e^{-t\gamma(\mathbf{X})}(t+1)^{r-1},\notag\\
    \|e^{t\mathbf{X}}\|&\overset{(3)}{\leq}\kappa(S)e^{-t\gamma(\mathbf{X})}\sum_{q=0}^{L-1}\frac{t^q}{q!}, \label{eq:qf-jordan-matrix}
\end{align}
where $(1)$ expands the finite exponential series and uses the triangle inequality with $||N_r||\leq1$, $(2)$ uses $\operatorname{Re}(z)\leq-\gamma(\mathbf{X})$ and compares coefficients with $(t+1)^{r-1}$, and $(3)$ takes the maximum over the Jordan blocks, with $r\leq L$, and uses $e^{t\mathbf{X}}=Se^{tJ}S^{-1}$.

For a uniform lower bound $\gamma(\mathbf X)\geq\gamma_0>0$, the inequality $e^{-at}(at)^j/j!\leq1$ with $a=\gamma_0/2$ gives
\begin{align*}
    \|e^{t\mathbf X}\|
    \leq\kappa(S)\left[\sum_{j=0}^{L-1}
        \left(\frac2{\gamma_0}\right)^j\right]e^{-\gamma_0t/2}.
\end{align*}
If $\kappa(S)=\operatorname{poly}(N)$ and $L=O(\log(N+1))$, the bracket is polynomial in $N$, proving rapid mixing with constant decay rate.

For the second estimate, $\mathbf{X}$ is real and $\mathbf{H}^T=-\mathbf{H}$, so \cref{app:eq:quasifree_X} gives $(\mathbf{X}+\mathbf{X}^T)/2=-4\operatorname{Re}(\mathbf{M})$. Since $\|e^{tA}\|\leq\|e^{t(A+A^\dagger)/2}\|$ for every matrix $A$ and every $t\geq0$ \cite[Theorem~IX.3.1]{bhatia2013matrix}, and $\operatorname{Re}(\mathbf{M})$ is real symmetric,
\begin{align}
    \|e^{t\mathbf{X}}\|\leq\bigl\|e^{-4t\operatorname{Re}(\mathbf{M})}\bigr\|=e^{-4t\lambda_{\min}(\operatorname{Re}(\mathbf{M}))}. \label{eq:qf-hermitian-part-bound}
\end{align}
In particular $\operatorname{Re}(\mathbf{M})\succeq\mu\BI$ with $\mu>0$ gives $\|e^{t\mathbf{X}}\|\leq e^{-4\mu t}$.

\end{proof}

\begin{proposition}[Repeated \cref{cor:sum-gaussian}: Stability under quasifree additions]
\label{app:cor:sum-gaussian}
Let $\LL$ and $\KK$ be quasifree fermionic Lindbladians on the same $m=2N$ Majoranas, with $\mathbf{M}$ matrices $\mathbf{M}_\LL$ and $\mathbf{M}_\KK$, and suppose that $\operatorname{Re}(\mathbf{M}_\LL)\succeq\mu\BI$ for some $\mu>0$. Then $\LL+\KK$ has a unique stationary state $\sigma_{\LL+\KK}$, and for every state $\rho$ and every $t\geq0$,
\begin{align}
    \lVert e^{t(\LL+\KK)}(\rho)-\sigma_{\LL+\KK}\rVert_1\leq me^{-4\mu t}\|\rho-\sigma_{\LL+\KK}\|_1.
\end{align}
\end{proposition}

\begin{proof}[Proof of \cref{cor:sum-gaussian}]
\prooflabel{cor:sum-gaussian}{proof:sum-gaussian}
The sum $\LL+\KK$ is quasifree, with quadratic Hamiltonian $\mathbf{H}_\LL+\mathbf{H}_\KK$ and jump set the union of the two jump sets, so $\mathbf{M}_{\LL+\KK}=\mathbf{M}_\LL+\mathbf{M}_\KK$. Each $\mathbf{M}$ is Hermitian and positive semidefinite, and $\operatorname{Re}(\mathbf{M})$ is real symmetric with $v^T\operatorname{Re}(\mathbf{M})v=v^\dagger\mathbf{M}v\geq0$ for real $v$, so $\operatorname{Re}(\mathbf{M}_\KK)\succeq0$ and therefore
\begin{align*}
    \operatorname{Re}(\mathbf{M}_{\LL+\KK})=\operatorname{Re}(\mathbf{M}_\LL)+\operatorname{Re}(\mathbf{M}_\KK)\succeq\mu\BI.
\end{align*}
Every finite-dimensional Lindbladian has a stationary state. \Cref{prop:quasifree_rapid_mixing} applied to $\LL+\KK$ gives the desired result. Since the right-hand side tends to zero, applying it to two stationary states proves uniqueness. Applying this estimate at every size with the same $\mu>0$ and $m=2N$ proves the family statement in \cref{cor:sum-gaussian}.
\end{proof}

\begin{lemma}[Quasifree stability from finite-time matrix contraction]
\label{lem:qf-matrix-stability}
Let $\LL,\GG$ be quasifree Lindbladians on the same $m$ Majoranas, with single-particle matrices $\mathbf X_\LL,\mathbf X_\GG$. Suppose that, for some $t_0>0$, $0<q<1$, and $\eta\geq0$,
\begin{align*}
    \|e^{t_0\mathbf X_\LL}\|\leq q,\qquad
    \|\mathbf X_\GG-\mathbf X_\LL\|\leq\eta,\qquad
    q':=q+t_0\eta<1.
\end{align*}
Then $\GG$ has a unique stationary state $\sigma_\GG$, and, for every state $\rho$ and $t\geq0$,
\begin{align*}
    \|e^{t\GG}(\rho)-\sigma_\GG\|_1
    \leq\frac{m}{q'}\exp\left[-\frac{\log(1/q')}{t_0}t\right]\|\rho-\sigma_\GG\|_1.
\end{align*}
For families $(\LL_N)_N$ and $(\GG_N)_N$ satisfying these bounds with $t_0,q,\eta$ independent of $N$, $(\GG_N)_N$ is rapidly mixing with constant decay rate.
\end{lemma}

\begin{proof}
The single-particle semigroup of every quasifree Lindbladian is contractive. Indeed, for $\mathbf X=-\mathbf H-4\operatorname{Re}\mathbf M$ and $v_s=e^{s\mathbf X}v$,
\begin{align}
    \frac{d}{ds}\|v_s\|_2^2
    =v_s^\dagger(\mathbf X+\mathbf X^T)v_s
    =-8v_s^\dagger\operatorname{Re}\mathbf M\,v_s\leq0.
    \label{eq:qf-ti-one-particle-contractivity}
\end{align}
Duhamel's formula therefore gives
\begin{align}
    \|e^{t\mathbf X_\GG}-e^{t\mathbf X_\LL}\|
    &\leq\int_0^t
    \|e^{(t-s)\mathbf X_\GG}\|\,
    \|\mathbf X_\GG-\mathbf X_\LL\|\,
    \|e^{s\mathbf X_\LL}\|\,ds
    \leq t\eta. \label{eq:qf-ti-duhamel}
\end{align}
In particular $\|e^{t_0\mathbf X_\GG}\|\leq q'$. Write $t=jt_0+u$, with $j=\lfloor t/t_0\rfloor$ and $0\leq u<t_0$. The semigroup property and contractivity imply
\begin{align}
    \|e^{t\mathbf X_\GG}\|
    \leq(q')^j
    \leq\frac1{q'}\exp\left[-\frac{\log(1/q')}{t_0}t\right].
    \label{eq:qf-ti-perturbed-decay}
\end{align}
A stationary state $\sigma_\GG$ exists in finite dimension. Applying \cref{lem:qf-oscillator-mixing} proves the bound. Applying it to another stationary state and taking $t\to\infty$ proves uniqueness.
\end{proof}

\begin{proposition}[Repeated \cref{prop:qf-translation-invariant-stability}: Stability of translation-invariant quasifree dynamics under Lindbladian perturbations]
\label{app:prop:qf-translation-invariant-stability}
Let $\faml{\LL}$ be a translation-invariant quasifree Lindbladian family on periodic lattices $\Lambda_N=(\mathbb Z/N\mathbb Z)^D$, with one fermionic mode per unit cell and uniformly bounded interaction strength $s_{\LL}$. Here $s_{\LL}:=\sup_N\sup_{x\in\Lambda_N}\sum_{Z\ni x}\|\LL_N^{(Z)}\|_{1\to1}<\infty$ for a decomposition $\LL_N=\sum_{Z\subseteq\Lambda_N}\LL_N^{(Z)}$ into quasifree Lindbladians supported on $Z$, with each norm computed on the modes in $Z$. Let $\mathbf X_{\LL_N}$ be its single-particle matrix. Suppose that for some $\gamma>0$, independently of $N$,
\begin{align}
    \max_{z\in\operatorname{spec}(\mathbf X_{\LL_N})}\operatorname{Re}z&\leq-\gamma\qquad\text{for every }N,
\end{align}
Then:
\begin{enumerate}
\item[a)] The family $(\LL_N)_N$ is rapidly mixing with constant decay rate.
\item[b)] For every interaction-strength bound $s_{\KK}\geq0$, there exists a constant $\varepsilon_*>0$, independent of $N$, such that, for every quasifree Lindbladian family $(\KK_N)_N$, not necessarily translation invariant, admitting a decomposition
\begin{align}
    \KK_N&=\sum_{Z\subseteq\Lambda_N}\KK_N^{(Z)},
    &\sup_N\sup_{x\in\Lambda_N}\sum_{Z\ni x}\|\KK_N^{(Z)}\|_{1\to1}&\leq s_{\KK}, \label{app:eq:qf-ti-local-perturbation}
\end{align}
where each $\KK_N^{(Z)}$ is a quasifree Lindbladian supported on $Z$ and each local norm is computed on the modes in $Z$, the family $(\LL_N+\varepsilon\KK_N)_N$ remains rapidly mixing with constant decay rate for all $0\leq\varepsilon\leq\varepsilon_*$.
\end{enumerate}
\end{proposition}

The threshold is uniform over all perturbation families obeying the prescribed strength bound, and the same $\varepsilon$ is used for every system size. The proposition includes spatially disordered perturbations with uniformly bounded local strength.

\begin{proof}[Proof of \cref{prop:qf-translation-invariant-stability}]
\prooflabel{prop:qf-translation-invariant-stability}{proof:qf-translation-invariant-stability}
The local-matrix estimate in \cref{eq:qf-ti-local-drift-bound} and the weighted-overlap argument in \cref{eq:qf-ti-weighted-overlap}, proved below for supported quasifree Lindbladian terms, apply equally to the decomposition of $\LL_N$. They give
\begin{align*}
    \|\mathbf X_{\LL_N}\|&\leq s_{\LL}\qquad\text{for every }N.
\end{align*}
The spatial Fourier transform $U_N$ is unitary and block diagonalizes the translation-invariant single-particle matrix into its two-dimensional momentum blocks:
\begin{align}
    U_N\mathbf X_{\LL_N}U_N^\dagger&=\bigoplus_{k\in\Lambda_N^*}\mathbf X_{\LL_N}(k),
    &\|e^{t\mathbf X_{\LL_N}}\|&=\max_{k\in\Lambda_N^*}\|e^{t\mathbf X_{\LL_N}(k)}\|, \label{eq:qf-ti-fourier-blocks}
\end{align}
where $\Lambda_N^*$ is the discrete momentum grid in $\mathbb T^D=[-\pi,\pi)^D$. Unitary invariance gives $\|\mathbf X_{\LL_N}(k)\|\leq s_{\LL}$ for every $N$ and $k\in\Lambda_N^*$. The assumed spectral bound for $\mathbf X_{\LL_N}$ also holds for each of these blocks.

For every $N$ and $k\in\Lambda_N^*$, take a Schur decomposition $\mathbf X_{\LL_N}(k)=Q_N(k)T_N(k)Q_N(k)^\dagger$. The matrices $Q_N(k)$ are unitary, while $T_N(k)$ is a $2\times2$ upper-triangular matrix with norm at most $s_{\LL}$ and diagonal entries with real parts at most $-\gamma$. Its off-diagonal matrix element has magnitude at most $s_{\LL}$, so the explicit exponential gives
\begin{align}
    \|e^{tT_N(k)}\|&\leq(2+s_{\LL}t)e^{-\gamma t}. \label{eq:qf-ti-schur-bound}
\end{align}
Indeed, the off-diagonal entry is an integral of the form $\int_0^t e^{(t-s)z_1}b e^{sz_2}\,ds$, whose magnitude is at most $s_{\LL}t e^{-\gamma t}$. Unitary invariance and \cref{eq:qf-ti-fourier-blocks} therefore give
\begin{align}
    \|e^{t\mathbf X_{\LL_N}}\|&\leq(2+s_{\LL}t)e^{-\gamma t}\leq C_{s_{\LL},\gamma}e^{-\gamma t/2}, \label{eq:qf-ti-one-particle-decay}
\end{align}
with $C_{s_{\LL},\gamma}$ independent of $N$. Applying \cref{prop:quasifree_rapid_mixing}, with the linear factor $m_N=2|\Lambda_N|$ in the number of Majoranas, proves rapid mixing with constant decay rate.

We now allow the perturbation to break translation invariance. First choose a time $t_0>0$, depending only on $s_{\LL}$ and $\gamma$, such that the right-hand side of \cref{eq:qf-ti-one-particle-decay} is at most $1/2$. Thus
\begin{align}
    \|e^{t_0\mathbf X_{\LL_N}}\|\leq\frac12 \qquad\text{for every }N. \label{eq:qf-ti-fixed-contraction}
\end{align}
To preserve this contraction, we bound the change in the single-particle matrix using the local strength of $\KK_N$. For each term $\KK_N^{(Z)}$, let $\mathbf X_{\KK_N}^{(Z)}$ be its matrix on the $2|Z|$ Majorana labels in $Z$, extended by zero on all other labels. Additivity of the Hamiltonian and dissipative matrices gives
\begin{align*}
    \mathbf X_{\KK_N}&:=\sum_Z\mathbf X_{\KK_N}^{(Z)},\\
    \LL_{\varepsilon,N}&:=\LL_N+\varepsilon\KK_N,
    &\mathbf X_{\LL_{\varepsilon,N}}&=\mathbf X_{\LL_N}+\varepsilon\mathbf X_{\KK_N}.
\end{align*}
We first bound each local matrix, and then sum these bounds using the overlap of the supports.

Fix $Z$, and list its Majoranas as $c_{Z,1},\ldots,c_{Z,2|Z|}$. Define the parity operator on these modes and, for a real vector $v\in\mathbb R^{2|Z|}$, the observable
\begin{align*}
    \Pi_Z&:=(-i)^{|Z|}c_{Z,1}\cdots c_{Z,2|Z|},
    &\widetilde c_Z(v)&:=i\Pi_Z\sum_{j=1}^{2|Z|}v_jc_{Z,j}.
\end{align*}
These are the local versions of \cref{eq:qf-parity-dressing}. The operators $i\Pi_Zc_{Z,j}$ are Hermitian and satisfy the Majorana relations, so
\begin{align*}
    \widetilde c_Z(v)^2&=\|v\|_2^2\BI_Z,
    &\|\widetilde c_Z(v)\|&=\|v\|_2.
\end{align*}
Since $\widetilde c_Z(v)$ is linear in the dressed Majoranas, $\delta_i\delta_j\widetilde c_Z(v)=0$ for all $i,j$. Thus \cref{eq:qf_adjoint_quasiderivations}, applied on $Z$, reduces to
\begin{align*}
    (\KK_N^{(Z)})^*(\widetilde c_Z(v))&=\widetilde c_Z(\mathbf X_{\KK_N}^{(Z)}v).
\end{align*}
Taking norms now relates the local matrix to the local generator:
\begin{align*}
    \|\mathbf X_{\KK_N}^{(Z)}v\|_2
    &=\|(\KK_N^{(Z)})^*(\widetilde c_Z(v))\|\\
    &\leq\|(\KK_N^{(Z)})^*\|_{\infty\to\infty}\|\widetilde c_Z(v)\|\\
    &=\|\KK_N^{(Z)}\|_{1\to1}\|v\|_2.
\end{align*}
The last equality uses trace--operator norm duality. Since $\mathbf X_{\KK_N}^{(Z)}$ is real, its operator norm is attained on real vectors. Taking the supremum over real unit vectors therefore gives
\begin{align}
    \|\mathbf X_{\KK_N}^{(Z)}\|&\leq\|\KK_N^{(Z)}\|_{1\to1}. \label{eq:qf-ti-local-drift-bound}
\end{align}
For a vector $u$, let $u_Z$ be its restriction to the Majorana labels in $Z$ and put $w_Z:=\|\KK_N^{(Z)}\|_{1\to1}$. The local-strength bound in \cref{app:eq:qf-ti-local-perturbation} gives
\begin{align}
    \sum_Zw_Z\|u_Z\|_2^2
    &=\sum_i|u_i|^2\sum_{Z\ni x_i}w_Z
    \leq s_{\KK}\|u\|_2^2, \label{eq:qf-ti-weighted-overlap}
\end{align}
where $x_i$ is the site of Majorana label $i$. Hence, for arbitrary vectors $u,v$,
\begin{align}
    |u^\dagger\mathbf X_{\KK_N}v|
    &\leq\sum_Zw_Z\|u_Z\|_2\|v_Z\|_2 \notag\\
    &\leq\left(\sum_Zw_Z\|u_Z\|_2^2\right)^{1/2}
    \left(\sum_Zw_Z\|v_Z\|_2^2\right)^{1/2}\notag\\
    &\leq s_{\KK}\|u\|_2\|v\|_2, \notag
\end{align}
Taking the supremum over unit vectors proves the size-independent estimate
\begin{align}
    \|\mathbf X_{\KK_N}\|\leq s_{\KK}. \label{eq:qf-ti-full-perturbation}
\end{align}

Apply \cref{lem:qf-matrix-stability} with $\GG=\LL_{\varepsilon,N}$, $q=1/2$, and $\eta=\varepsilon s_{\KK}$. Choose $\varepsilon_*>0$ such that $t_0\varepsilon_*s_{\KK}\leq1/4$. This choice depends only on $s_{\LL},\gamma,s_{\KK}$, so it works for every perturbation family satisfying the prescribed bound. Then \cref{eq:qf-ti-fixed-contraction,eq:qf-ti-full-perturbation} and the Duhamel estimate in the lemma give, for $0\leq\varepsilon\leq\varepsilon_*$,
\begin{align}
    \|e^{t_0\mathbf X_{\LL_{\varepsilon,N}}}\|\leq\frac34. \label{eq:qf-ti-perturbed-contraction}
\end{align}
Iterating this contraction as in the lemma gives a unique stationary state $\sigma_{\varepsilon,N}$ and
\begin{align}
    \|e^{t\LL_{\varepsilon,N}}(\rho)-\sigma_{\varepsilon,N}\|_1
    &\leq\frac43m_N\exp\left[-\frac{\log(4/3)}{t_0}t\right]\|\rho-\sigma_{\varepsilon,N}\|_1, \label{eq:qf-ti-perturbed-mixing}
\end{align}
where $m_N=2|\Lambda_N|$. The prefactor is linear in the system size and the decay rate is independent of $N$, completing the proof.
\end{proof}

\begin{corollary}[Uniform rapid mixing with constant decay rate under uniform dissipation]
\label{cor:qf-uniform-dissipation}\resulttoc{cor:qf-uniform-dissipation}{Uniform rapid mixing under uniform dissipation}
Let $\faml{\LL}$ be a family of quasifree fermionic Lindbladians with one fermionic mode per site. Suppose that there exist $\mu,R_0>0$, independent of the system size, such that every restriction to $A=b_x(R)$ with $x\in\Lambda$ and $R\geq R_0$ satisfies
\begin{align*}
    \operatorname{Re}(\mathbf M_A)\succeq\mu\BI.
\end{align*}
Then $\LL_A$ has a unique stationary state $\sigma_{\LL_A}$, and for every state $\rho$ on $A$ and every $t\geq0$,
\begin{align}
    \|e^{t\LL_A}(\rho)-\sigma_{\LL_A}\|_1\leq m_Ae^{-4\mu t}\|\rho-\sigma_{\LL_A}\|_1, \label{app:eq:qf-uniform-dissipation}
\end{align}
where $m_A=2|A|$ is the number of Majoranas in $A$. Thus the family is uniformly rapidly mixing with constant decay rate. The same bound holds with $\sigma_{\LL_A+\KK_A}$ for $\LL_A+\KK_A$ whenever $\faml{\KK}$ is a quasifree Lindbladian family on the same modes and $\KK_A$ retains its Hamiltonian terms and jumps supported entirely on $A$.
\end{corollary}

\begin{proof}[Proof of \cref{cor:qf-uniform-dissipation}]
\prooflabel{cor:qf-uniform-dissipation}{proof:qf-uniform-dissipation}
Fix $A=b_x(R)$ with $R\geq R_0$. Every finite-dimensional Lindbladian has a stationary state; let $\sigma_{\LL_A}$ be one for $\LL_A$.

Applying \cref{prop:quasifree_rapid_mixing} to this stationary state gives
\begin{align}
    \|e^{t\LL_A}(\rho)-\sigma_{\LL_A}\|_1 &\overset{(1)}{\leq}m_Ae^{-4t\lambda_{\min}(\operatorname{Re}\mathbf M_A)}\|\rho-\sigma_{\LL_A}\|_1,\notag\\
    &\overset{(2)}{\leq}m_Ae^{-4\mu t}\|\rho-\sigma_{\LL_A}\|_1,\notag\\
    &\overset{(3)}{=}2|A|e^{-4\mu t}\|\rho-\sigma_{\LL_A}\|_1, \notag
\end{align}
where $(1)$ is the quasifree mixing estimate, $(2)$ uses the uniform lower bound on $\operatorname{Re}\mathbf M_A$, and $(3)$ uses $m_A=2|A|$. Applying the bound to another stationary state and taking $t\to\infty$ proves uniqueness. The prefactor $2|A|$ is linear in the number of sites, and the decay rate $4\mu$ is independent of $A$ and $\Lambda$, proving uniform rapid mixing with constant decay rate.

For an addition, the retained jump sets give $\mathbf M_{\LL_A+\KK_A}=\mathbf M_{\LL_A}+\mathbf M_{\KK_A}$. Since $\operatorname{Re}\mathbf M_{\KK_A}\succeq0$, the same lower bound $\operatorname{Re}\mathbf M_{\LL_A+\KK_A}\succeq\mu\BI$ holds. The preceding argument applied to $\LL_A+\KK_A$ proves the remaining claim.
\end{proof}

\partbibliography

\paperpart{appendices-commuting_generators_mlsi}{1}
\section{Commuting generators and MLSI}\label{app:commuting}
The condition that a Lindbladian $\LL$ does not have any nonzero purely imaginary eigenvalues is equivalent to the infinite-time limit being well-defined.
\begin{lemma}[Well-defined infinite-time limit]\label{lem:infinite-time-limit}
Let $\LL$ be a Lindbladian on a finite-dimensional $\BB(\HH)$. The infinite-time limit $P_\LL=\lim_{t\to\infty}e^{t\LL}$, understood pointwise as $P_\LL(X)=\lim_{t\to\infty}e^{t\LL}(X)$ for every $X\in\BB(\HH)$, exists if and only if $\LL$ has no nonzero purely imaginary eigenvalue, i.e.\ $\operatorname{spec}(\LL)\cap i\mathbb{R}=\{0\}$.
\end{lemma}

\begin{proof}[Proof of \cref{lem:infinite-time-limit}]
\prooflabel{lem:infinite-time-limit}{proof:infinite-time-limit}
The spectral condition is necessary: if $\LL X=i\omega X$ with $\omega\neq0$ and $X\neq0$, then $e^{t\LL}(X)=e^{i\omega t}X$ oscillates, so the limit fails to exist already at this single $X$. It is also sufficient: let $\operatorname{spec}(\LL)\cap i\mathbb{R}=\{0\}$. Every eigenvalue of a Lindbladian has $\operatorname{Re}\lambda\le0$, so by assumption every $\lambda\neq0$ has $\operatorname{Re}\lambda<0$.

We first show $\ker\LL^2=\ker\LL$. If $\LL^2X=0$ then $Y:=\LL X$ lies in $\ker\LL$, so the exponential series terminates after two terms: $e^{t\LL}(X)=X+tY$. Each $e^{t\LL}$ is a channel and hence does not increase trace distance, $\|e^{t\LL}(A)-e^{t\LL}(B)\|_1\le\|A-B\|_1$, which for $B=0$ says that no trajectory can grow; this forces $Y=\LL X=0$. Consequently $\ker\LL\cap\operatorname{ran}\LL=\{0\}$: an element $Y=\LL X$ of the intersection satisfies $\LL^2X=\LL Y=0$, hence $X\in\ker\LL^2=\ker\LL$ and $Y=\LL X=0$. Together with $\dim\ker\LL+\dim\operatorname{ran}\LL=\dim\BB(\HH)$ this gives the decomposition
\begin{align}
    &\BB(\HH)=\ker\LL\oplus\operatorname{ran}\LL ,
\end{align}
in which $\operatorname{ran}\LL$ is invariant under $e^{t\LL}$ because $\LL$ commutes with it. Split $X\in\BB(\HH)$ accordingly as $X=X_0+X_1$. The kernel part stays fixed, $e^{t\LL}(X_0)=X_0$, while the remainder decays: the restriction of $\LL$ to $\operatorname{ran}\LL$ is invertible, so its eigenvalues are among the $\lambda\neq0$ and have $\operatorname{Re}\lambda<0$, whence $e^{t\LL}(X_1)\to0$. Therefore $e^{t\LL}(X)\to X_0$ for every $X$, i.e. the pointwise limit exists.
\end{proof}

It is not necessarily true that if $\LL,\KK$ have well-defined infinite-time limits, then so will their sum. However, this is true when $\LL,\KK$ commute: suppose that $\LL,\KK$ have infinite-time limits:
\begin{align}
    &P_\LL = \lim_{t\to\infty}e^{t\LL},\quad P_\KK = \lim_{t\to\infty}e^{t\KK}.
\end{align}
Then $e^{t(\LL+\KK)}=e^{t\LL}\circ e^{t\KK}=e^{t\KK}\circ e^{t\LL}$, and the limiting projections satisfy $P_{\LL+\KK}=P_\LL\circ P_\KK=P_\KK\circ P_\LL$:
\begin{align}
    &\|e^{t\LL}e^{t\KK}-P_\LL P_\KK\|_{1\to1}\le\|(e^{t\LL}-P_\LL)e^{t\KK}\|_{1\to1} +\|P_\LL(e^{t\KK}-P_\KK)\|_{1\to1}\overset{t\to\infty}{\longrightarrow}0,
\end{align}
because $e^{t\KK}$ and $P_\LL$ are contractions, the latter being a channel as a limit of channels, while $P_\LL=\lim_{t\to\infty}e^{t\LL}$ commutes with every $e^{s\KK}$ and hence with $P_\KK$; in particular $\LL+\KK$ has a well-defined infinite-time limit as well. We see below that their mixing time is bounded by the maximum of the two mixing times at half the target error.

\begin{lemma}[Commuting sum: mixing-time upper bound, \cite{KastoryanoWolf2012}]\label{lem:commuting-sum}\label{app:lem:commuting-sum}
Let $\HH$ be a finite-dimensional Hilbert space and $\LL$ and $\KK$ be commuting Lindbladians acting on $\SSS(\HH)$ each with a well-defined infinite-time limit. Then, for any $\rho\in\SSS(\HH)$
\begin{align}
    \|e^{t(\LL+\KK)}(\rho)-P_{\LL+\KK}(\rho)\|_1 \leq \|e^{t\LL}(\rho)-P_{\LL}(\rho)\|_1 +\|e^{t\KK}(\rho)-P_{\KK}(\rho)\|_1,
\end{align}
and the mixing times satisfy:
\begin{align}
    \tau_{\operatorname{mix}}^{\LL+\KK}(\delta) \leq \max\{\tau_{\operatorname{mix}}^{\LL}(\delta/2),\tau_{\operatorname{mix}}^{\KK}(\delta/2)\} .
\end{align}
\end{lemma}

\begin{proof}[Proof of \cref{lem:commuting-sum}]
\prooflabel{lem:commuting-sum}{proof:commuting-sum}
We follow the telescoping proof of \cite[Theorem~6]{KastoryanoWolf2012}, keeping the initial state fixed instead of taking the supremum over states. Since $\LL$ and $\KK$ commute we have
\begin{align}
    e^{t(\LL+\KK)} = e^{t\LL} \circ e^{t\KK} = e^{t\KK} \circ e^{t\LL} \quad \mathrm{and} \quad P_{\LL+\KK} = P_{\LL} \circ P_{\KK} = P_{\KK} \circ P_{\LL} .
\end{align}
Then
\begin{align}
    \|e^{t(\LL+\KK)} (\rho) - P_{\LL+\KK} (\rho) \|_1 &\overset{(1)}{\leq} \|e^{t(\LL+\KK)} (\rho) -  e^{t\KK} (P_{\LL}(\rho) )\|_1+\| e^{t\KK}(  P_{\LL}(\rho) )  - P_{\LL+\KK} (\rho) \|_1,\notag\\
    &\overset{(2)}{=}\|e^{t\KK}(e^{t\LL} (\rho) -P_{\LL}  (\rho ) )\|_1+\| P_{\LL} (e^{t\KK} (\rho)  - P_{\KK} (\rho) ) \|_1,\notag\\
    &\overset{(3)}{\leq} \|e^{t\LL} (\rho) -P_{\LL}  (\rho) \|_1+\|e^{t\KK} (\rho)  - P_{\KK} (\rho) \|_1,   \label{eq:tracenormbound:split}
\end{align}
where $(1)$ uses the triangle inequality, $(2)$  that $[P_{\LL}, e^{t \KK}] = 0$ and $(3)$ that the trace norm is contractive under the quantum channels $e^{t\KK},P_\LL$. For any time strictly larger than the maximum of the two mixing times at error $\delta/2$, both terms are at most $\delta/2$. Taking the infimum gives the mixing-time bound.
\end{proof}
\begin{lemma}[Repeated \cref{lem:commuting-unique-stationary}: Commuting addition to a relaxing generator]\label{app:lem:commuting-unique-stationary}
Let $\LL,\KK$ be commuting Lindbladians and suppose that $\LL$ has a unique stationary state $\sigma$ and a well-defined infinite-time limit. Then $\KK(\sigma)=0$, $\LL+\KK$ has the unique stationary state $\sigma$, and, for every $\rho\in\SSS(\HH)$ and $t\ge0$,
\begin{align*}
    \|e^{t(\LL+\KK)}(\rho)-\sigma\|_1\le \|e^{t\LL}(\rho)-\sigma\|_1.
\end{align*}
In particular, every rapid-mixing bound for $\LL$ is inherited by $\LL+\KK$.
\end{lemma}

\begin{proof}[Proof of \cref{lem:commuting-unique-stationary}]
\prooflabel{lem:commuting-unique-stationary}{proof:commuting-unique-stationary}
For every $s\geq0$, commutativity gives
\begin{align*}
    \LL(e^{s\KK}(\sigma))=e^{s\KK}(\LL(\sigma))=0.
\end{align*}
The state $e^{s\KK}(\sigma)$ is therefore stationary for $\LL$, so uniqueness gives $e^{s\KK}(\sigma)=\sigma$. Differentiating at $s=0$ yields $\KK(\sigma)=0$. Therefore $e^{t\KK}(\sigma)=\sigma$, and
\begin{align*}
    \|e^{t(\LL+\KK)}(\rho)-\sigma\|_1
    &=\|e^{t\KK}(e^{t\LL}(\rho))-e^{t\KK}(\sigma)\|_1\\
    &\le \|e^{t\LL}(\rho)-\sigma\|_1
\end{align*}
by trace-norm contractivity of $e^{t\KK}$. The right-hand side tends to zero because $P_\LL(\rho)$ is a stationary state of $\LL$ and therefore equals $\sigma$. Thus every initial state converges to $\sigma$ under $\LL+\KK$, which also proves uniqueness of its stationary state.
\end{proof}
\begin{lemma}[Repeated \cref{lem:commuting-kernel}: Commuting sum with kernel inclusion]\label{app:lem:commuting-kernel}
Let $\LL,\KK$ be commuting Lindbladians, suppose that $P_\LL=\lim_{t\to\infty}e^{t\LL}$ exists, and assume $\ker(\LL)\subseteq\ker(\KK)$. Then $P_{\LL+\KK}$ exists and equals $P_\LL$, and for all $\rho\in\SSS(\HH)$ and all $t\ge0$,
\begin{align*}
    \|e^{t(\LL+\KK)}(\rho)-P_{\LL+\KK}(\rho)\|_1\le \|e^{t\LL}(\rho)-P_\LL(\rho)\|_1.
\end{align*}
\end{lemma}

\begin{proof}[Proof of \cref{lem:commuting-kernel}]
\prooflabel{lem:commuting-kernel}{proof:commuting-kernel}
Kernel inclusion gives $\KK P_\LL=0$, hence $e^{t\KK}P_\LL=P_\LL$. Using commutativity and trace-norm contractivity, we obtain
\begin{align}
    \|e^{t(\LL+\KK)}(\rho)-P_\LL(\rho)\|_1
    &=\|e^{t\KK}(e^{t\LL}(\rho)-P_\LL(\rho))\|_1,\notag\\
    &\leq\|e^{t\LL}(\rho)-P_\LL(\rho)\|_1.\label{eq:commuting-kernel-contraction}
\end{align}
The right-hand side tends to zero, so the infinite-time limit of the sum exists and equals $P_\LL$. This proves the stated inequality.
\end{proof}
These lemmas imply: if $\LL,\KK$ commute and both are rapidly mixing, then $\LL+\KK$ is rapidly mixing; and either a unique stationary state for $\LL$ or the inclusion $\ker(\LL)\subseteq\ker(\KK)$ ensures that rapid mixing of $\LL$ alone is inherited by $\LL+\KK$. These statements can fail without commutativity, as the counterexamples of \cref{sec:instability_results} show.
\begin{lemma}[Reset perturbation]\label{lem:reset}
Let $\LL$ be a Lindbladian with a fixed point $\sigma\in\SSS(\HH)$, not necessarily full rank, and let $\KK$ be the reset Lindbladian to $\sigma$,
\begin{align}
    \label{eq:reset} &\KK(A)=\Tr[A]\,\sigma-A ,
\end{align}
which satisfies $\|\KK\|_{1\to1}\le2$, $\ker\KK=\C\sigma$ and $P_\KK(A)=\Tr[A]\,\sigma$. Then $[\LL,\KK]=0$ and, for every $\varepsilon>0$, every $\rho\in\SSS(\HH)$ and every $t\ge0$,
\begin{align}
    \label{eq:resetdecay} &\|e^{t(\LL+\varepsilon\KK)}(\rho)-\sigma\|_1\le e^{-\varepsilon t}\|\rho-\sigma\|_1 ,
\end{align}
irrespective of the mixing behaviour of $\LL$. In particular $\LL+\varepsilon\KK$ has the well-defined infinite-time limit $P_{\LL+\varepsilon\KK}(A)=\Tr[A]\,\sigma$ and its only fixed point among the states is $\sigma$.
\end{lemma}

\begin{proof}[Proof of \cref{lem:reset}]
\prooflabel{lem:reset}{proof:reset}
Let $\sigma=\sum_k p_k|k\rangle\langle k|$ and let $\{|j\rangle\}$ be any orthonormal basis. The jump operators $L_{kj}=\sqrt{p_k}\,|k\rangle\langle j|$ satisfy $\sum_{k,j}L_{kj}AL_{kj}^\dagger=\Tr[A]\,\sigma$ and $\sum_{k,j}L_{kj}^\dagger L_{kj}=\mathbf 1$, so the GKSL generator they define with vanishing Hamiltonian is exactly \cref{eq:reset}; and $\|\KK(A)\|_1\le|\Tr A|+\|A\|_1\le2\|A\|_1$.

Write $\mathcal{R}(A)=\Tr[A]\,\sigma$, so that $\KK=\mathcal{R}-\mathrm{id}$ with $\mathcal{R}^2=\mathcal{R}$, and $e^{t\LL}\circ\mathcal{R}=\mathcal{R}$ because $e^{t\LL}(\sigma)=\sigma$. The two generators commute,
\begin{align}
    &\LL\KK(A)=\Tr[A]\,\LL(\sigma)-\LL(A)=-\LL(A)=\Tr[\LL(A)]\,\sigma-\LL(A)=\KK\LL(A) ,
\end{align}
by $\LL(\sigma)=0$ and the fact that $\LL$ is trace-annihilating, so with $e^{s\KK}=e^{-s}e^{s\mathcal{R}}=e^{-s}\,\mathrm{id}+(1-e^{-s})\,\mathcal{R}$,
\begin{align}
    \label{eq:resetsemigroup} &e^{t(\LL+\varepsilon\KK)}=e^{t\LL}\circ e^{\varepsilon t\KK} =e^{-\varepsilon t}\,e^{t\LL}+(1-e^{-\varepsilon t})\,\mathcal{R},
\end{align}
where the displayed form of $e^{s\KK}$ already gives $P_\KK=\mathcal{R}$, while $\KK(A)=0$ forces $A=\Tr[A]\,\sigma$, i.e.\ $\ker\KK=\C\sigma$. Since $\|e^{t\LL}\|_{1\to1}\le1$, the first term of \cref{eq:resetsemigroup} vanishes as $t\to\infty$, so $e^{t(\LL+\varepsilon\KK)}\to\mathcal{R}$ and a state fixed by the semigroup equals $\mathcal{R}(\rho)=\sigma$. Applied to a state, \cref{eq:resetsemigroup} yields
\begin{align}
    \|e^{t(\LL+\varepsilon\KK)}(\rho)-\sigma\|_1& =e^{-\varepsilon t}\|e^{t\LL}(\rho)-\sigma\|_1\le e^{-\varepsilon t}\|\rho-\sigma\|_1 ,
\end{align}
the equality because $\mathcal{R}(\rho)=\sigma$ cancels the second term, and the bound by stationarity of $\sigma$ and trace-norm contractivity of $e^{t\LL}$.
\end{proof}
In the remainder of the section we show that if we additionally assume MLSI, the analogous results to the above lemmas hold for the MLSI of the sum. We will first need to introduce some notation. A \emph{conditional expectation} onto a unital $*$-subalgebra $\NN\subseteq\MM$ of a von Neumann algebra $\MM$ is a completely positive unital map $E^*:\MM\to\NN$ satisfying $E^*[X\Phi Y]=XE^*[\Phi]Y$ for $X,Y\in\NN$, $\Phi\in\MM$; its predual is denoted $E$. The \emph{decoherence-free subalgebra} of $\LL$ is the maximal algebra on which the dual semigroup acts as a $*$-automorphism \cite{Frigerio1978}:
\begin{align}
    \label{app:eq:dfs} \FF(\LL)=\Bigl\{X\in\BB(\HH): e^{t\LL^*}(X^\dagger X)=e^{t\LL^*}(X)^\dagger e^{t\LL^*}(X),\ e^{t\LL^*}(XX^\dagger)=e^{t\LL^*}(X)e^{t\LL^*}(X)^\dagger\ \ \forall t\ge0\Bigr\}.
\end{align}
\begin{proposition}[Repeated \cref{thm:CE}: Conditional expectation and infinite-time limit, Proposition 8 of \cite{carbone2013decoherence}, Theorem 19 \cite{carbone2015environment}]\label{app:thm:CE}
Assume $e^{t\LL}$ has a full-rank invariant state $\sigma$. Then there is a unique conditional expectation $E^*_\LL:\BB(\HH)\to\FF(\LL)$ with $E_\LL(\sigma)=\sigma$, and for every $X\in\BB(\HH)$,
\begin{align*}
    \lim_{t\to+\infty}e^{t\LL^*}(X-E^*_\LL[X])=0, \qquad\text{equivalently}\qquad \lim_{t\to+\infty}e^{t\LL}(\rho-E_\LL(\rho))=0\quad\forall\rho\in\SSS(\HH).
\end{align*}
If $\LL$ has no nonzero purely imaginary eigenvalues, then $E_\LL=P_\LL$.
\end{proposition}

\begin{proof}[Proof of \cref{thm:CE}]
\prooflabel{thm:CE}{proof:CE}
The existence and uniqueness of the conditional expectation, and the decay of the complementary component, follow from \cite[Proposition~8]{carbone2013decoherence} and \cite[Theorem~19]{carbone2015environment}. Suppose in addition that $\LL$ has no nonzero purely imaginary eigenvalues. The dual evolution on the finite-dimensional decoherence-free algebra $\FF(\LL)$ is a group of $*$-automorphisms, whose generator is diagonalizable with purely imaginary spectrum. The absence of nonzero purely imaginary eigenvalues therefore makes this evolution the identity, so $e^{t\LL^*}(E_\LL^*[X])=E_\LL^*[X]$. By linearity and the decay statement above,
\begin{align*}
    e^{t\LL^*}(X)
    &=e^{t\LL^*}(E_\LL^*[X])+e^{t\LL^*}(X-E_\LL^*[X])\\
    &=E_\LL^*[X]+e^{t\LL^*}(X-E_\LL^*[X])\longrightarrow E_\LL^*[X].
\end{align*}
Taking preduals gives $P_\LL=E_\LL$.
\end{proof}

For a Lindbladian $\LL$ with a full-rank invariant state and a well-defined infinite-time limit, we write $E_\LL=P_\LL$ for the state predual of the corresponding conditional expectation $E_\LL^*$. We use the fixed-point MLSI convention in \cref{eq:fixed-point-MLSI}; a full-rank fixed point need not be unique. Under these hypotheses,
\begin{align}
    \label{eq:ranE} &\FF(\LL)\overset{(1)}{=}\operatorname{ran}E^*_\LL \overset{(2)}{=}\operatorname{ran}P^*_\LL \overset{(3)}{=}(\ker P_\LL)^{\perp} \overset{(4)}{=}(\operatorname{ran}\LL)^{\perp} \overset{(5)}{=}\ker\LL^* ,
\end{align}
where $(1)$ combines the codomain in \cref{thm:CE} with the fact that a conditional expectation is a projection onto its target algebra, since taking $\Phi=Y=\BI$ in its definition gives $E^*[X]=XE^*[\BI]=X$ for $X\in\NN$; $(2)$ is $E_\LL=P_\LL$; $(3)$ is $\operatorname{ran}A^*=(\ker A)^{\perp}$ for the Hilbert-Schmidt inner product; $(4)$ uses that $P_\LL$ is the projection onto $\ker\LL$ along $\operatorname{ran}\LL$, by \cref{lem:infinite-time-limit}; and $(5)$ uses $(\operatorname{ran}\LL)^\perp=\ker\LL^*$, by adjoint duality. We now recall two useful results from the literature.
\begin{proposition}[Chain rule, Lemma 3.4 of \cite{JLR2019}]\label{prop:chainrule}
Let $E^*:\BB(\HH)\to\NN$ be a conditional expectation onto a unital $*$-subalgebra $\NN\subseteq\BB(\HH)$, and let $E$ be its state predual. Let $\sigma\in\SSS(\HH)$ satisfy $E(\sigma)=\sigma$. Then for every $\rho\in\SSS(\HH)$,
\begin{align}
    \label{eq:chainrule} D(\rho\|\sigma)=D(\rho\|E(\rho))+D(E(\rho)\|\sigma).
\end{align}
\end{proposition}
\begin{proposition}[Swap of fixed points, Lemma 2.1 of \cite{BCG2024}]\label{prop:swap_fixpoint}
Let $\LL$ be a Lindbladian with a well-defined infinite-time limit, and let $\sigma_1,\sigma_2\in\SSS(\HH)$ be full-rank states in its kernel. Then $\log\sigma_1-\log\sigma_2\in\FF(\LL)=\ker\LL^*$, so in particular
\begin{align}
    &E^*_\LL(\log\sigma_1-\log\sigma_2)=\log\sigma_1-\log\sigma_2,
\end{align}
and, for all $\rho\in\SSS(\HH)$,
\begin{align}
    \label{eq:swap} \Tr[\LL(\rho)(\log\sigma_1-\log\sigma_2)]=0.
\end{align}
\end{proposition}

\begin{proof}[Proof of \cref{prop:swap_fixpoint}]
\prooflabel{prop:swap_fixpoint}{proof:swap_fixpoint}
On a finite-dimensional Hilbert space $\HH$, the Hilbert space and decoherence-free subalgebra $\FF(\LL)$ can be written, up to unitary equivalence, as \cite{kadison1986fundamentals}:
\begin{align}
    &\HH = \bigoplus_{i\in I} \HH^{(1)}_i \otimes \HH_i^{(2)},\quad \FF(\LL) = \bigoplus_{i\in I}\BB(\HH_i^{(1)})\otimes \BI_{\HH_i^{(2)}},
\end{align}
and there is a family of density operators $\{\tau_i:i\in I\}$ such that for every full-rank state $\rho$:
\begin{align}
    &E_\LL(\rho) = \bigoplus_{i\in I}\Tr_{\HH_i^{(2)}}(P_i \rho P_i)\otimes \tau_i,
\end{align}
where for each $i$, $P_i$ is the projection onto $\HH_i^{(1)}\otimes \HH_i^{(2)}$; see \cite[Eq.~(2.10)]{BR2022}, originally proven in \cite{deschamps2016structure}; note that \cite{BR2022} writes $E_\NN$ for the map on observables and $E_{\NN*}$ for its predual, the opposite of the convention used here. Since $\sigma_1,\sigma_2\in\ker\LL$ and $E_\LL$ acts as the identity there, we can thus write the fixed points as:
\begin{align}
    \sigma_j&= E_\LL(\sigma_j)=\bigoplus_{i\in I}\tilde{\sigma}_{j,i}\otimes \tau_i, \text{ where }\tilde{\sigma}_{j,i} := \Tr_{\HH_i^{(2)}}(P_i \sigma_{j} P_i), \quad j=1,2,
\end{align}
with the \emph{same} family $\{\tau_i\}$ for both $j$, and with $\tilde{\sigma}_{j,i},\tau_i>0$ because $\sigma_j$ is full rank. Now using $\log(A\otimes B)=\log A\otimes\BI+\BI\otimes\log B$, we obtain:
\begin{align}
    \log\sigma_1-\log\sigma_2&=\bigoplus_i\log\tilde{\sigma}_{1,i}\otimes\BI_{\HH_i^{(2)}}+\bigoplus_i\BI_{\HH_i^{(1)}}\otimes \log\tau_{i}-\bigoplus_i\log\tilde{\sigma}_{2,i}\otimes\BI_{\HH_i^{(2)}}-\bigoplus_i\BI_{\HH_i^{(1)}}\otimes \log\tau_{i} ,\notag\\
    &=\bigoplus_i\bigl(\log\tilde{\sigma}_{1,i}-\log\tilde{\sigma}_{2,i}\bigr)\otimes\BI_{\HH_i^{(2)}} \ \in\ \FF(\LL)=\ker\LL^* = \operatorname{ran}E_\LL^*, \label{eq:loglogker}
\end{align}
since the cross terms cancel, the equalities of algebras being \cref{eq:ranE}; this yields the first claim. The second claim is:
\begin{align}
    &\Tr[\LL(\rho)(\log\sigma_1-\log\sigma_2)] \overset{(1)}{=}\Tr[\rho\,\LL^*(\log\sigma_1-\log\sigma_2)]\overset{(2)}{=}0,
\end{align}
where $(1)$ uses the Hilbert-Schmidt adjoint and $(2)$ uses \cref{eq:loglogker}.
\end{proof}
The next lemma isolates the only consequence of \cref{prop:swap_fixpoint} used below, in the generality in which it holds.
\begin{lemma}[Additivity of the entropy production]\label{lem:EP-additive}
Let $\LL,\KK$ be Lindbladians such that $\LL,\KK$ and $\LL+\KK$ have well-defined infinite-time limits, and let $\sigma$ be a full-rank state in $\ker\LL\cap\ker\KK$. Then, for every $\rho\succ0$, the states $E_\LL(\rho),E_\KK(\rho),E_{\LL+\KK}(\rho)$ are full rank and the entropy productions are additive,
\begin{align}
    \label{eq:EP-additive} &\operatorname{EP}_{\LL}(\rho)+\operatorname{EP}_{\KK}(\rho)=\operatorname{EP}_{\LL+\KK}(\rho) .
\end{align}
\end{lemma}

\begin{proof}[Proof of \cref{lem:EP-additive}]
\prooflabel{lem:EP-additive}{proof:EP-additive}
By linearity $\sigma\in\ker(\LL+\KK)$, so $\sigma$ is a full-rank state in $\ker\GG$ for each $\GG\in\{\LL,\KK,\LL+\KK\}$. With $c = \lambda_{\min}(\rho)/\lambda_{\max}(\sigma)$,
\begin{align}
    &E_\GG(\rho)-c\sigma \overset{(1)}{=}E_\GG(\rho)-cE_\GG(\sigma) \overset{(2)}{=}E_\GG(\rho -c\sigma) \overset{(3)}{\succeq} 0,
\end{align}
where $(1)$ uses $E_\GG(\sigma)=\sigma$, $(2)$ uses linearity of $E_\GG$ and $(3)$ that $E_\GG$ is a positive map together with $\rho-c\sigma\succeq0$, which holds for this $c$ because $\rho,\sigma$ are full-rank states; hence $E_\GG(\rho)\succeq c\sigma\succ0$, so the three entropy productions in \cref{eq:EP-additive} are well defined. Moreover $E_\GG=P_\GG$ projects onto $\ker\GG$ by \cref{lem:infinite-time-limit}, so $E_\GG(\rho)$ and $\sigma$ are two full-rank states in $\ker\GG$ and
\begin{align}
    \operatorname{EP}_{\LL}(\rho)+\operatorname{EP}_{\KK}(\rho) &\overset{(1)}{=} - \Tr[\LL(\rho)(\log \rho - \log \sigma)] - \Tr[\KK(\rho)(\log \rho - \log \sigma)],\notag\\
    &\overset{(2)}{=} - \Tr[(\LL+\KK)(\rho)(\log \rho - \log \sigma)],\notag\\
    &\overset{(3)}{=} \operatorname{EP}_{\LL+\KK}(\rho),
\end{align}
where $(1)$ applies \cref{prop:swap_fixpoint} to $\LL$ and to $\KK$ to swap the fixed points $E_\LL(\rho),E_\KK(\rho)$ for $\sigma$ in the entropy productions, $(2)$ uses linearity and $(3)$ applies it to $\LL+\KK$ to swap $\sigma$ for $E_{\LL+\KK}(\rho)$.
\end{proof}
The following proposition bounds the MLSI constant of a sum when the stationary projections commute.
\begin{proposition}[Repeated \cref{prop:mlsi-commuting}]\label{app:prop:mlsi-commuting}
Let $\HH$ be a finite-dimensional Hilbert space, and let $\LL$ and $\KK$ be Lindbladians on $\BB(\HH)$, each with a full-rank invariant state, a well-defined infinite-time limit, and a positive MLSI constant. Suppose that
\begin{align*}
    E_\LL E_\KK=E_\KK E_\LL.
\end{align*}
Then the infinite-time limit of $\LL+\KK$ exists, and
\begin{align}
    E_{\LL+\KK}&=E_\LL E_\KK,&
    \ker(\LL+\KK)&=\ker\LL\cap\ker\KK,\notag\\
    \alpha(\LL+\KK)&\geq\min\{\alpha(\LL),\alpha(\KK)\}. &&
\end{align}
If $E_\LL=E_\KK$, the stronger bound
\begin{align*}
    \alpha(\LL+\KK)\geq\alpha(\LL)+\alpha(\KK)
\end{align*}
holds.
\end{proposition}

\begin{proof}[Proof of \cref{prop:mlsi-commuting}]
\prooflabel{prop:mlsi-commuting}{proof:mlsi-commuting}
Put $E:=E_\LL E_\KK=E_\KK E_\LL$ and $\GG:=\LL+\KK$. Each $E_j$, $j\in\{\LL,\KK\}$, maps strictly positive states to strictly positive states: if $\sigma_j\succ0$ is invariant and $\rho\succeq c\sigma_j$ for some $c>0$, then $E_j(\rho)\succeq c\sigma_j$. Thus
\begin{align*}
    \sigma:=E(\BI/\dim\HH)
\end{align*}
is a full-rank state. Commutativity of the projections makes it stationary for both generators. Moreover,
\begin{align*}
    \operatorname{ran}E&=\ker\LL\cap\ker\KK,&
    E\GG&=\GG E=0.
\end{align*}
For example, $E\LL=E_\KK E_\LL\LL=0$ and $E\KK=E_\LL E_\KK\KK=0$; the range identity gives $\GG E=0$.

For $\rho\succ0$, the chain rule and data processing give
\begin{align}
    D(\rho\|E(\rho))
    &=D(\rho\|E_\LL(\rho))
      +D(E_\LL(\rho)\|E_\LL E_\KK(\rho)),\label{eq:Drho_ELK}\\
    &\leq D(\rho\|E_\LL(\rho))+D(\rho\|E_\KK(\rho)).
    \label{eq:rel_entropy_ineq}
\end{align}
The chain rule applies because $E(\rho)$ is fixed by $E_\LL$. The inequality is data processing for the channel $E_\LL$, applied to $\rho$ and $E_\KK(\rho)$: $D(E_\LL(\rho)\|E_\LL E_\KK(\rho))\leq D(\rho\|E_\KK(\rho))$. All the reference states are full rank by the positivity argument above. With $\tau:=E(\rho)$, \cref{prop:swap_fixpoint} applied separately to the two summands yields
\begin{align*}
    -\Tr[\GG(\rho)(\log\rho-\log\tau)]
    &=\EP_\LL(\rho)+\EP_\KK(\rho)\\
    &\geq2\min\{\alpha(\LL),\alpha(\KK)\}D(\rho\|\tau).
\end{align*}
Here $\tau$ is stationary for both summands, so this calculation does not require an infinite-time limit for $\GG$.

Write $\rho_t=e^{t\GG}(\rho)$. Since $E\rho_t=\tau$, the same estimate along the trajectory gives
\begin{align*}
    -\frac{d}{dt}D(\rho_t\|\tau)
    &\geq2\min\{\alpha(\LL),\alpha(\KK)\}D(\rho_t\|\tau).
\end{align*}
Gr\"onwall's inequality and Pinsker's inequality imply $\rho_t\to E(\rho)$. For an arbitrary state, approximate it by $(1-\delta)\rho+\delta\sigma$. Contractivity of $e^{t\GG}$ and $E$ bounds the two trace-norm approximation errors by $4\delta$, uniformly in time. Taking $t\to\infty$ and then $\delta\downarrow0$ proves convergence for every state. Thus $E_{\GG}=E$, giving the kernel identity, and the production estimate is the asserted MLSI. The entropy-decay inequality extends to all states by the same faithful approximation.

If $E_\LL=E_\KK=E$, the two MLSIs have the same relative-entropy denominator. Their sum therefore gives
\begin{align*}
    \EP_{\GG}(\rho)
    &\geq2\bigl(\alpha(\LL)+\alpha(\KK)\bigr)D(\rho\|E(\rho)),
\end{align*}
which proves the stronger bound.
\end{proof}
We next show that kernel inclusion preserves the stationary projection and the MLSI without requiring commutativity or convergence of the addition.
\begin{proposition}[Repeated \cref{prop:mlsi-kernel}]\label{app:prop:mlsi-kernel}
Let $\LL$ be a Lindbladian with a full-rank invariant state, a well-defined infinite-time limit $E_\LL$, and $\alpha(\LL)>0$. Let $\KK$ be a Lindbladian satisfying $\ker\LL\subseteq\ker\KK$. Then the infinite-time limit of $\LL+\KK$ exists, and
\begin{align*}
    E_{\LL+\KK}=E_\LL,\qquad
    \ker(\LL+\KK)=\ker\LL,\qquad
    \alpha(\LL+\KK)\geq\alpha(\LL).
\end{align*}
\end{proposition}

\begin{proof}[Proof of \cref{prop:mlsi-kernel}]
\prooflabel{prop:mlsi-kernel}{proof:mlsi-kernel}
Put $E:=E_\LL$, $\GG:=\LL+\KK$, and let $\sigma\succ0$ be an invariant state of $\LL$. Kernel inclusion implies that $\KK(\sigma)=0$ and $e^{t\KK}E=E$. We first show that $Ee^{t\KK}=E$ as well.

On observables, use the Hilbert-space norm
\begin{align*}
    \|A\|_\sigma^2:=\Tr(\sigma A^\dagger A).
\end{align*}
For a channel $\Phi$ preserving $\sigma$, Kadison--Schwarz gives
\begin{align*}
    \|\Phi^*(A)\|_\sigma^2
    &\leq\Tr[\sigma\Phi^*(A^\dagger A)]
    =\|A\|_\sigma^2.
\end{align*}
In particular, $E^*$ is a contractive idempotent and hence an orthogonal projection in this Hilbert space. Indeed, contractivity gives $\|u\|_\sigma\leq\|u+zv\|_\sigma$ for $u\in\operatorname{ran}E^*$, $v\in\ker E^*$, and every $z\in\mathbb C$, which forces $u\perp v$.

Set $T=e^{t\KK}$. The identity $TE=E$ gives $E^*T^*=E^*$. For $u\in\operatorname{ran}E^*$, orthogonality and contractivity imply
\begin{align*}
    \|T^*u\|_\sigma^2
    &=\|u\|_\sigma^2+\|(\mathrm{id}-E^*)T^*u\|_\sigma^2
    \leq\|u\|_\sigma^2.
\end{align*}
Thus $T^*u=u$, so $T^*E^*=E^*$ and $ET=E$. We have therefore shown that $Ee^{t\KK}=e^{t\KK}E=E$ for every $t\geq0$. Taking the derivative of each identity with respect to $t$ at $t=0$, and using $\left.\frac{d}{dt}e^{t\KK}\right|_{t=0}=\KK$, gives
\begin{align*}
    E\KK=\left.\frac{d}{dt}(Ee^{t\KK})\right|_{t=0}=0,\qquad
    \KK E=\left.\frac{d}{dt}(e^{t\KK}E)\right|_{t=0}=0.
\end{align*}
Since $E=P_\LL$ also satisfies $E\LL=\LL E=0$, adding the two generators yields
\begin{align*}
    E\GG=E\LL+E\KK=0,\qquad \GG E=\LL E+\KK E=0.
\end{align*}

For $\rho\succ0$, put $\tau:=E(\rho)$ and $\rho_t:=e^{t\GG}(\rho)$. Choosing $c>0$ with $\rho\succeq c\sigma$ shows that $\tau,\rho_t\succeq c\sigma\succ0$. Moreover, $E\rho_t=\tau$, and $\tau$ is stationary for both summands. The MLSI for $\LL$ gives
\begin{align}
    2\alpha(\LL)D(\rho_t\|\tau)\leq\EP_\LL(\rho_t).
    \label{eq:mlsi-L}
\end{align}
For every $\omega\succ0$, data processing for $e^{s\KK}$ fixing $\tau$ gives
\begin{align}
    D(e^{s\KK}(\omega)\|\tau)&\leq D(\omega\|\tau),
    \label{eq:DPI-monotone}\\
    -\Tr[\KK(\omega)(\log\omega-\log\tau)]&\geq0.
    \label{eq:DPI-der}
\end{align}
The second line follows by differentiating the first at $s=0$. Consequently,
\begin{align}
    -\frac{d}{dt}D(\rho_t\|\tau)
    &=-\Tr[\GG(\rho_t)(\log\rho_t-\log\tau)]\notag\\
    &\geq\EP_\LL(\rho_t)
    \geq2\alpha(\LL)D(\rho_t\|\tau).
    \label{eq:EP-drop-K}
\end{align}
Integration yields
\begin{align*}
    D(e^{t\GG}(\rho)\|E(\rho))
    \leq e^{-2\alpha(\LL)t}D(\rho\|E(\rho)).
\end{align*}
Pinsker's inequality proves convergence to $E(\rho)$. For an arbitrary state, use $(1-\delta)\rho+\delta\sigma$ and contractivity as in \cref{proof:mlsi-commuting}; this extends convergence to all states, and faithful approximation also extends the entropy-decay bound. Hence $E_{\GG}=E$, which proves the kernel identity. With this identification, \cref{eq:EP-drop-K} is the MLSI for $\GG$ with constant $\alpha(\LL)$.
\end{proof}
In particular, every Lindbladian addition preserving the unique full-rank stationary state of $\LL$ inherits its MLSI.
\begin{corollary}[MLSI of a reset perturbation]\label{cor:perturbation}
In the setting of \cref{lem:reset}, assume in addition that $\sigma$ is full rank. Then $\alpha(\KK)\geq\tfrac12$ and, for every $\varepsilon>0$,
\begin{align}
    &\alpha(\LL+\varepsilon\KK)\geq \tfrac{\varepsilon}{2} .
\end{align}
\end{corollary}

\begin{proof}[Proof of \cref{cor:perturbation}]
\prooflabel{cor:perturbation}{proof:perturbation}
By \cref{lem:reset} the kernel of $\KK$ is $\mathbb{C}\sigma$ and $E_\KK(\rho)=P_\KK(\rho)=\sigma$ for every state $\rho$, so for $\rho\succ0$ we have $\KK(\rho)=\sigma-\rho$ and
\begin{align}
    \operatorname{EP}_{\KK}(\rho) &=\Tr[(\rho-\sigma)(\log\rho-\log\sigma)] \notag\\
    &=D(\rho\|\sigma)+D(\sigma\|\rho)\ \geq\ D(\rho\|\sigma) ,
\end{align}
the entropy production of a reset being the symmetrised relative entropy. This is the MLSI for $\KK$ with constant $\tfrac12$, and scaling gives $\alpha(\varepsilon\KK)=\varepsilon\,\alpha(\KK)\geq\varepsilon/2$.

Apply \cref{prop:mlsi-kernel} with $\varepsilon\KK$ as the base generator and $\LL$ as the addition. The base has the full-rank stationary state $\sigma$, the limit $E_{\varepsilon\KK}(\rho)=\sigma$, and $\ker(\varepsilon\KK)=\mathbb C\sigma\subseteq\ker\LL$. Hence
\begin{align*}
    \alpha(\LL+\varepsilon\KK)\geq\alpha(\varepsilon\KK)\geq\varepsilon/2.
\end{align*}
\end{proof}
Only the MLSI needs $\sigma$ to be full rank: by \cref{eq:resetdecay} the perturbed generator mixes at rate $\varepsilon$ for an arbitrary fixed point $\sigma$, and without any assumption on $\LL$. For a full-rank $\sigma$ the decay \cref{eq:resetdecay} is moreover faster than what $\alpha(\LL+\varepsilon\KK)\geq\varepsilon/2$ yields through the standard relative-entropy argument, which carries an additional $\log\log\|\sigma^{-1}\|_\infty$.

\begin{lemma}\label{lem:mlsi-epsilon}
Let $\LL,\KK$ both have positive MLSI constants $\alpha(\LL),\alpha(\KK)>0$ and assume that both have the same full-rank state $\sigma$ in the kernel. Furthermore, assume that all of $\LL,\KK,\LL+\KK$ have well-defined infinite-time limits and that $\ker(\LL + \KK) = \ker(\LL) \cap \ker(\KK)$. Assume the following positivity order for the dual of the infinite time limits holds for some $0<\eta<\eta_0$, where $\eta_0$ is the first positive zero of the function $f$ below:
\begin{align}
    \label{eq:positivityorder} (1-\eta) E_{\LL+\KK}^* \preceq E_{\LL}^* \circ E_{\KK}^*\preceq (1+\eta) E_{\LL + \KK}^*
\end{align}
Then $\LL+\KK$ has an MLSI constant $\alpha(\LL+\KK)\geq f(\eta) \min\{\alpha(\LL),\alpha(\KK)\}$, where
\begin{align}
    f(\eta) = \left(\frac{1-\eta}{1+\eta}-\frac{\eta}{(1-\eta)(2\ln 2-1)}\right).
\end{align}
Thus $f(\eta)>0$ throughout the assumed interval; numerically, $\eta_0\approx0.204$.
\end{lemma}

\begin{proof}[Proof of \cref{lem:mlsi-epsilon}]
\prooflabel{lem:mlsi-epsilon}{proof:mlsi-epsilon}
For $0<\eta<\eta_0$, this follows from \cite[Lemma 2.3]{GJLL23} which states that if $E_{\LL+\KK}^*$ is a conditional expectation, $E_{\LL+\KK}^*\circ E^*_{\LL}\circ E^*_{\KK} = E_{\LL+\KK}^*$ and \cref{eq:positivityorder} holds, then
\begin{align}
    \label{eq:lemma2.3} D(\rho\| E_{\LL+\KK}(\rho)) \leq f(\eta)^{-1} D(\rho\| E_{\KK}\circ E_{\LL}(\rho)) \, .
\end{align}
In fact the full-rank state $\sigma$ lies in $\ker\LL\cap\ker\KK$ and hence in $\ker(\LL+\KK)$, so $E_{\LL+\KK}^*$ is a conditional expectation by \cref{thm:CE}. The second condition follows from $\ker(\LL + \KK) = \ker(\LL) \cap \ker(\KK)$, which gives $\operatorname{ran}E_{\LL+\KK}\subseteq\ker\LL=\operatorname{ran}E_\LL$ and $\operatorname{ran}E_{\LL+\KK}\subseteq\ker\KK=\operatorname{ran}E_\KK$, so that $E_{\LL} \circ E_{\LL +\KK} = E_{\LL + \KK}$ and $E_{\KK} \circ E_{\LL +\KK} = E_{\LL + \KK}$ and thus, after dualizing, $E_{\LL+\KK}^*\circ E^*_{\LL}\circ E^*_{\KK} = E_{\LL+\KK}^*$. Then, for $\rho\succ0$ we get
\begin{align}
    f(\eta)D(\rho\| E_{\LL+\KK}(\rho)) &\overset{(1)}{\leq} D(\rho\| E_{\KK}\circ E_{\LL}(\rho)),\notag\\
    &\overset{(2)}{=} D(\rho\| E_{\KK}(\rho))  + \Tr[\rho(\log E_{\KK}(\rho)-\log E_{\KK}\circ E_{\LL}(\rho))],\notag\\
    &\overset{(3)}{=}   D(\rho\| E_{\KK}(\rho))  + \Tr[E_{\KK}(\rho)(\log E_{\KK}(\rho)-\log E_{\KK}\circ E_{\LL}(\rho))],\notag\\
    &\overset{(4)}{=}   D(\rho\| E_{\KK}(\rho))  + D(E_{\KK}(\rho)\| E_{\KK}\circ E_{\LL}(\rho)),\notag\\
    &\overset{(5)}{\leq}   D(\rho\| E_{\KK}(\rho))  + D(\rho\| E_{\LL}(\rho)), \label{eq:two-entropies}
\end{align}
where $(1)$ is \cref{eq:lemma2.3}, $(2)$ adds and subtracts $\Tr[\rho\log E_{\KK}(\rho)]$ in the definition of the relative entropy, $(3)$ replaces $\rho$ by $E_{\KK}(\rho)$, which is allowed because $E_{\KK}(\rho)$ and $E_{\KK}\circ E_{\LL}(\rho)$ are two full-rank states in $\ker\KK$, so that \cref{prop:swap_fixpoint} applies to $\KK$, $(4)$ uses the definition of the relative entropy and $(5)$ the data-processing inequality for the channel $E_{\KK}$. It remains to turn the two MLSIs into an MLSI for $\LL+\KK$. Writing $\alpha_{\min}=\min\{\alpha(\LL),\alpha(\KK)\}$, we get for all $\rho\succ0$
\begin{align}
    2 f(\eta)\,\alpha_{\min} D(\rho \| E_{\LL+\KK}(\rho)) &\overset{(1)}{\leq} 2\alpha_{\min}\bigl(D(\rho\| E_{\LL}(\rho))+D(\rho\| E_{\KK}(\rho))\bigr),\notag\\
    &\overset{(2)}{\leq} \operatorname{EP}_{\LL}(\rho)+\operatorname{EP}_{\KK}(\rho),\notag\\
    &\overset{(3)}{=} \operatorname{EP}_{\LL+\KK}(\rho) ,
\end{align}
where $(1)$ is \cref{eq:two-entropies}, $(2)$ uses the MLSIs of $\LL$ and $\KK$ together with $\alpha_{\min}\leq\alpha(\LL),\alpha(\KK)$, and $(3)$ is \cref{eq:EP-additive}. Since $\rho\succ0$ was arbitrary, this is the MLSI for $\LL+\KK$ with constant $f(\eta)\alpha_{\min}$, i.e.\ $\alpha(\LL+\KK)\geq f(\eta)\min\{\alpha(\LL),\alpha(\KK)\}$, as claimed.
\end{proof}

\partbibliography

\paperpart{appendices-fast_mixing}{1}
\section{Stability of fast mixing in the \texorpdfstring{$\chi^2$}{chi-squared} divergence}\label{app:fast}
We first prove \cref{thm:tkrw}, which is an adaptation of Lemma 12 in \cite{TKRW2010}, since their formulation required $\LL$ to be primitive. This version will allow us to apply it to $\KK$, which need not be primitive: it only needs to satisfy $\KK(\sigma)=0$.

In this appendix, inequalities between self-adjoint superoperators are understood with respect to the Hilbert--Schmidt inner product: $\Phi\preceq\Psi$ means that $\langle A,(\Psi-\Phi)(A)\rangle\geq0$ for every operator $A$.

\begin{proposition}[$\chi^2$ contraction bound; adapted Lemma 12 in \cite{TKRW2010}]\label{thm:tkrw}\label{app:thm:tkrw}
Let $\mathcal L$ be a Lindbladian with a stationary state $\sigma\succ0$, and let $k\in K$. Write $\mathcal L^*$ for its Hilbert--Schmidt adjoint and define
\begin{align}
    \label{app:eq:Lambda} \Lambda_k(\mathcal L)=[\Omega_\sigma^k]^{1/2}\circ\mathcal L\circ[\Omega_\sigma^k]^{-1/2}+[\Omega_\sigma^k]^{-1/2}\circ\mathcal L^*\circ[\Omega_\sigma^k]^{1/2}.
\end{align}
Then $\Lambda_k(\mathcal L)\preceq0$ and $\Lambda_k(\mathcal L)(\sigma^{1/2})=0$. If $l_1^k(\mathcal L)\le0$ denotes its second largest eigenvalue, counting multiplicity, then every state $\rho$ satisfies
\begin{align}
    \label{app:eq:tkrw} \lVert e^{t\mathcal L}(\rho)-\sigma\rVert_1^2\le\chi_k^2(e^{t\mathcal L}(\rho),\sigma)\le e^{l_1^k(\mathcal L)t}\chi_k^2(\rho,\sigma),\qquad t\ge0.
\end{align}
\end{proposition}

\begin{proof}[Proof of \cref{thm:tkrw}]
\prooflabel{thm:tkrw}{proof:tkrw}
Write $\Omega=\Omega_\sigma^k$ and $T_t=e^{t\mathcal L}$. We first show that $\Lambda_k(\mathcal L)\preceq0$, then identify its stationary vector, and finally derive the decay bound.

Theorem 4 of \cite{TKRW2010} states that for every quantum channel $T$ and every matrix $A$,
\begin{align}
    \bigl\langle T(A),\Omega_{T(\sigma)}^k(T(A))\bigr\rangle\le\langle A,\Omega_\sigma^k(A)\rangle.
\end{align}
Since $T_t(\sigma)=\sigma$, applying this to $T=T_t$ gives
\begin{align}
    \langle T_t(A),\Omega(T_t(A))\rangle\le\langle A,\Omega(A)\rangle.
\end{align}
Using the Hilbert--Schmidt adjoint to move $T_t$ from the first argument to the second, we obtain
\begin{align}
    \bigl\langle A,(T_t^*\circ\Omega\circ T_t)(A)\bigr\rangle\le\langle A,\Omega(A)\rangle.
\end{align}
This holds for every $A$, so it is precisely the superoperator inequality $T_t^*\circ\Omega\circ T_t\preceq\Omega$.

At $t=0$ the two sides agree. Subtracting $\Omega$, dividing by $t>0$, and taking $t\downarrow0$ therefore gives
\begin{align}
    \mathcal L^*\circ\Omega+\Omega\circ\mathcal L\preceq0.
\end{align}
Because $\Omega^{-1/2}$ is self-adjoint, conjugation by $\Omega^{-1/2}$ preserves this inequality. Hence
\begin{align}
    \Lambda_k(\mathcal L)&=\Omega^{-1/2}\circ(\mathcal L^*\circ\Omega+\Omega\circ\mathcal L)\circ\Omega^{-1/2},\notag\\
    &=\Omega^{-1/2}\circ\mathcal L^*\circ\Omega^{1/2}+\Omega^{1/2}\circ\mathcal L\circ\Omega^{-1/2}\preceq0.
\end{align}
The two summands are adjoints of one another, so $\Lambda_k(\mathcal L)$ is self-adjoint.

We next identify a vector in its kernel. On matrices commuting with $\sigma$, the definition of $\Omega$ and the normalization $k(1)=1$ give $\Omega(A)=A\sigma^{-1}$. Its positive square root and inverse square root consequently satisfy
\begin{align}
    \Omega^{1/2}(\sigma)&=\sigma^{1/2},& \Omega^{-1/2}(\sigma^{1/2})&=\sigma,& \Omega^{1/2}(\sigma^{1/2})&=I.
\end{align}
Using $\mathcal L(\sigma)=0$ and $\mathcal L^*(I)=0$, we conclude that
\begin{align}
    \Lambda_k(\mathcal L)(\sigma^{1/2})=\Omega^{1/2}(\mathcal L(\sigma))+\Omega^{-1/2}(\mathcal L^*(I))=0.
\end{align}
Thus the largest eigenvalue of $\Lambda_k(\mathcal L)$ is zero. On the orthogonal complement of $\sigma^{1/2}$, its quadratic form is bounded above by its second largest eigenvalue $l_1^k(\mathcal L)$, counting multiplicity.

Now let $\rho(t)=T_t(\rho)$ and define
\begin{align}
    x(t)=\Omega^{1/2}(\rho(t)-\sigma).
\end{align}
By construction,
\begin{align}
    \|x(t)\|_2^2=\chi_k^2(\rho(t),\sigma).
\end{align}
Moreover, self-adjointness of $\Omega^{1/2}$ and preservation of the trace imply
\begin{align}
    \langle\sigma^{1/2},x(t)\rangle&=\langle\Omega^{1/2}(\sigma^{1/2}),\rho(t)-\sigma\rangle,\\
    &=\langle I,\rho(t)-\sigma\rangle=\Tr(\rho(t)-\sigma)=0.
\end{align}
Thus $x(t)$ lies in the orthogonal complement considered above.

Since $\mathcal L(\sigma)=0$, the evolution equation for $x(t)$ is
\begin{align}
    \dot x(t)&=\Omega^{1/2}(\dot\rho(t)),\notag\\
    &\overset{(1)}{=}\Omega^{1/2}(\mathcal L(\rho(t))),\notag\\
    &\overset{(2)}{=}\Omega^{1/2}(\mathcal L(\rho(t)-\sigma)),\notag\\
    &\overset{(3)}{=}\Omega^{1/2}(\mathcal L(\Omega^{-1/2}(x(t)))),\notag\\
    &=(\Omega^{1/2}\circ\mathcal L\circ\Omega^{-1/2})(x(t)),
\end{align}
where $(1)$ uses $\dot\rho(t)=\mathcal L(\rho(t))$, $(2)$ uses linearity and $\mathcal L(\sigma)=0$, and $(3)$ uses $\rho(t)-\sigma=\Omega^{-1/2}(x(t))$. Differentiating its squared norm and using the definition of $\Lambda_k(\mathcal L)$ therefore yields
\begin{align}
    \frac{\mathrm d}{\mathrm dt}\chi_k^2(\rho(t),\sigma)&\overset{(1)}{=}\frac{\mathrm d}{\mathrm dt}\|x(t)\|_2^2,\notag\\
    &\overset{(2)}{=}\langle\dot x(t),x(t)\rangle+\langle x(t),\dot x(t)\rangle,\notag\\
    &\overset{(3)}{=}\langle x(t),\Lambda_k(\mathcal L)(x(t))\rangle,\notag\\
    &\overset{(4)}{\le}l_1^k(\mathcal L)\|x(t)\|_2^2,\notag\\
    &\overset{(5)}{=}l_1^k(\mathcal L)\chi_k^2(\rho(t),\sigma),
\end{align}
where $(1)$ and $(5)$ use $\|x(t)\|_2^2=\chi_k^2(\rho(t),\sigma)$, $(2)$ uses the product rule, $(3)$ uses the evolution equation for $x(t)$, the Hilbert--Schmidt adjoint and the definition of $\Lambda_k(\mathcal L)$, and $(4)$ uses the variational principle on the orthogonal complement of $\sigma^{1/2}$, since $\langle\sigma^{1/2},x(t)\rangle=0$. Integrating this differential inequality gives
\begin{align}
    \chi_k^2(\rho(t),\sigma)\le e^{l_1^k(\mathcal L)t}\chi_k^2(\rho,\sigma).
\end{align}
Finally, Lemma 5 of \cite{TKRW2010} bounds the squared trace-norm distance by the divergence, so
\begin{align}
    \|\rho(t)-\sigma\|_1^2\le\chi_k^2(\rho(t),\sigma)\le e^{l_1^k(\mathcal L)t}\chi_k^2(\rho,\sigma).
\end{align}
\end{proof}

\begin{proposition}[Repeated \cref{prop:tkrw-stability}: Stability of the $\chi^2$ mixing bound]\label{app:prop:tkrw-stability}
Let $\mathcal L$ and $\mathcal K$ be Lindbladians satisfying $\mathcal L(\sigma)=\mathcal K(\sigma)=0$ for the same state $\sigma\succ0$. For every $k\in K$, using this same $\sigma$ to define all three operators, we have
\begin{align*}
    \Lambda_k(\mathcal L+\mathcal K)&=\Lambda_k(\mathcal L)+\Lambda_k(\mathcal K)\preceq\Lambda_k(\mathcal L),\\
    l_1^k(\mathcal L+\mathcal K)&\le l_1^k(\mathcal L)+l_1^k(\mathcal K).
\end{align*}
If $l_1^k(\mathcal L)+l_1^k(\mathcal K)<0$, then $\mathcal L+\mathcal K$ converges exponentially to $\sigma$, which is therefore its unique stationary state.
\end{proposition}

\begin{proof}[Proof of \cref{prop:tkrw-stability}]
\prooflabel{prop:tkrw-stability}{proof:tkrw-stability}
Fix $k\in K$ and write $\Omega=\Omega_\sigma^k$. All three operators $\Lambda_k(\mathcal L)$, $\Lambda_k(\mathcal K)$ and $\Lambda_k(\mathcal L+\mathcal K)$ are defined using this same $\Omega$. Since $\mathcal L(\sigma)=\mathcal K(\sigma)=0$, the state $\sigma$ is also stationary for $\mathcal L+\mathcal K$.

The equation $\Lambda_k(\LL+\KK) = \Lambda_k(\LL)+\Lambda_k(\KK)$ follows by linearity of $\Lambda_k$. Because $\mathcal K$ is a Lindbladian satisfying $\mathcal K(\sigma)=0$, \cref{thm:tkrw} gives $\Lambda_k(\mathcal K)\preceq0$. Thus, for every matrix $A$,
\begin{align}
    \langle A,\Lambda_k(\mathcal L+\mathcal K)(A)\rangle&=\langle A,\Lambda_k(\mathcal L)(A)\rangle+\langle A,\Lambda_k(\mathcal K)(A)\rangle,\notag\\
    &\le\langle A,\Lambda_k(\mathcal L)(A)\rangle.\label{eq:lambda-L-K-quadratic}
\end{align}
Equivalently, $\Lambda_k(\mathcal L+\mathcal K)\preceq\Lambda_k(\mathcal L)$.

We now compare the second largest eigenvalues. By \cref{thm:tkrw}, all three symmetrized operators are self-adjoint, negative semidefinite, and annihilate $\sigma^{1/2}$. Consequently, their second largest eigenvalues, counting multiplicity, are the largest Rayleigh quotients on the same subspace orthogonal to $\sigma^{1/2}$. Therefore,
\begin{align}
    l_1^k(\mathcal L+\mathcal K)&\overset{(1)}{=}\sup_{\substack{A\ne0\\
    \langle\sigma^{1/2},A\rangle=0}}\frac{\langle A,\Lambda_k(\mathcal L+\mathcal K)(A)\rangle}{\|A\|_2^2},\notag\\
    &\overset{(2)}{\le}\sup_{\substack{A\ne0\\
    \langle\sigma^{1/2},A\rangle=0}}\frac{\langle A,\Lambda_k(\mathcal L)(A)\rangle}{\|A\|_2^2}
    +\sup_{\substack{A\ne0\\
    \langle\sigma^{1/2},A\rangle=0}}\frac{\langle A,\Lambda_k(\mathcal K)(A)\rangle}{\|A\|_2^2},\notag\\
    &\overset{(3)}{=}l_1^k(\mathcal L)+l_1^k(\mathcal K),
\end{align}
where $(1)$ and $(3)$ use the variational principle, and $(2)$ uses additivity of $\Lambda_k$ and bounds the supremum of a sum by the sum of the suprema.

Finally, applying \cref{thm:tkrw} to $\mathcal L+\mathcal K$ gives, for every state $\rho$ and every $t\ge0$,
\begin{align}
    \|e^{t(\mathcal L+\mathcal K)}(\rho)-\sigma\|_1^2& \le e^{l_1^k(\mathcal L+\mathcal K)t}\chi_k^2(\rho,\sigma) \le e^{[l_1^k(\mathcal L)+l_1^k(\mathcal K)]t}\chi_k^2(\rho,\sigma).
\end{align}

If $l_1^k(\mathcal L)+l_1^k(\mathcal K)<0$, this bound tends to zero as $t\to\infty$. Since $\sigma\succ0$, the initial divergence is finite for every state $\rho$, so every state converges to $\sigma$. In particular, any stationary state $\tau$ of $\mathcal L+\mathcal K$ satisfies
\begin{align}
    \|\tau-\sigma\|_1^2=\|e^{t(\mathcal L+\mathcal K)}(\tau)-\sigma\|_1^2\le e^{[l_1^k(\mathcal L)+l_1^k(\mathcal K)]t}\chi_k^2(\tau,\sigma)\longrightarrow0,
\end{align}
and hence $\tau=\sigma$.
\end{proof}

\partbibliography

\paperpart{appendices-generic_perturbations}{1}
\section{Stability under generic perturbations}
\label{app:generic-perturbations}

We prove the generic perturbation estimate stated in \cref{lem:perturbation_fixed_point}. The first lemma turns mixing of states into contraction on traceless Hermitian operators. We then apply Duhamel's formula at a fixed time and iterate the resulting contraction. Related Duhamel bounds comparing perturbed and unperturbed trajectories were established in \cite[Theorem~6]{SzehrWolf2013}.

\begin{lemma}[Mixing implies contraction on the traceless subspace]
\label{lem:mixing_traceless}
Let $\mathcal L$ be a Lindbladian with fixed point $\sigma$. Assume that there exist $C,\lambda>0$ such that, for every state $\rho$ and every $t\geq0$,
\begin{align}
    \lVert e^{t\mathcal L}(\rho)-\sigma\rVert_1\le Ce^{-\lambda t}.
\end{align}
Then, for every traceless Hermitian operator $X$,
\begin{align}
    \lVert e^{t\mathcal L}(X)\rVert_1\le Ce^{-\lambda t}\lVert X\rVert_1. \label{eq:mixing-traceless-hermitian}
\end{align}
For an arbitrary traceless operator $X$, the same bound holds with $2C$ in place of $C$.
\end{lemma}

\begin{proof}[Proof of \cref{lem:mixing_traceless}]
\prooflabel{lem:mixing_traceless}{proof:mixing_traceless}
Let $X=X^\dagger$ be traceless and nonzero. Write $X=X_+-X_-$ for its positive and negative parts. Their supports are orthogonal, and tracelessness gives
\begin{align}
    \operatorname{tr}(X_+)=\operatorname{tr}(X_-)=:\alpha=\frac{\lVert X\rVert_1}{2}.
\end{align}
Thus $\rho_\pm:=X_\pm/\alpha$ are states. Applying the mixing assumption to these two states yields
\begin{align}
    \lVert e^{t\mathcal L}(X)\rVert_1&=\alpha\lVert e^{t\mathcal L}(\rho_+)-e^{t\mathcal L}(\rho_-)\rVert_1,\notag\\
    &\le\alpha\bigl(\lVert e^{t\mathcal L}(\rho_+)-\sigma\rVert_1+\lVert e^{t\mathcal L}(\rho_-)-\sigma\rVert_1\bigr),\notag\\
    &\le2\alpha Ce^{-\lambda t},\notag\\
    &=Ce^{-\lambda t}\lVert X\rVert_1.
\end{align}
This proves \cref{eq:mixing-traceless-hermitian}. The case $X=0$ is immediate.

For an arbitrary traceless operator, write
\begin{align}
    X=X_H+iX_A,\qquad X_H:=\frac{X+X^\dagger}{2},\qquad X_A:=\frac{X-X^\dagger}{2i}.
\end{align}
Both parts are traceless and Hermitian, and the triangle inequality gives $\lVert X_H\rVert_1,\lVert X_A\rVert_1\le\lVert X\rVert_1$. Therefore,
\begin{align}
    \lVert e^{t\mathcal L}(X)\rVert_1&\le\lVert e^{t\mathcal L}(X_H)\rVert_1+\lVert e^{t\mathcal L}(X_A)\rVert_1,\notag\\
    &\le Ce^{-\lambda t}\bigl(\lVert X_H\rVert_1+\lVert X_A\rVert_1\bigr),\notag\\
    &\le2Ce^{-\lambda t}\lVert X\rVert_1.
\end{align}
\end{proof}

\begin{lemma}[Repeated \cref{lem:perturbation_fixed_point}: Mixing under a small contractive perturbation]
\label{app:lem:perturbation_fixed_point}
Let $\mathcal L_0$ be a Lindbladian with a unique stationary state $\sigma_0$ and a well-defined infinite-time limit. Set $T:=\tau_{\operatorname{mix}}^{\mathcal L_0}(1/2)$, and let $V$ be a Hermiticity-preserving, trace-annihilating linear map.
Define the restricted induced trace norm on traceless Hermitian operators by
\begin{align*}
    \|V\|_{1\to1,0}:=\sup_{\substack{X=X^\dagger,\ \Tr X=0\\X\neq0}}
    \frac{\|V(X)\|_1}{\|X\|_1}.
\end{align*}
For $\varepsilon\geq0$, set $\mathcal L_\varepsilon:=\mathcal L_0+\varepsilon V$ and suppose that $e^{t\mathcal L_\varepsilon}$ is a trace-norm contraction on Hermitian operators for every $t\geq0$. If
\begin{align}
    \varepsilon \|V\|_{1\to1,0}\leq\frac{1}{4T}, \label{app:eq:generic-perturbation-threshold}
\end{align}
then $\mathcal L_\varepsilon$ has a unique stationary state $\sigma_\varepsilon$ and, for all states $\rho$ and $t\geq0$,
\begin{align}
    \lVert e^{t\mathcal L_\varepsilon}(\rho)-\sigma_\varepsilon\rVert_1
    \leq\frac43e^{-\gamma t}\lVert\rho-\sigma_\varepsilon\rVert_1,
    \qquad \gamma:=\frac{\log(4/3)}{T}. \label{app:eq:generic-perturbation-bound}
\end{align}
Mixing times $T=O(\operatorname{polylog}N)$ imply rapid mixing; $T=O(1)$ gives constant decay rate.
\end{lemma}

\begin{proof}[Proof of \cref{lem:perturbation_fixed_point}]
\prooflabel{lem:perturbation_fixed_point}{proof:perturbation_fixed_point}
Continuity at $T=\tau_{\operatorname{mix}}^{\mathcal L_0}(1/2)$ gives $\sup_\rho\|e^{T\mathcal L_0}(\rho)-\sigma_0\|_1\leq1/2$. For a traceless Hermitian $X$, write $X=\frac12\|X\|_1(\rho_+-\rho_-)$ using its positive and negative parts, as in \cref{proof:mixing_traceless}. The triangle inequality then gives
\begin{align}
    \lVert e^{T\mathcal L_0}(X)\rVert_1\le\frac12\lVert X\rVert_1. \label{eq:generic-unperturbed-contraction}
\end{align}
Duhamel's formula gives
\begin{align}
    \bigl(e^{T\mathcal L_\varepsilon}-e^{T\mathcal L_0}\bigr)(X)
    =\varepsilon\int_0^T e^{(T-s)\mathcal L_0}V e^{s\mathcal L_\varepsilon}(X)\,ds.
\end{align}
Both semigroups are trace-norm contractions on Hermitian operators, and $e^{s\mathcal L_\varepsilon}(X)$ is traceless and Hermitian. Thus the difference has norm at most $\varepsilon T \|V\|_{1\to1,0}\lVert X\rVert_1$, and
\begin{align}
    \lVert e^{T\mathcal L_\varepsilon}(X)\rVert_1
    &\overset{(1)}{\leq}\lVert e^{T\mathcal L_0}(X)\rVert_1
      +\lVert(e^{T\mathcal L_\varepsilon}-e^{T\mathcal L_0})(X)\rVert_1,\notag\\
    &\overset{(2)}{\leq}\left(\frac12+\varepsilon T \|V\|_{1\to1,0}\right)\lVert X\rVert_1
    \overset{(3)}{\leq}\frac34\lVert X\rVert_1.
\end{align}
Here $(1)$ is the triangle inequality, $(2)$ uses \cref{eq:generic-unperturbed-contraction} and the preceding Duhamel bound, and $(3)$ uses \cref{app:eq:generic-perturbation-threshold}.
Writing $t=nT+s$ with $n=\lfloor t/T\rfloor$ and $0\leq s<T$, iteration and contractivity give
\begin{align}
    \lVert e^{t\mathcal L_\varepsilon}(X)\rVert_1
    \leq\left(\frac34\right)^n\lVert X\rVert_1
    \leq\frac43e^{-\gamma t}\lVert X\rVert_1.
\end{align}
The perturbed semigroup maps states to states: its outputs are Hermitian, have trace one and trace norm at most one, hence are positive. Such a finite-dimensional state-preserving semigroup has a stationary state $\sigma_\varepsilon$, which is unique by applying the contraction bound to the difference of two stationary states. Substituting $X=\rho-\sigma_\varepsilon$ proves \cref{eq:generic-perturbation-bound}.
\end{proof}

\partbibliography

\paperpart{appendices-classical_case}{1}
\section{MLSI and rapid mixing for classical Markov chains}
\label{app:classical-case}

There are two different statements that are sometimes summarized by saying that, classically, modified logarithmic Sobolev inequalities and rapid mixing are equivalent. The first is a general semigroup identity: MLSI is equivalent to exponential contraction of relative entropy. The second is a more specialized theorem for finite-range Hamiltonians, their Gibbs states and Glauber dynamics, uniform over finite volumes and boundary conditions.

Let $P_t=e^{tL}$ be an irreducible continuous-time Markov semigroup on a finite set $\Omega$, with invariant measure $\pi$.  For simplicity suppose that it is reversible, and write
\begin{align}
    \operatorname{Ent}_\pi(f) &=\pi[f\log f]-\pi[f]\log\pi[f],\\
    \mathcal E(f,g)&=-\pi[fLg].
\end{align}
If $f_t$ is the density of a law evolved for time $t$ with respect to $\pi$, then
\begin{align}
    \label{eq:classical-entropy-dissipation} \frac{d}{dt}\operatorname{Ent}_\pi(f_t) =-\mathcal E(f_t,\log f_t).
\end{align}
Consequently, the modified logarithmic Sobolev inequality
\begin{align}
    \label{eq:classical-mlsi} \alpha_{\mathrm m}\operatorname{Ent}_\pi(f) \leq \mathcal E(f,\log f),\qquad f>0,
\end{align}
is equivalent to the entropy-contraction estimate
\begin{align}
    \label{eq:classical-entropy-contraction} \operatorname{Ent}_\pi(P_t f) \leq e^{-\alpha_{\mathrm m}t}\operatorname{Ent}_\pi(f) \quad(t\geq0).
\end{align}
Indeed, \cref{eq:classical-mlsi} and \cref{eq:classical-entropy-dissipation} give \cref{eq:classical-entropy-contraction} by Gronwall's lemma, while differentiating \cref{eq:classical-entropy-contraction} at $t=0$ gives the converse.  Thus MLSI is exactly equivalent to exponential contraction of relative entropy with prefactor one \cite{BobkovTetali2006}.

As in the quantum case, the MLSI implies, via Pinsker's inequality, an exponential decay in total variation:
\begin{align}
    \label{eq:classical-mlsi-tv} \lVert P_t(x,\cdot)-\pi\rVert_{\mathrm{TV}} \leq \polylog(\min_x{\pi(x)}) e^{-\alpha_{\mathrm m}t/2}.
\end{align}
For a fixed bounded finite-spin interaction on a volume $\Lambda$, let $\pi_*:=\min_x\pi(x)$. Then $\log(1/\pi_*)=O(|\Lambda|)$, so a uniform positive MLSI constant gives $t_{\mathrm{mix}}(\varepsilon)=O(\log(|\Lambda|/\varepsilon))$ when each site is updated at rate one.

The Gross, or standard, logarithmic Sobolev inequality is instead
\begin{align}
    \label{eq:classical-gross-lsi} \alpha_2\operatorname{Ent}_\pi(g^2) \leq 2\mathcal E(g,g).
\end{align}
For reversible finite chains, the hierarchy is
\begin{align}\label{eq:classical-hierarchy}\text{Gross LSI}\ \Longrightarrow\ \text{MLSI}\ \Longrightarrow\ \text{Poincar\'e inequality (spectral gap)}.\end{align}
For the first implication, apply Gross LSI to $g=\sqrt f$ and use $(a-b)(\log a-\log b)\geq4(\sqrt a-\sqrt b)^2$ to obtain
\begin{align*}
    \alpha_2\operatorname{Ent}_\pi(f)
    \leq 2\mathcal E(\sqrt f,\sqrt f)
    \leq \tfrac12\mathcal E(f,\log f).
\end{align*}
For the second, substitute $f=1+\varepsilon g$ with $\pi(g)=0$ into MLSI and expand to second order in $\varepsilon$. This gives
\begin{align*}
    \tfrac{\alpha_{\mathrm m}}2\operatorname{Var}_\pi(g)\leq\mathcal E(g,g),
\end{align*}
which is a Poincar\'e inequality. Neither reverse implication holds uniformly over general families \cite{DiaconisSaloffCoste1996,BobkovTetali2006}.

An exponential total-variation estimate with a size-independent decay rate implies a uniform spectral gap for reversible chains. Suppose that, for every initial state $x$ and all $t\geq0$,
\begin{align}
    \label{eq:classical-tv-gap}
    \lVert P_t(x,\cdot)-\pi\rVert_{\mathrm{TV}}\leq C_Ne^{-\gamma t},
\end{align}
where $C_N<\infty$ is independent of time and $\gamma>0$ is independent of system size. Let $\lambda=\operatorname{gap}(L)$. Since $L$ is self-adjoint on $L^2(\pi)$, choose an eigenfunction $Lf=-\lambda f$ with $\lVert f\rVert_\infty=1$ and a state $x_*$ with $|f(x_*)|=1$. Stationarity gives $\pi(f)=0$, so
\begin{align*}
    e^{-\lambda t}
    &=|P_tf(x_*)-\pi(f)|\\
    &\leq 2\lVert P_t(x_*,\cdot)-\pi\rVert_{\mathrm{TV}}
    \leq 2C_Ne^{-\gamma t}.
\end{align*}
Letting $t\to\infty$ gives $\lambda\geq\gamma$.

More can be said in the finite-range Glauber setting of \cite{StroockZegarlinski1992a}. Fix a bounded finite-range Hamiltonian on $\mathbb Z^d$ with a finite single-site spin space, and consider its canonical single-site Glauber (heat-bath) dynamics.
Its Gibbs specification is the family of finite-volume Gibbs distributions obtained by fixing the spins outside each volume. Take each site to update at rate one, and require all constants below to be uniform over every finite region and every admissible boundary condition. In this setting, \cite[Theorem~1.8(c)]{StroockZegarlinski1992a}, together with the Dobrushin--Shlosman theory summarized in \cite[Theorem~8.8]{GuionnetZegarlinski2003}, the chain \cref{eq:classical-hierarchy} and the argument above identify the following equivalent properties:
\begin{enumerate}
\item the Dobrushin--Shlosman strong mixing condition;
\item a uniform positive spectral gap for the finite-volume Glauber generators;
\item a uniform positive Gross log-Sobolev constant;
\item a uniform positive MLSI constant;
\item uniform rapid mixing with constant decay rate.
\end{enumerate}

\partbibliography

\addtocontents{toc}{\protect\setcounter{tocdepth}{2}}
\hypersetup{bookmarksdepth=2}

\printbibliography[ heading=bibintoc, title=Bibliography, label=bibliography, ]

\end{document}